\documentclass[12pt]{article} \usepackage{amsthm}
\usepackage{latexsym} \usepackage{eucal} 
\usepackage{amssymb,amsmath,amsfonts,amssymb}
\usepackage{latexsym} \usepackage{eucal} \usepackage{srcltx}
\usepackage{amssymb,amsmath,amsfonts,amssymb}
\usepackage{amssymb,amsmath,amsfonts,amssymb}
\usepackage{amsmath,amsthm,amscd,amssymb}
\usepackage{latexsym}
\usepackage {verbatim}

\usepackage{graphics}
\usepackage{amsmath}
\usepackage{amsfonts,amssymb,amsthm,eucal}
\usepackage{pstricks,multido,graphicx,psfrag}
-\marginparwidth \oddsidemargin -4mm \evensidemargin
\oddsidemargin
\advance\hoffset by 5mm

\def\@abssec#1{\vspace{.05in}\footnotesize \parindent .2in
{\bf #1. }\ignorespaces}

\newtheorem{theorem}{Theorem}[section]
\newtheorem{thm}[theorem]{Theorem}
\newtheorem{lemma}[theorem]{Lemma}
\newtheorem{proposition}[theorem]{Proposition}
\newtheorem{corollary}[theorem]{Corollary}
\newtheorem{definition}[theorem]{Definition}
\newtheorem{remark}[theorem]{Remark}

\def\1{{\bf 1}}

\chardef\set=35

\def\C{\mathbb C}

\def\x{{\vec x}}
\def\y{{\vec y}}

\def\u{{\bf u}}
\def\p{{\vec p}}

\def\j{{\bf j}}

\def\n{{\bf n}}

\def\k{{\vec \varkappa}}

\def\b{{\bf b}}

\def\q{{\bf q}}
\def\l{{\bf l}}
\def\a{{\bf a}}
\def\h{{\vec h}}
\def\s{{\bf s}}
\def\v{{\bf v}}
\def\m{{\bf m}}

\def\R{{\mathbb R}}
\def\Z{{\mathbb Z}}

\def\D{{\bf D}}

\def\DD{{\cal D}}

\def\E{{\cal E}}
\def\RR{{\cal R}}

\def\embb{\supset \kern-11.6pt\lower.335ex\hbox{$\scriptscriptstyle <$} \;}

\def\M{{\cal M}}

\def\W{{\cal W}}
\def\SS{{\cal S}}

\def\B{{\cal B}}

\def\OO{{\cal O}}
\def\MM{{\cal M}}

\def\g{{\bf g}}

\begin{document}
\title{Extended States for the Schr\"odinger Operator with Quasi-periodic Potential in Dimension Two}
\author{Yulia Karpeshina\footnote{Research partially supported by USNSF Grant DMS 1201048.} and Roman Shterenberg}


\maketitle

\begin{abstract} We consider a Schr\"odinger operator $H=-\Delta+V(\x)$ in
dimension two with a quasi-periodic potential $V(\x)$. We prove that
the absolutely continuous spectrum of $H$ contains a semiaxis and
there is a family of generalized eigenfunctions at every point of
this semiaxis with the following properties. First, the
eigenfunctions are close to plane waves $e^{i\langle \k,\x\rangle }$
at the high energy region. Second, the isoenergetic curves in the
space of momenta $\k$ corresponding to these eigenfunctions have a
form of slightly distorted circles with holes (Cantor type
structure). A new method of multiscale analysis in the momentum
space is developed to prove these results.

     The result is based on the previous paper \cite{KaSh}  on quasiperiodic polyharmonic operator $(-\Delta)^l+V(\x)$, $l>1$. We address here technical complications arising in the case $l=1$. However, this text is self-contained and can be read without familiarity with \cite{KaSh} .\end{abstract}
\tableofcontents

\newpage

\section{Introduction}
We study  an operator
\begin{equation}
    H=-\Delta+V(\x) \label{main0}
    \end{equation}
    in dimension two, $V(\x)$ being a quasi-periodic
potential:
\begin{equation}\label{V}
V=\sum\limits_{\s_1,\s_2\in\Z^2,\,\s_1+\alpha\s_2\in{\cal
S}_Q}V_{\s_1,\s_2}e^{2\pi i\langle \s_1+\alpha \s_2,\x\rangle},
\end{equation} where $\alpha$ is irrational number and ${\cal S}_Q$ is a finite
set.  To simplify the construction we put some additional conditions on $\alpha $ and ${\cal S}_Q$, see the beginning of Section \ref{Section 2}.

The one-dimensional situation $d=1$ is thoroughly
   investigated in discrete and continuum settings,
   see e.g. \cite{DiSi}--\cite{FK} and references there.
   It is known that a one-dimensional  quasi-periodic Schr\"{o}dinger operator demonstrates spectral and transport
   properties which
   are \underline{not} close to those of a periodic operator.
    The spectrum of the quasi-periodic
   operator is, as a rule, a Cantor set, while in the periodic case, it has a band structure.
   In the periodic case the spectrum is absolutely continuous, while in the
   quasi-periodic case, it
   can have any nature: absolutely continuous, singular continuous and  pure point.
   The transition between different types  of spectrum can happen even with a small change of
   a coefficient in a quasi-periodic operator (see \cite{J1}, \cite{FK}--\cite{Avila2}).
    The mechanism of the difference in
   spectral behavior between periodic and quasi-periodic cases can be explained by a
   phenomenon
    which is known as resonance tunneling in quantum mechanics. It is associated with small
denominators appearing in formal series of perturbation theory. Since the spectrum of
    the one-dimensional Laplacian is thin  (multiplicity 2),  resonance
    tunneling can produce an effect strong enough to destroy
    the spectrum. If a potential is periodic, then  resonance tunneling  produces
    gaps in the spectrum near the points $\lambda _n=(\pi n/a)^2$, $n\in \Z$, $a$ being the
    period of the potential. If the potential is quasi-periodic, then it
    can be thought as a sort of combination of  infinite number of periodic potentials,
     each
    of them producing gaps near its own $\lambda _n$-s. Since the set of all
    $\lambda _n$-s can be
    dense, the number of points surrounded by gaps can be dense too. Thus, the spectrum gets a
    Cantor like structure. The properties of the operator in the high energy region for the continuum case $d=1$ are studied in
   \cite{DiSi}-\cite{MP}, \cite{E}. The KAM method is used to prove  absolute continuity of the
      spectrum and existence of quasiperiodic solutions at high
      energies.

    There are  important results on the density of states, spectrum,
     localization concerning the  quasi-periodic operators in $\Z^d$ and, partially, in
     $\R^d$, $d>1$, e.g. \cite{Sh1}--\cite{Bou2}.
        However, it is still much less known  about (\ref{main0}) than about its
    one-dimensional analog. The properties of the spectrum in the high energy region, existence of extended states and
    quantum transport are  still the wide open problems
    in the multidimensional case. It is worth noticing here the discrete results \cite{Pos} and  \cite{DamG} which show that for discrete model the spectrum is pure point for a wide class of limit-periodic potentials and thus, no absolutely continuous spectrum is present.

 Here we study properties of the spectrum and
eigenfunctions of (\ref{main0}) in the high energy region. We prove
the following results for the case $d=2$. This is the extension of the analogous results proven by the authors \cite{KaSh} for the quasiperiodic operator $(-\Delta )^l+V$, $l>1$, $d=2$.
    \begin{enumerate}
    \item The spectrum of the operator (\ref{main0})
    contains a semiaxis.

    This is a generalization of  a renown Bethe-Sommerfeld conjecture, which states that in the case of a periodic
    potential and $d\geq 2$, the spectrum of \eqref{main0} contains a semiaxis.
    There is a variety of proofs for the periodic case, the earliest one is \cite{14r}, for the most general case see \cite{PS}. For a limit-periodic periodic potential, being periodic in one direction,  the conjecture is proved in \cite{SS}. For a general case
    of limit-periodic potential the conjecture is proven in \cite{KL1}--\cite{KL3}.


    \item There are generalized eigenfunctions $\Psi_{\infty }(\k, \x )$,
    corresponding to the semi-axis, which are close to plane waves:
    for every $\k $ in an extensive subset $\cal{G} _{\infty }$ of
$\R^2$ (see \eqref{full} below), there is
    a solution $\Psi_{\infty }(\k, \x)$ of the  equation
    $H\Psi _{\infty }=\lambda _{\infty }\Psi _{\infty }$ which can be
described by
    the formula:
    \begin{equation}
    \Psi_{\infty }(\k, \x)
    =e^{i\langle \k, \x \rangle}\left(1+u_{\infty}(\k,
    \x)\right), \label{qplane}
    \end{equation}
    \begin{equation}
    \|u_{\infty}\|_{L_{\infty }(\R^2)}\underset{|\k| \rightarrow
     \infty}{=}O\left(|\k|^{-\gamma _1}\right),\ \ \ \gamma _1>0,
    \label{qplane1}
    \end{equation}
    where $u_{\infty}(\k, \x)$ is a quasi-periodic
    function, namely a point-wise convergent series of exponentials $e^{2\pi i\langle\n+\alpha \m,\x\rangle}$, $\n ,\m \in \Z^2$.
    The  eigenvalue $\lambda _{\infty }(\k)$, corresponding to
    $\Psi_{\infty }(\k, \x)$, is close to $|\k|^{2}$:
    \begin{equation}
    \lambda _{\infty }(\k)\underset{|\k| \rightarrow
     \infty}{=}|\k|^{2}+
    O\left(|\k|^{-\gamma _2}\right),\ \ \ \gamma _2>0. \label{16a}
    \end{equation}
     The ``non-resonant" set $\cal{G} _{\infty }$ of
       vectors $\k$, for which (\ref{qplane}) -- (\ref{16a}) hold, is
       a Cantor type set: ${\cal G} _{\infty }=\cap _{n=1}^{\infty }{\cal G}
_n$,
       where $\{{\cal G} _n\}_{n=1}^{\infty}$ is a decreasing sequence of
sets in $\R^2$. Each ${\cal G} _n$ has a finite number of holes  in each
bounded
       region. More and more holes appear when $n$ increases,
       however
       holes added at each step are of smaller and smaller size.
       The set $\cal{G} _{\infty }$ is extensive in the sense that it satisfies the estimate:
       \begin{equation}\left|\cal{G} _{\infty }\cap
        \bf B_R\right|\underset{R \rightarrow
     \infty}{=}|{\bf B_R}| \bigl(1+O(R^{-\gamma _3})\bigr),\ \ \ \gamma
_3>0,\label{full}
       \end{equation}
       where $\bf B_R$ is the disk of radius $R$ centered at the
       origin, $|\cdot |$ is the Lebesgue measure in $\R^2$.

       \item The set $\cal{D}_{\infty}(\lambda)$,
defined as a level (isoenergetic) set for $\lambda _{\infty }(
\k)$,
$$ {\cal D} _{\infty}(\lambda)=\left\{ \k \in \cal{G} _{\infty }
:\lambda _{\infty }(\k)=\lambda \right\},$$ is proven
to be a slightly distorted circle with infinite number of holes. It
can be described by  the formula:
 \begin{equation} {\cal
D}_{\infty}(\lambda)=\left\{\k:\k=\varkappa_{\infty}(\lambda, \vec
\nu){\vec \nu},
    \ {\vec \nu} \in {\cal B}_{\infty}(\lambda)\right
    \}, \label{Dinfty}
    \end{equation}
where ${\cal B}_{\infty }(\lambda )$ is a subset of the unit circle
$S_1$. The set ${\cal B}_{\infty }(\lambda )$ can be interpreted as
the set of possible  directions of propagation for  almost plane
waves (\ref{qplane}). The set ${\cal B}_{\infty }(\lambda )$ has a
Cantor type structure and an asymptotically full measure on $S_1$ as
$\lambda \to \infty $:
\begin{equation}
L\bigl({\cal B}_{\infty }(\lambda )\bigr)\underset{\lambda
\rightarrow
     \infty}{=}2\pi +O\left(\lambda^{-\gamma _4}\right),\ \ \
     \gamma_4>0,
\label{B}
\end{equation}
here and below $L(\cdot)$ is a length of a curve. The value
$\varkappa_{\infty }(\lambda ,{\vec \nu} )$ in (\ref{Dinfty}) is the
``radius" of ${\cal D}_{\infty}(\lambda)$ in a direction ${\vec \nu}
$. The function $\varkappa_{\infty }(\lambda ,{\vec \nu}
)-\lambda^{1/2}$ describes the deviation of ${\cal
D}_{\infty}(\lambda)$ from the perfect circle of the radius
$\lambda^{1/2}$. It is proven that the deviation is asymptotically
small:
\begin{equation}
\varkappa_{\infty }(\lambda ,{\vec \nu} )\underset{\lambda
\rightarrow
     \infty}{=}\lambda^{1/2}+O\left(\lambda^{-\gamma _5 }\right),
\ \ \ \gamma _5>0. \label{h}
\end{equation}

\item The branch of the spectrum of the operator (\ref{main0}) corresponding to the generalized eigenfunctions $\Psi_{\infty }(\k, \x
)$ is absolutely continuous.

\end{enumerate}


To prove the results listed above we  suggest a method which can be
described as {\em multiscale analysis in the space of momenta}. This
is a development of a method applied in \cite{KaSh} for operator $(-\Delta )^l+V$, $l>1$, $d=2$. The present case $l=1$  has technical complications, comparing with $l>1$. These complications are in the first steps of an approximation procedure. They are related to the fact, that $V$ is a stronger perturbation with respect to
$-\Delta $ than with respect to $(-\Delta )^l$, $l>1$. The method is also related to that developed in
\cite{KL1}--\cite{KL3} for limit-periodic potentials. The essential difference is that
in \cite{KL1}--\cite{KL3} there was constructed a modification of KAM method, where the
space variable  $\x$ still plays some role (e.g. in the uniform  in
$\x$ approximation of a limit-periodic potential by periodic ones),
while in the present situation all considerations are happening in
the space of the dual variable $\k$. The KAM methods in \cite{KL1}--\cite{KL3}
and here are motivated by \cite{2}--\cite{3}, where  modifications of KAM method are used for
periodic problems. Multiscale analysis which we apply here is essentially
analogous to the original multiscale method developed in \cite{FrSp}
(see also \cite{BG}, \cite{74}) for the proof of localization. The
essential difference is that
 in \cite{FrSp}, \cite{BG}, \cite{74} the multiscale procedure is constructed with respect to space variable $\x$ to prove localization, while we construct a multiscale procedure in the space of momenta $\k$ to prove delocalization.

        Here is a brief description of the iteration procedure which leads to the results described above.
        Indeed,  let $\k \in \R^2$. We consider a set of finite linear combinations of plane waves
        $e^{i\langle\k+2\pi(\n+\alpha \m),\x\rangle}$,
        $ \n,\m \in \Z^2$.
         The set is invariant under action of the differential expression (\ref{main0}). Let $H(\k)$ be a
         matrix describing action of  (\ref{main0}) in the linear set of the exponents.
          Obviously,
          $$H(\k)=H_0(\k)+V,\ \  H_0(\k)_{(\n,\m), (\n',\m ')}=|\k+2\pi(\n+\alpha \m)|_{\R^2}^2\delta _{(\n,\n')}
          \delta_{ (\m,\m')},$$ $$V_{(\n,\m), (\n',\m')}=V_{\n-\n', \m-\m'}.$$
          Next, we consider  an expanding sequence of finite sets $\Omega_n$ in the space $\Z^2\times \Z^2$ of indices
          $(\n,\m)$: $\Omega_n\subset \Omega_{n+1}$, $\lim _{n\to \infty } \Omega_n= \Z^2\times \Z^2$. Let $P_n$ be
          the characteristic projection of  set $\Omega_n$ in the space $\ell^{2}(\Z^2\times \Z^2)$. We consider a sequence
          of finite matrices $H^{(n)}(\k)=P_nH(\k)P_n$. Each matrix corresponds to a finite dimensional operator in
          $\ell^{2}(\Z^2\times \Z^2)$, given that the operator acts as zero on $(I-P_n){\ell}^2$.
          For each $n$ we construct a ``non-resonant" set ${\cal G}_n$ in the space $\R^2$ of momenta $\k$, such that:
          if $\k \in {\cal G}_n$, then $H ^{(n)}(\k)=P_nH(\k)P_n$ has an eigenvalue $\lambda _n(\k)$ and
          its spectral projector $\E _n(\k)$ which can be described by perturbation formulas with respect to the previous
          operator $H ^{(n-1)}(\k)$. If $\k \in \cap _{n=1}^{\infty} {\cal G} _n$ then $\lambda _n(\k)$ and $\E _n(\k)$
          have  limits. The linear combinations of the exponentials, corresponding to the projectors $\E _n(\k)$, have a
          point-wise limit in $\x$, the limit being a generalized eigenfunction of \eqref{main0}.
         The generalized eigenfunction is close to the plane wave $e^{i\langle\k,\x\rangle}$ in the high energy region.

          Each matrix $H ^{(n)}$ is considered as a perturbation of a matrix $\hat H ^{(n)}$, the latter has a block
          structure, i.e., consists of a variety of blocks $H^{(s)}(\k +2\pi(\n +\alpha \m))$, $s=1,...,n-1$, and, naturally,
          some diagonal terms. Blocks with different indices $(s)$ have    sizes of different orders of magnitude
          (the size increasing
          with $s$). Thus we have a multiscale structure  in the definition of $\hat H ^{(n)}$.
          We use $\hat H ^{(n)}(\k)$ as a starting operator to construct perturbation series for $H^{(n)}(\k)$.
          At a step $n$ we apply our knowledge of spectral properties of $H^{(s)}(\k +2\pi(\n'+\alpha \m'))$, $s=1,...,n-1$,
          $\n', \m' \in  \Z^2$, obtained in the previous steps, to
        describe spectral properties of  $H^{(n)}(\k +2\pi(\n +\alpha \m))$, $\n,\m\in \Z^2$ and to construct ${\cal G} _n$.

At step one we use a regular perturbation theory and elementary geometric considerations to prove the following results. There is a set ${\cal G}_1\subset \R^2$ such that: if $\k\in {\cal G}_1$, then the operator $H^{(1)}(\k)$ has a single eigenvalue close to the unperturbed one:
\begin{equation} \lambda ^{(1)}(\k)\underset{|\k| \rightarrow
     \infty}{=}|\k|^{2}+
    O\left(|\k|^{-\gamma _2}\right),\ \ \ \gamma _2>0.\label{lambda-1} \end{equation}
     A normalized eigenvector $\u^{(1)}$ is also close to the unperturbed one: $\u^{(1)}=\u^{(0)}+\tilde \u^{(1)}$, where
     $(\u^{(0)})_{(\n,\m)}=\delta _{\n,{\bf 0}}\delta_{ \m,{\bf 0}}$ and the $l^{1}$-norm of $\tilde \u^{(1)}$ is small:
     $\|\tilde \u^{(1)}\|_{l^{1}}<|\k|^{-\gamma _1}$, $\gamma _1>0$. It follows that:
 \begin{equation}
    \Psi_1 (\k, \x)
    =e^{i\langle \k, \x \rangle}+\tilde u_{1}(\k,
    \x),\ \ \
\|\tilde u_{1}\|_{L_{\infty }(\R^2)}\underset{|\k| \rightarrow
     \infty}{=}O(|\k|^{-\gamma _1}),\ \
\ \gamma _1>0, \label{na}
\end{equation}
    where  $\Psi_1 (\k, \x)$, $\tilde u_{1}(\k, \x)$ are the linear combinations of the exponentials corresponding to
    vectors $\u^{(1)}$ and $\tilde \u^{(1)}$, respectively. It is shown that function $\Psi_1 (\k, \x)$
    satisfies the equation for eigenfunctions with a good accuracy:
    \begin{equation} \label{eqforeigenfunctions-1}-\Delta \Psi_1+V\Psi_1=|\k|^{2}\Psi_1+f_1,\ \ \|f_1\|_{L_{\infty }(\R^2)}\underset{|\k| \rightarrow
     \infty}{=}O(|\k|^{-\gamma _6}),\ \
\ \gamma _6>0.
\end{equation}
Relation \eqref{lambda-1} is differentiable:
\begin{equation} \nabla \lambda ^{(1)}(\k)\underset{|\k| \rightarrow
     \infty}{=}2\k+
    O\left( |\k|^{-\gamma _7}\right),\ \ \ \gamma _7>0.\label{dlambda-1} \end{equation}
     Next, we construct a sequence ${\cal G}_n$, $n\geq 2$, such  for any $\k \in {\cal G}_n$ the operator $H^{(n)}(\k)$ has a single eigenvalue $\lambda ^{(n)}(\k)$ in a super exponentially small neighborhood of $\lambda ^{(n-1)}(\k)$:
\begin{equation} \lambda ^{(n)}(\k)\underset{|\k| \rightarrow
     \infty}{=}\lambda ^{(n-1)}(\k)
   + O\left(|\k|^{-|\k|^{\gamma _8n}}\right),\ \ \ \gamma _8>0.\label{lambda-n} \end{equation}
   Similar estimates hold for the eigenvectors and the  corresponding functions $\Psi_n(\k,\x)$:
   \begin{equation}
   \Psi _n(\k,\x)=\Psi _{n-1}(\k,\x)+\tilde u _n(\k,\x),\ \ \ \|\tilde u_{n}\|_{L_{\infty }(\R^2)}\underset{k \rightarrow
     \infty}{=} O\left(|\k|^{-|\k|^{\gamma _9n}}\right),\ \ \ \gamma _9>0.\ \ \label{na-n}
\end{equation}
\begin{equation} \label{eqforeigenfunctions-1*}-\Delta \Psi_n+V\Psi_n=\lambda ^{(n)}(\k)\Psi_n+f_n,\ \ \|f_n\|_{L_{\infty }(\R^2)}\underset{|\k| \rightarrow
     \infty}{=}O\left(|\k|^{-|\k|^{\gamma _{10}n}}\right),\ \ \
\ \gamma _{10}>0.
\end{equation}
Formula \eqref{lambda-n} is differentiable with respect to $\k$:
\begin{equation} \nabla \lambda ^{(n)}(\k)\underset{|\k| \rightarrow
     \infty}{=}\nabla \lambda ^{(n-1)}(\k)+O\left(|\k|^{-|\k|^{\gamma _8n}}\right),\ \ \
\ \gamma _8>0.\label{dlambda-n} \end{equation}
In fact, for large $n$ estimates \eqref{lambda-n} -- \eqref{dlambda-n} are even stronger.
The non-resonant set ${\cal G} _{n}$
        is proven to be
       extensive in $\R^2$:
       \begin{equation}
       \left|{\cal G} _{n}\cap
       \bf B_R\right|\underset{R \rightarrow
     \infty}{=}|{\bf B_R}|\bigl(1+O(R^{-\gamma _3})\bigr). \label{16b}
       \end{equation}
       Estimates (\ref{lambda-n}) -- (\ref{16b}) are uniform in $n$.

The set ${\cal D}_{n}(\lambda)$ is defined as the level
(isoenergetic) set for the non-resonant eigenvalue $\lambda ^{(n)}(\k)$:
$$ {\cal D} _{n}(\lambda)=\left\{ \k \in {\cal G} _n:\lambda ^{(n)}(\k)=\lambda \right\}.$$
This set is proven to be a slightly distorted circle with a
finite number of holes. The set
${\cal D} _{n}(\lambda)$ can be described by the formula:
\begin{equation}
{\cal D}_{n}(\lambda)=\left\{\k:\k=
    \varkappa^{(n)}(\lambda, {\vec \nu})\vec \nu ,
    \ \vec \nu  \in {\cal B}_{n}(\lambda)\right\}, \label{Dn}
    \end{equation}
where ${\cal B}_{n}(\lambda )$ is a subset  of the unit circle
$S_1$. The set ${\cal B}_{n}(\lambda )$ can be interpreted as the
set of possible directions of propagation for  almost plane waves $\Psi _n(\k,\x)$, see
(\ref{na}), (\ref{na-n}). It has an asymptotically full measure on $S_1$ as
$\lambda \to \infty $:
\begin{equation}
L\bigl({\cal B}_{n}(\lambda )\bigr)\underset{\lambda \to \infty
}{=}2\pi +O\left(\lambda^{-\gamma _4}\right). \label{Bn}
\end{equation}
Each set ${\cal B}_{n}(\lambda)$ has only a finite number of holes,
however their number is growing with $n$. More and more holes of a
smaller and smaller size are added at each step. The value
$\varkappa^{(n)}(\lambda ,{\vec \nu} )-\lambda^{1/2}$ gives the
deviation of ${\cal D}_{n}(\lambda)$ from the perfect circle of the
radius $\lambda^{1/2}$  in the direction ${\vec \nu} $. It is
proven that the deviation is asymptotically small:
\begin{equation}
\varkappa^{(n)}(\lambda ,{\vec \nu})
=\lambda^{1/2}+O\left(\lambda^{- \gamma _5}\right),\ \ \ \
\frac{\partial \varkappa^{(n)}(\lambda ,{\vec \nu})}{\partial
\varphi }=O\left(\lambda^{- \gamma _{11} }\right),\ \
\gamma_5,\gamma _{11}>0,\label{hn}
\end{equation}
$\varphi $ being an angle variable, ${\vec \nu} =(\cos \varphi ,\sin
\varphi )$.  Estimates (\ref{Bn}), (\ref{hn}) are uniform in $n$.

On each step more and more points are excluded from the
non-resonant sets ${\cal G} _n$, thus $\{ {\cal G} _n \}_{n=1}^{\infty }$ is a
decreasing sequence of sets. The set ${\cal G} _\infty $ is defined as the
limit set: ${\cal G} _\infty=\cap _{n=1}^{\infty }{\cal G} _n $. It
has an infinite number of holes, but nevertheless satisfies the
relation (\ref{full}). For every $
\k \in {\cal G} _\infty $ and every $n$, there is a generalized
eigenfunction of $H^{(n)}$ of the type  (\ref{na}), (\ref{na-n}). It is
proven that  the sequence of
$\Psi _n(\k, \x)$ has a limit in $L_{\infty }(\R^2)$ when $
\k \in {\cal G} _\infty $.
The function $\Psi _{\infty }(\k, \x)
=\lim _{n\to \infty }\Psi _n(\k, \x)$ is a generalized
eigenfunction of $H$. It can be written in the form
(\ref{qplane}) -- (\ref{qplane1}).
Naturally, the corresponding eigenvalue $\lambda _{\infty }(\k) $ is
the limit of $\lambda ^{(n)}(\k )$ as $n \to \infty $.

It is shown that $\{{\cal B}_n(\lambda)\}_{n=1}^{\infty }$  is a
decreasing sequence of sets,  on each step more and more directions
being excluded. We consider the limit ${\cal B}_{\infty}(\lambda)$
of ${\cal B}_n(\lambda)$:
    \begin{equation}{\cal B}_{\infty}(\lambda)=\bigcap_{n=1}^{\infty} {\cal
    B}_n(\lambda).\label{Dec8a}
    \end{equation}
    This set has a Cantor type structure on the unit circle.
    It is proven that ${\cal B}_{\infty}(\lambda)$ has an asymptotically
    full measure on the unit circle (see (\ref{B})).
    We prove
    that the sequence $\varkappa^{(n)}(\lambda ,{\vec \nu} )$, $n=1,2,... $,
describing the
     isoenergetic curves ${\cal D}_n(\lambda)$, quickly converges as $n\to
\infty$. We show that ${\cal D}_{\infty}(\lambda)$ can be described
as the limit of  ${\cal D}_n(\lambda)$ in the sense (\ref{Dinfty}),
where $\varkappa_{\infty}(\lambda, \vec \nu )=\lim _{n \to \infty}
\varkappa^{(n)}(\lambda, \vec \nu )$ for every $\vec \nu  \in {\cal
B}_{\infty}(\lambda)$. It is shown that the derivatives of the
functions $\varkappa^{(n)}(\lambda, \vec \nu )$ (with respect to the
angle variable on the unit circle) have a limit as $n\to \infty $
for every $\vec \nu  \in {\cal B}_{\infty}(\lambda)$. We denote this
limit by $\frac{\partial \varkappa_{\infty}(\lambda ,\vec
\nu)}{\partial \varphi }$. Using (\ref{hn}), we  prove that
    \begin{equation}\frac{\partial \varkappa_{\infty}(\lambda ,\vec \nu)}{\partial
\varphi }=O\left(\lambda^{- \gamma _{11} }\right).\label{Dec9a}
\end{equation} Thus, the limit curve ${\cal D}_{\infty}(\lambda)$ has a
tangent vector in spite of its Cantor type structure, the tangent
vector being the limit of corresponding tangent vectors for ${\cal
D}_n(\lambda)$ as $n\to \infty $.  The curve  ${\cal
D}_{\infty}(\lambda)$ looks as
  a slightly distorted circle with
infinite number of holes for every sufficiently large $\lambda $, $\lambda >\lambda _*(V)$. It immediately follows that $[\lambda _*, \infty )$ is in the spectrum of $H$ (Bethe-Sommerfeld conjecture).

The main technical difficulty to overcome is the construction of
   non-resonant sets
${\cal B} _n(\lambda)$ for every fixed sufficiently large $\lambda $,
$\lambda >
\lambda_0 (V)$,
where $\lambda_0(V)$ is the same for all $n$. The set
${\cal B} _n(\lambda)$ is obtained  by deleting a ``resonant" part
from
${\cal B}_{n-1}(\lambda)$.
Definition of  ${\cal B} _{n-1}(\lambda)\setminus {\cal B}_{n}(\lambda)$
includes
eigenvalues of $H^{(n-1)}(\k)$. To describe  ${\cal B} _{n-1}(\lambda)\setminus
{\cal B}_{n}(\lambda)$ one
has to consider
  not only non-resonant
eigenvalues of the type (\ref{lambda-1}),  (\ref{lambda-n}), but also
resonant
eigenvalues, for which no suitable  formulas are known. Absence of
formulas causes difficulties in estimating the size of ${\cal B}
_{n-1}(\lambda)\setminus {\cal B}_n(\lambda)$.
       To treat this problem we start with introducing
        an angle variable $\varphi \in [0,2\pi
       )$,  $\vec \nu  = (\cos \varphi ,
       \sin \varphi )\in S_1$ and consider sets ${\cal B}_n(\lambda)$ in
terms of this
       variable.
        Next, we show that
the resonant set  ${\cal B} _{n-1}(\lambda)\setminus
{\cal B}_{n}(\lambda)$ can be described as the set of zeros of  functions of the type
        \begin{equation} \label{May5-14}\det \Bigl(H^{(s)}\bigl(\k _{n-1}(\varphi )+2\pi(\n+\alpha
        \m)
       \bigr)-\lambda-\varepsilon \Bigr), \ \ \ s=1,...,n-1,\ \ \ (\n,\m)\in \Omega_n\setminus{(\bf 0,\bf 0)},\end{equation}
       where $\k _{n-1}(\varphi )$ is a vector-function
describing ${\cal D}_{n-1}
       (\lambda)$: $\k _{n-1}(\varphi )=\varkappa_{n-1}(\lambda
,{\vec \nu} ){\vec \nu} $. To obtain ${\cal B} _{n-1}(\lambda)\setminus
       {\cal B}_{n}(\lambda)$ we take all values of $\varepsilon $ in a small
interval
       and $(\n,\m)$ in some subset of $\Omega_n$.
        Further, we extend our
       considerations to
        a complex neighborhood $\varPhi _0$ of $[0,2\pi
       )$. We show that the    determinants are analytic functions of
        $\varphi $ and, by this,
         reduce the
        problem of estimating the size of the resonant set
         to
        a problem in complex analysis. We use theorems for analytic
functions to count
          zeros of the determinants and to investigate how far  the zeros move
when
         $\varepsilon $ changes. It enables us to estimate the size
         of the zero set of the determinants, and, hence, the size of
         the non-resonant set $\varPhi _n\subset \varPhi _0$,  which is
defined as a
         non-zero set for the determinants.
          Proving that the non-resonant set $\varPhi _n$
         is sufficiently large, we
          obtain estimates (\ref{16b}) for ${\cal G} _n$ and (\ref{Bn}) for
         ${\cal B}_n$, the set  ${\cal B}_n$ corresponding to the real part of
         $\varPhi _n$.  The essential part of constructing the nonresonant set is estimating the number of $(\n,\m)$ in $\Omega_n$ for which \eqref{May5-14} holds for a fixed
         real $\varphi $ and a sufficiently small values of $\varepsilon $. To obtain such an estimate we investigate geometric properties of the quasiperiodic lattice
         $\{ \n +\alpha \m,\ (\n,\m)\in \Omega_n\}$. Namely, we use algebraic and geometric considerations to obtain an estimate the number of lattice points inside a thin semi-agberaic set corresponding to zeros of
         \eqref{May5-14} when $\varepsilon $ in a small neighborhood of zero.

          To obtain $\varPhi _n$ we delete  from $\varPhi _0$ more and more
discs (holes) of smaller and
         smaller radii at each step. Thus, the non-resonant set $\varPhi
         _n\subset \varPhi _0$ has a structure of Swiss
         Cheese. Deleting  a resonance set from $\varPhi _0$ at each
         step of the recurrent procedure  we call a ``Swiss
         Cheese Method".  The essential
         difference of our method from constructions of non-resonant sets  in similar
         situations before (see e.g. \cite{2}--\cite{3}, \cite{B3}) is that
we construct a
         non-resonant set not only in the whole space of a parameter
         ($\k\in \R^2$ here), but also on isoenergetic curves
         ${\cal D}_n(\lambda )$ in
         the space of the parameter, when $\lambda $ is sufficiently large.
Estimates for the
         size of non-resonant sets on a curve require more subtle
         technical considerations than those sufficient for
         description of a non-resonant set in the whole space of
         the parameter. But as a reward, such estimates enable us to
         show that every isoenergetic set for $\lambda>\lambda_0$ is not empty and
         thus, to prove Bethe-Sommerfeld conjecture.

         The plan of the paper is the following. Preliminary considerations are in Section 2. Sections 3 -- 7
         describe steps of the recurrent procedure.  Steps I,II are designed to start the procedure. Step I is quite simple and completely analogous to that in the case of 
         the polyharmonic operator $(-\Delta )^l +V$, $l>1$, \cite{KaSh}.  However, already the preparation for Step II (Section 3.5) is essentially more complicated in the present case of Schr\"odinger operator ($l=1$) than in the case of the polyharmonic operator. It is caused by the fact, that $V(\x)$ is a relatively stronger perturbation of $(-\Delta )^l$  at high energies in the case $l=1$ than in the case $l>1$. Mathematically, this means that the small denominator problem is more intricate for $l=1$. Even in the case of a periodic $V(\x)$, Bloch eigenvalues of $(-\Delta )^l$  are located much denser at high energies in the case $l=1$ than in the case  $l>1$. This means that description of their perturbation by a potential is a more challenging problem when $l=1$. In the case of a quasi-periodic potentials these difficulty, naturally,  persists.  The structure of the model operator for Step II is essentially more complicated here than that in \cite{KaSh}.  In Step II we have to use the second condition on the potential (see the next section), which we did not impose in the case $l>1$.  To construct the model block operator in Step II, we use one-dimensional periodic operators defined by ``directional components" of $V(\x)$ (by the second condition on the potential, they are periodic).     Naturally, the proof of convergence of the perturbation series in Step II (Section 4.1) is more elaborated for $l=1$. 
         Construction of the model operator in Step III requires understanding properties of a quasiperiodic lattice  $\n +\alpha \m$ in a resonant set of $\k$. Resonant sets have more complicate structure here, than those in the case $l>1$. Therefore, the geometric part in the preparation for Step III  (Section 4.3) is also trickier.
          Step III is already typical for the recurent procedure,
         however still uses some ``non-typical" estimates from
         Steps I,II. This step is essentially the same for both $l>1$ and $l=1$ cases, except just a few places where results from Steps I,II are used. Step IV is completely typical: all other steps of the recurrent procedure differ from Step IV only by the change of indices. They are completely analogous for the cases $l=1$ and $l>1$. The proofs of
         convergence of the iteration procedure and of the results 1 -- 3, listed at the beginning of the introduction, are  in Section
         8. They are completely analogous to those in the polyharmonic case \cite{KaSh}.The result 4 on absolutely continuity of the  spectrum is
         proven in Section 9. The proof is also completely similar to that in \cite{KaSh}. Section 10 (Appendices) contains technical lemmas. For the sake of convenience we provide the list of the main notations at the end of the text (Section 11). In this text, for completeness of consideration, we provide all the proofs, even those which are analogous to the case $l>1$.

\vspace{5mm} \noindent {\bf Acknowledgements.} The authors are very
grateful to Prof. Parnovski for useful
discussions.

\section{Preliminary Remarks \label{Section 2}}

We consider two-dimensional quasi-periodic Schr\"odinger operator \eqref{main0},
which is perturbation of the free operator
$H_0:=-\Delta$. Here the potential $V$ has the form \eqref{V},
where  $\alpha$ is an irrational number and ${\cal S}_Q={\cal S}_Q(\alpha)$ is finite
set. To simplify the construction
we will assume that the set ${\cal S}_Q$ and $\alpha$ are not
degenerate in some sense. More precisely, we impose the following
conditions:

1) If $\s_1,\s_2\in\Z^2$ are such that $\s_1+\alpha\s_2\in{\cal
S}_Q$ then $|\s_1|+|\s_2|\leq Q$.

2) If $\s_1,\s_2,\s_1',\s_2'\in\Z^2$ are such that
$\s_1+\alpha\s_2,\s_1'+\alpha\s_2'\in{\cal S}_Q$ and
$\s_1+\alpha\s_2=c_*(\s_1'+\alpha\s_2')$ then $c_*$ is rational. This means that if there are several vectors in ${\cal S}_Q$
with the same direction then they form a subset of a periodic
one-dimensional lattice. Without loss of generality, we will assume that ${\cal
S}_Q$ contains generating vectors $\s_1+\alpha\s_2$ of all present
directions as well as all their integer multipliers
$n(\s_1+\alpha\s_2)$, $n\in\Z$, such that (see condition 1))
$|n|(|\s_1|+|\s_2|)\leq Q$. In particular, the set ${\cal S}_Q$ is
symmetric with respect to ${\bf 0}$ and ${\bf 0}\in{\cal S}_Q$. We
will assume though that $V_{{\bf 0},{\bf 0}}=0$. As shown below in Lemma~\ref{psnorms} the period of every one-dimensional sublattice in ${\cal S}_Q$ is not smaller than $CQ^{-\mu}$.

3) $\alpha:\ 0<\alpha<1$ is irrational and irrationality measure $\mu$ of $\alpha$ is
finite: $\mu<\infty$ (in other words, this means that $\alpha$ is
not a Liouville number). Note also, that $\mu\geq 2$ for any
irrational number $\alpha $.


4) There are $N_0,N_1>0$
 such that if $|n_1|+|n_2|+|n_3|>N_1$ then
\begin{equation}
n_1+\alpha n_2+\alpha^2 n_3=0\ \ \  \hbox{or}\ \ \ |n_1+\alpha
n_2+\alpha^2 n_3|> (|n_1|+|n_2|+|n_3|)^{-N_0}. \label{condition}
\end{equation}

Note that the fourth condition is automatically satisfied for a quadratic irrational $\alpha $:\begin{equation}
n'_1+\alpha
n'_2+\alpha^2 n'_3=0 \label{quadirr}
\end{equation} for a triple $(n'_1, n'_2,n'_3)$. Such triple is unique up to trivial multiplication (if exists), otherwise $\alpha$ is rational. If $n_1+\alpha
n_2+\alpha^2 n_3\not=0$ for some other triple $(n_1, n_2,n_3)$, then \eqref{condition} holds automatically, since otherwise $\mu =\infty $.
The fourth condition is needed to estimate from below the angle between two non-colinear vectors $\s_1+\alpha\s_2$ and $\s_1'+\alpha\s_2'$ by
a negative power of $|\s_1+\alpha\s_2| +|\s_1'+\alpha\s_2'|$.

It follows from the definition of the irrationality measure that 1)
{\it For any $\epsilon>0$ there exists a constant $C_{\varepsilon} $
such that for any irreducible rational number
$\frac{\tilde{M}}{\tilde {N}}$ we have}
\begin{equation}\label{geq} \left|\alpha-\frac{\tilde{M}}{\tilde
{N}}\right|\geq\frac{C_{\varepsilon}}{\tilde {N}^{\mu+\epsilon}}.
\end{equation}

2) {\it For any $\epsilon>0$ there exists a sequence $\frac{M}{N}$
of irreducible rational numbers such that}
\begin{equation}\label{leq}
\left|\alpha-\frac{M}{N}\right|\leq\frac{1}{N^{\mu-\epsilon}}.
\end{equation}
For simplicity we will often take $\epsilon=1$.

For every pair of integer vectors $\s_1,\,\s_2\in\Z^2$ we consider
$\p_\s:=2\pi(\s_1+\alpha\s_2)$. We introduce the norm
$$ |\|\p_{\s}\||:=|\s_1|+|\s_2|.
$$
We will also use the notation $p_{\s}:=|\p_{\s}|$ and
$\p_{\s}=p_{\s}(\cos\varphi_{\s},\sin\varphi_{\s})$.
\begin{lemma}\label{psnorms}
For every $\p_\s\not=0$ we have
\begin{equation}\label{above}
p_\s\leq2\pi|\|\p_{\s}\||,
\end{equation}
\begin{equation}\label{below} p_{\s}\geq
2\pi C_\varepsilon\,\||\p_\s\||^{-(\mu-1+\epsilon)}. \end{equation}
\end{lemma}
\begin{proof} The estimate \eqref{above} is obvious. To prove \eqref{below} we notice that
if $\s_2=0$ then $p_{\s}=2\pi|\s_1|\geq2\pi$. Let
$\s_1=(s_{11},s_{12})$, $\s_2=(s_{21},s_{22})$. If, for example,
$s_{21}\not=0$ then from \eqref{geq} and definition of
$|\|\p_{\s}\||$ we obtain
\begin{equation} \label{12}
\begin{split}&
p_{\s}\geq2\pi\left|s_{11}+\alpha s_{21}\right|=
2\pi|s_{21}||\alpha+\frac{s_{11}}{s_{21}}|\geq\cr & 2\pi
C_\varepsilon|s_{21}|^{-\mu-\epsilon+1}\geq 2\pi
C_\varepsilon\,\||\p_\s\||^{-(\mu-1+\epsilon)}.
\end{split}
\end{equation}
\end{proof}

We introduce vector $\k(\varphi
):=(\varkappa_1,\varkappa_2)=\varkappa\vec
\nu:=\varkappa(\cos\varphi,\sin\varphi)$. Similar agreement will be
used for other vectors. Let $H(\k)=H(\varkappa,\varphi)$ be the
"fiber" operator acting in $l^2(\Z^4)$ with its matrix elements
given by $$
(H(\k))_{\s,\s+\q}=|\k+\p_{\s}|^{2}\delta_{\s,\s+\q}+V_{\p_\q}. $$
Here $V_{\p_\q}:=V_{\q_1,\q_2}$ and (see \eqref{V})
\begin{equation}\label{V_q=0} V_{\p_\q}=0,\ \ \mbox{when  }
\||\p_\q\||>Q,\ \ \ (Q<\infty) .\end{equation} To simplify the
notation in what follows we will write $V_\q$ instead of $V_{\p_\q}$
when it does not lead to confusion.

By $C,c$ we denote constants depending only on $V$, by $C_0,c_0$ we denote absolute constants.

\section{Step I}

\subsection{Operator $H^{(1)}$}
Let $\delta$ be some small parameter, $0<\delta <(10^5\mu )^{-1}$.
We will also assume that $2\delta N_0\leq 1/2$, $N_0$ being defined by \eqref{condition}.
We put
\begin{equation} \Omega(\delta):=\{\m\in\Z^4:\ \ |\|\p_{\m}|\|\leq k^\delta\},\ \ \ \tilde\Omega(\delta):=\{\m\in\Z^4:\ \ |\|\p_{\m}|\|\leq 4k^\delta\}. \label{May19-14} \end{equation}

By $P(\delta)$ we denote an orthogonal (diagonal) projection in
$l^2(\Z^4)$ on the set of elements supported in $\Omega(\delta)$. We
call it the characteristic projector of $ \Omega(\delta)$. The
dimension of the projector is equal to the number of elements in $
\Omega(\delta)$ and, obviously, does not exceed $(8k^\delta )^4$. We
have $$
\big(P(\delta)H_0(\k)P(\delta)\big)_{\m,\n}=|\k+\p_{\m}|^{2}\delta_{\m,\n}\,\chi_{\Omega(\delta)}(\m),
$$ where as usual $\chi_{\Omega(\delta)}(\m)$ is the characteristic
function of the set $\Omega(\delta)$. We are going to consider
 $H^{(1)}(\k)=P(\delta)H(\k)P(\delta)$ as a perturbation of the operator
$P(\delta)H_0(\k)P(\delta)$.

\subsection{Perturbation Formulas}

Now we construct a ``non-resonant" set of $\varphi $, for which the operator $H^{(1)}(\k (\varphi ))$ can be constructively considered as a perturbation of $H^{(1)}_0(\k (\varphi ))$ corresponding to $V=0$.
In what follows $\tau $ is an auxiliary
parameter $\frac{1}{32}\leq \tau \leq 32$.

\begin{lemma}[Geometric] \label{L:G1} For every $k>800$, there is a subset $\omega^{(1)} (k,\delta
,\tau )$ of the interval $[0,2\pi )$ such that: \begin{enumerate}
\item For every $\varphi \in \omega^{(1)} (k,\delta ,\tau )$ and $\m\in
\tilde\Omega (\delta )\setminus \{0\}$, the following inequality
holds:
\begin{equation} \left||\vec k(\varphi )+\p_{\m}|^{2}-k^{2}\right|>\tau k^{1-40\mu
\delta },\ \ \ \vec k:=k(\cos \varphi, \sin \varphi). \label{G1-1}\end{equation} \item For every $\varphi $ in
the real $\frac{\tau }{16}k^{-(40\mu +1)\delta }$-neighborhood of
$\omega^{(1)} (k,\delta ,\tau )$ and $\varkappa\in \R:
|\varkappa-k|<\frac{\tau }{16}k^{-40\mu \delta }$, a slightly weaker
inequality holds for $\k(\varphi )=\varkappa(\cos \varphi ,\sin
\varphi )$ and $\m \in \tilde\Omega (\delta )\setminus \{0\}$:
\begin{equation}
\left||\k(\varphi )+\p_{\m}|^{2}-{k}^{2}\right|>\frac{\tau
}{2}k^{1-40\mu \delta }. \label{G1-2}\end{equation}
\item The set $\omega^{(1)} (k,\delta ,\tau )$ has an asymptotically full
measure in $[0,2\pi )$ as $k\to \infty $. Namely,
\begin{equation}
|\omega^{(1)} (k,\delta ,\tau )|=2\pi +O(k^{-37\mu \delta}),\ \
{k\to \infty }. \label{G1-3}
\end{equation}\end{enumerate}
\end{lemma}
\begin{corollary} \label{C:L:G1} If $\varphi $ is in
the real $\frac{\tau }{16}k^{-(40\mu +1)\delta }$-neighborhood of $\omega^{(1)}
(k,\delta ,\tau )$  and $z$ is on the circle
\begin{equation} \label{G-4}
C_1=\{z:|z-k^{2}|=\frac{\tau}{4}k^{1-40\mu \delta }\},
\end{equation}
then  the following inequality holds for all $\m \in \tilde\Omega
(\delta )$:
\begin{equation}
\left||\vec k(\varphi) +\p_{\m}|^{2}-z\right|\geq \frac{\tau
}{4}k^{1-40\mu \delta }, \ \ z\in C_1. \label{G1-5}\end{equation}
\end{corollary}

The lemma is proved in Section \ref{GC} (Corollaries \ref{Parts 1,2}
and \ref{Part 3}.) The corollary from the lemma is proven at the end
of Section \ref{GC}. Note that in Section \ref{GC} we construct
non-resonance set of $\varphi $ in the set of complex numbers. Such
complex non-resonance set we need for construction of further steps
of approximation.

Let $r=1,2...$ and
\begin{equation}\label{g} g^{(1)}_r({\k}):=\frac{(-1)^r}{2\pi
ir}\hbox{Tr}\oint_{C_1}\left((P(\delta)(H_0({\k})-zI)P(\delta))^{-1}VP(\delta)\right)^rdz,
\end{equation} \begin{equation}\label{G}
G^{(1)}_r({\k}):=\frac{(-1)^{r+1}}{2\pi
i}\oint_{C_1}\left((P(\delta)(H_0({\k})-zI)P(\delta))^{-1}VP(\delta)\right)^r(P(\delta)(H_0({\k})-zI)P(\delta))^{-1}dz.
\end{equation}
Note that $g^{(1)}_1(\k)=0$ since $V=0$. Coefficient
$g^{(1)}_2(\k )$ admits representation:
\begin{equation}\begin{aligned}
g^{(1)}_2(\k)
    &=\sum _{\q\in \Omega (\delta )\setminus \{0\}}| V_\q| ^2(|\k|^{2}-
        |{\k}+\p_\q|^{2})^{-1} \\
    &=-\sum _{\q\in \Omega (\delta )\setminus \{0\}}\frac{| V_\q| ^2
       |\p_\q|^{2}}{(|{\k}|^{2}-
        |{\k}+\p_\q|^{2})(|{\k}|^{2}-
        |{\k}-\p_\q|^{2})},\label{2.14}
 \end{aligned}\end{equation}
 From now on  $\|A \|_1$ means the norm of an operator $A$ in the trace class.

\begin{theorem} \label{Thm1} Suppose $\varphi $ is in
the real  $\frac{\tau }{16}k^{-(40\mu +1)\delta }$-neighborhood of
$\omega^{(1)} (k,\delta,\tau )$ and   $\varkappa\in\R$,
$|\varkappa-k|\leq \frac{\tau}{16}k^{-40\mu \delta }$,
$\k=\varkappa(\cos \varphi ,\sin \varphi )$. Then, for sufficiently
large $k>k_0(V,\delta ,\tau )$ there exists a single eigenvalue of
$H^{(1)}({\k})$ in the interval $\varepsilon _1( k,\delta,\tau )=(
k^{2}-\frac{\tau}{2}k^{1-40\mu \delta }, k^{2}+\frac{\tau
}{2}k^{1-40\mu \delta })$. It is given by the absolutely converging
series:
\begin{equation}\label{eigenvalue}\lambda^{(1)}({\k})=\varkappa^{2}+\sum\limits_{r=2}^\infty
g^{(1)}_r({\k}).\end{equation} For coefficients $g^{(1)}_r({\k})$
the following estimates hold:
\begin{equation}\label{estg} |g^{(1)}_r({\k})|\leq
(Ck)^{-(r-1)(1-40\mu\delta)+4\delta}. \end{equation}
 Moreover,
 \begin{equation}\label{estg_2}
|g^{(1)}_2({\k})|\leq Ck^{-2+(80\mu+6)\delta}. \end{equation} The
corresponding spectral projection is given by the series:
\begin{equation}\label{sprojector}
\E^{(1)}({\k})=\E_0({\k})+\sum\limits_{r=1}^\infty G^{(1)}_r({\k}),
\end{equation} $\E_0({\k})$ being the unperturbed spectral
projection. The operators $G^{(1)}_r({\k})$ satisfy the estimates:
\begin{equation}
\label{jan27}
\left\|G^{(1)}_r({\k})\right\|_1<(Ck)^{-r(1-44\mu\delta)}.
\end{equation}
Matrix elements of $G^{(1)}_r({\k})$ satisfy the following
relations:
\begin{equation}
G^{(1)}_r({\k})_{\s\s'}=0,\ \ \mbox{if}\ \ \
rQ<\||\p_\s\||+\||\p_{\s'}\||. \label{zeros} \end{equation}
\end{theorem}
\begin{corollary} \label{corthm1} For the perturbed eigenvalue and its spectral
projection the following estimates hold:
 \begin{equation}\label{perturbation}
\lambda^{(1)}({\k})=\varkappa^{2}+O\left(k^{-2+(80\mu+6)\delta}\right),
\end{equation}
\begin{equation}\label{perturbation*}
\left\|\E^{(1)}({\k})-\E_0({\k})\right\|_1<ck^{-1+44\mu\delta}.
\end{equation}
Matrix elements of spectral projection $\E^{(1)}(\k)$ also satisfy
the estimate:
\begin{equation}
\label{matrix elements}
\left|\E^{(1)}(\k)_{\s\s'}\right|<(Ck)^{-d^{(1)}(\s,\s')},\ \ \
d^{(1)}(\s,\s')=Q^{-1}\left(\||\p_\s\||+\||\p_{\s'}\||\right)(1-44\mu
\delta ).
\end{equation}
\end{corollary}
The last estimate easily follows from the formula \eqref{zeros} and
estimate \eqref{jan27}.
\begin{proof}The proof is based on expansion of the resolvent in
perturbation series on the circle $C_1$. Indeed, let us consider
the series
\begin{equation}
\left(H^{(1)}-z \right)^{-1}=\sum ^{\infty }
_{r=0}\left(H_0^{(1)}-z\right)^{-1}\left( -P(\delta)VP(\delta
)\left(H_0^{(1)}-z\right)^{-1}\right)^{r} \label{seriesforresolvent}
\end{equation}
where $H_0^{(1)}=P(\delta )H_0$ and $z\in C_1$. It easily follows
from \eqref{G1-5} that
\begin{equation}\left\|\left(H_0^{(1)}(\k)-z\right)^{-1}\right\|<\frac{8}{\tau
}k^{-1+40\mu \delta }.   \label{ocenka}
\end{equation}
Hence,
\begin{equation}\left\|\left(H^{(1)}(\k)-z\right)^{-1}\right\|<\frac{16}{\tau
}k^{-1+40\mu \delta }   \label{ocenka*}
\end{equation}
for sufficiently large $k$.
Substituting the series into the formula $\E^{(1)}
(\k)=-\frac{1}{2\pi i}\oint _{C_1}(H^{(1)}(\k)-z)^{-1}dz$ and
integrating term-wise, we arrive at
 \eqref{sprojector}. Estimates \eqref{jan27} easily follow from \eqref{ocenka} and the obvious inequality $\|P(\delta)\|_1\leq (2k^{\delta })^4$.
 It follows $\E^{(1)}= \E_0+O(k^{-1+44\mu \delta })$. This means that there is
 a single eigenvalue of $H^{(1)}(\k)$ inside $C_1$.   In a similar way (using \eqref{g}, \eqref{2.14} and $V=0$) we obtain
 the formula for the eigenvalue and \eqref{estg}, \eqref{estg_2}, for details see \cite{K}. To prove \eqref{zeros} we consider the operator $A=VP(\delta
)\left(H_0^{(1)}-z\right)^{-1}$ and represent it as $A=A_0+A_1+A_2$,
where $A_0=\left(P(\delta )-\E_0({\k})\right)A \left(P(\delta
)-\E_0({\k})\right)$, $A_1=\left(P(\delta )-\E_0({\k})\right)A
\E_0({\k})$, $A_2= \E_0({\k})A \left(P(\delta )-\E_0({\k})\right)$.
It is easy to see that $\E_0({\k})A \E_0({\k})=0$ because of $V=0$.
Note that
$$\oint _{C_1}\left(H_0^{(1)}-z\right)^{-1}A_0^r dz=0,$$ since the
integrand is a holomorphic function inside $C_1$. Therefore,
$$G^{(1)}_r({\k})=\frac{(-1)^{r+1}}{2\pi i}\sum _{j_1,...j_r=0,1,2,\
j_1^2+...+j_r^2\neq 0}\oint
_{C_1}\left(H_0^{(1)}-z\right)^{-1}A_{j_1}.....A_{j_r} dz. $$ At
least one of indices in each term is equal to $1,2$. We take into
account that $(A_2)_{\s\s'}=(A_1)_{\s'\s}=0$ if $\s\neq 0$ and
$A_{\s\s'}=0$ if $\||\p_{\s-\s'}\||>Q$. It follows that
$G^{(1)}_r({\k})_{\s\s'}$ can differ from zero only if $rQ\geq
\||\p_\s\||+\||\p_{\s'}\||$.

\end{proof}


It will be shown (Corollary \ref{Corollary 3.8}) that coefficients
$g^{(1)}_r({\k})$ and operators $G^{(1)}_r({\k})$ can be
analytically extended into the complex $\frac{\tau }{16}k^{-(40\mu
+1)\delta }$-neighborhood of $\omega^{(1)}(k,\delta, \tau )$ as
functions of $\varphi $
 and to the complex
$\frac{\tau }{8}k^{-(40\mu+1)\delta }-$ neighborhood of $k$ as
functions of $\varkappa$, estimates \eqref{estg}, \eqref{estg_2},
\eqref{jan27} being preserved. Now, we use formulae \eqref{g},
\eqref{eigenvalue} to extend
$\lambda^{(1)}({\k})=\lambda^{(1)}(\varkappa,\varphi)$ as an
analytic function. Obviously, series \eqref{eigenvalue} is
differentiable. Using Cauchy integral we get the following lemma.
\begin{lemma} \label{L:derivatives-1}Under
conditions of Theorem \ref{Thm1} the following estimates hold when
$\varphi $ is in $\omega^{(1)}(k,\delta, \tau )$ or its complex
$\frac{\tau }{32}k^{-(40\mu +1)\delta }$-neighborhood and
$\varkappa$ is in the complex $\frac{\tau }{16}k^{-40\mu\delta
}$-neighborhood of $\varkappa=k$ :
 \begin{equation}\label{perturbation-C}
\lambda^{(1)}({\k})=\varkappa^{2}+O\left(k^{-2+(80\mu+6)\delta}\right),
\end{equation}
\begin{equation}\label{estgder1}
\frac{\partial\lambda^{(1)}}{\partial\varkappa}=2\varkappa +
O\left(k^{-2 +(120\mu+6)\delta}\right), \ \
\frac{\partial\lambda^{(1)}}{\partial \varphi }=
 O\left(k^{-2 +(120\mu+7)\delta}\right),\end{equation}
\begin{equation}\label{estgder2} \begin{split} &
\frac{\partial^2\lambda^{(1)}}{\partial\varkappa^2}=2+O\left(k^{-2+(160\mu+6)\delta}\right),\cr
& \frac{\partial^2\lambda^{(1)}}{\partial\varkappa\partial \varphi
}= O\left(k^{-2+(160\mu+7)\delta}\right) ,\ \
\frac{\partial^2\lambda^{(1)}}{\partial\varphi ^2}=
O\left(k^{-2+(160\mu+8)\delta}\right).
\end{split}\end{equation}
\end{lemma}






\subsection{\label{GC}Geometric Considerations}
In this section we prove Lemma \ref{L:G1} and its corollary.
However, we will prove a version of this lemma  for a complex set of
$\varphi $. We need this complex version for further steps. Lemma
\ref{L:G1} is a simple corollary of the result proven in this
section. We will use the notation $|\a|^2_\R:=(\a,\a)_\R$ where
$(\a,\b)_\R:=a_1b_1+a_2b_2$ when $\a, \b \in \C^2$. It is easy to
see that $|\k (\varphi )+\p_{\m}|^2_{\R}$ is an analytic extension
in $\varkappa$ and $\varphi$ of $$
|\k+\p_{\m}|^2=\varkappa^2+p_{\m}^2+2\varkappa
p_{\m}\cos(\varphi-\varphi_{\m})
$$ defined for real $\varkappa,\varphi $. Note that $|\cdot |$ is the
canonical norm in $\C$ or $\R^2$. For every fixed $k\geq1$ and
$\frac{1}{32}\leq \tau\leq 32$, we describe a resonance set
$\OO^{(1)}=\OO^{(1)}(k,\tau )$ of $\varphi\in \C$. We put
\begin{equation} \label{52a} \OO^{(1)}(k,\tau
):=\cup_{\m\in\tilde\Omega(\delta)\setminus\{0\}}\OO_{\m}(k,\tau
),\end{equation} where
\begin{equation}\label{resonance} \begin{split}& \OO_{\m}(k,\tau
):=\{\varphi\in\C:\ \ \left||\vec k+\p_{\m}|^2_{\R}-k^2\right|\leq
\tau k^{1-40\mu\delta}\}=\cr & \{\varphi\in\C:\ \
\left|p_{\m}^2+2kp_{\m}\cos(\varphi-\varphi_{\m})\right|\leq \tau
k^{1-40\mu\delta}\}. \end{split} \end{equation} In most cases
parameter $\tau$ will be equal to $1$. But sometimes we will use
different choice of $\tau$. It easily follows from the definition
\eqref{resonance} and the estimate \eqref{above} that for any
$\varkappa\in\C$ such that $|\varkappa-k|\leq1$ and any $\varphi\in
\OO_{\m}(k,\tau )$ we have
\begin{equation}\label{complex}
\left||p_{\m}^2+2\varkappa
p_{\m}\cos(\varphi-\varphi_{\m})|-|p_{\m}^2+2kp_{\m}\cos(\varphi-\varphi_{\m})|\right|\leq
\frac{\tau}{4} k^{1-40\mu\delta}, \end{equation} provided
$2(1+40\mu)\delta\leq1$ and $k\geq 800$ which will be assumed in
what follows.

Let  $\W_0:=\{\varphi \in \C: |\Im \varphi |<1\}.$ We introduce a complex non-resonant set:
\begin{equation} \label{W1} \W^{(1)}(k,\tau ):=\W_0 \setminus \OO^{(1)}(k,\tau
). \end{equation} Clearly, it  is open. We also note that the set $\OO^{(1)}\cap[0,2\pi]$ is
symmetric, i.e.
$\OO^{(1)}\cap[0,2\pi]+\pi\,(\hbox{mod}\,2\pi)\,=\OO^{(1)}\cap[0,2\pi]$,
since $\varphi _{-\m}=\varphi _\m+\pi $.
We define $\omega^{(1)} (k,\delta, \tau )$ as a real part of
$\W^{(1)} (k,\delta, \tau )$:
\begin{equation}\omega^{(1)} (k,\delta, \tau )=\W^{(1)}(k,\tau )\cap [0,2\pi ).
\label{omega} \end{equation}
\begin{lemma} \label{L:jan28} Let $\varphi $ be in $\W^{(1)}(k,\tau )$, then
\begin{equation}\label{jan28a}
\left||\vec k (\varphi )+\p_{\m}|^2_{\R}-k^2\right|\geq \tau
k^{1-40\mu\delta}\mbox{   for all  }\m\in \tilde\Omega (\delta
)\setminus\{0\}.\end{equation} If $\varphi $ is in the complex
$k^{-(40\mu +1)\delta }$-neighborhood of $\W^{(1)}(k,\tau )$ and
$\varkappa\in \C: |\varkappa-k|<\frac{\tau }{8}k^{-40\mu \delta }$.
Then, for $\k =\varkappa (\cos \varphi ,\sin \varphi )$ the
following estimate holds:
\begin{equation}\label{jan28b}
\left||\k (\varphi )+\p_{\m}|^2_{\R}-k^2\right|\geq \frac{\tau}{2}
k^{1-40\mu\delta}\mbox{   for all  }\m\in \tilde\Omega (\delta
)\setminus\{0\}.\end{equation}\end{lemma} The lemma easily follows
from \eqref{resonance} and \eqref{complex}.
\begin{corollary} \label{Parts 1,2} Parts 1 and 2 of Lemma \ref{L:G1} hold.
\end{corollary}
\begin{corollary} \label{Corollary 3.8} Coefficients $g^{(1)}_r({\k})$ and
operators $G^{(1)}_r({\k})$ can be analytically extended into the
complex $\frac{\tau }{16}k^{-(40\mu +1)\delta }$-neighborhood of
$\omega^{(1)}(k,\delta, \tau )$ as functions of $\varphi $
 and to the complex
$\frac{\tau }{16}k^{-(40\mu+1)\delta }-$ neighborhood of $k$ as
functions of $\varkappa$, estimates \eqref{estg}, \eqref{estg_2},
\eqref{jan27} being preserved. \end{corollary}

\begin{lemma} The  measure of the resonance set
$\OO^{(1)}\cap[0,2\pi]$ satisfies the estimate:
\begin{equation}\label{meas1} meas(\OO^{(1)}\cap[0,2\pi])\leq  C_0k^{-37\delta\mu}.
\end{equation}
\end{lemma}
\begin{corollary} \label{Part 3} Part 3 of Lemma \ref{L:G1} holds. \end{corollary}
\begin{proof}

Let $\m \neq 0$ and $\varphi^{\pm}_{\m}$ be two (mod $2\pi$)
solutions of the equation $$
p_{\m}^2+2kp_{\m}\cos(\varphi-\varphi_{\m})=0. $$ Obviously,
$\varphi^{\pm}_{\m}-\varphi_{\m}=\pm \frac{\pi }{2}+O(k^{-1+\delta
})$.  Put $$ \Phi^{\pm}_{\m}:=\{\varphi\in\C:\ \
|\varphi-\varphi^{\pm}_{\m}|\leq \tau k^{-39\delta\mu}\}.
$$ Then, taking into account \eqref{below}, it is not difficult to see that $\OO_{\m}\subset
\cup_{\pm,j\in\Z}(\Phi^{\pm}_{\m}+2\pi j)$. Thus,
\begin{equation}\label{meas} meas(\OO^{(1)}\cap[0,2\pi])\leq4\tau
k^{-39\delta\mu}(8k^\delta)^4\leq C k^{-37\delta\mu}.
\end{equation} \end{proof}

{\bf Proof of Corollary \ref{C:L:G1}.} Let $C_1:=\{z\in\C:\ \
|z-k^{2}|=\frac{\tau }{4}k^{1-40\mu\delta}\}$ be the contour around
eigenvalue $k^{2}$ of the unperturbed operator $H_0(\vec k)$. Then
it follows from \eqref{jan28a} that for any
$\varphi\in\W^{(1)}(k,\tau)$,
$\m\in\tilde\Omega(\delta)\setminus\{0\}$,
and $z:\ |z-k^{2}|\leq \frac{\tau }{4}k^{1-40\mu\delta}$ we have
\begin{equation}\label{perturbest} \begin{split}&
||\vec k+\p_{\m}|^{2}_{\R}-z|\geq ||\vec
k+\p_{\m}|^{2}_{\R}-k^{2}|-\frac{\tau}{4}k^{1-40\mu\delta}\geq \cr &
{\tau} k^{1-40\mu\delta}-\frac{\tau}{4}k^{1-40\mu\delta}\geq
\frac{\tau}{4} k^{1-40\mu\delta},
\end{split}
\end{equation} for sufficiently large $k$. For $\m=0$ the estimate
follows from the definition of $C_1$.

\subsection{\label{IS1}Isoenergetic Surface for Operator $H^{(1)}$}

\begin{lemma}\label{ldk} \begin{enumerate}
\item For every sufficiently large $\lambda $, $\lambda :=k^{2}$, and $\varphi $ in the real $\frac{\tau }{32} k^{-(40\mu +1)\delta }$-neighborhood
of $\omega^{(1)}(k,\delta, \tau )$ , there is a unique
$\varkappa^{(1)}(\lambda, \varphi )$ in the interval
$I_1:=[k-\frac{\tau }{32}k^{-40 \mu \delta },k+\frac{\tau
}{32}k^{-40 \mu \delta }]$, such that
    \begin{equation}\label{2.70}
    \lambda^{(1)} \left(\k
^{(1)}(\lambda ,\varphi )\right)=\lambda ,\ \ \k ^{(1)}(\lambda
,\varphi ):=\varkappa^{(1)}(\lambda ,\varphi )\vec \nu(\varphi).
    \end{equation}
\item  Furthermore, there exists an analytic in $ \varphi $ continuation  of
$\varkappa^{(1)}(\lambda ,\varphi )$ to the complex  $\frac{\tau
}{32} k^{-(40\mu +1)\delta }$-neighborhood of
$\omega^{(1)}(k,\delta, \tau )$ such that $\lambda^{(1)} (\k
^{(1)}(\lambda, \varphi ))=\lambda $. Function
$\varkappa^{(1)}(\lambda, \varphi )$ can be represented as
$\varkappa^{(1)}(\lambda, \varphi )=k+h^{(1)}(\lambda, \varphi )$,
where
\begin{equation}\label{dk0} |h^{(1)}|=O(k^{-3+(80\mu +6)
\delta }), \end{equation}

\begin{equation}\label{dk}
\frac{\partial{h}^{(1)}}{\partial\varphi}=O\left(k^{-3+(120\mu+7)\delta}\right),\
\ \ \ \
\frac{\partial^2{h}^{(1)}}{\partial\varphi^2}=O\left(k^{-3+(160\mu+8)\delta}\right),
\end{equation}
\begin{equation}\label{dk1} \frac{\partial \varkappa^{(1)}}{\partial \lambda }=\frac{1}{2k}\left(1+O(k^{-3+(120\mu +6)\delta })\right). \end{equation}\end{enumerate}

 \end{lemma} \begin{proof} \begin{enumerate} \item Let us prove
 existence of  $\varkappa^{(1)}(\lambda, \varphi ) $. By Theorem \ref{Thm1}, there
exists an eigenvalue $\lambda ^{(1)} (\k)$, given by
(\ref{eigenvalue}), for all $\varkappa$ in the interval $I_1$.  Let
    ${\cal L}^{(1)}(\varphi ):=
    \{\lambda^{(1)}(\k) : \varkappa \in
I_1\}.$ Using the definition of $I_1$, \eqref{perturbation}, and
continuity of $\lambda^{(1)}(\k)$  in $\varkappa$, we
easily obtain
    $
    {\cal L}^{(1)}(\varphi ) \supset
    [k^{2}-t,k^{2}+t]$,
    $t=c_1k^{1-40\mu \delta}$, $
    0<c_1 \neq c_1(k).
    $
     Hence,  there exists a $\varkappa^{(1)}$ such that
$\lambda^{(1)}(\k^{(1)})=k^{2}$, $\varkappa^{(1)} \in I_1$.

 Now we show that there is only one $\varkappa^{(1)}$  in the interval $I_1$ satisfying
(\ref{2.70}). Indeed, by \eqref{estgder1},
    $
    \dfrac{\partial \lambda^{(1)}(\k)}{\partial \varkappa } \geq 2k\bigl( 1+o(1)
\bigr)$. This implies that $\lambda^{(1)}(\k)$ is monotone with
respect to $\varkappa$ in $I_1$. Thus, there is only one $\varkappa
\in I_1$ satisfying ~(\ref{2.70}).

\item  We  consider $\lambda^{(1)}\left(\k
(\varphi )\right)$ as a function of complex variable $\varkappa$ in
the disc $|\varkappa-k|<\frac{\tau }{32}k^{-40\mu \delta }$. Taking
into account \eqref{perturbation-C} and applying Rouch\'{e}'s
Theorem, we obtain that for any $\varphi$ in $\frac{\tau
}{32}k^{-(40\mu +1)\delta }$-neighborhood of $\omega^{(1)}(k,\delta,
\tau )$ there exists unique value of $\varkappa^{(1)}(\varphi )$
such that $|\varkappa^{(1)}(\varphi )-k|<\frac{\tau }{32}k^{-40\mu
\delta }$ and $\lambda^{(1)}\left(\k^{(1)}(\varphi
)\right)=\lambda:=k^{2}$. Actually (see \eqref{perturbation-C}),
\begin{equation} |\varkappa^{(1)}(\varphi )-k|=O(k^{-3+(80\mu +6) \delta
}).\label{kappa1} \end{equation} Then it follows from
\eqref{estgder1} and  implicit function theorem that
$\varkappa^{(1)}(\varphi)$ is locally analytic. Combined with
uniqueness this implies global analyticity.

 The
estimate \eqref{dk0} follows from \eqref{kappa1}.  Applying standard
arguments with the Cauchy formula we obtain \eqref{dk}. Using
\eqref{estgder1} we get \eqref{dk1}.
\end{enumerate}
\end{proof}

Let us consider the set of points in $\R^2$ given by the formula:
$\k=\k^{(1)} (\varphi), \ \ \varphi \in \omega^{(1)} (k,\delta, \tau
)$. By Lemma \ref{ldk} this set of points is a slightly disturbed
circle with holes. All the points of this curve satisfy
the equation $\lambda^{(1)} (\k ^{(1)}(\lambda, \varphi ))=k^{2}$.
We call it isoenergetic surface of the operator $H^{(1)}$ and denote
by ${\cal D}_{1}(\lambda)$. The ``radius"
$\varkappa^{(1)}(\lambda, \varphi )$ of ${\cal D}_{1}(\lambda)$
monotonously increases with $\lambda $, see \eqref{dk1}.


\subsection{Preparation for Step II. Construction
of the Second Nonresonant Set}
\subsubsection{Model Operator for Step II \label{MOforStep2}}
Here we will describe an operator $PHP$, see \eqref{PHP}, which will
be used  for constructing perturbation series in the second step.
The operator $PHP$ has a block structure.

Let $r_1$ be some fixed number $2<r_1$. An upper bound on $r_1$ we
will introduce in Step II. We defined $\OO _\m$ by formula
(\ref{resonance}) for all $\m $: $0<\||\p_\m\||\leq 4k^{\delta }$.
Now,   we define $\OO _\m$ for $\m $: $4k^{\delta}<\||\p_\m\||\leq k^{r_1}$ by the following formula:
\begin{equation}\label{resonance1} \begin{split}& \OO_{\m}(k,\tau
):=\{\varphi\in\C:\ \ \left||\vec k+\p_{\m}|^2_{\R}-k^2\right|\leq
\tau k^{\delta_*}\}=\cr & \{\varphi\in\C:\ \
\left|p_{\m}^2+2kp_{\m}\cos(\varphi-\varphi_{\m})\right|\leq \tau
k^{\delta_*}\},\ \ \delta_*=10^4\mu\delta <1/10. \end{split} \end{equation}
 Note that  the
right-hand part in the inequality here is smaller than the
corresponding one in \eqref{resonance}. Obviously, $\OO _\m$
contains the whole interval $[0,2\pi )$ for sufficiently small
$p_\m$. As in Step I let $\varphi^{\pm}_{\m}$ be two (mod $2\pi$)
solutions of the equation
\begin{equation} p_{\m}^2+2kp_{\m}\cos(\varphi-\varphi_{\m})=0.
\label{Jan23a} \end{equation}

\begin{lemma} \label{L:3.1} The set $\OO _\m (k, \tau )$ has the following properties:
\begin{enumerate}
\item If $p_\m>4k$, then $\W_0\cap \OO _\m (k, \tau )=\emptyset $.
\item If $k^{-1+2\delta_* }\leq p_\m\leq 4k$ and $|4k^2-p_\m^2|>4\tau k^{\delta_* }$, then
$\OO_{\m}\subset \cup_{\pm,j\in\Z}(\Phi^{\pm}_{\m}+2\pi
j)$, where
$$ \Phi^{\pm}_{\m}:=\left\{\varphi\in\C:\ \
|\varphi-\varphi^{\pm}_{\m}|\leq \frac{\tau k^{-1+\delta_*
}}{p_\m\sqrt{1-p_\m^2(2k)^{-2}}}\right\},
$$
and $\Phi^{+}_{\m}\cap \Phi^{-}_{\m}=\emptyset $.
\item
If  $|4k^2-p_\m^2|\leq 4\tau k^{\delta_* }$, then $\OO_{\m}\subset
\cup_{\pm,j\in\Z}(\Phi^{\pm}_{\m}+2\pi j)$, where
$$ \Phi^{\pm}_{\m}:=\left\{\varphi\in\C:\ \
|\varphi-\varphi^{\pm}_{\m}|\leq 32\tau k^{-1+\delta_*/2 }\right\}.
$$
\end{enumerate}
\end{lemma}
In the proof we use the Taylor series with respect to $\varphi $ for
$|\vec k(\varphi
)+\p_\m|_\R^2-k^{2}$ near its zeros, see Appendix 1.\\

Let $\varphi _0\in[0,2\pi)\setminus\OO^{(1)}(k,8)$, where
$\OO^{(1)}(k,8)$ is given by \eqref{52a}. We define $\MM(\varphi
_0)\subset \Z^4$ as follows:  \begin{equation} \label{M} \MM(\varphi
_0):=\{\m:\ \ \ 0<\,|\|\p_{\m}|\|\leq k^{r_1}\ \hbox{and}\ \varphi
_0\in\OO_{\m}(k,1)\}. \end{equation} We will also need a larger set
$$ \MM'(\varphi _0):=\{\m:\ \ \ 0<\,|\|\p_{\m}|\|\leq 2k^{r_1}\
\hbox{and}\ \varphi _0\in\OO_{\m}(k,1)\}. $$ In fact, $\MM(\varphi
_0)$, $\MM'(\varphi _0)$ do not include
$\m:\,|\|\p_{\m}|\|<4k^{\delta }$, since $\varphi
_0\in[0,2\pi)\setminus\OO^{(1)}(k,8)$.

We split $\MM (\varphi _0)$ into two
 components  $\MM:=\MM_1\cup\MM_2$.
By definition, $\m\in \MM_1$ if $$\min _{\m'\in \MM' (\varphi _0),
\m'\neq \m}\||\p _{\m-\m'}\||>k^{\delta }.$$   Let $\MM_2=\MM
\setminus \MM_1$. Next, let $\tilde{\MM}_{\m}$ be $(k^{\delta
}/3)$-neighborhood of $\m\in \MM$ in $\||\cdot \||$ norm:
\begin{equation}\tilde{\MM}_{\m}:=\{\n:\ \ \ |\|\p_{\n-\m}|\|<
k^{\delta}/3\ \hbox{for a given}\ \m\in\MM(\varphi _0)\}. \label{51a} \end{equation}
 Obviously,
$$\tilde{\MM}_{\m}(\varphi _0)\cap \tilde{\MM}_{\m'}(\varphi
_0)=\emptyset,\ \ \ \mbox{for any }\m \in \MM_1\mbox{ and }\m'\in
\MM',\ \ \m'\neq \m.$$ Let  $\tilde{\MM}_1(\varphi _0)$ be
$(k^{\delta }/3)$-neighborhood of $\MM_1$ in $\||\cdot \||$ norm:
$$\tilde{\MM}_1(\varphi _0):=\cup _{\m \in \MM_1(\varphi
_0)}\tilde{\MM}_{\m}(\varphi_0) =\{\n:\ \ \ |\|\p_{\n-\m}|\|<
k^{\delta}/3\ \hbox{for some}\ \m\in\MM_1(\varphi _0)\}.$$

\bigskip

Let us introduce an equivalence relation in $\M'$. We say $\m_0 \sim
\m_0'$ if there is a sequence $\m_j\in \MM'$, $j=1,...,J,$ connecting these two points ($\m_J=\m_0'$), such that for each $j$
$\min_{k:\,0\leq k<j}\||\p _{\m_j-\m_{k}}\||\leq 3k^{\delta }$.  We denote the equivalence class
containing $\m \in \MM_2$ by  $\MM_2^{(\m)}$. By definition of
$\MM_2$ such equivalence class contains at least one more element.

\begin{lemma}\label{2.12} Let $\m_0\in \MM_2$ and $\m_j \in \M'$, $j=1,2$, are such that all $\m_{j}$, $j=0,1,2$, are different  and
$\min_{k:\,k\neq j}\||\p _{\m_j-\m_{k}}\||\leq 3k^{\delta }$ for all
$j=0,1,2$. Then, $\p _{\m_1-\m_{0}}=c\p _{\m_2-\m_{0}}$ with some
$c\in\R$.\end{lemma}
\begin{corollary}  All points $\p_{\m'}$ with $\m'$ being in
a class of equivalence $\MM_2^{(\m)}$ are situated on a line.
\label{CL:4}\end{corollary}
The proof of the lemma is in Appendix 2.

Obviously, for any pair $\m,\m'\in \MM_2$ either $ \MM_2^{(\m)}=
\MM_2^{(\m')}$ or $ \MM_2^{(\m)}\cap \MM_2^{(\m')}=\emptyset .$ We
can enumerate different equivalence classes $ \MM_2^{(\m)}$ by an
index $j$ and denote them by $ \MM_2^{j}$, $j=1,...,J_0$. By
construction, $\MM _2\subset \cup _{j=1}^{J_0} \MM_2^{j}\subset \MM
'$.
By Lemma \ref{2.12} each class $ \MM_2^{j}$ has a "direction", which is the direction of the corresponding line. We denote it by $\p_\q$, $\p_\q=\p _{\m_1-\m_{0}}$, the direction, naturally, being defined up to a constant multiplier. Obviously, all vectors $\p _{\m}$, $\m \in  \MM_2^{j}$, have the same projection on a direction $\vec \nu _{\q}^{\bot }$ orthogonal to $\vec\nu_\q:=\p _{\q}/p_\q$. We define the ``orthogonal component" $t_{\q}^{\bot }$ of $\MM_2^{j}$ as
\begin{equation} \label{ortcomp}
t_{\q}^{\bot }=\left(\vec k(\varphi _0)+\p _{\m},\vec \nu _{\q}^{\bot }\right), \ \  \m \in  \MM_2^{j}. \end{equation}

\begin{definition} \label{def} We call $ \MM_2^{j}$ trivial if one of two conditions holds:

1) no vector in $\SS _Q$ has the same direction as $ \MM_2^{j}$ or

2) there is $\q\in \SS _Q$ which has the same direction as $ \MM_2^{j}$ and $t_{\q}^{\bot }$ satisfies the inequality:
\begin{equation} \left|k^2-(t_{\q}^{\bot })^2\right|>\frac18 k^{\delta _*}.\label{Jan1-13} \end{equation}

Otherwise, we call a cluster non-trivial.
In a non-trivial case, without the loss of generality, we can consider that the directional vector $\q$ of $ \MM_2^{j}$ belongs to $\SS _Q$  and it is a generating vector in its direction (see the definition of $\SS _Q$). \end{definition}
When a cluster is trivial, it can be treated by a method quite similar to that in the case of a polyharmonic operator. For non-trivial clusters we will use an additional consideration, involving a  periodic operator in one dimension. It is a periodic and not a quasiperiodic operator, because of condition $2$ on $V(\x)$ at the beginning of Section \ref{Section 2}. Treatment of non-trivial clusters is the main reason we need condition $2$. We did not need condition $2$ on $V(\x)$ for the polyharmonic case \cite{KaSh}, since in that case all
clusters $ \MM_2^{j}$ could be treated the same way as  trivial ones here.

Now we introduce a further split of every $\MM_2^{j}$. Suppose $\MM_2^{j}$ is non-trivial. For
$\m,\m'\in\MM_2^{j}$ we say $\m\sim\m'$ if $\m-\m'=n\q$, $n\in \Z$, $\q$ is the  directional vector of $\MM_2^{j}$  in $\SS _Q$ as described above.
Hence,
every $\MM_2^{j}$ is represented as a disjoint union of such
subsets, $\MM_2^{j}=\cup_s\MM_2^{j,s}$.
We notice that every set $\MM_2^{j,s}$ consists of vectors of
the form $\m_{j,s}+n\q$,
where $n\in\Z$ and such that
$\m_{j,s}+n\q\in\MM_2^{j}$, $\q=\q(j)$. By \eqref{M},  \eqref{resonance1},
\begin{equation}
\left|\left|\vec k (\varphi _0)+\p_{\m_{j,s}}+n\p_{\q}\right|^2-k^2\right|<k^{\delta _*}. \label{Jan1a-13}
\end{equation}
Considering the inequality opposite to \eqref{Jan1-13}, we obtain that each $\MM_2^{j,s}$ can be described by the formula:
\begin{equation} \MM_2^{j,s}=\left\{\m_{j,s}+n\q:0\leq \left(\vec k(\varphi _0)+\p _{\m_{j,s}},\vec \nu _{\q}\right)<p_{\q}, \ n\in \Z,\  n_-^{j,s}\leq n\leq n_+^{j,s}\right\},
\label{Mjs}
\end{equation} where $n_-^{j,s}<0$, $n_+^{j,s}>0$, $\frac12 p_\q^{-1}k^{\delta _*/2}<|n_-^{j,s}|,n_+^{j,s}<2p_\q^{-1}k^{\delta _*/2}$ and $$\left|\left|\vec k (\varphi _0)+\p_{\m_{j,s}}+n_{\pm }^{j,s}\p_{\q}\right|^2-k^2\right|=k^{\delta _*}+O(k^{\delta _*/2}).$$ It is easy to see also that
\begin{equation} |n_{\pm }^{j,s}-n_{\pm }^{j,s'}|\leq 1. \label{Mjs*}
\end{equation} for any pair $s,s'$.
We will refer the point $\m_{j,s}$ as the central point of $\MM_2^{j,s}$. We also have $|\p_\q|\geq CQ^{-\mu}$.

In the case when $\MM_2^{j}$ is trivial, we  consider that each $\MM_2^{j,s}$ contains just one point.

In a non-trivial case the reduction $H^{j,s}(\vec k(\varphi _0)):=P_{j,s}H(\vec k(\varphi _0))P_{j,s}$ (here $P_{j,s}$ is the
diagonal projection in $l^2$ corresponding to  $\MM_2^{j,s}$) can be described by the  matrix:
\begin{equation}\label{May7-14a} H^{j,s}_{nn'}=|\vec k(\varphi _0)+\p_{\m_{j,s}}+n\p_{\q(j)}|^2_\R\delta_{n,n'}+V_{(n-n')\q(j)}.
\end{equation}
with $\m_{j,s}+n\q,\,\m_{j,s}+n'\q\in\MM_2^{j,s}$ ($n_-^{j,s}\leq n, n'\leq n_+^{j,s}$).  Note that there is a ``separation of variables", i.e., this matrix can be rewritten as
\begin{equation}\label{May7b}H^{j,s}=\tilde
H^{j,s}(t_\q)+(t_\q^\perp)^2I, \end{equation}
where \begin{equation} \label{May7-14c}
\tilde
H^{j,s}_{n_1,n_2}:=(t_\q+n_1p_{\q})^2\delta_{n_1,n_2}+V_{(n_1-n_2)\q},
\end{equation}
\begin{equation}
t_\q:=(\vec k(\varphi_0)+\p_{\m_{j,s}},\vec\nu_\q),\ \ \ t_\q^\perp:=(\vec
k(\varphi_0)+\p_{\m_{j,s}},\vec\nu_\q^\perp).
\label{Nov7-12} \end{equation}
Note that $t_\q$ depends both on $j$ and $s$, while $ t_\q^\perp$ depends on $j$ only (we omit indices $j,s$ for shortness).
By construction, $\q=\q(j)\in \SS _Q$ and  $0\leq t_\q<p_{\q}$. Let us consider also an infinite version of \eqref{May7-14c} $\tilde
H^{j,s}_{per}$, i.e. $\left(\tilde
H^{j,s}_{per}\right)_{nn'}$ is given by \eqref{May7-14c} for all $n,n'\in\Z$. Clearly, $\tilde
H^{j,s}_{per}$ corresponds to a one-dimensional Schroedinger operator with a periodic potential and quasimomentum $t_\q$. Obviously, this construction works not only for
$\varphi _0$, but also for any  $\varphi $ in the $2k^{-2-40\mu\delta}$-neighborhood of $\varphi _0$. We are going to investigate properties of $H^{j,s}$ using known properties of  $\tilde
H^{j,s}_{per}$.


\begin{lemma}\label{per1}
Let $|\Re \varphi-\varphi_0|\leq 2k^{-2-40\mu\delta}$. Then $\Im|\vec
k(\varphi)+\p_\n|^2_\R$ have the same sign for all
$\n\in\MM_2^{j,s}$ when $\Im\varphi>0$ (or $\Im\varphi<0$)
and the following inequality holds:
\begin{equation}\label{1per}
\left|\Im|\vec k(\varphi)+\p_\n|^2_\R\right|\geq
k^{2-85\mu\delta}|\Im\varphi|.
\end{equation}
\end{lemma}
\begin{corollary}\label{C:per1}  If $\varphi $ is real and $|\varphi-\varphi_0|\leq 2k^{-2-40\mu\delta}$, then $H^{j,s}(\vec k(\varphi ))$ is monotonous in $\varphi $ and its eigenvalues $\lambda _n(\varphi )$ satisfy the estimates:
\begin{equation} \label{1per*} \left|\frac{\partial \lambda _n(\varphi ) }{\partial \varphi }\right|\geq k^{2-85\mu\delta}.\end{equation}
\end{corollary}
{\em Proof of the corollary.}  Since $|\vec k(\varphi)+\p_\n|^2_\R$ is a holomorphic function of $\varphi $ in the rectangle $|\Re \varphi-\varphi_0|\leq 2k^{-2-40\mu\delta}$,
the inequality \eqref{1per} yields that \begin{equation}\left|\frac{\partial  |\vec k(\varphi)+\p_\n|^2_\R}{\partial \varphi }\right|\geq k^{2-85\mu\delta}\label{1per**}
\end{equation} for real $\varphi $ and the derivative has the same sign for all $\n\in \MM _2^{j,s}$.  Now, the obvious relation $\frac{\partial  H^{j,s}}{\partial \varphi }=\frac{\partial  H_0^{j,s}}{\partial \varphi }$ yields \eqref{1per*}.

\begin{proof}
It is enough to consider the case $\Im\varphi>0$. The other case is
just the complex conjugated.
Let $\m \in \MM_2^{j}$.  Obviously,
$$
\Im|\vec k(\varphi)+\p_{\m}|^2_\R=\Im 2(\vec
k(\varphi),\p_{\m})=-2kp_{\m}\sin(\Re\varphi-\varphi_{\m})\sinh(\Im\varphi),
$$
where $\varphi_{\m}$ is the angle corresponding to
$\p_{\m}$. It suffices to show that $p_{\m}>k^{1-42\mu \delta }$ and
\begin{equation} \label{Nov6-12}|\sin(\Re\varphi-\varphi_{\m})|>k^{-42\mu \delta }. \end{equation}
Indeed, let $\p_\q $ be a vector, such that  $0<\||\p_\q\||<k^{\delta }$ and
$$
||\vec
k(\varphi_0)+\p_\m|^2_\R-k^2|\leq
k^{\delta_*},\ \ \ ||\vec
k(\varphi_0)+\p_{\m+\q}|^2_\R-k^2|\leq
k^{\delta_*}.
$$
Such a  $\p_\q $  exists, since $\m \in \MM_2^{j}$.
It follows:
\begin{equation}2\left|(\vec
k(\varphi_0)+\p_\m, \p_{\q})\right |\leq 2k^{\delta _*}+(2\pi)^2k^{2\delta}. \label{Nov6-12a}
\end{equation}
Since $\varphi _0 \in \W^{(1)}(k,\tau )$, we have $\left|(\vec
k(\varphi_0), \p_{\q})\right |>\frac{1}{8}k^{1-40\mu \delta }$, see Lemma \ref{L:jan28} for $\m=\q$. Using the last two inequalities, we obtain $\left|(\p_\m, \p_{\q})\right |>\frac{1}{16}k^{1-40\mu \delta }$. Hence, $p_\m>k^{1-42\mu \delta }$.
It remains to prove \eqref{Nov6-12}.
  Assume that $|\sin(\Re\varphi-\varphi_{\m})|\leq
k^{-42\mu\delta}$. It follows from the assumption $|\Re\varphi-\varphi _0|<2k^{-2-40\mu \delta }$ that
\begin{equation}\label{2per}
|\sin(\varphi_0-\varphi_{\m})|\leq 2k^{-42\mu\delta}.
\end{equation}
Since $\m \in \MM _2^j$,
\begin{equation}
||\vec
k(\varphi_0)+\p_\m|^2_\R-k^2|=|2kp_\m\cos(\varphi_0-\varphi_\m)+p_\m^2|\leq
k^{\delta_*}. \label{Nov1-2012} \end{equation}
This and \eqref{2per} yield $\p_\m=-2\vec k(\varphi_0)
+O(k^{1-42\mu\delta})$. Then,
$$
|2(\vec
k(\varphi_0)+\p_\m,\p_\q)+p_\q^2|=|(\vec
k(\varphi_0),\p_\q)|+O(k^{1-41\mu\delta})\geq
\frac{1}{16}k^{1-40\mu\delta}.
$$
This contradicts to \eqref{Nov6-12a}.
The contradiction completes the proof of \eqref{Nov6-12}. Assume $\MM_2^{j}$ is trivial. Then, $\MM_2^{j,s}$ contains just one point and the statement of the lemma follows immediately.
If $\MM_2^{j}$ is non-trivial, then considering as above, we obtain \eqref{1per} for every $\n \in \MM_2^{j,s}$. It remains to prove only that $\Im|\vec k(\varphi)+\p_{\n}|^2_\R$ has the same sign for all $\n \in \MM_2^{j,s}$. By the definition of $\MM_2^{j,s}$, every
$\n\in\MM_2^{j,s}$ can be written  as $\n=\m_{j,s}+n\q$, $|n|=O(k^{\delta _*/2})$ (see \eqref{Mjs}).
 Next,
$$
\Im|\vec k(\varphi)+\p_{\m_{j,s}+n\q}|^2_\R=\Im 2(\vec
k(\varphi),\p_{\m_{j,s}+n\q})=-2kp_{\m_{j,s}+n\q}\sin(\Re\varphi-\varphi_{\m_{j,s}+n\q})\sinh(\Im\varphi),
$$
where $\varphi_{\m_{j,s}+n\q}$ is the angle corresponding to
$\p_{\m_{j,s}+n\q}$. Since $|n|p_\q\leq C(Q)k^{\delta_*/2}$ and $p_{\m_{j,s}}\geq
k^{1-42\mu\delta}$, we get
$\varphi_{\m_{j,s}+n\q}=\varphi_{\m_{j,s}}+O(k^{-1+42\mu\delta +\delta _*/2})$. Thus,
$$
\Im|\vec
k(\varphi)+\p_{\m_{j,s}+n\q}|^2_\R=-2kp_{\m_{j,s}+n\q}\left(\sin(\Re\varphi-\varphi_{\m_{j,s}})+O(k^{-1+42\mu\delta +\delta _*/2})\right)\sinh(\Im\varphi).
$$
Using \eqref{Nov6-12}, we obtain that the right-hand side of the last inequality has the same sign for all $n: \m_{j,s} +n\q \in \MM ^{j,s}_2$.
\end{proof}

\begin{lemma}\label{per2}
For every $j$ and $s$ the resolvent $(H^{j,s}(\vec k(\varphi))-k^2)^{-1}$
has no more than two poles  in the strip $\{\varphi :|\Re \varphi-\varphi_0|\leq
2k^{-2-40\mu\delta}, |\Im \varphi |<1\}$, counting multiplicity \footnote{By the multiplicity of a pole of the resolvent we mean the multiplicity of the zero of
$\det\left(H^{j,s}-k^2\right)$ as function of $\varphi$. }. The poles are situated on the real axis.
\end{lemma}
\begin{proof}  If $\MM^{j,s}_2$ consists of just one element $\m$, then $H^{j,s}$ has one matrix element and
\begin{equation}H^{j,s}_{\m\m}-k^2=|\vec
k(\varphi)+\p_{\m}|^2_\R-k^2=2kp_{\m}\cos(\varphi-\varphi_{\m})+p_\m^2.\label{Jan4-13}\end{equation}
Obviously, the number of roots of this function in the strip does not exceed two. By \eqref{1per}, they are both real.
Assume that $\MM^{j,s}_2$ contains more than one element (the case of a non-trivial cluster). Obviously, $H^{j,s}-(H^{j,s})^*=H_0^{j,s}-(H^{j,s}_0)^*$. Considering now \eqref{1per},
we conclude that the resolvent $(H^{j,s}(\vec k(\varphi))-k^2)^{-1}$ may have only real poles $\varphi $.
The poles are described by equations
$\lambda_n(\varphi)=k^2$, $n=1,\dots$, the eigenvalues $\lambda_n$
of $H^{j,s}$ being numerated in the increasing order counting
multiplicity. By Corollary \ref{C:per1},
$\partial\lambda_n/\partial\varphi\not=0$ and the derivatives have the same sign for
all $n$. Assume, for definiteness, that
$\partial\lambda_n/\partial\varphi>0$. It follows that the order of
each pole $\varphi_k$  of the resolvent is less or equal to the multiplicity of the
corresponding eigenvalue $\lambda_{n(k)}(\varphi_k)$. If for two
poles we have $\varphi_k<\varphi_{k'}$, then $n(k')<n(k)$ because of
monotonicity of $\lambda_n(\varphi)$. Thus, the number of poles
$\varphi_k$ (counting multiplicity) in the disk
$|\varphi-\varphi_0|\leq 2k^{-2-40\mu\delta}$ does not exceed the
number of eigenvalues (counting multiplicity) satisfying
$\lambda_n(\varphi)=k^2$ for some $\varphi\in\R:\
|\varphi-\varphi_0|\leq 2k^{-2-40\mu\delta}$. Note that by
perturbation theory it is enough to estimate the number of
eigenvalues of $H^{j,s}(\vec k(\varphi_0))$ in the interval
$[k^2-4k^{-40\mu\delta},k^2+4k^{-40\mu\delta}]$.

Using the ``separation of the variables" \eqref{May7b}, we rewrite
$(H^{j,s}(\vec k(\varphi))-k^2)^{-1}$ in the form $(\tilde
H^{j,s}(t_\q)+(t_\q^\perp)^2-k^2)^{-1}$.
We denote eigenvalues of $\tilde H^{j,s}(t_\q)$ by
$\tilde\lambda_n(t_\q)$. Thus,  $\lambda_n(\varphi)=\tilde\lambda_n(t_\q)+\left(t_\q^\perp \right)^2$.
It suffices to prove that there are no more than two eigenvalues $\tilde \lambda _n(t_\q)$ (counting multiplicity) in the interval
$[k^2-(t_\q^\perp)^2-4k^{-40\mu\delta},k^2-(t_\q^\perp)^2+4k^{-40\mu\delta}]$ .   Suppose there are at least three.
For a non-trivial cluster, $|k^2-(t_\q^\perp)^2|\leq \frac18 k^{\delta_*}$ at $\varphi_0$, see Definition \ref{def}.
 It follows from the defnition of  $ \M^{j,s}_2$ that {\it all} $n$
such that $(t_\q+np_{\q})^2\leq \frac 78 k^{\delta_*}$ are present in the
matrix of $\tilde H^{j,s}(t_\q)$. Now, by the perturbation theory, all eigenvalues of $\tilde H^{j,s}(t_\q)$ such that $|\tilde
\lambda_n(t_\q)|\leq \frac18k^{\delta_*}$ can be approximated by
eigenvalues of the corresponding {\it infinite} matrix  $\tilde H^{j,s}_{per}(t_\q)$  (cf.
\eqref{May7-14c}) with an accuracy
$O(k^{-\frac{\delta_*}{C(Q)}k^{\delta_*/2}})$. Infinite matrix
$\tilde H^{j,s}_{per}(t_\q)$ corresponds to the one-dimensional
periodic Schr\"odinger operator with the period $p_\q^{-1}$. We
denote the eigenvalues of the matrix $\tilde H^{j,s}_{per}(t_\q)$
by $\tilde \lambda_n^{per}(t_\q)$. We have
\begin{equation}\label{5per}
\tilde \lambda_n^{per}(t_\q)=\tilde
\lambda_{n}(t_\q)+O(k^{-\frac{\delta_*}{C(Q)}k^{\delta_*/2}}).
\end{equation}
Thus, if the operator $H^{j,s}(\vec k(\varphi_0))$ has three eigenvalues in
the interval $[k^2-4k^{-40\mu\delta},k^2+4k^{-40\mu\delta}]$, then
$\tilde H^{j,s}_{per}(t_\q)$ has at least three eigenvalues in
the interval
$[k^2-(t_\q^\perp)^2-8k^{-40\mu\delta},k^2-(t_\q^\perp)^2+8k^{-40\mu\delta}].$
Since every eigenvalue
belongs to a different zone  of the corresponding one-dimensional periodic operator, it means that there is a zone in the
spectrum with the length less than $16k^{-40\mu\delta}$. We also
notice that $t_\q$ in the matrix $\tilde H^{j,s}_{per}(t_\q)$
plays the role of quasi-momentum, while the corresponding one-dimensional
periodic operator {\it does not} depend on $k$. Now, we arrive at
the contradiction if $k^{-40\mu\delta}$ is sufficiently small
(obviously, depending on $V$ only).
\end{proof}
\begin{definition} \label{Nov27-12}
We call a set $\MM_2^{j,s}$ strongly resonant if the
corresponding resolvent $(H^{j,s}(\vec k(\varphi))-k^2)^{-1}$ has at least
one pole in the disk $|\varphi-\varphi_0|\leq 2k^{-2-40\mu\delta}$.
Otherwise the set is called weakly resonant.\end{definition}
 We say that two
strongly resonant sets $\MM_2^{j,s}$ and $\MM_2^{j',s'}$ are
adjacent, if the $\||\cdot \||$ distance between them is no more than $3k^\delta$.
We say that two sets  $\MM_2^{j,s}$  and $\MM_2^{j,s' }$ belong to the same cluster if they can be connected via some path of subsequently
adjacent sets.   This defines an equivalence relation. Thus, strongly resonant sets $\MM_2^{j,s}$  form clusters (equivalence classes).
The cluster, containing a strongly resonant set $\MM_2^{j,s}$,  we  denote by $\MM_{2,str}^{j,s}$. Obviously, $\MM_{2,str}^{j,s}=\MM_{2,str}^{j,s'}$ when
$\MM_2^{j,s}$ and $\MM_2^{j,s'}$ are from the same cluster.
\begin{lemma}\label{per3}
A cluster $\MM_{2,str}^{j,s}$ of strongly resonant  sets  contains no more than
two sets $\MM_2^{j,s}$.
\end{lemma}
\begin{proof}  By the definition of $\MM _2^j$, all $\MM_2^{j,s}$ in a cluster belong to the same $\MM _2^j$.
First, assume $\MM _2^j$ is trivial. Then,  by definition, each $\MM_2^{j,s}$ consists of just one point $\m$ and, by the Definition \ref{Nov27-12},
$\left||\vec k (\varphi _0)+\p_\m|^2_{\R}-k^2\right|=O(k^{-40\mu \delta })$.  Suppose there are more than two $\MM_2^{j,s}$ in the cluster. Then, there are $\n,\n+\q _1
, \n+\q_2\in \MM _2^j $ such that $\q_1\neq \q_2$, $0<\||\p_{\q_1}\||, \||\p_{\q_2}\||<3k^{\delta }$, and
$$\left||\vec k (\varphi _0)+\p_\n|^2_{\R}-k^2\right|=O(k^{-40\mu \delta}), \ \ \left||\vec k (\varphi _0)+\p_{\n+\q_i}|^2_{\R}-k^2\right|=O(k^{-40\mu \delta}), \ i=1,2.$$
It follows that
\begin{equation}\left|2(\vec k(\varphi _0)+\p_\n,\p _{\q_i})+ p ^2_{\q _i}\right|=O(k^{-40\mu \delta}),\ \ \ i=1,2.\label{Nov8-12}
\end{equation}
By Lemma \ref{2.12}, $\q_1$ is parallel to $\q_2$. Now \eqref{Nov8-12} yields $p _{\q _1-\q _2}=|p _{\q _1}-p_{\q _2}|=O\left(k^{-38\mu \delta }\right)$.
 Using \eqref{below}, we obtain  $\q _1=\q_2$, since  $\||\p_{\q _1}\||, \||\p_{\q_2}\||<3k^{\delta }$. However, $\q_1\neq \q_2$ by construction.  Thus, the lemma is proven in the case of a trivial $\MM _2^j$.

Suppose   $\MM _2^j$ is not trivial and  $\vec p _\q$ is its directional vector, $\q \in \SS _Q$. Let $\MM_2^{j,s_i}$, $i=0,1,2$,  be any three neighboring elements in a $3k^{\delta}$-cluster. Let $\m _i$ be the central point of  $\MM_2^{j,s_i}$, $i=0,1,2$.  All $\MM_2^{j,s_i}$  are subsets of  periodic one-dimensional lattices
with the same period $\p _{\q}$.
By
Corollary \ref{CL:4}, $\p_{\m _0-\m _i}=w_i\p_\q$, $w_i\in \R$, $i=1,2$. Since $\m_i$ are central points, we have $|w_i|<1$, see \eqref{Mjs}.
Let us show that at least one of the numbers $w_1,w_2, w_1-w_2$ is an integer.  Suppose it is not so.  By Definition \ref{Nov27-12} and elementary perturbation consideration, each operator $H^{j,s_i}(\vec k (\varphi _0))$ has an eigenvalue $\lambda _{n_i}^{( s_i)}(\varphi _0)$ such that $\lambda _{n_i}^{( s_i)}(\varphi _0)=k^2+
O(k^{-40\mu \delta})$. Using \eqref{May7b}, we obtain:
\begin{equation}\tilde \lambda_{n_i}(t_\q^{(s_i)})+(t_\q^\perp)^2=k^2+
O(k^{-40\mu \delta}), \label{3per} \end{equation}
where $\tilde \lambda_{n_i}(t_\q^{(s_i)})$ is an eigenvalue of $\tilde H^{j,s_i}$.
By
\eqref{5per},  these eigenvalues of $\tilde H^{j,s_i}$, $i=0,1,2$, can be
approximated by the corresponding eigenvalues of $\tilde H^{j,s}_{per}$
with the high accuracy $O(k^{-\frac{\delta_*}{C(Q)}k^{\delta_*/2}})$.
It follows  from \eqref{3per} that
\begin{equation}\label{6per}
\begin{split}
&|\tilde\lambda_{n_0}^{per}(t_\q^{(s_0)})-\tilde\lambda_{n_1}^{per}(t_\q^{(s_1)})|=O(k^{-40\mu \delta}),\cr &
\ \ \  |\tilde\lambda_{n_0}^{per}(t_\q^{(s_0)})-\tilde\lambda_{n_2}^{per}(t_\q^{(s_2)})|=O(k^{-40\mu \delta}).
\end{split}
\end{equation}
All $\tilde\lambda_{n_i}^{per}(t_\q^{(s_i)})$, $i=0,1,2$, are eigenvalues of the same
fiber matrix $\tilde H^{j}_{per}$, but at different values of the
quasi-momentum $t_\q^{(s_i)}$ (see \eqref{Nov7-12}).
Recall that the corresponding periodic operator does not depend on
$k$. Properties of the one-dimensional periodic
Schr\"odinger operators (i.e. non-degeneracy of the band functions) and \eqref{6per},  imply that for sufficiently large $k$, depending on $V$ only:
$$
|(t_\q^{(s_0)}\pm
t_\q^{(s_i)})_{mod\,p_\q}|=O(k^{-20\mu \delta}),\ \ i=1,2.
$$
Substituting the expression for $t_\q$ (see \eqref{Nov7-12}) into the above estimate and considering that $\p_{\m _0-\m _i}=w_i\p_\q$, $i=1,2$, yields that
either $|w_i-n_i|=O(k^{-20\mu \delta})$ or $|2t_\q^{(s_0)}
p_\q^{-1}-w_i-n_i|=O(k^{-20\mu \delta})$ for some $n_i \in \Z$, $i=1,2$. Suppose the first relation holds for $i=1$. By properties of $w_1$,
$p_{\m_0-\m_1-n_1\q}=O(k^{-20\mu \delta})$.  Let us estimate $ \||\p_{\m_0-\m_1-n_1\q}\||$ from above.  Since $\m_0$ and $\m_1$ are central points of adjacent sets $\MM_2^{j,s_i}$, we have $\||\p_{\m_0-\m_1}\||<3C(Q)k^{\delta }$, see Appendix 3.
Since $|w_1|<1$, and $|w_1-n_1|=O(k^{-20\mu \delta})$, we have
$|n_1|<2$. It follows that $\|| \p_{\m_0-\m_1-n_1\q}\||=O(k^{\delta })$. Therefore, by \eqref{below}, $p_{\m_0-\m_1-n_1\q}>k^{-2\mu \delta}$. This contradicts to the inequality
$p_{\m_0-\m_1-n_1\q}=O(k^{-20\mu \delta})$ obtained just above. Therefore,
$\m_0-\m_1-n_1\q={\bf 0}$,  that is $w_1=n_1$. This contradicts to the initial assumption that $w_1$  is not an integer. If $|w_2-n_2|=O(k^{-20\mu \delta})$, we arrive to the the analogous contradiction for $w_2$. It remains to assume that
 $|2t_\q^{(s_0)}
p_\q^{-1}-w_i-n_i|=O(k^{-20\mu \delta})$ for both $i=1,2$.   It follows $|
w_1-w_2-(n_2-n_1)|=O(k^{-20\mu \delta})$ and, hence, $p_{\m _1-\m _2-n\q}=O(k^{-20\mu \delta})$, $n=n_2-n_1$.  Let us estimate $\||\p_{\m _1-\m _2-n\q}\||$ from above.
 Since $\m_1$ and $\m_2$ are central points, we have $\||\p_{\m_1-\m_2}\||<6C(Q)k^{\delta }$, see Appendix 3. Since $|
w_1-w_2-(n_2-n_1)|=O(k^{-20\mu \delta})$, we have $|n|<3$.  It follows that $\|| \p_{\m_1-\m_2-n\q}\||=O(k^{\delta })$. Therefore, by \eqref{below}, $p_{\m_1-\m_2-n\q}>k^{-2\mu \delta}$. This contradicts to the inequality
$p_{\m_1-\m_2-n\q}=O(k^{-20\mu \delta})$ obtained just above.
Hence, we obtain $w_1-w_2=n$.
Thus, we proved  that at least two out of each three neighboring sets $\MM ^{j,s}_2$ in a cluster are shifted with respect to each other by $n\p_{\q}$, $n\in \Z$, and, thus, (see \eqref{Mjs})  they coincide. Moreover, we proved that if $w_i$  is not an integer, then it is equal to $2t_\q^{(s_0)}
p_\q^{-1}+n_i$, $n _i\in \Z$. 
\end{proof}

\subsubsection{Estimates for the Resolvent of the Model Operator}

We are going to obtain  estimates for the resolvent $\Big(P(r_1)\big(H(\k^{(1)}(\varphi))-k^2\big)P(r_1)\Big)^{-1}$. We start with the following lemma.
\begin{lemma}\label{single}
Let $|\varphi-\varphi_0|\leq k^{-2}$. If $\m$ does not belong to $\MM(\varphi_0)$ then
\begin{equation}\label{single1}
\left|(H_0(\vec k(\varphi))-k^2)_{\m\m}\right|\geq\frac12 k^{\delta_*},\ \ \ \ \
\left|(H_0(\k^{(1)}(\varphi))-k^2)_{\m\m}\right|\geq\frac12 k^{\delta_*}.
\end{equation}
\end{lemma}
\begin{proof} The proof follows from the definition of $\MM(\varphi_0)$ (see \eqref{M}, \eqref{resonance1}), \eqref{dk0} and simple perturbative arguments.
\end{proof}

Let $P_\m$ be the orthogonal projector corresponding to the set  $\tilde\MM_\m$, see \eqref{51a}. We consider $\m\in\MM_1$. The following lemma is an analogue of Lemma 3.14 from \cite{KaSh} (statements 1-3). The proof is very similar, however, we present it here for completeness. The statement 4 of Lemma 3.14 from \cite{KaSh} is now quite different and will require additional arguments which we present below. It is the main place where we have the difference between the Schr\"odinger operator and polyharmonic operators considered in \cite{KaSh}.

\begin{lemma}\label{L:estnonres1} Let $\varphi _0\in \omega ^{(1)}(k,8)$.
\begin{enumerate}
\item If $\m\in \M_1(\varphi _0): p_\m>4k^{\delta}, |2k-p_\m| \geq 1$, then,
the operator $$\left(P_\m\left(H\big( \k ^{(1)}(\varphi
)\big)-k^{2}I\right)P_\m\right)^{-1}$$ has no more than one pole in
the disk $|\varphi -\varphi _0|<2k^{-2-\delta (40\mu +1) }$. The
following estimate holds: \begin{equation} \label{Mon3-1}
\left\|\left(P_\m\left(H\big(\k ^{(1)}(\varphi
)\big)-k^{2}I\right)P_\m\right)^{-1}\right\|<ck^{-1}\varepsilon
_0 ^{-1}, \  \ \varepsilon _0=\min\{\varepsilon , k^{-2-(40\mu +1)
\delta}\}, \end{equation} when $\varphi $ is in the smaller disk
$|\varphi -\varphi _0|<k^{-2-\delta (40\mu +1)}$, $\varepsilon $
being the distance from $\varphi $ to the nearest pole of the
operator.
\item If $\m\in \M: |2k-p_\m| <1$, then, in fact, $\m\in \M_1$ and the operator $$\left(P_\m\left(H\big( \k ^{(1)}(\varphi
)\big)-k^{2}I\right)P_\m\right)^{-1}$$ has no more than two poles
in the disk $|\varphi -\varphi _0|<2k^{-2-\delta (40\mu +1)}$. The
following estimate holds: \begin{equation} \label{Mon3-3}
\left\|\left(P_\m\left(H\big(\k ^{(1)}(\varphi
)\big)-k^{2}I\right)P_\m\right)^{-1}\right\|<ck^{-2}\varepsilon
_0^{-2},\ \ \varepsilon _0=\min\{\varepsilon, k^{-2-\delta(40\mu
+1)}\},
\end{equation}  when $\varphi $ is in the smaller disk $|\varphi -\varphi _0|<k^{-2-\delta (40\mu
+1)}$, $\varepsilon $ being the distance from $\varphi $ to the
nearest pole of the operator.
\item If $\m\in \M: p_\m<4k^{\delta}$, then, in fact, $\m\in \M_1$ and the operator
$$\left(P_\m\left(H\big( \k ^{(1)}(\varphi
)\big)-k^{2}I\right)P_\m\right)^{-1}$$ has no more than one pole in
the disk $|\varphi -\varphi _0|<2k^{-2-\delta (40\mu +1)}$. The
following estimate holds: \begin{equation} \label{Mon3-2}
\left\|\left(P_\m\left(H\big(\k ^{(1)}(\varphi
)\big)-k^{2}I\right)P_\m\right)^{-1}\right\|<
8k^{-1}p_\m^{-1}\varepsilon _0 ^{-1},\  \  \varepsilon _0=\min
\{\varepsilon , k^{-1 +\delta }\},
\end{equation} when $\varphi $ is in the smaller disk $|\varphi -\varphi _0|<k^{-2-\delta (40\mu
+1)}$, $\varepsilon $ being the distance from $\varphi $ to the
nearest pole of the operator.
\end{enumerate}
\end{lemma}
\begin{proof}\begin{enumerate} \item Let $|2k-p_\m|\geq 1$,
$p_\m>4k^{\delta }$.
 Clearly only
the case $p_\m<4k$ is significant, since otherwise $\OO _{\m}(k,1)$
cannot intersect the disc $|\varphi -\varphi _0|<k^{-2-40(\mu
+1)\delta }$ by Lemma \ref{L:3.1}. It is easy to see that the set
$\OO _{\m}(k,1)$ consists of two separate discs $\OO _{\m}^{\pm}(k,1)$, the
distance between them being greater than $ck^{-1/2}$. Let us assume for definiteness $\varphi _0\in \OO _{\m}^{+}(k,1)$. This means  the disc  $|\varphi -\varphi
_0|<k^{-2-\delta (40\mu +1)}$ does not intersect  $\OO _{\m}^-(k,1)$. Let us first show that the operator
\begin{equation}\left(P_\m\left(H _0\big(\k ^{(1)}(\varphi
)\big)-k^{2}I\right)P_\m\right)^{-1} \label{freeres} \end{equation}
has exactly one pole inside $\OO _{\m}^+(k,1)$, which is, in fact,
inside $\OO _{\m}^+(k,1/4)$.  Note that $\k ^{(1)}(\varphi )$ is
defined in $\OO_\m^+(k,1)$, since the size of $\OO_\m^+(k,1)$ is
much less than that of any circle in $\OO^{(1)} $. It satisfies the
estimate $\k ^{(1)}(\varphi )= \vec k (\varphi )+o(k^{-2})$ in $\OO
_\m ^+$, see \eqref{dk0}. If $\varphi _0\in \OO_\m ^+(k,1)\setminus \OO_\m^+(k,1/4)$,
then the estimates $\left||\k ^{(1)}(\varphi _0
)+\p_{\m+\q}|^2-k^2\right|>\frac{1}{4} k^{\delta_* }$, hold for
$0\leq \||\p_\q\||<k^{\delta }$ (see the definition of $\MM_1(\varphi
_0)$) and these estimates can be extended to the $(k^{-2-(40\mu +1) \delta
})$-neighborhood of $\varphi _0$ ($\frac14$ becomes $\frac18$).
Thus,
\begin{equation} \label{Mon3a*} \left\|\left(P_\m\left(H_0\big(\k
^{(1)}(\varphi
)\big)-k^{2}I\right)P_\m\right)^{-1}\right\|<ck^{-\delta_*}
\end{equation}
$$\mbox{when}\ \ |\varphi -\varphi _0|<k^{-2-(40\mu +1)\delta },\ \
\varphi _0\in \OO_\m ^+(k,1)\setminus \OO_\m^+(k,1/4).$$
Clearly the resolvent \eqref{Mon3a*} does not have poles in the set
$|\varphi -\varphi _0|<k^{-2-(40\mu +1)\delta }$. The estimate
(\ref{Mon3-1})  with $\varepsilon _0=k^{-2-(40\mu +1)\delta}$ follows
from (\ref{Mon3a*}) and Hilbert identity.

 Now, suppose that $\varphi _0\in \OO _\m^+(k,\frac{1}{4})$.
The function $| \vec k (\varphi )+\p_\m|^{2}_\R-k^{2}$ has a
single zero inside $\OO_\m^+(k,\frac{1}{4})$. Using Rouch\'{e}'s
Theorem, we obtain that $|\k ^{(1)}(\varphi )+\p_\m|^{2}_\R-k^{2}$
also has a single zero inside $\OO_\m^+(k,\frac{1}{4})$. Note that
the following inequality holds in $\OO_\m^+(k,\frac{1}{4})$ for
$0<\||\p_\q\||<k^{\delta }$: $$\left||\k ^{(1)}(\varphi
)+\p_{\m+\q}|^{2}_\R-k^{2}\right|>\frac{1}{4}k^{\delta_*
}.$$ Indeed, if $\left||\k ^{(1)}(\varphi
)+\p_{\m+\q}|^{2}_\R-k^{2}\right|\leq
\frac{1}{4}k^{\delta_* }$ for some $\q\neq (0,0)$ and
$\varphi \in \OO_\m^+(k,\frac{1}{4})$, then
\begin{equation} \left|2(\k ^{(1)}(\varphi
)+\p_{\m},\p_\q)_\R+p_\q^2\right|<\frac{1}{2}k^{\delta_*
}.\label{Mon21}\end{equation}
Considering that the size of $\OO_\m^+(k,\frac{1}{4})$ is
$\frac{k^{-1+\delta_* }}{p_\m \sqrt{1-p_\m^2(2k)^{-2}}}(1+o(1))$
and that $p_\m>4k^{\delta}>4p_\q/2\pi$, we obtain the inequality
analogous to (\ref{Mon21}) for $\varphi _0$ with $\frac{3}{4}$
instead of $\frac{1}{2}$. This contradicts to the assumption
$\m \in \MM _1(\varphi _0)$. Thus,
 the following inequality holds for all
$\q:\||\p_\q\||<k^{\delta }$ including $\q=(0,0)$: $$\left||\k
^{(1)}(\varphi )+
\p_{\m+\q}|^{2}_\R-k^{2}\right|\geq\frac{1}{4}k^{\delta_*
},$$ when $\varphi $ is on the boundary of
$\OO_\m^+(k,\frac{1}{4})$. Hence, the resolvent $$\left(P_\m\left(H
_0\big(\k ^{(1)}(\varphi )\big)-k^{2}I\right)P_\m\right)^{-1}$$ of
the free operator $P_\m H_0$ has exactly one pole inside $\OO_\m
^+(k,\frac{1}{4})$ and \begin{equation} \left\|\left(P_\m\left(H
_0\big(\k ^{(1)}(\varphi )\big)-k^{2}I\right)P_\m\right)^{-1}\right
\|\leq 4k^{-\delta_* }, \label{Mon5} \end{equation} when
$\varphi $ is on the boundary on the disc $\OO _\m
^+(k,\frac{1}{4})$. Considering that the dimension of $P_\m$ does
not exceed $16k^{4\delta }$ we obtain:
\begin{equation} \left\|\left(P_\m\left(H _0\big(\k ^{(1)}(\varphi
)\big)-k^{2}I\right)P_\m\right)^{-1}\right \|_1\leq
64k^{-\delta_*+4\delta }. \label{Mon5a} \end{equation} It
remains to prove the analogous result for the perturbed operator
$H$.  We introduce the determinant $$D(\varphi )=\det
\left(P_\m\left(H \big(\k ^{(1)}(\varphi
)\big)-k^{2}I\right)P_\m\left(H _0\big(\k ^{(1)}(\varphi
)\big)-k^{2}I\right)^{-1}P_\m \right).$$ Obviously, $D(\varphi
)=\det(I+A)$, where $I,A:P_\m L_2(\Z^2)\to P_\m L_2(\Z^2)$
$$A(\varphi )=P_\m V\left(H _0\big(\k ^{(1)}(\varphi
)\big)-k^{2}I\right)^{-1}P_\m .$$ Taking into account that
$$D(\varphi )=\frac{\det \left(P_\m \left(H \big(\k ^{(1)}(\varphi
)\big)-k^{2}I\right)P_\m \right)}{\det \left(P_\m \left(H_0\big(\k
^{(1)}(\varphi )\big)-k^{2}I\right)P_\m \right)},$$ we see that
$D(\varphi )$ is a meromorphic function inside $\OO_\m
^+(k,\frac{1}{4})$. Next, we employ a well-known inequality for the
determinants, see \cite{RS}: \begin{equation} \left|\det (I+A)-\det
(I+B)\right|\leq \|A-B\|_1exp (\|A\|_1+\|B\|_1+1),\ \ A,B\in \bf
S_1. \label{determinants}\end{equation}  Putting $A=A(\varphi )$,
$B=0$, we obtain \begin{equation}\label{determinants1}\left|\det (I+A)-1\right|\leq \|A\|_1exp
(\|A\|_1+1).\end{equation} It is easy to see that $$\|A\|_1\leq
\|V\|\|P_\m\left(H_0 \big(\k ^{(1)}(\varphi
)\big)-k^{2}I\right)^{-1}P_\m\|_1.$$ Considering the estimate
(\ref{Mon5a}) for the resolvent of the free operator, we obtain
$\|A_1(\varphi )\|_1<1/200$ on the boundary of
$\OO^+_\m(k,\frac{1}{4})$ for sufficiently large $k$. By
Rouch\'{e}'s Theorem, $D(\varphi )$ has only one zero in $\OO_\m^+
(k,\frac{1}{4})$. Thus, $\det P_\m\left(H \big(\k ^{(1)}(\varphi
)\big)-k^{2}I\right)P_\m$ has exactly one zero in
$\OO^+_\m(k,\frac{1}{4})$. Using this, we immediately obtain that
operator $\left(P_\m\left(H \big(\k ^{(1)}(\varphi
)\big)-k^{2}I\right) P_\m\right)^{-1}$ has one pole inside
$\OO_\m^+(k,\frac{1}{4})$. Considering the estimate for the free
resolvent and using Hilbert
 identity, we immediately obtain,
\begin{equation}\left\|\left(P_\m\left(H (\k ^{(1)}(\varphi
))-k^{2}I\right) P_\m \right)^{-1}\right\|\leq 8k^{-\delta_* } \label{Mon9} \end{equation} for all $\varphi $ on the
boundary of $\OO_\m^+(k,\frac{1}{4})$. Taking into account that the
size of $\OO^+_\m(k,\frac{1}{4})$ does not exceed $k^{-1+\delta_*
}$, we obtain:
\begin{equation} \left\|\left(P_\m\left(H \big(\k ^{(1)}(\varphi
)\big)-k^{2}I\right)P_\m \right)^{-1}\right\|\leq 8k^{-
\delta_* }(k^{-1+\delta_* }/\varepsilon
)\label{Jan10}\end{equation} when $\varphi \in
\OO_\m^+(k,\frac{1}{4})$  on the distance $\varepsilon $,
 from the pole. If $|\varphi -\varphi _0|<k^{-2-(40\mu +1)\delta }$, but $\varphi \not \in \OO_\m^+(k,\frac{1}{4})$,
 then $\varphi $ is on the distance less than $k^{-2-(40\mu +1)\delta }$ from the boundary of $\OO_\m^+(k,\frac{1}{4})$, since $\varphi _0$ is inside
  $\OO_\m^+(k,\frac{1}{4})$. The estimate  \eqref{Mon9} holds on the boundary and stable with respect to such a
  small perturbation of $\varphi $. Thus, estimate \eqref{Mon3-1} is proven.


\item Let $\m\in \MM$, $|2k-p_\m|<1$. Then,
$\OO_\m^+$ and $\OO_\m^-$ can overlap. The case $\varphi _0\in
\OO_\m (k,1)\setminus \OO_\m(k,1/4)$ we consider in the same way as
for $|2k-p_\m|\geq 1$. Suppose $\varphi _0\in \OO_\m(k,1/4)$.
Combining $\left|\bigl|\vec k(\varphi _0
)+\p_\m\bigr|_\R^2-k^2\right|<\frac{1}{4}k^{\delta_* }$ with
$|2k-p_\m|<1$, we obtain that the vectors $2\vec k(\varphi )$ and
$-\p_\m $ are close: $$|2\vec k(\varphi _0)+\p_\m|_\R^2<5k.$$
Therefore, $( \vec k(\varphi _0)+\p_\m , \p_\q)_\R =-( \vec
k(\varphi _0), \p_\q)_\R +O(k^{1/2+\delta })$ for all
$0<\||\p_\q|\|<k^{\delta }$. Considering that the size of $\OO_\m ^{\pm }$
does not exceed $ck^{-1+\delta_*/2 }$ (Lemma \ref{L:3.1}) and the
distance between $\OO_\m^+$ and $\OO_\m^-$ is $O(k^{-1/2})$, we
obtain the analogous estimate for all $\varphi $ in $\OO_\m$:
$$(\vec k(\varphi )+\p_\m , \p_\q )_\R=-(\vec k(\varphi ),
\p_\q)_\R+O(k^{1/2+\delta })\ \ \mbox{ for all
  }0<\||\p_\q\||<k^{\delta }.$$ It immediately follows:
$$|\vec k(\varphi )+\p_{\m+\q}|^2_\R-|\vec k(\varphi )+\p_{\m}|^2_\R=
|\vec k(\varphi )-\p_{\q}|^2_\R-|\vec k(\varphi )|^2_\R
+O(k^{1/2+\delta })$$ for all $0<\||\p_\q\||<k^{\delta }$. The size
of $\OO_\m$ is much smaller than that of $\OO_{-\q}$ and
  $\OO_\m$ is not completely in $\OO_{-\q}$. Hence,
  $$\left||\vec k(\varphi )+\p_{\m+\q}|^2_\R-|\vec k(\varphi
  )+\p_{\m}|^2_\R\right|>\frac{1}{2}k^{1-40\mu \delta }$$
  for all $\varphi \in \OO_\m$. Considering that $\left||\vec k(\varphi
  )+\p_{\m}|^2_\R-k^2\right|\leq
  \frac{1}{4}k^{\delta_* }$ in $\OO_\m$, we obtain
  $$\left||\vec k(\varphi )+\p_{\m+\q}|^2_\R-k^2\right|>\frac{1}{4}k^{1-40\mu \delta }, \ \ \mbox{when}\ \varphi \in \OO_\m \ \ \mbox{ for all
  }0<\||\p_\q\||<k^{\delta }.$$
  In particular, this means $\m\in \MM_1(\varphi _0)$, since $|\varphi -\varphi _0|<k^{-2-\delta (40\mu +1)}$. Considering that
  $\left||\vec k(\varphi
  )+\p_{\m}|^2_\R-k^2\right|=
  \frac{1}{4}k^{\delta_* }$ on the boundary of $\OO_\m$, we
  obtain \eqref{Mon5} and \eqref{Mon5a}.
Considering as before, we show that the resolvent (\ref{freeres})
has at most two poles inside $\OO_\m $. It follows that $$
\left\|\left(P_\m\left(H \big(\k ^{(1)}(\varphi
)\big)-k^{2}I\right)P_\m \right)^{-1}\right\|\leq k^{-
\delta_* }(ck^{-1+\delta_*/2 }/\varepsilon )^2$$ when $\varphi $ is
on the distance $\varepsilon $ from the pole.


\item Let $\m\in \MM$, $0<p_\m\leq 4k^{\delta },\ \m\not \in \Omega
(\delta)$. The case $\varphi _0 \in {\cal O_\m}(k,1)\setminus {\cal
O_\m}(k,\frac{1}{4})$ is considered the same way as in the previous
steps, see (\ref{Mon3a*}).  From now on we assume $\varphi _0 \in
{\cal O}_{\m}(k,\frac{1}{4})$. There is an eigenvalue $\lambda^{(1)}
\big(\k^{(1)}(\varphi )+\p_\m\big)$ of $P_\m H(\k^{(1)}(\varphi
))P_\m$ given by the perturbation series. Indeed,
$$\left|\bigl|\vec k(\varphi
_0)+\p_\m\bigr|_\R^2-k^2\right|<\frac{1}{4}k^{\delta_* },$$
since $\varphi _0 \in {\cal O}_{\m}(k,\frac{1}{4})$. Considering
that $\varphi _0 \not \in {\cal O}^{(1)}(k,8)$, we easily obtain
that $\left|( \vec k (\varphi _0 ),\p_{\q})_\R\right|\gtrsim
k^{1-40\mu \delta }$ for all $\q\in \Omega (\delta )\setminus
\{0\}$. Taking into account that and $p_\m\leq 4k^{\delta }$ we
arrive at the estimate:
$$\left|\bigl|\vec k (\varphi _0 )+\p_{\m+\q}\bigr|_\R^2-k^{2}\right|\gtrsim
k^{1-40\mu \delta }$$ for all $\q\in \Omega (\delta )\setminus
\{0\}$ and any $\varphi _0\in \omega^{(1)}(k,8)\cap {\cal
O_\m}(k,\frac{1}{4})$. It follows  $\m\in \MM_1$. By Lemma \ref{ldk},
$\k^{(1)}(\varphi )$ is defined in $
\frac{1}{4}k^{-(40\mu+1)\delta}$-neighborhood of $\omega^{(1)}(k,8)$, which
we denote by $\tilde{\cal W}^{(1)}(k,\frac14)$.
It is easy
to show that the estimates similar to the last two hold for
$\k^{(1)}(\varphi )$, $\varphi \in \tilde{\cal W}^{(1)}(k,\frac14))\cap
{\cal O_\m}(k,\frac{1}{2})$ . Therefore,
\begin{equation}\label{metka1}\left|\bigl|\k^{(1)}(\varphi )+\p_\m\bigr|_\R^{2}-k^{2}\right|\leq
\frac{1}{2}k^{\delta_* },\end{equation} \begin{equation} \label{metka2}\left|\bigl|\k^{(1)}(\varphi
)+\p_{\m+\q}\bigr|^{2}-k^{2}\right|\geq \frac 12 k^{1-40\mu \delta
}\end{equation} for all $\q\in \Omega (\delta )\setminus \{0\}$. It follows from
the last two estimates, that the perturbation series for
$\lambda^{(1)} \big(\k^{(1)}(\varphi )+\p_\m\big)$ and
$\lambda^{(1)} \big(\k^{(1)}(\varphi )\big)$ converge. Both are holomorphic functions of $\varphi $ in $\tilde{\cal W}^{(1)}(k,\frac14)\cap
{\cal O_\m}(k,\frac{1}{2})$. Using
Rouch\'{e}'s Theorem, it is not difficult to show (for details see
Appendix 4, Lemma~\ref{L: Appendix1}) that the equation
\begin{equation}\lambda^{(1)} \big(\k^{(1)}(\varphi
)+\p_\m\big)=k^{2}+\varepsilon _0,\ \ \ |\varepsilon _0|\leq p_\m
k^{\delta },\label{25} \end{equation} has no more than two solutions
$\varphi^\pm(\varepsilon_0) $ in the $\tilde{\cal W}^{(1)}(k,\frac18)\cap
{\cal O}_{\m}(k,\frac{1}{2})$. They satisfy the estimates:
 \begin{equation} \label{87a} \big|\varphi^{\pm }(\varepsilon_0)-\varphi _\m ^{\pm
}\big|<4k^{-1+2\delta }.\end{equation} Considering that $\varphi _\m ^{\pm
}=\varphi_\m\pm \pi/2+ O(k^{-1+\delta })$, we see that  the distance
between two solutions is approximately equal to $\pi $. For any
$\varphi \in \tilde{\cal W}^{(1)}(k,\frac14)\cap {\cal
O}_{\m}(k,\frac{1}{2})$ satisfying the estimate $\big|\varphi
-\varphi _\m ^{\pm }\big|<k^{-\delta },$
\begin{equation}\frac{\partial }{\partial \varphi }\lambda^{(1)}
\big(\k^{(1)}(\varphi )+\p_\m\big)=\pm 2p_\m k(1+o(1)),
\label{26}
\end{equation} for details see Appendix 4, Lemma~\ref{L:Appendix 2}. Therefore (for details
see Appendix 4, Lemma~\ref{L:3.7.1}),
\begin{equation}\big|\lambda^{(1)} \big(\k^{(1)}(\varphi
)+\p_\m\big)-k^{2}\big|\geq k p_\m\varepsilon
\label{Mon10}\end{equation} if $\varphi \in \tilde{\cal
W}^{(1)}(k,\frac18)\cap {\cal O}_{\m}(k,\frac{1}{2})$ is outside $\tilde
{\cal O}_{\m,\varepsilon}^+ \cup \tilde {\cal
O}_{\m,\varepsilon}^-$, here and below $\tilde {\cal
O}_{\m,\varepsilon }^{\pm }$ are the open discs of the radius
$\varepsilon $, $0<\varepsilon <k^{-1+\delta}$ , centered at
$\varphi^\pm (0)$. It is shown in Appendix 4, Lemma~\ref{L:July5}
that
\begin{equation} \left \|(\lambda ^{(1)}(\y(\varphi
))-k^{2})\left(P_\m(H(\k^{(1)}(\varphi))-k^{2})P_\m\right)^{-1}\right\|\leq
8, \ \ \ \y(\varphi ):= \k^{(1)}(\varphi )+\p _\m, \label{July3a'''}
\end{equation}
for any $\varphi $ in $\tilde{\cal
W}^{(1)}(k,\frac18)\cap {\cal O}_{\m}(k,\frac{1}{2})$.

If $|\varphi -\varphi _0|<2k^{-2-\delta (40\mu +1)}$ and $\varphi _0\in \omega (k,\delta ,8)\cap {\cal O}_{\m}(k,\frac{1}{4})$, then $\varphi \in \tilde{\cal
W}^{(1)}(k,\frac18)\cap {\cal O}_{\m}(k,\frac{1}{2})$ and, hence, \eqref{Mon10}, \eqref{July3a'''} hold. Now \eqref{Mon3-2} easily follows from \eqref{Mon10} and \eqref{July3a'''}.
\end{enumerate}
\end{proof}

Next, we consider the resolvent reduced to a weakly resonant $\MM_2^{j,s}$.
\begin{lemma}\label{weak}
Let $\MM_2^{j,s}$ be weakly resonant and let $|\varphi-\varphi_0|\leq \frac32 k^{-2-40\mu\delta}$. Then
\begin{equation}\label{weak1}
\|(H^{j,s}(\k^{(1)}(\varphi))-k^2)^{-1}\|\leq k^{214\mu\delta}.
\end{equation}
\end{lemma}
\begin{proof}
By \eqref{dk0} and perturbative arguments it is enough to show that
\begin{equation}\label{weak1'}
\|(H^{j,s}(\vec k(\varphi))-k^2)^{-1}\|=O(k^{213\mu\delta}).
\end{equation}
 Let $|\Im \varphi |>k^{-2-128\mu \delta }$, then \eqref{weak1'} follows from \eqref{1per}. Suppose $|\Im \varphi |\leq k^{-2-128\mu\delta }$.
 Let  $\varphi _*=\Re \varphi $. Assume first $\|(H^{j,s}(\vec k(\varphi _*))-k^2)^{-1}\|<k^{127\mu \delta }$. Using perturbative arguments, we obtain
 $\|(H^{j,s}(\vec k(\varphi))-k^2)^{-1}\|<2 k^{127\mu \delta }$ and, hence, \eqref{weak1'} holds. It remains to assume
 $\|(H^{j,s}(\vec k(\varphi _*))-k^2)^{-1}\|\geq k^{127\mu \delta }$.
 Considering \eqref{1per*}, we obtain that there is a real pole $\varphi _{**}$ of the resolvent in the $k^{-2-42\mu \delta }$-neighborhood of $\varphi _*$. It follows that $\varphi $ is in the
 $k^{-2-41\mu\delta}$-neighborhood of a pole $\varphi _{**}$ of the resolvent. This contradicts to the definition of the weakly resonant $\MM^{j,s}_2$. Thus, \eqref{weak1'} and, hence, \eqref{weak1} are proven.
\end{proof}
Now, we discuss the part of the resolvent corresponding to strongly resonant blocks. Let us consider a trivial $\MM_2^{j}$ and a cluster $\MM_{2,str}^{j,s}$ of strongly resonant
sets, see Definition \ref{Nov27-12} and the text after. We denote by $\tilde  \MM_{2,str}^{j,s}$ the $\frac13k^\delta$-neighborhood (in $\||\cdot\||$-norm) of  $\MM_{2,str}^{j,s}$.
By definition of a cluster, such neighborhoods of different clusters are disjoint.

Let $P_{str}$ be the projection corresponding to  a  $\tilde \MM_{2,str}^{j,s}$. We consider the operator $P_{str}H(\k^{(1)}(\varphi))P_{str}$. By $d_{str}(\n)$ we denote the $\||\cdot\||$-distance from a point $\n$ to the nearest strongly resonant $\MM^{j,s}_2$, i.e.,
$$
d_{str}(\n):=\min\limits_{\m}\||\p_{\n-\m}\||,
$$
where the minimum is over all $\m$ in strongly resonant $\MM^{j,s}_2$.

\begin{lemma}\label{trivial} Let $\MM_2^{j}$ be trivial.
Then, the  resolvent $(P_{str}(H(\k^{(1)}(\varphi))-k^2)P_{str})^{-1}$ has at most  $4$ poles in the disc $|\varphi-\varphi_0|\leq k^{-2-40\mu\delta}$. The resolvent obeys the estimate
\begin{equation}\label{trivial1}
\left\|\left(P_{str}\left(H\left(\k^{(1)}(\varphi)\right)-k^2\right)P_{str}\right)^{-1}\right\|\leq k^{127\mu\delta}\left(\frac{k^{-2-41\mu\delta}}{\varepsilon_0}\right)^4,
\end{equation}
where $\varepsilon_0=\min\{\varepsilon,\,k^{-2-41\mu\delta}\}$ with $\varepsilon$ being the distance from $\varphi$ to the nearest pole of the resolvent, $\k^{(1)}(\varphi)$ being defined in Section \ref{IS1}.

Moreover, the following estimate for the matrix elements holds:
\begin{equation}\label{trivial2}
\begin{split}
&\left |\left(P_{str}\left(H\left(\k^{(1)}(\varphi)\right)-k^2\right)P_{str}\right)^{-1}_{\m\m'}-\left(P_{str}\left(H_0\left(\k^{(1)}(\varphi)\right)-k^2\right)P_{str}\right)^{-1}_{\m\m'}\right|\leq \cr & \left(\frac{k^{-2-41\mu\delta}}{\varepsilon_0}\right)^4k^{-\delta_* d(\m,\m')}+k^{-\delta_*/9},
\end{split}
\end{equation}
where
\begin{equation}\label{trivial3}
d(\m,\m')=\frac{d_{str}(\m)+d_{str}(\m')}{20Q},\ \ \ \ \ \ \  d_{str}(\m)+d_{str}(\m')>20Q.
\end{equation}
\end{lemma}
\begin{proof}
First we consider the resolvent of the free operator: $(P_{str}(H_0(\vec k(\varphi))-k^2)P_{str})^{-1}$. Its poles are zeros of the function \eqref{Jan4-13} for $\m$ being
in the $\frac13 k^{\delta }$-neighborhood of a strongly resonant cluster (in particular, either $\m\in\MM_2^j$ or $\m\not\in\MM(\varphi_0)$).  It follows from Lemma~\ref{single} that \eqref{Jan4-13}  does not have zeros in $|\varphi-\varphi_0|\leq k^{-2-40\mu\delta}$
for $\m\not\in\MM^j_2$. Let $\m  \in \MM_2^{j}$.
Then, \eqref{Jan4-13}  has no more than two zeros. Using Lemma \ref{per3}, we obtain that  the resolvent of the free operator  has no more than four poles.

Let $\tilde\OO$ be the union of the discs of radius $k^{-2-41\mu\delta}$ around poles of $(P_{str}(H_0(\vec k(\varphi))-k^2)P_{str})^{-1}$ in the $2k^{-2-40\mu\delta}$-neighborhood of $\varphi_0$. Obviously, there is a circle $|\varphi-\varphi_0|= c k^{-2-40\mu\delta}$ with $1\leq c\leq\frac32$ which does not intersect $\tilde\OO$. Further we reduce our considerations to the disc bounded by this circle, since it contains $|\varphi-\varphi_0|<k^{-2-40\mu\delta}$ and any connected component of $\tilde\OO$ which intersect this disc is strongly inside it.

We  show that in the disc $|\varphi-\varphi_0|\leq c k^{-2-40\mu\delta}$ the resolvent of the perturbed operator $(P_{str}(H(\vec k(\varphi))-k^2)P_{str})^{-1}$ has no poles outside $\tilde\OO$ and it has the same number of poles in $\tilde\OO$ as the resolvent of the free operator. Indeed, let us consider  the perturbation series for the resolvent:
\begin{equation}\label{Jan12a-13}
\begin{split}
&\left(P_{str}\left(H(\vec k(\varphi))-k^2\right)P_{str}\right)^{-1}=\cr &\sum\limits_{r=0}^\infty (P_{str}(H_0(\vec k(\varphi))-k^2)P_{str})^{-1}\left(-V(P_{str}(H_0(\vec k(\varphi))-k^2)P_{str})^{-1}\right)^r.
\end{split}
\end{equation}
The series converges if $\|VA\|<1$, where \begin{equation}\label{trivial3a}
A(\vec k(\varphi)):=(P_{str}(H_0(\vec k(\varphi))-k^2)P_{str})^{-1}V(P_{str}(H_0(\vec k(\varphi))-k^2)P_{str})^{-1}.
\end{equation}
We prove that
\begin{equation}\label{trivial4}
\|A(\vec k(\varphi))\|\leq k^{132\mu\delta-\delta_*/2}<k^{-\delta_*/3},
\end{equation}
when $\varphi \not \in \tilde\OO$.
First, using \eqref{1per**}, we easily obtain that
\begin{equation} \label{Jan4b-13}\left||\vec k(\varphi)+\p_\m|^2-k^2\right|\geq k^{-126\mu\delta }\ \mbox{ when } \m \in \MM_2^{j},\ \varphi \not \in  \tilde\OO, \
|\varphi-\varphi_0|<c k^{-2-40\mu\delta}.\end{equation}
Combining \eqref{single1} and \eqref{Jan4b-13}, we obtain:
\begin{equation}
\|(P_{str}(H_0(\vec k(\varphi))-k^2)P_{str})^{-1}\|\leq k^{126\mu\delta} \mbox{ when }\ \varphi \not \in  \tilde\OO, \
|\varphi-\varphi_0|<c k^{-2-40\mu\delta}. \label{Jan4c-13}
\end{equation}
Next, let $\MM_2^{j}$ be such that corresponding direction $\q$ does not belong to ${\cal S}_Q$ (the first option from the definition of a trivial $\MM_2^{j}$).
Assume $\m,\,\m'\in\MM^j_2$. Note that $A_{\m\m'}$ can differ from zero only if $\m-\m'\in{\cal S}_Q\setminus\{{\bf 0}\}$. On the other hand $\m-\m'=w\q$, where $\q$ is the directional vector of $\MM^j_2$, $w\in \R$. This contradicts to assumption $\q\not\in{\cal S}_Q$. Thus, if $A_{\m\m'}\not=0$ we have either $\m\not\in\MM^j_2$ or $\m'\not\in\MM^j_2$. Using \eqref{single1}, we get
$$
|A_{\m\m'}|\leq 2k^{-\delta_*}\|V\|k^{126\mu\delta}.
$$
 It is easy to see that the number of $\m$ in a cluster of boxes does not exceed $ck^{4\delta }$, Therefore, \eqref{trivial4} holds.
Now, let $\q\in{\cal S}_Q$, but
\begin{equation}\label{trivial5}
|k^2-(\vec k+\p_\m,\vec\nu_\q^\perp)^2|\geq k^{\delta_*}/8,
\end{equation}
when $\m\in\MM^j_2$ (the second  option from Definition \ref{def} of a trivial $\MM_2^{j}$). Let us show that
\begin{equation}\label{trivial5a}
||\vec k+\p_\m|_\R^2-k^2|+||\vec k+\p_{\m'}|_\R^2-k^2|>k^{\delta_*/2-\delta}.
\end{equation}
Suppose the inequality does not hold.
As above, if $A_{\m\m'}\not=0$ then $\m-\m'=n_0\q$, $n_0\in\Z\setminus\{{0}\}$ and $|n_0|\||\p_\q\||\leq Q$. It follows:
\begin{equation}\label{trivial6}
|2(\vec k+\p_\m,n_0p_\q)+(n_0\p_\q)^2|\leq 2k^{\delta_*/2-\delta}.
\end{equation}
Considering \eqref{trivial5}  and the inequality opposite to \eqref{trivial5a}, we obtain
$$
|(\vec k+\p_\m,\vec\nu_\q)|^2>\frac{1}{16}k^{\delta_*},
$$
which contradicts \eqref{trivial6}. Hence, \eqref{trivial5a} holds when $A_{\m\m'}\not=0$. Using \eqref{Jan4b-13} and the inequality  \eqref{trivial5a}, we obtain \eqref{trivial4} for the second case of a trivial $\MM^j_2$.

Estimates \eqref{trivial4} and \eqref{Jan4c-13} yield that the series \eqref{Jan12a-13} converges and \begin{equation}
\|(P_{str}(H(\vec k(\varphi))-k^2)P_{str})^{-1}\|< 2k^{126\mu\delta},\  \mbox{ when }\ \varphi \not \in  \tilde\OO, \
|\varphi-\varphi_0|<c k^{-2-40\mu\delta}.  \label{Jan4d-13}
\end{equation}
Next, we show that $\left( P_{str}\left( H_0(\vec k(\varphi)) -k^{2}\right) P_{str}\right)^{-1}$ and $\left( P_{str}\left( H(\vec k(\varphi)) -k^{2}\right) P_{str}\right)^{-1}$ have the same number of poles (counting algebraic multiplicity) inside $\tilde \OO$. Indeed, we introduce $H _{\beta }=P_{str}(H_0+\beta V)P_{str}$, $0\leq \beta  \leq 1$.
The series for the resolvent converges on the boundary of $\tilde\OO$ uniformly in $\beta $.
Thus, the determinant $D_\beta:=\det
\left(H_{\beta }(\vec k(\varphi)
)-k^{2}I\right)$ is the polynomial in $\beta$ uniformly bounded from below on the boundary of $\tilde\OO$. Now, it follows from continuity that in each connected component of
$\tilde \OO$ the determinant $D_\beta$ has the same number of zeros  for all $0\leq\beta\leq1$.
Thus, $(P_{str}(H(\vec k(\varphi))-k^2)P_{str})^{-1}$ has the same number of poles in $\tilde\OO$ as the unperturbed operator $(P_{str}(H_0(\vec k(\varphi))-k^2)P_{str})^{-1}$, i.e. not more than $4$. We note that in what follows we will often use similar arguments without additional comments.

Now, we can apply \eqref{dk0} and pertubative arguments to \eqref{trivial4} --  \eqref{Jan4d-13} to obtain
\begin{equation}\label{trivial4'}
\|A((\k^{(1)}(\varphi))\|<k^{-\delta_*/4},\  \mbox{ when } \varphi \not \in \tilde\OO, \
|\varphi-\varphi_0|<c k^{-2-40\mu\delta},
\end{equation}
\begin{equation}
\left\|\left(P_{str}\left(H_0\left(\k^{(1)}(\varphi)\right)-k^2\right)P_{str}\right)^{-1}\right\|\leq 2k^{126\mu\delta}, \mbox{ when }\ \varphi \not \in  \tilde\OO, \
|\varphi-\varphi_0|<c k^{-2-40\mu\delta}, \label{Jan4c-13'}
\end{equation}
 \begin{equation}
\|(P_{str}(H(\k^{(1)}(\varphi))-k^2)P_{str})^{-1}\|< 4k^{126\mu\delta},\  \mbox{ when }\ \varphi \not \in  \tilde\OO, \
|\varphi-\varphi_0|<c k^{-2-40\mu\delta}.  \label{Jan4d-13'}
\end{equation}
To show that $\left( P_{str}\left( H(\k^{(1)}(\varphi)) -k^{2}\right) P_{str}\right)^{-1}$ and $\left( P_{str}\left( H(\vec k(\varphi)) -k^{2}\right) P_{str}\right)^{-1}$ have the same number of poles (counting algebraic multiplicity) inside $\tilde \OO$ we apply the arguments as above for the determinant $D_t:=\det
\left(H((1-t)\vec k(\varphi)+t\k^{(1)}(\varphi)
)-k^{2}I\right)$, $0\leq t\leq1$. Finally,
using the maximum principle in $\tilde \OO$ and \eqref{Jan4d-13'}, we get \eqref{trivial1}.

Next, we prove \eqref{trivial2}. We rewrite a matrix element of $(P_{str}(H(\k^{(1)}(\varphi))-k^2)P_{str})^{-1}$ in the form
\begin{equation} \label{Jan5c-13}
\begin{split}
& (P_{str}(H(\k^{(1)}(\varphi))-k^2)P_{str})^{-1}_{\m\m'}=\cr &\sum\limits_{r=0}^R\left((P_{str}(H_0(\k^{(1)}(\varphi))-k^2)P_{str})^{-1}\left(-V(P_{str}(H_0(\k^{(1)}(\varphi))-k^2)P_{str})^{-1}\right)^r\right)_{\m\m'}+\cr &
\left((P_{str}(H(\k^{(1)}(\varphi))-k^2)P_{str})^{-1}\left(-V(P_{str}(H_0(\k^{(1)}(\varphi))-k^2)P_{str})^{-1}\right)^{R+1}\right)_{\m\m'},
\end{split}
\end{equation}
where we choose $R:=\left[\frac{d_{str}(\m')}{Q}\right]-2$ (here, without the loss of generality we assume $d_{str}(\m')\geq d_{str}(\m)$).  Note that $\m'\not \in \MM^{j,s}_2$  (strongly resonant) by \eqref{trivial3}.
Let us consider the sum in the right-hand side. By induction, one can easily see that the sum in the right-hand side part is holomorphic in $\tilde\OO$ ($R:=\left[\frac{d_{str}(\m')}{Q}\right]-2$), since \eqref{V_q=0} is valid. If $\varphi \not \in \tilde \OO$, then we use \eqref{Jan4c-13'} for $\m,\m'$ and \eqref{trivial4'}  to show:
\begin{equation}
\begin{split}
\left|\left((P_{str}(H_0(\k^{(1)}(\varphi))-k^2)P_{str})^{-1}\left(V(P_{str}(H_0(\k^{(1)}(\varphi))-k^2)P_{str})^{-1}\right)^r\right)_{\m\m'}\right|<\cr 2k^{126\mu\delta -\delta _*r/8}, \ \ \ r\geq 1.
\end{split}\label{Jan5b-13}
\end{equation}
   Since the sum  is holomorphic in $\tilde\OO$ it can be estimated by the maximum principle inside $\tilde\OO$.  Using \eqref{trivial1}, \eqref{Jan4c-13'} and \eqref{trivial4'}, we obtain that the last term in the right-hand side of \eqref{Jan5c-13} is bounded by $\left(\frac{k^{-2-41\mu\delta}}{\varepsilon_0}\right)^4k^{-\delta_* d(\m,\m')}$. Now \eqref{trivial2} easily follows.

We also notice that from \eqref{Jan4c-13'} and the maximum principle one has the estimate (cf. \eqref{Jan5b-13}, \eqref{weak1})
\begin{equation}\label{trivial7}
|(P_{str}(H_0(\k^{(1)}(\varphi))-k^2)P_{str})^{-1}_{\m\m}|\leq 2k^{126\mu\delta},\ \ \ \mbox{when}\ d_{str}(\m)>0,
\end{equation}
uniformly in the disc $|\varphi-\varphi_0|\leq k^{-2-40\mu\delta}$.
\end{proof}

Now, we consider the case of a non-trivial $\MM^j_2$. Let $\MM^{j,s}_{2,str}$ be a strongly resonant cluster, see Definition \ref{Nov27-12} and the text after. We need to consider its $k^{\delta } $-neighborhood in $\||\cdot \||$-norm. According to Lemma  \ref{per3} and the definition of a cluster, such a neighborhood  contains no more than two strongly resonant sets $\MM^{j,s}_2$.
A slight technical complication is that such a neighborhood can intersect weak clusters. If  a weak cluster is completely inside the neighborhood, it is not a problem. If a weak cluster sticks out of the neighborhood we have to "attach" it to the neighborhood as a whole. Here are more technical details.  Let $\MM^{j,s}_{2,str}$ be  a strongly resonant cluster.
We consider its $\frac16 C(Q)^{-1}k^\delta$-neighborhood, where $C(Q)$ is as in Appendix 3 (without the loss of generality, $C(Q)\geq 1)$. We also consider a slightly bigger
$\frac15 C(Q)^{-1}k^\delta$-neighborhood of $\MM^{j,s}_{2,str}$. If a weakly resonant $\MM^{j,s''}_2$ intersects the bigger neighborhood, then   we attach the whole  $\MM^{j,s''}_2$  to the smaller $\frac16 C(Q)^{-1}k^\delta$-neighborhood of $\MM^{j,s}_{2,str}$.  We call this object the extended $\frac16 C(Q)^{-1}k^\delta$-neighborhood of $\MM^{j,s}_{2,str}$ and denote by
$\tilde \MM^{j,s}_{2,str}$. In other words, our extended neighborhood contains the ``body", which is $\frac16 C(Q)^{-1}k^\delta$-neighborhood of $\MM^{j,s}_{2,str}$ and the branches, which are all weakly resonant $\MM_2^{j,s''}$ intersecting the bigger $\frac15 C(Q)^{-1}k^\delta$-neighborhood of $\MM^{j,s}_{2,str}$. Thus, this branches can be, in fact, disjoint from the body. At the same time, such definition will be very convenient later as (by construction of $\MM_2^{j,s''}$) our extended neighborhood is not connected by potential $V$ with any new weakly resonant $\MM^{j,s''}_2$ (see below the construction of the model operator).   Note that any weak ``branch" $\MM_2^{j,s''}$ considered above can be included into the $\frac 13 k^\delta$-neighborhood of $\MM^{j,s}_{2,str}$, see Appendix 3. This means that the extended neighborhood $\tilde \MM^{j,s}_{2,str}$
belongs to the $\frac13 k^{\delta }$-neighborhood of $\MM^{j,s}_{2,str}$ , but, generally speaking, does not coincide with it. Note that extended neighborhoods  of any two different  clusters $\MM^{j,s}_{2,str}$ are disjoint. It follows from the definition of a cluster and the fact that each extended neighborhood belongs to the $\frac13 k^{\delta }$-neighborhood of $\MM^{j,s}_{2,str}$.

By definition, $\MM^{j,s}_{2,str}=\cup_s\MM^{j,s}_2$, the union is taken over all $\MM^{j,s}_2$  belonging to the cluster. Let $P_{j}$ be the projection corresponding to this union, while $P_{j,s}$ are the projections corresponding to sets $\MM^{j,s}_2$. We put $H^{j,s}:=P_{j,s}HP_{j,s}$, $\tilde H^j=\oplus_s H^{j,s}$ and $\hat H:=\tilde H^j+H_0(P_{str}-P_j)$, where $P_{str}$ is the projection corresponding to $\tilde \MM^{j,s}_{2,str}$.
Hence,
$$
P_{str}HP_{str}=\hat H+W,\ \ \ W:=P_{str}VP_{str}-\sum_s P_{j,s}VP_{j,s}.
$$

\begin{lemma}\label{nontrivial}  Let $\MM_2^{j}$ be non-trivial.
The resolvent $(P_{str}(H(\k^{(1)}(\varphi))-k^2)P_{str})^{-1}$ has at most $4$ poles in the disc $|\varphi-\varphi_0|\leq k^{-2-40\mu\delta}$. It obeys the estimate
\begin{equation}\label{nontrivial1}
\|(P_{str}(H(\k^{(1)}(\varphi))-k^2)P_{str})^{-1}\|\leq k^{215\mu\delta}\left(\frac{k^{-2-41\mu\delta}}{\varepsilon_0}\right)^4,
\end{equation}
where $\varepsilon_0=\min\{\varepsilon,\,k^{-2-41\mu\delta}\}$ with $\varepsilon$ being the distance from $\varphi$ to the nearest pole of the resolvent.

Moreover, the following estimate for the matrix elements holds:
\begin{equation}\label{nontrivial2}
\begin{split}&
\left|(P_{str}(H(\k^{(1)}(\varphi))-k^2)P_{str})^{-1}_{\m\m'}-(P_{str}(\hat H(\k^{(1)}(\varphi))-k^2)P_{str})^{-1}_{\m\m'}\right|\leq \cr & \left(\frac{k^{-2-41\mu\delta}}{\varepsilon_0}\right)^4k^{-\delta_* d(\m,\m')}+k^{-\delta_*/9},
\end{split}
\end{equation}
where
\begin{equation}\label{nontrivial3}
d(\m,\m')=\frac{d_{str}(\m)+d_{str}(\m')}{20C(Q)Q},\ \ \ \ \ \ \ d_{str}(\m),\ d_{str}(\m')>10C(Q)Q.
\end{equation}
Here $C(Q)$ is the constant from Appendix 3.
\end{lemma}
\begin{proof}
Let $A:=(\hat H-k^2)^{-1}W(\hat H-k^2)^{-1}$ and $\tilde\OO$ be the $k^{-2-41\mu\delta}$-neighborhood of the poles of $(\hat H(\vec k(\varphi))-k^2)^{-1}$ in the $2k^{-2-40\mu\delta}$-neighborhood of $\varphi_0$. Obviously, such poles are just poles of $(H^{j,s}(\vec k(\varphi))-k^2)^{-1}$ corresponding to strongly resonant $\MM^{j,s}_2$. In particular, we have no
more than $4$ poles and (cf. above) there exists $c\in[1,3/2]$ such that the circle $|\varphi-\varphi_0|= c k^{-2-40\mu\delta}$ does not intersect $\tilde\OO$. Considering as in the proof of Lemma~\ref{weak} we obtain the estimate
 \begin{equation}
 \|(\hat H(\vec k(\varphi))-k^2)^{-1}\|=O(k^{213\mu\delta}),\ \mbox{when}\ |\varphi-\varphi_0|\leq c k^{-2-40\mu\delta}\ \mbox{and }\ \varphi\not\in\tilde\OO. \label{Jan8-13}\end{equation}
We just note that condition $\varphi\not\in\tilde\OO$ ensures that the proof holds in the case of a strongly resonant $\MM_2^{j,s}$.

If $\m'\not\in\cup_s\MM^{j,s}_2$ then
$$
|((\hat H-k^2)^{-1})_{\m\m'}|=|\delta_{\m\m'}((H_0-k^2)^{-1})_{\m\m'}|\leq k^{-\delta_*}.
$$
Thus,
\begin{equation}\label{nontrivial4}
\|A(P_{str}-P_j)\| \leq \|V\|k^{214\mu\delta-\delta_*}\leq k^{-\delta_*/2}.
\end{equation}
Next,
$$
P_jAP_j=\sum\limits_{s,s'}(H^{j,s}-k^2)^{-1}P_{j,s}WP_{j,s'}(H^{j,s'}-k^2)^{-1}.
$$
Let us show that $P_jAP_j=0$. Indeed,
$$
P_{j,s}WP_{j,s'}=P_{j,s}VP_{j,s'}-P_{j,s}\left(\sum_{s''} P_{j,s''}VP_{j,s''}\right)P_{j,s'}.
$$
Let us show that the last expression is equal to zero. It is obvious for $s=s'$. If $s\not=s'$ we have $P_{j,s}WP_{j,s'}=P_{j,s}VP_{j,s'}$. Consider $V_{\m\m'}$ with $\m\in\MM^{j,s}_2$, $\m'\in\MM^{j,s'}_2$. We have $\m-\m'=w\q$. If $V_{\m\m'}\not=0$, then $\m-\m'\in{\cal S}_Q$ and hence, by the second condition on the potential, $w$ is integer. But this means $\MM^{j,s}_2=\MM^{j,s'}_2$. Hence, $P_{j,s}WP_{j,s'}=0$ for every $s,\,s'$. Therefore, $P_{j}AP_{j}=0$. Combining this with \eqref{nontrivial4}, we obtain the estimate analogous to \eqref{trivial4}.  Further, applying arguments from the proof of Lemma~\ref{trivial}  and \eqref{Jan8-13}, we obtain \eqref{nontrivial1}.

It remains to estimate the matrix elements. We have
\begin{equation}\label{decomp}
\begin{split}
& (P_{str}(H(\k^{(1)}(\varphi))-k^2)P_{str})^{-1}_{\m\m'}=\cr &\sum\limits_{r=0}^R\left((\hat H(\k^{(1)}(\varphi))-k^2)^{-1}\left(-W(\hat H(\k^{(1)}(\varphi))-k^2)^{-1}\right)^r\right)_{\m\m'}+\cr &
\left((P_{str}(H(\k^{(1)}(\varphi))-k^2)P_{str})^{-1}\left(-W(\hat H(\k^{(1)}(\varphi))-k^2)^{-1}\right)^{R+1}\right)_{\m\m'},
\end{split}
\end{equation}
where $R:=\left[\frac{d(\m')}{C(Q)Q}\right]-2$. Here, as before, without the loss of generality we assume that ${d}(\m')\geq{d}(\m)$.  Let us show that the sum in the right hand side part is holomorphic in $\tilde\OO$.  It is easy to see that this sum is a combination (including  products and sums) of the operators of the type \begin{equation}
(H^{j,s_1}-k^2)^{-1}\left(W(P_{str}-P_j)(H_0-k^2)^{-1}\right)^{t}W(H^{j,s_2}-k^2)^{-1},\ t\geq 0,\label{Feb28-13} \end{equation}
where  $(P_{str}-P_j)(H_0-k^2)^{-1}$ is just a diagonal part of $\hat H(\k^{(1)}(\varphi))-k^2)^{-1}$. Operators of the type \eqref{Feb28-13} without the end terms also can be present.
Since $W_{\q\q'}=0$ when $\||\p_{\q-\q'}\||>Q$,  the operator \eqref{Feb28-13} can differ from zero only  when the $\||\cdot\||$-distance between $\MM^{j,s_1}_2$ and $\MM^{j,s_2}_2$ is less than $(t+1)Q$. Therefore (see Appendix 3), the $\||\cdot\||$-distance between their central points is less than $C(Q)(t+1)Q$. Consequently, if any term in the finite sum of $\eqref{decomp}$ includes a strongly resonant block of $\hat H$, this term is equal to zero unless the $\||\cdot\||$-distance between the central point of the strongly resonant $\MM^{j,s}_2$ and the central point of a weakly resonant block  containing $\m'$ is less than $RC(Q)Q$, if  $\m'$ is in a resonant block.
By Appendix 3, $\||\cdot\||$-distance between $\m'$ and the strongly resonant $\MM^{j,s}_2$ is not greater than $RC(Q)Q+\||\p_\q\||\leq (R+1)C(Q)Q$, otherwise the corresponding term is just zero.  This  contradicts to the definition of $R$, $R:=\left[\frac{d(\m')}{C(Q)Q}\right]-2$. If $\m'$ is not in a resonant block, then the considerations are, obviously, similar, just simpler.

Now, the proof of \eqref{nontrivial2} can be completed in the same way as in the proof of Lemma~\ref{trivial}.

We also notice that from \eqref{Jan8-13}, \eqref{dk0} and the maximum principle one has the estimate (cf. \eqref{weak1})
\begin{equation}\label{nontrivial7}
|(P_{str}(\hat H(\k^{(1)}(\varphi))-k^2)P_{str})^{-1}_{\m\m'}|\leq k^{214\mu\delta},\ \ \mbox{when}\  d_{str}(\m)+ d_{str}(\m')>0,
\end{equation}
uniformly in the disc $|\varphi-\varphi_0|\leq k^{-2-40\mu\delta}$.
\end{proof}

Now, we put
\begin{equation}\label{defP}P=\sum _{\m \in
\MM_1}P_{\m}+\sum_{j, trivial}\sum _{t}P_{2, str}^{j,t}+\sum_{j,non-trivial}\left(\sum _{s}P_{2, weak}^{j,s}+\sum _{t}P_{2, str}^{j,t}\right),\end{equation} where $P_{\m}$ are diagonal projectors corresponding to the sets $\tilde \MM _{\m}$, $P_{2, weak}^{j,s}$ correspond to
 weakly resonant $\MM _2^{j,s}$  and  $P_{2, str}^{j,t}$ correspond to all different  $\tilde \MM _{2,str}^{j,t}$  introduced before Lemmas \ref{trivial} and \ref{nontrivial}.  In \eqref{defP} we take the sum over all (non-trivial) weakly resonant sets, which are disjoint from any  $\tilde \MM _{2,str}^{j,t}$ and over all $\m \in\MM_1$, which are not in any  $\tilde \MM _{2,str}^{j,t}$ or a weakly resonant set of a non-trivial $\MM _2^{j}$.
 \begin{lemma}\label{ortogonal} The projectors in
 \eqref{defP} are mutually  orthogonal:
 \begin{equation}\label{PVP-1}
P_{\m}VP_{\m'}=P_{\m}VP_{2, weak}^{j,s}=P_{\m}VP_{2,str}^{j,t}=0, \ \hbox{when}\ \m,\m'\in\MM_1,\  \m \neq \m',\end{equation}
\begin{equation}\label{PVP-4}
P_{2,str}^{j,t}VP_{2,str}^{j',t'}=0,\  \hbox{when}\
 (j,t)\not=(j',t'), \end{equation}
 where $\MM _2^{j}$ and $\MM _2^{j'}$ are trivial or non-trivial; and
\begin{equation}  P_{2,weak}^{j,s}VP_{2,weak}^{j',s'}=0, \ \hbox{when}\ (j,s)\not=(j',s'),\label{PVP-2} \end{equation}
  where $\MM _2^{j}$ and $\MM _2^{j'}$ are both non-trivial; and
\begin{equation} \label{PVP-3}
P_{2,weak}^{j,s}VP_{2,str}^{j',t}=0,
\end{equation}
where $\MM _2^{j}$ is non-trivial and $\MM _2^{j'}$ is trivial or non-trivial.
 \end{lemma}
\begin{corollary} The operator PHP has a block structure defined by the projectors in \eqref{defP}:
\begin{equation}
\label{PHP}
\begin{split}& PHP=\sum _{\m \in
\MM_1}P_{\m}HP_{\m}+\sum_{j, trivial}\sum _{t}P_{2, str}^{j,t}HP_{2, str}^{j,t}+\cr &\sum_{j,non-trivial}\left(\sum _{s}P_{2, weak}^{j,s}HP_{2, weak}^{j,s}+\sum _{t}P_{2, str}^{j,t}HP_{2, str}^{j,t}\right).
\end{split}\end{equation} \end{corollary}
\begin{proof} Relations \eqref{PVP-1} follow from the definition of $\MM_1$ and
 the obvious inequality $Q<k^{\delta }/3$. Similar consideration yields \eqref{PVP-4}-\eqref{PVP-3} for $j\neq j'$. Further we assume $j=j'$.

 The relation  \eqref{PVP-4} follows from the definition of a cluster and the fact that each $\tilde \MM_{2, str}^{j,t}$
 belongs to $k^{\delta }$-neighborhoods of  $\MM_{2, str}^{j,t}$

 Let us prove \eqref{PVP-2}.  Suppose \eqref{PVP-2} does not hold. Then, there is $\m \in \MM_{2, weak}^{j,s}$, $\m '\in \MM_{2, weak}^{j,s'}$, such that $V_{\m-\m'}\neq 0$. Since $\m,\m'\in \MM_{2}^{j}$, we have $\m-\m'=c\q$, where $\q$ is the direction of $\MM_{2}^{j}$.  Property 2) of potential $V$ implies that $c$ is an integer. Now, using the definition of $\MM_{2}^{j,s}$, we obtain $s=s'$. This contradicts to the assumption $s\neq s'$ and, hence, proves \eqref{PVP-2}.

Let us prove \eqref{PVP-3}. Indeed, if $\MM_{2,weak}^{j,s}$ is connected by $V$ with the main body of $\tilde \MM_{2,str}^{j,t}$, then, by construction, such $\MM_{2,weak}^{j,s}$ belongs to $\tilde \MM_{2,str}^{j,t}$. If $\MM_{2,weak}^{j,s}$ is connected by $V$ with a branch of $\tilde \MM_{2,str}^{j,t}$, which is another weakly resonant $\MM_{2,weak}^{j,s'}$, then, considering as in the proof of \eqref{PVP-2}, we obtain that this is the same branch. Therefore, again, $\MM_{2,weak}^{j,s}\subset \tilde  \MM_{2,str}^{j,t}$. Since the summation in
\eqref{defP} is only over $\MM_{2,weak}^{j,s}$ which are disjoint from $\tilde \MM_{2,str}^{j,t}$, we arrive to contradiction. Thus, \eqref{PVP-3} is proven. \end{proof}

 Since
\eqref{G1-1} holds for any $\m\in \tilde \Omega (\delta )\setminus
\{0\}$ (see \eqref{May19-14}), we have $\MM (\varphi _0)\cap \tilde \Omega (\delta
)=\emptyset $. This means that the $\||\cdot \||$-distance between
$\tilde \MM (\varphi _0)$ and $ \Omega (\delta )$ is no less than
$3k^{\delta }$. Hence,
\begin{equation}P(\delta)VP_{\m}=P(\delta )VP_{2,weak}^{j,s}=P(\delta )VP_{2,str}^{j,t}=0.
 \label{PVP*}
\end{equation}

We conclude this subsection with the following corollaries.
\begin{lemma}\label{estnonres} Let $\varphi _0\in {\omega ^{(1)}}(k,8)$. Then, the operator
$\left(P\left(H\big( \k ^{(1)}(\varphi
)\big)-k^{2}I\right)P\right)^{-1}$ has no more than $64k^{4r_1}$
poles in the disk $|\varphi -\varphi _0|<2k^{-2-\delta (40\mu +1)}$.
 The
following estimate holds: \begin{equation} \label{Mon3}
\left\|\left(P\left(H\big(\k ^{(1)}(\varphi
)\big)-k^{2}I\right)P\right)^{-1}\right\|<k^{2\mu
r_1}\varepsilon _0^{-1}+ck^{-2}\varepsilon  _0^{-2}+k^{-7}\varepsilon _0
^{-4},  \  \varepsilon _0=\min
\{\varepsilon , k^{-2-41\mu\delta }\},
\end{equation}  when $\varphi $ is in the smaller disk  $|\varphi -\varphi _0|<k^{-2-\delta
(40\mu +1)}$, $\varepsilon $ being the distance from $\varphi $ to
the nearest pole of the operator.\end{lemma}
\begin{corollary}\label{estnonres1}
If $\varepsilon =k^{-r_1'}$, $r_1'\geq\mu r_1$, then
\begin{equation} \label{Mon3a} \left\|\left(P\left(H\big(\k
^{(1)}(\varphi
)\big)-k^{2}I\right)P\right)^{-1}\right\|<\frac{1}{32}k^{4r_1'},\end{equation}
\begin{equation} \label{Mon3a'}
\left\|\left(P\left(H\big(\k^{(1)}(\varphi
)\big)-k^{2}I\right)P\right)^{-1}\right\|_1<\frac 12 k^{4r_1'+4r_1}.\end{equation}  \end{corollary}
 The first formula follows from \eqref{Mon3}. The second formula follows from the fact that the
dimension of $P$ does not exceed $16k^{4r_1}$.\\
\begin{proof}
Indeed, the number of blocks in $PHP$ (see \eqref{PHP})  does not
exceed $16k^{4r_1}$ (the number of elements in $\Omega (r_1)$). The
resolvent of each block has no more than  four poles. Therefore, the
resolvent of $PHP$ has no more than $64k^{4r_1}$ poles. Using
Lemmas~\ref{single},  \ref{L:estnonres1}, \ref{weak}, \ref{trivial}, \ref{nontrivial} and, using that $p_\m>k^{-2\mu r_1}$ in
\eqref{Mon3-2}, we obtain the lemma. \end{proof}

\subsubsection{Resonant and Nonresonant Sets for Step II \label{GSII}}

We divide $[0,2\pi )$ into $[2\pi k^{2+\delta(40\mu+1)}]+1$
intervals $\Delta_l^{(1)}$ with the length not bigger than
$k^{-2-\delta(40\mu+1)}$. If a particular interval belongs to
$\OO^{(1)}(k,8)$ we ignore it; otherwise, let
$\varphi_0^{(l)}$ be a point inside the $\Delta_l^{(1)}$, $\varphi_0^{(l)}\not\in\OO^{(1)}(k,8)$. Let
$$\W_l^{(1)}=\{\varphi \in \W^{(1)}:\ |\varphi -\varphi
_0^{(l)}|<2k^{-2-\delta(40\mu+1)}\}.$$ Clearly, neighboring sets
$\W_l^{(1)}$ overlap (because of the multiplier 2 in the
inequality), they cover the $2k^{-2-\delta(40\mu+1)}$-neighborhood of $\omega ^{(1)}(k,\delta,8)$. We denote this neighborhood by $\hat \W^{(1)}(k,2)$. For each $
\varphi $ in the neighborhood there is a $l$ such that $|\varphi
-\varphi _{0}^{(l)}|<k^{-2-\delta(40\mu+1)}$. We consider the poles of
the operator $\left(P(\varphi _0^{(l)})\left(H(\k^{(1)}(\varphi))-k^{2}\right)P(\varphi
_0^{(l)})\right)^{-1}$ in a $\W_l^{(1)}$ and denote them by $\varphi _{lm}$,
$m=1,...,M_l$. By Lemma \ref{estnonres}, $M_l<64k^{4r_1}$. Next,
let $\OO^{(2)}_{lm}$ be the disc of the radius $k^{-r_1'}$ around
$\varphi _{lm}$, $r_1'>\mu r_1$.
\begin{definition} The set
\begin{equation}\OO^{(2)}=\cup _{lm}\OO^{(2)}_{lm} \label{O2}
\end{equation}
we call the second resonant set. The set
\begin{equation}\W^{(2)}= \hat \W^{(1)}(k,2)\setminus \OO^{(2)}\label{W2}
\end{equation}
is called the second nonresonant set. The set
\begin{equation}\omega^{(2)}= \W^{(2)}\cap [0,2\pi) \label{w2}
\end{equation}
is called the second real nonresonant set. \end{definition}
\begin{lemma}\label{L:geometric2}Let  $r_1'>\mu r_1$,
$\varphi \in \W^{(2)}$ and $\varkappa \in \C:
|\varkappa-\varkappa^{(1)}(\varphi )|<k^{-4r'_1-1-\delta }$.
Then,
\begin{equation} \label{Mon3a**}
\left\|\left(P\left(H\big(\k(\varphi
)\big)-k^{2}I\right)P\right)^{-1}\right\|<k^{4r_1'},\end{equation}
\begin{equation} \label{Mon3a***}
\left\|\left(P\left(H\big(\k(\varphi
)\big)-k^{2}I\right)P\right)^{-1}\right\|_1<k^{4r_1'+4r_1},\end{equation}
where $P$ is the projection \eqref{defP} corresponding to the
interval $\Delta _l^{(1)}$ containing $\Re \varphi $.
\end{lemma}
\begin{proof} For $\k=\k^{(1)}(\varphi )$ the lemma follows
immediately from the definition of $\W^{(2)}$ and Corollary
\ref{estnonres1}. Considering the Hilbert identity, it is easy to
see that estimates \eqref{Mon3a} and \eqref{Mon3a'} are  stable with
respect to perturbation of $\varkappa^{(1)}$ of order
$k^{-4r_1'-1-\delta}$. This stability ensure \eqref{Mon3a**} and
\eqref{Mon3a***}.
 \end{proof}

 By total size of
the set $\OO^{(2)}$ we mean the sum of the sizes of its connected
components.
\begin{lemma} \label{L:O2size} Let $r_1'\geq (\mu +4)r_1$. Then, the size of each connected component of $\OO^{(2)}$ is less
than $128k^{4r_1-r_1'}$. The total size of $\OO^{(2)}$ is less than
$10^3k^{2+\delta(40\mu+1)+4r_1-r_1'}$, where $2+\delta(40\mu+1)+4r_1-r_1'<0$.
\end{lemma}
\begin{corollary} \label{C:O2size} If  a connected component of $\OO^{(2)}$ intersects $[0,2\pi)$ or its
$\frac{1}{2}k^{-2-\delta(40\mu+1)}$-neighborhood, then it is strictly inside
$\tilde \W^{(1)}$. \end{corollary}
\begin{proof} Indeed, each set $\W_l^{(1)}$ contains no more than
$64k^{4r_1}$ discs $\OO_{lm}$. Therefore,  the total size of $\OO^{(2)}\cap
\W_l^{(1)}$ is less than $128k^{-r_1'+4r_1}$. Considering that
$128k^{-r_1'+4r_1}$ is much smaller than the length of $\Delta _l^{(1)}$, we
obtain that there is no connected components which go across the
whole set $\W_l^{(1)}$ and the size of each connected component of
$\OO^{(2)}$ is less than $128k^{4r_1-r_1'}$. Considering that
$l<7k^{2+\delta(40\mu+1)}$, we obtain the required estimate for the
total size of $\OO^{(2)}$. \end{proof}


We will also need the estimates for the resolvent in the
neighborhood of $\m=0$. 
Let $C_2$ be a circle in the complex plane:
 \begin{equation} C_2=\{z\in \C: |z-k^{2}|=
\frac 12 k^{-4r_1'}\}. \label{C-2} \end{equation}
Using the definition of $\k^{(1)}(\varphi)$
we obtain the following lemma.
\begin{lemma}\label{estnonres0} Let
$\varphi\in\W^{(1)}_l$. Then,
$$
\left\|\Big(P(\delta)\big(H(\k^{(1)}(\varphi))-z\big)P(\delta)\Big)^{-1}\right\|\leq
8k^{4r_1'}. $$ \end{lemma}
The proof is completely analogous to the proof of corresponding Lemma 3.21 from \cite{KaSh} and we omit it here.

Let
\begin{equation}\tilde{P}(\varphi_0^{(l)}):=P(\varphi_0^{(l)})+P(\delta),  \label{defP_j} \end{equation}
\begin{lemma} Let  $\varphi $ be in the  $k^{-r_1'-\delta}$-neighborhood of $ \W^{(2)}\cap W^{(1)}_l$  and  $P$, $\tilde P$ be constructed for  the
interval $\Delta _l^{(1)}$ containing $\Re \varphi $.  Then,
\begin{equation} PVP(\delta)=0, \label{May29-14a}\end{equation}
\begin{equation}\tilde PH(\k^{(1)}(\varphi ))\tilde
P=  P(\delta)HP(\delta )+PHP \label{May29-14b}\end{equation} 
\begin{equation}\label{estfull}
\left\|\left(\tilde{P}\big(H(\k^{(1)}(\varphi))-z\big)\tilde{P}\right)^{-1}\right\|\leq
16k^{4r_1'},\ \ \ \mbox{when }z\in C_2.\end{equation}
\end{lemma}
\begin{proof} Formula \eqref{May29-14a} follows from  \eqref{defP} and \eqref{PVP*}. Using \eqref{defP_j} and  \eqref{May29-14a}, we obtain \eqref{May29-14b}.
We  notice that the statement of the Lemma~\ref{L:geometric2}
still holds (up to the multiplier 2 at the r.h.s.), when we use $z\in C_2$ instead
of $k^{2}$.
We also use
Lemma~\ref{estnonres0}. At last, considering from the beginning the
discs $\OO^{(2)}_{jm}$ with radius $\frac12 k^{-r_1'}$ instead of
$k^{-r_1'}$ one can easily see that similar estimates (up to another multiplier 2) hold in $k^{-r_1'-\delta}$-neighborhood of ${\cal
W}^{(2)}$. \end{proof}

\section{Step II}
\subsection{Operator $H^{(2)}$. Perturbation Formulas}
Let $P(r_1)$ be an orthogonal projector onto $\Omega(r_1):=\{\m:\
|\|\p_\m\||\leq k^{r_1}\}$ and $H^{(2)}=P(r_1)HP(r_1) $.  From now
on we assume \begin{equation}r_1'=40\mu r_1+2,\ \ \ 2<r_1<k^{\delta /8}. \label{Aug13-1} \end{equation} 
 Let  $\varphi $ be in the  $k^{-r_1'-\delta}$-neighborhood of $ \W^{(2)}\cap W^{(1)}_l$. We
consider $H^{(2)}(\k^{(1)}(\varphi ))$ as a perturbation of
\begin{equation}\label{gulf1} \tilde H^{(1)}=\tilde PH(\k^{(1)}(\varphi ))\tilde
P+\left(P(r_1)-\tilde P\right)H_0(\k^{(1)}(\varphi
))\left(P(r_1)-\tilde P\right),\end{equation}where $\tilde P$ is defined in \eqref{defP_j} and corresponds to the interval $\Delta _l^{(1)}$
containing $\Re\varphi $. By \eqref{May29-14b} and \eqref{PHP},  the first term on the right-hand side of \eqref{gulf1} has a block structure. 
The second term in \eqref{gulf1} is, obviously, diagonal. Thus, $\tilde H^{(1)}$ has a block-diagonal structure.
Let $W$ be the perturbation of $\tilde H^{(1)}$, i.e, $W=H^{(2)}-\tilde H^{(1)}$. It is easy to see that:
\begin{equation}W=P(r_1)VP(r_1)-\tilde PV\tilde P. \label{W}\end{equation}
By analogy with \eqref{g}, \eqref{G},
\begin{equation}\label{g2} g^{(2)}_r({\k}):=\frac{(-1)^r}{2\pi
ir}\hbox{Tr}\oint_{C_2}\left(W(\tilde
H^{(1)}({\k})-zI)^{-1}\right)^rdz,
\end{equation} \begin{equation}\label{G2}
G^{(2)}_r({\k}):=\frac{(-1)^{r+1}}{2\pi i}\oint_{C_2}(\tilde
H^{(1)}({\k})-zI)^{-1}\left(W(\tilde
H^{(1)}({\k})-zI)^{-1}\right)^rdz.
\end{equation}

Next theorem is the analogue of Theorem 4.1 from \cite{KaSh}.

\begin{theorem} \label{Thm2} Suppose $\varphi $ is in
the real  $k^{-r_1'-\delta }$-neighborhood of $\omega
^{(2)}(k,\delta,\tau )$ and $\varkappa\in\R$,
$|\varkappa-\varkappa^{(1)}(\varphi )|\leq k^{-4r'_1-1-\delta}$,
$\k=\varkappa(\cos \varphi ,\sin \varphi )$. Then, for sufficiently
large $k>k_1(V,\delta ,\tau )$ there exists a single eigenvalue of
$H^{(2)}({\k})$ in the interval\\ $\varepsilon _2( k,\delta,\tau
)=\left( k^{2}-\frac 12 k^{-4r_1'},
k^{2}+\frac 12 k^{-4r_1'}\right)$. It is given by the absolutely
converging series:
\begin{equation}\label{eigenvalue-2}\lambda^{(2)}({\k})=\lambda^{(1)}({\k})+
\sum\limits_{r=2}^\infty g^{(2)}_r({\k}).\end{equation} For
coefficients $g^{(2)}_r({\k})$ the following estimates hold:
\begin{equation}\label{estg2} |g^{(2)}_r({\k})|<k^{-k^{\delta
}(2Q)^{-1}}k^{-\delta_*r/16}.
\end{equation}
The corresponding spectral projection is given by the series:
\begin{equation}\label{sprojector-2}
\E ^{(2)}({\k})=\E^{(1)}({\k})+\sum\limits_{r=1}^\infty
G^{(2)}_r({\k}), \end{equation} $\E^{(1)}({\k})$ being the spectral
projection of $H^{(1)}(\k)$. The operators $G^{(2)}_r({\k})$ satisfy
the estimates:
\begin{equation}
\label{Feb1a} \left\|G^{(2)}_r({\k})\right\|_1<k^{-k^{\delta
}(4Q)^{-1}}k^{-\delta_*r/16}.
\end{equation}
\begin{equation}G^{(2)}_r({\k})_{\s\s'}=0,\ \ \mbox{if}\ \ 10rk^{\delta_*}<\||\p_\s\||+\||\p_{\s'}\|| \label{Feb6a}
\end{equation}
\end{theorem}
\begin{corollary} \label{corthm2} For the perturbed eigenvalue and its spectral
projection the following estimates hold:
 \begin{equation}\label{perturbation-2}
\lambda^{(2)}({\k})=\lambda^{(1)}({\k})+ O\left(k^{-k^{\delta
}(2Q)^{-1}}\right),
\end{equation}
\begin{equation}\label{perturbation*-2}
\left\|\E^{(2)}({\k})-\E^{(1)}({\k})\right\|_1<k^{-k^{\delta
}(4Q)^{-1}},
\end{equation}
\begin{equation}
\left|\E^{(2)}({\k})_{\s\s'}\right|<k^{-d^{(2)}(\s,\s')},\ \
\mbox{when}\ \||\p_\s\||>k^{\delta } \mbox{\ or }
\||\p_{\s'}\||>k^{\delta },\label{Feb6b}
\end{equation}
$$d^{(2)}(\s,\s')=\frac{\delta_*}{160}(\||\p_\s\||+\||\p_{\s'}\||)k^{-\delta_* }+k^{\delta
}(4Q)^{-1}.$$
\end{corollary}
Formulas \eqref{perturbation-2} and \eqref{perturbation*-2} easily
follow from \eqref{eigenvalue-2}, \eqref{sprojector-2} and
\eqref{estg2} and \eqref{Feb1a}. The estimate \eqref{Feb6b} follows
from \eqref{sprojector-2}, \eqref{Feb1a} and \eqref{Feb6a}. Indeed,
using these estimates, we obtain
$\left|\left(\E^{(2)}({\k})-\E^{(1)}({\k})\right)_{\s\s'}\right|
<k^{-d^{(2)}(\s,\s')}$. Considering that $\E^{(1)}({\k})_{\s\s'}=0$
when $\||\p_\s\||>k^{\delta }$ or $\||\p_{\s'}\||>k^{\delta }$, we
arrive at  \eqref{Feb6b}.
\begin{proof}
Let $P':=P(r_1)-\tilde{P}$. By \eqref{gulf1}, \eqref{W},
$$\tilde{H}^{(1)}\big(\k ^{(1)}(\varphi)\big):=\tilde{P}H\big(\k
^{(1)}(\varphi)\big)\tilde{P}+P'H_0\big(\k ^{(1)}(\varphi)\big)P',\ \ \
W:=P'VP'+P'V\tilde{P}+\tilde{P}VP'.$$
We will often omit $\k ^{(1)}(\varphi)$ in the arguments when it
cannot lead to confusion. By \eqref{estfull}, we have
\begin{equation}\label{step2raz}
\left\|(\tilde{H}^{(1)}-zI)^{-1}\right\| <16k^{4r'_1}.
\end{equation} Let us consider the perturbation series
\begin{equation}\label{step2dva}
(H^{(2)}-z)^{-1}=\sum_{r=0}^\infty(\tilde H^{(1)}-z)^{-1}\left(-W(\tilde H^{(1)}-z)^{-1}\right)^r.
\end{equation}
Put
$$\tilde A:=-(\tilde H^{(1)}-z)^{-1}W(\tilde H^{(1)}-z)^{-1}.$$
To check the convergence it is enough to show that
\begin{equation}\label{||A||2}\|\tilde A\|<k^{-\delta_*/8}.
\end{equation}
Estimates \eqref{step2raz} and \eqref{||A||2} yield
\begin{equation}\label{step2raz*}
\left\|({H}^{(2)}-zI)^{-1}\right\| <32k^{4r'_1}.
\end{equation}
To prove \eqref{||A||2} it suffices to check
\begin{equation}\label{||A||2-2}\|P'\tilde AP'\|<4\|V\|k^{-\delta_*/2+215\mu \delta }.
\end{equation}
\begin{equation}\label{||A||2-3}\|P'\tilde A\tilde P\|<2\|V\|k^{-\delta_*/4+214\mu \delta }.
\end{equation}
Let us prove \eqref{||A||2-2}.
By construction,  $P'=P_{nonres}+P_{triv,weak}$, where $P_{nonres}$ is the projection on $\Omega (r_1)\setminus \MM$ and
$P_{triv,weak}$ is the projection on those weak points in trivial clusters $\MM_2^j$, which do not belong to $\frac13 k^{\delta }$-neighborhoods of strong points.
Using Lemmas \ref{single} and \ref{weak}, we easily obtain \begin{equation}
\|P_{nonres}\tilde AP_{nonres}\|<4\|V\|k^{-2\delta _*},\ \ \|P_{nonres}\tilde AP_{triv, weak}\|, \|P_{triv, weak}\tilde AP_{nonres}\|<2\|V\|k^{-\delta _*+214\mu \delta}. \label{May11-13}
\end{equation}
To obtain  $\|P_{triv, weak}\tilde AP_{triv, weak}\|<2\|V\|k^{-\delta_*/2+215\mu \delta }$, we use the arguments from the proof of Lemma~\ref{trivial}. Namely, if $\MM_2^j$ corresponds to the first option in the Definition~\ref{def} then $P_{triv, weak}\tilde AP_{triv, weak}=0$. If $\MM_2^j$ corresponds to the second option in the Definition~\ref{def} then the inequality opposite to \eqref{trivial5a} is valid for any pair $\m,\m'\in \MM^j_{2,\,triv, weak}$. Combining this with Lemma~\ref{weak} we get the sought estimate. Thus,
\eqref{||A||2-2} is proven.


Next, we prove \eqref{||A||2-3}. By \eqref{PVP-1}-\eqref{PVP-2}, it is enough to check

\begin{equation}
\label{||A||2-4}\|P'\tilde AP_{\m}\|<8\|V\|k^{-\delta _*+214\mu\delta}, \ \m \in \MM _1\cup {\bf 0},\ \mbox {where  } P_{\bf 0}:=P(\delta ),
\end{equation}
\begin{equation}
\label{||A||2-5}\|P'\tilde AP_{2,\,triv, str}^{j,s}\|<2\|V\|k^{-\delta _*/4+214\mu \delta},
\end{equation}
\begin{equation}
\label{||A||2-6}\|P'\tilde AP_{2,\,nontriv, weak}^{j,s}\|<2\|V\|k^{-\delta _*+214\mu \delta},
\end{equation}
\begin{equation}
\label{||A||2-7}\|P'\tilde AP_{2,\,nontriv, strong}^{j,s}\|<2\|V\|k^{-\delta _*/4+214\mu \delta},
\end{equation}
 To prove \eqref{||A||2-4} we represent $(\tilde H^{(1)}-z)^{-1}P_{\m}$ as follows:
\begin{align}\label{May11-13b}
&(\tilde H^{(1)}-z)^{-1}P_{\m}=
\sum_{r=0}^{R_0}\left(-(H_0-z)^{-1}P_{\m}VP_{\m}\right)^r(H_0-z)^{-1}P_{\m}+\cr &
\left(-(H_0-z)^{-1}P_{\m}VP_{\m}\right)^{R_0+1}(\tilde
H^{(1)}-z)^{-1}P_{\m}, \end{align}
 where $R_0$ to be fixed later.
Then,
\begin{equation}\label{step25}
\begin{split}& \|P'VP_{\m}(\tilde H^{(1)}-z)^{-1}\|\leq
\sum_{r=0}^{R_0}\left\|B_r\right\|+
\left\|P'V\left((H_0-z)^{-1}P_{\m}VP_{\m}\right)^{R_0+1}\right\|\|(\tilde
H^{(1)}-z)^{-1}P_{\m}\|, \cr
&
B_r:=P'V\left((H_0-z)^{-1}P_{\m}VP_{\m}\right)^r(H_0-z)^{-1}P_{\m}.
\end{split} \end{equation}
Note that $B_r=P'B_rP_{\m}$ and matrix elements $(B_r)_{\j\s}$ are
equal to zero if $|\|\p_\j-\p_\s\||>Q(r+1)$ (see \eqref{V_q=0}).
Thus, the only non-trivial elements $(B_r)_{\j\s}$ are such that $$
\j\in\Omega(r_1)\setminus\left(\tilde{\MM}(\varphi _0)\cup \Omega (\delta )\right),\ \ \
\s\in \tilde{\MM}_{\m}(\varphi _0)\cup \Omega (\delta ), \ \ |\|\p_\j-\p_\s\||\leq Q(r+1). $$
Let $r:Q(r+1)\leq k^\delta/6$. It follows that  $(B_r)_{\j\s}=0$ if $\s=\m$ or $\s={\bf 0}$, since  such $\s$ have the distance greater than $\frac{1}{3}k^{\delta }$ from $\j$. If
$\s \neq \m$ or $\s \neq 0$, then $\left||\k^{(1)}(\varphi) +\p_{\s}|_{\R}^{2}-z\right|>\frac 14 k^{\delta_* }$.  Therefore, for $r:Q(r+1)\leq k^\delta/6$ we have:
$$ \|B_r\|\leq
(4\|V\|k^{-\delta_*})^{r+1},\ \ \
\left\|P'V\left((H_0-z)^{-1}P_{\m}V{P}_{\m}\right)^{r+1}\right\|\leq
\|V\|\big(4\|V\|k^{-\delta_*}\big)^{r+1}. $$ Now, we fix $R_0:=[k^\delta /(6Q)]-1$. Then the condition $Q(r+1)\leq
k^\delta/6$ is satisfied for all $r\leq R_0$ and $$ \|P'V{P}_{\m}(\tilde H^{(1)}-z)^{-1}\|\leq
\sum_{r=0}^{R_0}(4\|V\|k^{-\delta_*})^{r+1}+\|V\|\big(4\|V\|k^{-\delta_*}\big)^{R_0+1}16k^{4r'_1}.
$$  Assuming that $k$ is
large enough (in particular, $\frac{\delta_*
k^\delta}{6Q}>5r'_1$) and using Lemma~\ref{weak} , we obtain \eqref{||A||2-4}.

Let us prove \eqref{||A||2-5}.  Noting that  $P^{j,s}_{2,\,triv,str}$ is the projection into the $k^{\delta }$-neighborhood of a strongly resonant point and
$V_{\q}=0$ when $\||\p_\q\||>Q$, we conclude that we can insert a projection into the formula for $\tilde A$; namely
$P'\tilde AP^{j,s}_{2,str}=P'(H_0-z)^{-1}VP''(\tilde H^{(1)}-z)^{-1}$, where $P''$ is the projection
 corresponding to the points situated outside  the $(k^{\delta }-Q)$-neighborhood of a strongly resonant point.  Using Lemmas \ref{single} and \ref{weak}, we obtain
 $\|P'(H_0-z)^{-1}\|<2k^{214\mu \delta}$. Using \eqref{trivial2} with $d(\m,\m')>(k^{\delta }-Q)/(20Q)$, and $\varepsilon _0=k^{-r_1'}$, $(r_1'<k^\delta \delta_*/(100Q))$, we obtain
 $\|P''(\tilde H^{(1)}-z)^{-1}\|<\|V\|k^{-\delta _*/4}$. Now \eqref{||A||2-5} easily follows.

Let us prove \eqref{||A||2-6}. Obviously, $P'\tilde AP_{2,\,nontriv, weak}^{j,s}=P_{nonres}\tilde AP_{2,\,nontriv, weak}^{j,s}$.
Using Lemmas \ref{single} and \ref{weak}, we arrive at  \eqref{||A||2-6}.

The proof of \eqref{||A||2-7} is analogous to the proof of \eqref{||A||2-5}, with \eqref{nontrivial2} being used instead of \eqref{trivial2}.

To prove \eqref{Feb1a} we consider the operator
$A=W\left(\tilde H^{(1)}-z\right)^{-1}$ and
represent it as $A=A_0+A_1+A_2$, where
$A_0=\left(P(r_1)-\E^{(1)}({\k})\right)A \left(P(r_1)
-\E^{(1)}({\k})\right)$, $A_1=\left(P(r_1 )-\E^{(1)}({\k})\right)A
\E^{(1)}({\k})$, $A_2= \E^{(1)}({\k})A
\left(P(r_1)-\E^{(1)}({\k})\right)$. Note that we have
$\E^{(1)}({\k})W\E^{(1)}({\k})=0$, because of \eqref{W}. It is easy to see that by construction $A_0$ is holomorphic inside $C_2$ (see, e.g. Lemma~\ref{L:geometric2} and Theorem~\ref{Thm1}). Hence,
$$\oint _{C_2}\left(\tilde H^{(1)}-z\right)^{-1}A_0^r dz=0.$$
Therefore,
\begin{equation} \label{Feb1} G^{(2)}_r({\k})=\frac{(-1)^{r}}{2\pi i}\sum
_{j_1,...j_r=0,1,2,\ j_1^2+...+j_r^2\neq 0}I_{j_1...j_r},\ \ \ \
I_{j_1...j_r}:=\oint _{C_2}\left(\tilde
H^{(1)}-z\right)^{-1}A_{j_1}.....A_{j_r} dz.
\end{equation} At least one of indices in each term is equal to 1 or 2.
Let us show that
\begin{equation} \label{A_2}
\|A_2\|_1<ck^{-k^{\delta }(2Q)^{-1}}k^{214\mu\delta}.
\end{equation} First, we notice that
$\E^{(1)}W(P(r_1)-\E^{(1)})=\E^{(1)}WP'$ by \eqref{W} and \eqref{PVP*}. It suffices
to show that
\begin{equation}\label{Feb6}\|\E^{(1)}WP'\|_1<ck^{-k^{\delta
}(2Q)^{-1}},\end{equation} since $\|P'\left(\tilde
H^{(1)}-z\right)^{-1}\|=\|P'\left(
H_0-z\right)^{-1}\|<2k^{214\mu\delta}$ for $z\in C_2 $.
Indeed,
$$\left(\E^{(1)}WP'\right)_{\s\s'}=\sum _{\s'':\ \||\p_{\s''}\||\leq k^{\delta},\ \||\p_{\s''-\s'}\||\leq
Q}\E^{(1)}_{\s\s''}W_{\s''-\s'}$$ when $\||\p_{\s'}\||>k^{\delta }$
and it is equal to zero otherwise. Hence,
$$\left|\left(\E^{(1)}WP'\right)_{\s\s'}\right|\leq \|W\|
\sum _{\s'':\ k^{\delta }-Q\leq \||\p_{\s''}\||\leq k^{\delta
}}\E^{(1)}_{\s\s''}$$ if $\||\p_{\s'}\||<k^{\delta }+Q$ and zero
otherwise. Using \eqref{matrix elements}, we obtain
\begin{equation} \label{E1-1}\left|\left(\E^{(1)}WP'\right)_{\s\s'}\right|<
ck^{4\delta }\max _{\||\p_{\s''}\||>k^{\delta
}-Q}k^{-d^{(1)}(\s,\s'')}.\end{equation} It easily follows:
$$\left|\left(\E^{(1)}WP'\right)_{\s\s'}\right|<ck^{4\delta }k^{-(1-40\mu\delta)(k^{\delta }Q^{-1}-1+\||\p_\s\||Q^{-1})}$$
when  $\||\p_{\s'}\||< k^{\delta }+Q$, and zero otherwise. It
follows $\left\|\E^{(1)}WP'\right\|<ck^{-k^{\delta
}(2Q)^{-1}}$. Considering that $\E^{(1)}$ is a
one-dimensional projection, we obtain the same estimate for $\bf S_1
$-norm, namely, \eqref{Feb6}. Thus, we have proved \eqref{A_2}.
Let us estimate $I_{j_1...j_r}$. Suppose one of the indices is equal to 2.
Substituting \eqref{A_2} into \eqref{Feb1} and taking into account \eqref{||A||2}, \eqref{step2raz},
we obtain:
$$\left\|I_{j_1...j_r}\right\|<ck^{-k^{\delta
}(2Q)^{-1}}k^{214\mu\delta}k^{-\frac18\delta_* (\frac{r}{2}-1)}k^{4r_1'}<k^{-k^{\delta
}(4Q)^{-1}}k^{-\delta_*r/16}.$$
(More precisely, our $A_2$ splits the integrand in $I_{j_1\dots j_r}$ into two parts; for each part we use \eqref{||A||2} for every product of two $A_{j_k}$ and \eqref{step2raz} for the last single $A_{j_s}$ or $(\tilde H^{(1)}-z)^{-1}$; we also take into account the length of the circle.)
Note that the operator
$A_1$ is always followed by $A_2$ unless $A_1$ occupies the very
last position in the product. Thus, it remains to consider the case
$A_{j_1}.....A_{j_r}= A_0^{r-1}A_1$. It is easy to see that
$$\left(\tilde
H^{(1)}-z\right)^{-1}A_{0}^{r-1}A_{1}=\left(\left(\tilde
H^{(1)}-\bar z\right)^{-1}A_2(\bar z)A_{0}^{r-1}(\bar z)\right)^*.$$
This implies the estimate for this case too.
Therefore,
$$\left\|G^{(2)}_r({\k})\right\|<k^{-k^{\delta
}(4Q)^{-1}}k^{-\delta_*r/16}.$$ The same estimate
can be written for the $\bf S_1$ norm of this operator, since
$\E^{(1)}$ is one-dimensional.

Let us obtain the estimate for $g_r({\k})$.
Obviously,\begin{equation} \label{Feb1'}
g^{(2)}_r({\k})=\frac{(-1)^r}{2\pi ir}\sum _{j_1,...j_r=0,1,2,\
j_1^2+...+j_r^2\neq 0}Tr\oint _{C_2}A_{j_1}.....A_{j_r} dz.
\end{equation}
Note that each term contains both $A_1$ and $A_2$, since we compute
the trace of the integral. Using \eqref{Feb6}, we obtain:
$\|A_1\|_1<cb^{-1}_2k^{-k^{\delta }(2Q)^{-1}}$, where $b_2$ is the radius of $C_2$.  Combining
this estimate with \eqref{A_2} and \eqref{||A||2}, \eqref{step2raz}, we obtain
\eqref{estg2} for $r\geq 2$. Finally, applying \eqref{g2} in the
case $r=1$, we see that $g^{(2)}_1({\k})=0$, since
$\E^{(1)}W\E^{(1)}=0$.

To prove \eqref{Feb6a} it's enough to notice that the biggest block
of $\tilde H^{(1)}$ has the size not greater than $2k^{\delta_*}$.
\end{proof}

It is easy to see that coefficients $g^{(2)}_r({\k})$ and operators
$G^{(2)}_r({\k})$ can be analytically extended into the complex
$k^{-r_1'-\delta}$-neighborhood of $\omega ^{(2)}$ (in fact, into
$k^{-r_1'-\delta}$-neighborhood of $\W^{(2)}$) as functions of
$\varphi $ and to the complex $(k^{-4r'_1-1-\delta})-$
neighborhood of $\varkappa =\varkappa^{(1)}(\varphi )$ as functions
of $\varkappa$, estimates \eqref{estg2}, \eqref{perturbation-2}
being preserved. Now, we use formulae \eqref{g2},
\eqref{eigenvalue-2} to extend
$\lambda^{(2)}({\k})=\lambda^{(2)}(\varkappa,\varphi)$ as an
analytic function. Obviously, series \eqref{eigenvalue-2} is
differentiable. Using Cauchy integral and Lemma
\ref{L:derivatives-1} we get the following lemma.
\begin{lemma} \label{L:derivatives-2}Under conditions of Theorem \ref{Thm2} the following
estimates hold when $\varphi \in \omega ^{(2)}(k,\delta )$ or its
complex $k^{-r_1'-\delta}$-neighborhood and $\varkappa\in \C:$
$|\varkappa-\varkappa^{(1)}(\varphi )|<k^{-4r'_1-1-\delta}:$
\begin{equation}\label{perturbation-2c}
\lambda^{(2)}({\k})=\lambda^{(1)}({\k})+ O\left(k^{-k^{\delta
}(2Q)^{-1}}\right),
\end{equation}
\begin{equation}\label{estgder1-2k}
\frac{\partial\lambda^{(2)}}{\partial\varkappa}=\frac{\partial\lambda^{(1)}}{\partial\varkappa}
+ O\left(k^{-k^{\delta }(2Q)^{-1}}k^{4r_1'+1+\delta
}\right), \end{equation}
\begin{equation}\label{estgder1-2phi}\frac{\partial\lambda^{(2)}}{\partial \varphi }=\frac{\partial\lambda^{(1)}}{\partial \varphi }+O\left(k^{-k^{\delta }(2Q)^{-1}}k^{r_1'+\delta }\right),
 \end{equation}
\begin{equation}\label{estgder2-2} \frac{\partial^2\lambda^{(2)}}
{\partial\varkappa^2}= \frac{\partial^2\lambda^{(1)}}
{\partial\varkappa^2}+O\left(k^{-k^{\delta
}(2Q)^{-1}}k^{8r_1'+2+2\delta }\right), \end{equation}
\begin{equation} \label{gulf2} \frac{\partial^2\lambda^{(2)}}
{\partial\varkappa\partial \varphi
}=\frac{\partial^2\lambda^{(1)}}{\partial\varkappa\partial \varphi
}+ O\left(k^{-k^{\delta }(2Q)^{-1}}k^{5r_1'+1+2\delta
}\right),
\end{equation}
\begin{equation} \label{gulf3}
\frac{\partial^2\lambda^{(2)}}{\partial\varphi ^2}=\frac{\partial^2\lambda^{(1)}}{\partial\varphi ^2}+O\left(k^{-k^{\delta }(2Q)^{-1}}k^{2r_1'+2\delta}\right).
\end{equation}\end{lemma}

\subsection{\label{IS2}Isoenergetic Surface for Operator $H^{(2)}$}

\begin{lemma}\label{ldk-2} \begin{enumerate}
\item For every sufficiently large $\lambda $, $\lambda :=k^{2}$, and $\varphi $ in the real  $\frac{1}{2} k^{-r_1'-\delta }$-neighborhood
of $\omega^{(2)}(k,\delta, \tau )$ , there is a unique
$\varkappa^{(2)}(\lambda, \varphi )$ in the interval
$I_1:=[\varkappa^{(1)}(\lambda, \varphi )-\frac{1
}{2}k^{-4r_1'-1-\delta },\varkappa^{(1)}(\lambda, \varphi
)+\frac{1 }{2}k^{-4r_1'-1-\delta },]$, such that
    \begin{equation}\label{2.70-2}
    \lambda^{(2)} \left(\k
^{(2)}(\lambda ,\varphi )\right)=\lambda ,\ \ \k ^{(2)}(\lambda
,\varphi ):=\varkappa^{(2)}(\lambda ,\varphi )\vec \nu(\varphi).
    \end{equation}
\item  Furthermore, there exists an analytic in $ \varphi $ continuation  of
$\varkappa^{(2)}(\lambda ,\varphi )$ to the complex  $\frac{1}{2}
k^{-r_1'-\delta }$-neighborhood of $\omega^{(2)}(k,\delta, \tau )$
such that $\lambda^{(2)} (\k ^{(2)}(\lambda, \varphi ))=\lambda $.
Function $\varkappa^{(2)}(\lambda, \varphi )$ can be represented as
$\varkappa^{(2)}(\lambda, \varphi )=\varkappa^{(1)}(\lambda, \varphi
)+h^{(2)}(\lambda, \varphi )$, where
\begin{equation}\label{dk0-2} |h^{(2)}(\varphi )|=
O\left(k^{-k^{\delta }(2Q)^{-1}}k^{-1}\right),
\end{equation}
\begin{equation}\label{dk-2}
\frac{\partial{h}^{(2)}}{\partial\varphi}= O\left(k^{-k^{\delta
}(2Q)^{-1}}k^{r_1'+\delta-1 }\right),\ \ \ \ \
\frac{\partial^2{h}^{(2)}}{\partial\varphi^2}=O\left(k^{-k^{\delta
}(2Q)^{-1}}k^{2r_1'+2\delta-1 }\right).
\end{equation} \end{enumerate}\end{lemma}
\begin{proof}  The proof is completely analogous to that of Lemma \ref{ldk}, estimates \eqref{perturbation-2c} --\eqref{gulf3} being used. \end{proof}



Let us consider the set of points in $\R^2$ given by the formula:
$\k=\k^{(2)} (\varphi), \ \ \varphi \in \omega ^{(2)}(k,\delta, \tau
)$. By Lemma \ref{ldk-2} this set of points is a slight distortion
of ${\cal D}_{1}$. All the points of this curve satisfy
the equation $\lambda^{(2)} (\k ^{(2)}(\varphi ))=k^{2}$. We call
it isoenergetic surface of the operator $H^{(2)}$ and denote by
${\cal D}_{2}$.


\subsection{Preparation for Step III - Geometric Part. Properties of the Quasiperiodic Lattice}\label{geomIII}
Let
\begin{equation} \label{Aug25-1}
\SS (k, \varepsilon _0)=\left\{ \k \in \R^2:\left\|\left(H^{(1)}(\k)-k^{2}\right)^{-1}\right\|>\varepsilon _0^{-1}\right\}.\end{equation}
In this section we prove that the number of the lattice points $ \k _0+\p_\m$,
$\||\p_\m\||<k^{r_1}$  in $\SS (k, \varepsilon _0)$  does not exceed $Ck^{\frac{2r_1}{3}+1}$ when $\varepsilon _0$ is  sufficiently small and  $\k  _0$ is fixed. For this we split
$\SS $ into two subsets: `` non-resonant" and ``resonant", the non-resonant set being just a vicinity of $\DD _1(k^{2})$.
An estimate for the number of lattice points in the non-resonant set is proven in Lemma \ref{L:number of points-1}. Estimates for the number of lattice points in different types of resonant sets are proven in Lemmas \ref{4.10}, \ref{t4.10}, \ref{2t4.10}, \ref{nontriv4.10}.
 These estimates play an important
role in the further construction.
\subsubsection{General Lemmas} We consider
$\p_\m=2\pi(\s_1+\alpha\s_2)$ with integer vectors $\s_j$ such that
$|\s_j|\leq 4k^{r_1}$.

It is easy to see that there exists a pair $(q,p)\in\Z^2$ such that
$0<q\leq 4k^{r_1}$ and
\begin{equation}\label{q}
|\alpha q+p|\leq \frac14 k^{-r_1}.
\end{equation}
We choose a pair $(p,q)$ which gives the best approximation. In
particular, $p$ and $q$ are mutually prime. Put
$\epsilon_q:=\alpha+\frac{p}{q}$.
We have (see \eqref{geq}, \eqref{q})
\begin{equation}k^{-2r_1\mu}\leq |\epsilon_q| \leq \frac14 q^{-1}k^{-r_1}.\label{epsilon_q}\end{equation}
Here we assume that $k$ is large enough (depending on $\alpha$ only) so that \eqref{geq} ensures the lower inequality in \eqref{epsilon_q}. In what follows we will often make similar assumptions without additional remarks.

We  write any $\s_2$ in the form\begin{equation}\s_2=q\s_2'+\s_2'' \label{s1}
\end{equation}
with  integer vectors $\s_2'$ and $\s_2''$, $0\leq (\s_2'')_j< q$
for $j=1,2$. Hence, $|(\s_2')_j|\leq 4k^{r_1}/q+1$. It follows
$$
(2\pi)^{-1}\p_\m=(\s_1-p\s_2')+(-\frac{p}{q}\s_2''+\epsilon_q\s_2'')+\epsilon_q
q\s_2'.
$$
Denote $\s:=\s_1-p\s_2'$. Then $|\s|\leq 8k^{r_1}$. The number of
different vectors $\tilde{\s}:=-\frac{p}{q}\s_2''+\epsilon_q\s_2''$
is not greater than $(2q)^2$. For each fixed pair $\tilde \s,\ \s$
we obtain a lattice parameterized by $\s_2'$. We call this lattice a
cluster corresponding to given $\tilde \s,\ \s$. Each cluster,
obviously, is a square lattice with the step $\epsilon _qq$. It
contains no more than $\left(9k^{r_1}q^{-1}\right)^2$ elements,
since $|(\s_2')_j|\leq 4k^{r_1}q^{-1}+1$, $j=1,2$. The size of each
cluster is less than $5|\epsilon _q|k^{r_1}$. If $\epsilon _q$
satisfies slightly stronger inequality, than \eqref{epsilon_q} then
clusters don't overlap, see the following lemma.

     \begin{lemma}\label{Lattice-1}Suppose that $\epsilon _q$ satisfies the
     inequality
     \begin{equation}|\epsilon_q|\leq \frac{1}{64}q^{-1}k^{-r_1}.\label{epsilon_q'}\end{equation}
     Then, the size of each cluster is less that $\frac{1}{8q}$. The distance between clusters is greater than
     $\frac{1}{2q}$. \end{lemma}\begin{proof}Let us estimate the
     distance between  points $\s_2'=\bf 0$ of two different
     clusters. Indeed, $\s-\frac{p}{q}\s_2''\neq \bf 0$, since $p\s_2''$,
see \eqref{s1},
     is not a multiple of $q$. Therefore,  $\left|(\s-\frac{p}{q}\s_2'')_j\right|\geq \frac{1}{q} $, $j=1,2$.
     Considering that
     $0\leq (\s_2'')_j<q$, $j=1,2$, we obtain that the distance between two points
     where $\s_2'=\bf 0$ is greater than $\frac{1}{q}-|\epsilon _q|q$, that is greater
     than $\frac{15}{16q}$. The size of each cluster is obviously
     less than $|\epsilon _q|q(4k^{r_1}q^{-1}+1)\leq\frac{1}{8q}$. Thus, two clusters cannot overlap,
     the distance between them being greater than
     $\frac{1}{2q}$.\end{proof}
     We need two more properties of the lattice $\p_\m$,
     $\||\p_{\m}\||<2k^{r_1}$.
     \begin{lemma}\label{Lattice-2} The number of vectors $\p_\m$,
      satisfying the inequalities $\||\p_{\m}\||<2k^{r_1}$,
     $p_{\m}<|\epsilon _q|qk^{r_1/3}$, does not exceed
     $k^{2r_1/3}$.\end{lemma}
     \begin{proof} Suppose vectors $\p_{\m}$ and $\p_{\m'}$ satisfy
     the conditions of the lemma. Then, $\||\p_{\m}-\p_{\m'}\||<4k^{r_1}$.
     By definition of $\epsilon _q$, $(2\pi )^{-1}|\p_{\m}-\p_{\m'}|\geq |\epsilon _q|q$. Thus, the distance
     between the points $\p_{\m},\p_{\m'}$ is greater than $2\pi |\epsilon _q|q$ and each point can be
      surrounded by the disc  of the radius $\pi|\epsilon _q|q$,
      the discs being disjoint. Dividing the area of the disc of the radius $2 |\epsilon _q|qk^{r_1/3}$
      (we increased radius to take into account points $\p_\m$ near the boundary of the disc $p_{\m}<|\epsilon _q|qk^{r_1/3}$) by the area of a disc of the radius $\pi |\epsilon _q|q$, we
      obtain that the number of vectors satisfying the inequality
     $p_{\m}<|\epsilon _q|qk^{r_1/3}$ does not exceed
     $k^{2r_1/3}$.\end{proof}
     \begin{lemma} \label{Lattice-3} Suppose $q$ in the inequality
     \eqref{q} satisfies the estimate $q>k^{2r_1/3}$. Then, the
     number of vectors  $\p_\m$,
     $\||\p_{\m}\||<2k^{r_1}$, satisfying the inequality
     $p_{\m}<k^{-2r_1/3}$ does not exceed
     $2^{12}\cdot k^{2r_1/3}$.\end{lemma}
     \begin{proof} First assume
     $|\epsilon_q|>\frac{1}{64}q^{-1}k^{-r_1}$. Then, dividing
     the area of the disc of the radius $2 k^{-2r_1/3}$
      by the area of a disc of the radius $\pi |\epsilon _q|q>\frac{1}{32}k^{-r_1}$, we
      obtain that the number of vectors satisfying the inequality
     $p_{\m}<k^{-2r_1/3}$ does not exceed
     $2^{12}k^{2r_1/3}$.

     Second, we consider the case $|\epsilon_q|\leq
     \frac{1}{64}q^{-1}k^{-r_1}$. According to Lemma
     \ref{Lattice-1}, the clusters do not overlap. The distance
     between clusters is greater than $\frac{1}{2q}$. Therefore, dividing the area of a disc with radius $\frac32 k^{-2r_1/3}$ by the area of a disc with radius $\frac{1}{4q}$, the last number being smaller than $\frac14 k^{-2r_1/3}$ by the conjecture of the lemma, we obtain that the
     number of clusters intersecting the disc of the radius $k^{-2r_1/3}$ is
     less than $\left(6k^{-2r_1/3}q\right)^2$. Each cluster contains
     less than $\left(9k^{r_1}q^{-1}\right)^2$ points. Therefore, the
     total number of of vectors  $\p_\m$,
     $\||\p_{\m}\||<2k^{r_1}$, satisfying the inequality
     $p_{\m}<k^{-2r_1/3}$ does not exceed
     $\left(6k^{-2r_1/3}q\right)^2\cdot \left(9k^{r_1}q^{-1}\right)^2
     <2^{12}\cdot k^{2r_1/3}$. \end{proof}
     \subsubsection{Lattice Points in the Nonresonant Set\label{Lattice Points in a Nonresonant Set}}
     \begin{lemma} \label{L:number of points-1}
     Let $N(k, r_1, \k _0,\varepsilon _0)$  be the number of points $\k  _0+\p_{\n}$,
$\||\p_{\n}\||<k^{r_1}$ in the $\varepsilon _0$-neighborhood of
     ${\cal D}_1(k^{2})$, where $\varepsilon _0=k^{-5\mu r_1}$ and
     $\k  _0\in\R^2$ being fixed. Then,
     $$ N(k,r_1,\k _0,\varepsilon _0)<1000  \cdot k^{\frac{2r_1}{3}+1}.$$
     \end{lemma}
     \begin{proof}Let us consider the segment $\p_{\n-\n'}$ between
     two points $\k  _0+\p_{\n}$ and $\k  _0+\p_{\n'}$ in the neighborhood.
     Obviously, $\||\p_{\n-\n'}\||<2k^{r_1}$ and $p_{\n-\n'}>k^{-\mu
     r_1}>>\varepsilon _0$. This means that the direction of the
     segment cannot be orthogonal to the curve (in fact they are almost parallel to the curve) and each end can be
     assigned its own angle coordinate $\varphi _{\n}, \varphi _{\n'}$, $\varphi _{\n}\neq \varphi _{\n'}$. We
     enumerate the points $\k  _0+\p_{\n}$ in the order of increasing $\varphi
     _{\n}$ and connect neighboring points by segments. First we
     consider the segments with the length greater or equal to
     $\frac{1}{64}k^{-\frac{2r_1}{3}}$. Since the length of ${\cal D}_1(k^{2})$ does
     not exceed $3\pi k$, the number of such segments does not
     exceed $650 k^{\frac{2r_1}{3}+1}$.

      It remains to estimate the
     number of segments with the length less than
     $\frac{1}{64}k^{-\frac{2r_1}{3}}$.
 First, we prove that no two
     segments $\p_{\n_1-\n'_1}$, $\p_{\n_2-\n'_2}$ can be equal to
     each other.
We use concavity of the curve ${\cal D}_1(k^{2})$ and
      a small size $\varepsilon _0$ of its neighborhood. We show that for every
     $\p_{\n_1-\n'_1}$ with both ends in the neighborhood, there is a point on the
     curve where the tangent vector is parallel to $\p_{\n_1-\n'_1}$. Since the tangent vector
      changes monotonously with $\varphi $, no two vectors
      $\p_{\n_1-\n'_1}$ can have the same direction. Indeed, let us consider a segment $\p_{\n_1-\n'_1}$.
      Let $(x,y)$ be local coordinates associated with  $\p_{\n_1-\n'_1}$,  the beginning
      of the segment being at the origin and the end having
      the coordinates $(\tau, 0)$, $\tau=p_{\n_1-\n'_1}$. The curve is described by the equation $y=y(x)$. It easily follows
      from Lemma \ref{ldk} that $y'(x)=o(1)$ and the curvature $\kappa $ of the curve ${\cal D}_1(k^{2})$ is
      $\frac{1}{k}(1+o(1))$ at all  points of the curve. Using the formula
      $\kappa (x)=|y''(x)|\left(1+y'(x)^2\right)^{-3/2}$, we easily obtain $y''(x)=-\frac{1}{k}(1+o(1))$.
      Using a Taylor formula, we get $y(\tau )=y(0)+y'(0)\tau
      -\frac{1}{2k}(1+o(1))\tau^2$. Note that $|y(0)|, |y(\tau
      )|<2\varepsilon _0$, since both ends of the segment are in the $\varepsilon
      _0$-neighborhood of the curve. Considering also that $\tau
      >k^{-r_1\mu }$ and the estimate on $\varepsilon _0$, we
      conclude: $\frac{\tau }{k}=2y'(0)(1+o(1))+O(k^{-4r_1\mu })$.
      Substituting this into the Taylor formula \begin{equation}
      y'(\tau
      )=y'(0)-\frac{\tau}{k}(1+o(1)),\label{y'}
      \end{equation}
       we obtain: $y'(\tau
      )=-y'(0)(1+o(1))+O(k^{-4r_1\mu })$.
      If $y'(\tau
      )$ and $y'(0)$ have the same sign or one of them is zero, the last relation yields $|y'(\tau
      )|+|y'(0)|=O(k^{-4r_1\mu })$. This contradicts to \eqref{y'},
      since $\tau >k^{-r_1\mu }$.
      Therefore, $y'(\tau
      )$ and $y'(0)$ have different signs. Considering that $y'(x)$
      is continuous, we obtain that there is a point $x_0$ in $(0,\tau
      )$ such that $y'(x_0)=0$. This means that the isoenergetic
      curve at this point is parallel to $\p_{\n_1-\n'_1}$.

      To finish the proof of the lemma  we consider two cases. Suppose $q$ in the inequality
     \eqref{q} satisfies the estimate $q>k^{2r_1/3}$. Then, by Lemma \ref{Lattice-3}, the
     number of vectors  $\p_\n$,
     $\||\p_{\n}\||<2k^{r_1}$, satisfying the inequality
     $p_{\n}<\frac{1}{64}k^{-2r_1/3}$ does not exceed
     $2^{12}\cdot k^{2r_1/3}$. Since each of them can be used only once, the
     total number of short segments  does not exceed $2^{12}\cdot k^{2r_1/3}$.

     Let  $q\leq k^{2r_1/3}$.
     If $|\epsilon _q|>\frac{1}{64}q^{-1}k^{-r_1}$. Then, obviously,
     $\frac{1}{64}k^{-2r_1/3}<|\epsilon _q|qk^{r_1/3}$. Applying Lemma
     \ref{Lattice-2}, we obtain that the number of segments with the
     length less than $\frac{1}{64}k^{-2r_1/3}$ is less than
     $k^{2r_1/3}$. Since each of them can be used only once, the
     total number of short segments  does not exceed $k^{2r_1/3}$.
     It remains to consider the case $q\leq
     k^{2r_1/3}$, $|\epsilon _q|\leq \frac{1}{64}q^{-1}k^{-r_1}$. By
     Lemma \ref{Lattice-1}, clusters are well separated. Considering
     that the distance between clusters is greater than
     $\frac{1}{2q}$ and the size of each cluster is less than
     $\frac{1}{8q}$, we obtain that no more than $8\pi qk$ clusters
     can intersect $\varepsilon _0$-neighborhood of ${\cal D}_1(k^{2})$.
     The part of the curve inside the clusters has the length $L_{in}$ which is less
     than the double size of a cluster $10|\epsilon _q|k^{r_1}$ (the curve is concave) multiplied by
     the number of clusters $8\pi qk$, i.e., $L_{in}<80\pi |\epsilon
     _q|qk^{r_1+1}$.  Next, the segments with the length less than
     $\frac{1}{2}k^{-2r_1/3}$ cannot connect different clusters, since the
     distance between clusters is greater than $\frac{1}{2q}\geq
     \frac{1}{2}k^{-2r_1/3}$. Therefore, any segment of the length
     less than $\frac{1}{2}k^{-2r_1/3}$ is inside one
     cluster. If we consider the segments with the length greater
     than $|\epsilon _q|qk^{r_1/3}$, then the number of such segments
     is less than $L_{in}/|\epsilon _q|qk^{r_1/3}$, i.e., it is less
     than $80 \pi k^{2r_1/3+1}$. By Lemma \ref{Lattice-2}, the total
     number of segments of the length less than $|\epsilon
     _q|qk^{r_1/3}$ is less than $k^{2r_1/3}$. Each of them can be used only once. Thus, the total
     number of segments is less than $ 300 k^{2r_1/3+1}$.
\end{proof}
\subsubsection{Lattice Points in the Resonant Set\label{Lattice
Points in a Resonant Set}}
Here we introduce and investigate some properties of the subsets of \eqref{Aug25-1} associated with strongly resonant clusters $\MM_{2,str}^{j,s}$.
We remind that, by Lemma~\ref{per3}, such cluster can contain no more than two sets $\MM_{2}^{j,s}$. We have to consider several cases (corresponding to nontrivial $\MM_2^j$ and trivial $\MM_2^j$ of two types in accordance with Definition~\ref{def}).  The goal is to prove Lemmas \ref{4.10}, \ref{t4.10}, \ref{2t4.10} and \ref{nontriv4.10}. They follow from Lemmas \ref{4.9}, \ref{t4.9} and  \ref{4.9nn}, describing properties of  the isoenergetic surface in the resonant region, and Lemmas \ref{Lattice-1}-\ref{Lattice-3}. Lemma \ref{4.9n} is preparatory for \ref{4.9nn}.  As usual, where it does not lead to confusion, we will use identical notation for similar (but not coinciding) auxiliary objects from different cases. We believe that this way it makes easier to follow  parallel proofs for all cases.

{\it Trivial case. I.} First, consider $\q$, $\||\p_\q\||\leq k^\delta$, such that no vector in $\SS_Q$ is a multiple of $\q$. We put
\begin{equation}\label{Rq1}
\begin{split}
& \RR_{\q,\,triv}^1=\RR_\q:=\{\k\in\R^2:\ ||\k|^2-k^2|<k^{-40\mu\delta};\cr &
||\k+\p_{\q'}|^2-k^2|>2k^{-40\mu\delta}\ \hbox{for}\ \q'=c\q,\ c\not=0,\ \||\p_{\q'}\||\leq k^\delta;\cr &
 ||\k+\p_{\q''}|^2-k^2|>k^{\delta_*}\ \hbox{for}\ \q''\not=c\q,\ \||\p_{\q''}\||\leq k^\delta\}.
\end{split}
\end{equation}
Obviously, $\RR_\q=\RR_{\tilde\q}$ if $\tilde\q=c\q$. We will also assume that
\begin{equation}\label{Rq1-a}||\k+\p_{\q'}|^2-k^2|\leq k^{\delta_*} \ \mbox{for at least one $\q'=c\q,\ c\not=0$}, \end{equation} otherwise, we just deal with a non-resonant situation.
Let $P_\q:=P(\delta)$, $P(\delta )$ being defined at the beginning of Step I (we use different notation here to make it similar to constructions below for other cases). We consider the operator $P_\q\big(H(\k)-(k^2+\varepsilon_0')I\big)P_\q$, $|\varepsilon_0'|\leq k^{-165\mu\delta}$, and its determinant $D(\k,k^2+\varepsilon_0')$. Let ${S}_\q\subset \RR_\q$ be the set
\begin{equation}\label{Sq}
{S}_\q(\varepsilon _0):=\{\k\in \RR_\q:\ D(\k,k^2+\varepsilon'_0)=0\ \hbox{for some}\ |\varepsilon'_0|\leq\varepsilon_0\},\  \ 0< \varepsilon_0\leq k^{-165\mu\delta}.
\end{equation}
Let $\k=(\tau_1,\tau_2)$ where $\tau_1$ is the coordinate in the direction of $\q$ and $\tau_2$ is the coordinate in the direction orthogonal to $\q$. 
We have
\begin{lemma}\label{4.9} The sets $D(\k,k^2)=0$ and ${S}_\q(\varepsilon _0)$ have the following properties in $\RR_\q$:

1) The equation $D(\k,k^2)=0$ describes at most two curves in $\RR_\q$ which are represented by $\tau_2=f_i(\tau_1)$ where $|f_i'|\leq k^{\delta_*+\mu\delta-1}$, $i=1,2$.

2) The set ${S}_\q(\varepsilon _0)$ belongs to $\cup_{i=1,2}S_i$, $$S_i(k,\varepsilon_0):=\{\k:\ |\tau_2-f_i(\tau_1)|<2\varepsilon_0 k^{-1},\ |\tau_1|\leq k^{\delta_*+\mu\delta}\}.$$

3) Every curve $\tau_2=f_i(\tau_1)$ contains no more than $2^{31}k^{8\delta}$ inflection points.

4) Let $\l$ be a segment of a straight line,
\begin{equation}\l=\{\k=(\tau_1, \beta _1 \tau _1+\beta _2),\  \tau _{1,0}\leq \tau _1\leq \tau _{1,0}+\eta
\},\ \ |\tau_{1}|<k^{\delta_*+\mu\delta},\label{segment-1}\end{equation}
such that both of its ends belong to ${S}_{i}(k,\varepsilon _0)$ and
    $2\varepsilon _0<\eta^2k^{-1}$, $0<\eta<k^{-41\mu\delta}$. Then, there is an
inner part $\l'$ of the segment
 which is not in $S_{i}(k,\varepsilon _0)$. Moreover, there is a point $(\tau _{1*},\tau _{2*})$ in $\l'$ such that
 $f'_i(\tau _{1*})=\beta _1$, i.e., the curve and the segment have
 the same direction when $\tau _{1}=\tau _{1*}$.
\end{lemma}
\begin{proof}
Let us prove that operator $P_\q H(\vec \varkappa )P_\q$ has exactly one eigenvalue $\lambda_1(\vec \varkappa)$ satisfying $|\lambda_1(\vec \varkappa)-k^2|<\frac32 k^{-40\mu\delta}$. Indeed, we surround $k^2$ by the circle $\tilde C$ of the radius $\frac32 k^{-40\mu\delta}$ and consider the perturbation series for  $(P_\q (H(\vec \varkappa)-z)P_\q)^{-1}$, $z\in\tilde C$, with respect to the resolvent of the free operator. Considering the definitions of $\RR_\q$, $\tilde C$ and the assumption that no vector in $\SS_Q$ is a multiple of $\q$, we easily obtain:
\begin{equation}\label{VV}
\|(P_\q(H_0-z)P_\q)^{-1}V(P_\q(H_0-z)P_\q)^{-1}\|\leq k^{-\delta_*+41\mu\delta}.
\end{equation}
Therefore,  the series converges. The convergence of the series implies that $P_\q HP_\q$ has the same number of eigenvalues inside $\tilde C$ as $P_\q H_0P_\q$. i.e. exactly one. We also have $\lambda_1(\vec \varkappa )=|\k|^2+O(k^{-\delta_*+41\mu\delta})$ and this asymptotic formula can be differentiated. From \eqref{Rq1} it follows that $\lambda_1(\vec \varkappa )=k^2+\varepsilon'_0$, $|\varepsilon'_0|<\frac32 k^{-40\mu\delta}$.

Next, since $||\k+\p_{\q'}|^2-k^2|\leq k^{\delta_*}$ for at least one $\q'\not={\bf 0}$ parallel to $\q$, we obtain that
$\left|\frac{\partial\lambda_1}{\partial\tau_1}\right|=O(k^{\delta_*+\delta})$ and $\frac{\partial\lambda_1}{\partial\tau_2}=2\tau _2(1+o(1))$. This gives us the first and the second statements of the lemma (for more details see the proof of Lemma 4.9 in
\cite{KaSh}). The third statement  can be proven in complete analogy with the proof of the similar statement from Lemma 4.9 from \cite{KaSh}.

To prove the fourth statement we introduce operators $\hat H_0$ and $E_0$ such that $(\hat H_0)_{\n\n}=(H_0)_{\n\n}$ for $\n\not={\bf 0}$ and $(\hat H_0)_{\bf 00}=1$; $(E_0)_{\n\n}=0$ for $\n\not={\bf 0}$ and $(E_0)_{\bf 00}=1$. Obviously, $\hat H_0=H_0(P_\q-E_0)+E_0$.  We will use this operator to``normalize" the determinant $D$.  Indeed, let \begin{equation}\label{A}
\hat{D}(\k,k^2):=\frac{\hbox{det}\Big(P_\q\big(H(\k)-k^2 I\big)P_\q\Big)}{\hbox{det}\left(P_\q\big(\hat H_0(\k)-k^2(I-E_0)\big)P_\q\right)}.
\end{equation}
Obviously, $D=0$ in $\RR_\q$ if and only if $\hat D=0$. We consider $\k$ being in the segment $\l$. It follows from the estimate for $f'_i$ that $\beta_1=O(k^{\delta_*+\mu\delta-1})$ and
$|\beta_2|=k(1+o(1))$.
 Let $D_0:=\hat D$ for $V=0$:
$$
D_0:=\frac{\hbox{det}\Big(P_\q\big(H_0(\k)-k^2I\big)P_\q\Big)}{\hbox{det}\Big(P_\q\big(\hat H_0(\k)-k^2(I-E_0)\big)P_\q\Big)}=\det (\hat H_0(\k)-k^2I)E_0.$$
We easily see: $$
D_0=|\k|^2-k^2=\tau_1^2+(\beta_1\tau_1+\beta_2)^2-k^2.
$$
Clearly, $D_0$ has no more than two roots $\tau _1$.  Assume these roots are separated by the distance less than $2k^{-41\mu\delta}$ (in another case the proof is analogous, but simpler). We consider their $k^{-41\mu\delta}$-neighborhood which we denote by $T$ (in the case of separated roots, the size of the neighborhood is $k^{-82\mu\delta}d^{-1}$, $d$ being the distance between the roots). Obviously,
\begin{equation}\label{D0below}
|D_0|>k^{-83\mu\delta}\ \hbox{on}\ \partial T.
\end{equation}
If $\tau _1\in T$, then the point $(\tau _1, \beta_1\tau_1+\beta_2)$ satisfy the inequalities in \eqref{Rq1}.  A simple computation yields:
$$\frac{\hat D}{D_0}=\det (I+B), \  \ B=\big((|\k|^2-k^2)E_0+I-E_0\big)^{-1}A,$$
where $A:=\Big(P_\q\big(\hat H_0-k^2(I-E_0)\big)P_\q\Big)^{-1/2}V\Big(P_\q\big(\hat H_0-k^2(I-E_0)\big)P_\q\Big)^{-1/2}$, the square root being chosen arbitrary. It is easy to see (cf. \eqref{VV}) that
\begin{equation}\label{Anew}
\|A\|\leq k^{-\delta_*/2+21\mu\delta}.
\end{equation}
Using  \eqref{D0below}, we obtain: $\left\| B \right\|_1<k^{-\delta_*/2+105\mu\delta} $. By the inequality $\|\det(I+B)-1\|<\|B\|_1e^{2+\|B\|_1}$, $B$ being in the trace class, we easily obtain:
\begin{equation}\label{**}
\hat{D}=D_0\big(1+O(k^{-\delta_*/4})\big)\ \hbox{on}\ \partial T.
\end{equation}
By Rouch\'{e}'s Theorem $\hat D$ has the same number of zeros in $T$ as $D_0$ (and they are in $\frac12 k^{-41\mu\delta}$-neighborhood of the zeros of $D_0$). Therefore,
$\hat D$ can be represented in the form
\begin{equation}\label{155}
\hat D(\tau_1,\beta_1\tau_1+\beta_2)=\tilde f(\tau_1)\prod\limits_{j=1}^2(\tau_1-\tau_1^{(j)}),\ \tau_1\in T.
\end{equation}
From \eqref{D0below}, \eqref{**} and the minimum principle in $T$ it follows that $|\tilde f|>k^{-2\mu\delta}$ when $\tau_1\in T$. 

Let us consider the segment $\l$. By \eqref{155}, there is a point $(\tau'_1,\beta_1\tau'_1+\beta_2)\in\l$ where
$$
|\hat D|>k^{-2\mu\delta}\eta^2.
$$
Obviously, $(D_0)_{\tau_2}=2\tau _2$. Using the estimates similar to \eqref{Anew}, we obtain  $|\hat D'_{\tau_2}|\leq 3k$. Since by the definition of the curve $\hat D\Big(\big(\tau'_1,f_i(\tau'_1)\big),k^2 \Big)=0$, we have
$$
|f_i(\tau'_1)-(\beta_1\tau'_1+\beta_2)|>\frac13 k^{-2\mu\delta-1}\eta^2.
$$
Since $2\varepsilon_0/k<\frac13 k^{-2\mu\delta-1}\eta^2$, there are points in $\l$ which are outside $S_i$. At one of these points the function $|f_i(\tau'_1)-(\beta_1\tau'_1+\beta_2)|$ attains its maximum value. At this point the curve and the segment are parallel.
\end{proof}
Let $N^{(1)}_{\q}(k ,r_1,\k _0,\varepsilon_0)$ be the number of points $\k _0
+\p_{\n}$, $\||\p_{\n}\||<k^{r_1}$ in $S_{\q}(k,\varepsilon _0)$, $\k _0$ being fixed.
\begin{lemma} \label{4.10} Let $\delta_* <r_1<\infty $. If $\varepsilon _0<k^{-16\mu r_1}$ then the number
of points $N^{(1)}_{\q}(k ,r_1,\k_0,\varepsilon_0)$ admits the estimate
\begin{equation}\label{est4.10}N^{(1)}_{\q}(k ,r_1,\k_0,\varepsilon_0)\leq
2^{44}k^{2r_1/3+8\delta_* }.\end{equation}\end{lemma}
\begin{proof} The proof of the lemma is analogous to that of Lemma~4.10 from \cite{KaSh}.
\end{proof}

Now, let
\begin{equation}\label{tRq1}
\begin{split}
& \tilde\RR_{\q,\,triv}^1=\tilde\RR_\q:=\{\k\in\R^2:\ ||\k|^2-k^2|<k^{-40\mu\delta};\cr &||\k+\p_{\tilde\q}|^2-k^2|<2k^{-40\mu\delta}\ \hbox{for one}\ \tilde\q=\tilde c\q,\ \tilde c\not=0,\ \||\p_{\tilde\q}\||\leq k^\delta;\cr &
||\k+\p_{\q'}|^2-k^2|>4k^{-40\mu\delta}\ \hbox{for}\ \q'=c\q,\ c\not=0,\tilde c,\ \||\p_{\q'}\||\leq k^\delta;\cr &
 ||\k+\p_{\q''}|^2-k^2|>k^{\delta_*}\ \hbox{for}\ \q''\not=c\q,\ \||\p_{\q''}\||\leq k^\delta\}.
\end{split}
\end{equation}
Obviously, $\tilde\RR_{\q,\,triv}^1=\tilde\RR_{\tilde\q,\,triv}^1$. So, one can assume $\tilde c=1$, i.e. $\tilde\q=\q$. Let $\tilde P_\q$ be the projector onto the set
$$\{\n:\ \||\p_\n\||\leq k^\delta\ \hbox{or}\ ||\p_{\n-\tilde\q}\||\leq k^\delta\}.$$
We consider the operator $\tilde P_\q(H(\k)-(k^2+\varepsilon_0')I)\tilde P_\q$ and its determinant $\tilde D(\k,k^2+\varepsilon_0')$, $|\varepsilon_0'|\leq k^{-165\mu\delta}$. Let $\tilde{S}_\q\subset \tilde\RR_\q$ be the set
\begin{equation}\label{tSq}
\tilde{S}_\q:=\{\k\in \tilde\RR_\q:\ \tilde D(\k,k^2+\varepsilon'_0)=0\ \hbox{for some}\ |\varepsilon'_0|\leq\varepsilon_0\},\ \ 0<\varepsilon_0\leq  k^{-165\mu\delta}.
\end{equation}
The following lemma describes properties of $\tilde{S}_\q$.

\begin{lemma}\label{t4.9} The sets $\tilde D(\k,k^2)=0$ and $\tilde{S}_\q$ have the following properties in $\tilde\RR_\q$:

1) The equation $\tilde D(\k,k^2)=0$ describes at most four curves in $\tilde\RR_\q$ which are represented by $\tau_2=f_i(\tau_1)$ where $\k=(\tau_1,\tau_2)$ and $|f_i'|\leq k^{\delta_*+\mu\delta-1}$, $i=1,2,3,4$.

2) The set $\tilde{S}_\q$ belongs to $\cup_{i}S_i$, $$S_i(k,\varepsilon_0):=\{\k:\ |\tau_2-f_i(\tau_1)|<2\varepsilon_0 k^{-1},\ |\tau_1|\leq k^{\delta_*+\mu\delta}\}.$$

3) Every curve $\tau_2=f_i(\tau_1)$ contains no more than $2^{31}k^{8\delta}$ inflection points.

4) Let $\l$ be a segment of a straight line,
\begin{equation}\l=\{\k=(\tau_1, \beta _1 \tau _1+\beta _2),\  \tau _{1,0}<\tau _1<\tau _{1,0}+\eta
\},\ \ |\tau_{1}|<k^{\delta_*+\mu\delta},\label{segment-11}\end{equation}
such that both its ends belong to ${S}_{i}(k,\varepsilon _0)$,
    $2\varepsilon _0<\eta^4k^{-1}$, $0<\eta<k^{-41\mu\delta}$. Then, there is an
inner part $\l'$ of the segment
 which is not in $S_{i}(k,\varepsilon _0)$. Moreover, there is a point $(\tau _{1*},\tau _{2*})$ in $\l'$ such that
 $f'_i(\tau _{1*})=\beta _1$, i.e., the curve and the segment have
 the same direction when $\tau _{1}=\tau _{1*}$.
\end{lemma}
\begin{proof}
The proof repeats the arguments of that from Lemma~\ref{4.9} with obvious changes (two blocks instead of just one), since no vector in $S_{Q}$ is a multiple of $\q$.
\end{proof}

Let $\tilde N^{(1)}_{\q}(k ,r_1,\k _0,\varepsilon_0)$ be the number of points $\k _0
+\p_{\n}$, $\||\p_{\n}\||<k^{r_1}$ in $\tilde S_{\q}(k,\varepsilon _0)$, $\k _0$ being fixed. The analogue of Lemma~\ref{4.10} holds.
\begin{lemma} \label{t4.10} Let $\delta_* <r_1<\infty $. If $0<\varepsilon _0<k^{-16\mu r_1}$, then the number
of points $\tilde N^{(1)}_{\q}(k ,r_1,\k_0,\varepsilon_0)$ admits the estimate
\begin{equation}\label{test4.10}\tilde N^{(1)}_{\q}(k ,r_1,\k_0,\varepsilon_0)\leq
2^{44}k^{2r_1/3+8\delta_* }.\end{equation}\end{lemma}

{\it Trivial case. II.} Now we consider the case when $\q\in\SS_Q$ but $|k^2-(t^\perp_\q)^2|>\frac18 k^{\delta_*}$, where $t^\perp_\q=(\vec \varkappa, \vec \nu ^\perp_\q)$, $\vec \nu ^\perp_\q$ beind a unit vector orthogonal to $\vec p_{\q}$, see Definition~\ref{def}. More precisely, we put
\begin{equation}\label{Rq2}
\begin{split}
& \RR_{\q,\,triv}^2=\RR_\q:=\{\k\in\R^2:\ ||\k|^2-k^2|<k^{-40\mu\delta};\ |k^2-(\k,\vec\nu^\perp_\q)^2|>\frac18 k^{\delta_*};\cr &
||\k+\p_{\q'}|^2-k^2|>2k^{-40\mu\delta}\ \hbox{for}\ \q'=c\q,\ c\not=0,\ \||\p_{\q'}\||\leq k^\delta;\cr &
 ||\k+\p_{\q''}|^2-k^2|>k^{\delta_*}\ \hbox{for}\ \q''\not=c\q,\ \||\p_{\q''}\||\leq k^\delta\}
\end{split}
\end{equation}
and
\begin{equation}\label{tRq2}
\begin{split}
& \tilde\RR_{\q,\,triv}^2=\tilde\RR_\q:=\{\k\in\R^2:\ ||\k|^2-k^2|<k^{-40\mu\delta};\ |k^2-(\k,\vec\nu^\perp_\q)^2|>\frac18 k^{\delta_*};\cr & ||\k+\p_{\tilde\q}|^2-k^2|<2k^{-40\mu\delta}\ \hbox{for a certain}\ \tilde\q=\tilde c\q,\ \tilde c\not=0,\ \||\p_{\tilde\q}\||\leq k^\delta;\cr &
||\k+\p_{\q'}|^2-k^2|>4k^{-40\mu\delta}\ \hbox{for}\ \q'=c\q,\ c\not=0,\tilde c,\ \||\p_{\q'}\||\leq k^\delta;\cr &
 ||\k+\p_{\q''}|^2-k^2|>k^{\delta_*}\ \hbox{for}\ \q''\not=c\q,\ \||\p_{\q''}\||\leq k^\delta\}.
\end{split}
\end{equation}
Here, one cannot assume that $\tilde c=1$ (unlike the trivial case I) since $\q\in\SS_Q$ does not imply $\tilde\q\in\SS_Q$ and thus $\tilde\RR_{\q,\,triv}^2\not=\tilde\RR_{\tilde\q,\,triv}^2$ in general; but we can and will assume that $\q$ is the generating vector of the corresponding direction. We introduce $P_\q,\ D,\ \tilde P_\q,\ \tilde D$ as before. Lemmas~\ref{4.9}--\ref{t4.10} hold with the same proofs as above. The only difference is the proof of the estimates similar to \eqref{VV} and \eqref{Anew} providing the convergence of perturbation series. Here, we provide the details which are identical for $\RR_\q$ and $\tilde\RR_\q$, so we use $\RR_\q$ for definiteness. We check that
\begin{equation}\label{VVa}
\left|\left((P_\q(H_0-z)P_\q)^{-1}V(P_\q(H_0-z)P_\q)^{-1}\right)_{\n\n'}\right|\leq k^{-\delta_*/2+42\mu\delta}.
\end{equation}
 If $\n$ or $\n'$ is not parallel to $\q$ the estimate is obvious. Assume that $\n=c\q$ and $\n'=c'\q$. The left hand side of \eqref{VVa} differs from zero only if $\n-\n'\in\SS_Q$. Thus, $c-c'\in\Z$ and $|c-c'|\leq Q$. In this case, considering as in the proof of  \eqref{trivial5a} and using the second inequality in \eqref{Rq2},
 we obtain:
$$
||\k+\p_\n|_\R^2-k^2|+||\k+\p_{\n'}|_\R^2-k^2|>k^{\delta_*/2-\delta},
$$
which proves \eqref{VVa}. The estimate similar to \eqref{Anew} can be proven in the same way.

Let $N^{(2)}_{\q}(k ,r_1,\k _0,\varepsilon_0)$ be the number of points $\k _0
+\p_{\n}$, $\||\p_{\n}\||<k^{r_1}$ in $S_{\q}(k,\varepsilon _0)$ or $\tilde S_{\q}(k,\varepsilon _0)$, $\k _0$ being fixed. The analogue of Lemma~\ref{4.10} holds.
\begin{lemma} \label{2t4.10} Let $\delta_* <r_1<\infty $. If $0<\varepsilon _0<k^{-16\mu r_1}$, then the number
of points $N^{(2)}_{\q}(k ,r_1,\k_0,\varepsilon_0)$ admits the estimate
\begin{equation}\label{2test4.10}N^{(2)}_{\q}(k ,r_1,\k_0,\varepsilon_0)\leq
2^{44}k^{2r_1/3+8\delta_* }.\end{equation}\end{lemma}


{\it Nontrivial case.} Let $\q\in\SS_Q$ be one of two generating vectors of a direction. We put
\begin{equation}\label{Rqn}
\RR_{\q,nontriv}=\RR_\q:=\{\k\in \R^2:\ ||\k|^2-k^2|\leq k^{\delta_*},\ \ |k^2-(\k,\vec\nu_\q^\perp)^2|\leq\frac18 k^{\delta_*}\},
\end{equation}
and introduce coordinates $\tau_1,\,\tau_2$ as before.  As in the previous cases, we will describe here properties of $S_{\q}(k,\varepsilon _0)$. However, prior considering
$P_{\q}HP_{\q}$, we  consider $P_*HP_*$,
where $P_*$ is a smaller projection than $P_{\q}$. The projection $P_*$ has a property, that $P_*HP_*$ admits a separation of variables in the direction of $\q$ and its orthogonal, thus $(P_*VP_*)_{\l \j}=0$, when $\l-\j\neq c\q$. Properties of $P_*HP_*$ can be described in terms of Schr\"odinger operator in dimension one.
This operator is just periodic, because of  condition 2 on potential $V$, see \eqref{V} and below. After investigating properties of $P_*HP_*$, we use perturbative arguments to prove that $P_{\q}HP_{\q}$ has analogous properties. This is the main place in the paper, where condition 2 on the potential is needed, since it allows to use well-known properties of
periodic Schr\"odinger operator in dimension one, instead of a quasi-periodic one.

We consider $\k \in \RR_{\q}$ and the set:
\begin{equation}
{\cal M_*}( \k)=\{n\q, n\in\Z :|\tau_1+ n p_{\q }|^2\leq \frac 98 k^{\delta_*}\}. \label{M*}\end{equation}
This definition is independent on $\tau_2$ and  piece-wise constant in $\tau_1$ with maybe one or two steps for every interval of $\tau_1$ of the size $\pi p_\q/4$.  In what follows we proceed with the constructions locally, in the intervals of $\tau_1$ of the size $\pi p_\q/4$, all the statements and estimates being uniform with respect to the position of such intervals. Without the loss of generality
we further assume that ${\cal M_*}$ does not depend on $\k$ in the given local interval of $\tau_1$.\footnote{We note that in fact ${\cal M_*}$ is essentially the same as  a $\MM_2^{j,s}(\k)$, see \eqref{Mjs}.  However, we prefer to use the notation ${\cal M_*}$, while working with abstract geometric objects. The precise
connection between ${\cal M_*}$ and sets $\MM_2^{j,s}(\k)$ will be established later.}
We consider $P_*HP_*$, where $P_*$ is the projector, corresponding to $\cal M_* (\k )$: $(P_*)_{\m \m}=1$ iff $\m \in  {\cal M_*} $. The operator $P_*HP_*$ admits the separation of variables in the direction of $\q$ and its orthogonal.
As usual (see \eqref{May7-14c} and below), by $\tilde H_{per}$ we will denote the corresponding one-dimensional periodic Schr\"odinger operator; $\tilde\lambda_i^{per}$ being its eigenvalues.

As before, we introduce the operator $P_*(H(\k)-k^2-\varepsilon_0')P_*$, its determinant $D_{aux}(\k,k^2+\varepsilon_0')$ and corresponding set $S_\q(\varepsilon _0)\subset\RR_\q$ (see \eqref{Sq}). \footnote{\label{f1} We notice  that  the determinant $D_{aux}$ is  periodic in $\tau_1$ since shift of $\k$ is properly compensated by the shift of $\MM_*$.}

\begin{lemma}\label{4.9n} Let $0<\varepsilon_0\leq k^{-165\mu\delta}$. The set $D_{aux}(\k,k^2)=0$ and ${S}_\q$ have the following properties in $\RR_\q$:

1) For every $|\varepsilon'_0|\leq\varepsilon_0$ the equation $D_{aux}(\k,k^2)=0$ describes at most four curves in $\RR_\q$ which are represented by $\tau_2=f_i(\tau_1)$ where $\k=(\tau_1,\tau_2)$ and $|f_i'|\leq k^{\delta_*+\mu\delta-1}$, $i=1,2,3,4$.

2) The set ${S}_\q$ belongs to $\cup_{i}S_i$, $$S_i(k,\varepsilon_0):=\{\k:\ |\tau_2-f_i(\tau_1)|<2\varepsilon_0 k^{-1},\ |\tau_1|\leq k^{\delta_*+\mu\delta}\}.$$

3) Every curve $\tau_2=f_i(\tau_1)$ contains no more than $2^{31}k^{8\delta}$ inflection points.

4) Let $\l$ be a segment of a straight line,
\begin{equation}\l=\{\k=(\tau_1, \beta _1 \tau _1+\beta _2),\  \tau _{1,0}<\tau _1<\tau _{1,0}+\eta
\},\ \ |\tau_{1}|<k^{\delta_*+\mu\delta},\label{segment-1n}\end{equation}
such that both its ends belong to ${S}_{i}(k,\varepsilon _0)$,
    $2\varepsilon _0<\eta^{6}k^{-1}$, $0<\eta<k^{-41\mu\delta}$. Then, there is an
inner part $\l'$ of the segment
 which is not in $S_{i}(k,\varepsilon _0)$. Moreover, there is a point $(\tau _{1*},\tau _{2*})$ in $\l'$ such that
 $f'_i(\tau _{1*})=\beta _1$, i.e., the curve and the segment have
 the same direction when $\tau _{1}=\tau _{1*}$.
\end{lemma}
\begin{proof}
The projection $P_*$ is defined in such a way that variables for  $P_*H(\vec \varkappa )P_*$ can be separated,  eigenvalues being described by the formula $\lambda _i=\tilde\lambda_i(\tau _1)+\tau _2^2$. Obviously, the set $D_{aux}(\vec \varkappa , k^2)=0$  is described by isoenergetic curves  $\tilde\lambda_i(\tau _1)+\tau _2^2=k^2$.  The formula $\tau _2=f_i(\tau _1)$ and the estimate for $f_i'$ easily follow, as well as the second statement of the lemma.

To prove the first statement, it remains to show that  the number of curves $\tau _2=f_i(\tau _1)$ does not exceed two.
Indeed (see \eqref{5per}), all eigenvalues of $P_*HP_*$ such that $|\lambda_i-k^2|\leq\frac18 k^{\delta_*}$ can be approximated with the high accuracy $O(k^{-\frac{\delta_*}{C(Q)}k^{\delta_*/2}})$ by the eigenvalues of $\tilde H_{per}$ (more precisely, by $\tilde\lambda_i^{per}(\tau_1)+\tau_2^2$). Therefore,
\begin{equation}\label{1-}\lambda_i(\k)=k^2\end{equation} implies
 \begin{equation}
 \label{Feb14} \tilde\lambda_i^{per}(\tau_1)+\tau_2^2=k^2+O(k^{-\frac{\delta_*}{C(Q)}k^{\delta_*/2}}). \end{equation} Note that then $|\tilde\lambda_i^{per}(\tau_1)|=O(k^{\delta_*})$. There are no more than two eigenvalues $\tilde\lambda_i^{per}$ in the $k^{-\frac{\delta_*}{C(Q)}k^{\delta_*/2}}$-neighborhood of $k^2-\tau_2^2$, and consequently, no more than two eigenvalues $\lambda_i$ satisfying \eqref{1-}. Hence, the number of the curves $\tilde\lambda_i(\tau _1)+\tau _2^2=k^2$ does not exceed two, each corresponding to two curves $\tau _2=f_i(\tau _1)$.
The third statement can be proved in complete analogy with the proof of the similar statement from Lemma 4.9 in \cite{KaSh}.

It remains to prove the fourth statement.  Let $\tau_2=\beta_1\tau_1+\beta_2$. It follows from the estimate for $f'_i$ that $\beta_1=O(k^{\delta_*+\mu\delta-1})$ and $|\beta_2|=k(1+o(1))$. According to the above, to investigate $S_i\cap \l$, we have to consider the equations $\tilde\lambda_i (\tau_1)+(\beta_1\tau_1+\beta_2)^2=k^2+\varepsilon_0'$ with $|\tilde\lambda_i(\tau_1)|=O(k^{\delta_*})$, $|\varepsilon_0'|<\varepsilon_0$.
Obviously, for each real $\tau_1$ the last equality can hold for no more than two eigenvalues simultaneously (let's say for $\tilde\lambda_{i_0}$ and $\tilde\lambda_{i_0+1}$), and for all other eigenvalues we have $|\tilde\lambda_i(\tau_1)+(\beta_1\tau_1+\beta_2)^2-k^2|\geq Cp_\q |i-i_0|$.
Effectively, the estimate from below for $D_{aux}$ will follow from the estimate for each pair of close eigenvalues.

We'll need to distinguish the case of low energies (i.e. $|\tilde\lambda_{i_0}(\tau_1)| \leq \Lambda$) and large energies (i.e. $|\tilde\lambda_{i_0}(\tau_1)|\geq \Lambda$),  $\Lambda$ being a large constant (to be specified later), which depends on $V$ only.
Let $|\tilde\lambda_{i_0}(\tau_1)|\leq \Lambda $. With a proper choice of numeration, we obtain that the corresponding $\tilde\lambda_{i_0}^{per}(\tau_1),  \tilde\lambda_{i_0+1}^{per}(\tau_1)$ are analytic functions of $\tau _1$ in a vicinity of $\tau _{1,0}$. By \cite{Kor},  $|(\tilde\lambda^{per}_{i_0}(\tau_1))''|+|(\tilde\lambda_{i_0}^{per}(\tau_1))'''|\neq 0$ and, hence,
 it is separated from zero by a constant $c(V,\Lambda )$:
 $|({\tilde\lambda_{i_0}^{per}}(\tau_1))''|+|({\tilde\lambda_{i_0}^{per}}(\tau_1))'''|>c(V,\Lambda )$ in a vicinity of $\tau _{1,0}$, the size of vicinity depending on $V$ and $\Lambda $ , but not $k$.
 This means (we recall that $\beta_1=O(k^{-1/2})$) that $\tilde\lambda_{i_0}^{per}(\tau_1)+(\beta _1\tau _1+\beta _2)^2-k^2$ has no more than three zeros in the $\sigma_{small}$-neighborhood of
  $\tau _{1,0}$, $\sigma_{small}=\sigma_{small}(V, \Lambda )$. The analogous fact holds for $\tilde\lambda_{i_0+1}^{per}$,
while for all other $i$-s  the analogous expressions are separated from zero by a constant.  Hence, denoting by $\Sigma$ the $\sigma_{small}/6$-neighborhood of the zeros of $\tilde\lambda_{i_0}^{per}(\tau_1)+(\beta _1\tau _1+\beta _2)^2-k^2$ and $\tilde\lambda_{i_0+1}^{per}(\tau_1)+(\beta _1\tau _1+\beta _2)^2-k^2$ we have
$$\|(\tilde H_{per}(\tau _1)+(\beta _1\tau _1+\beta _2)^2-k^2I)^{-1})\|<C(V,\Lambda)\ \ \ \hbox{on}\ \partial\Sigma.$$
Using perturbative arguments, we easily obtain that the analogous estimate holds for $P_*HP_*$ and the corresponding resolvent has  no more than 6 poles inside. Therefore,
\begin{equation}\|(P_*(H(\tau _1)+(\beta _1\tau _1+\beta _2)^2-k^2)P_*)^{-1})\|<C_1(V,\Lambda )r^{-6},\ \ 0< r<\sigma_{small}/6, \label{Nov29-13a}
\end{equation}
where $r$ is the distance to the nearest pole.
It follows that there is a point in the interval $(\tau _{1,0}, \tau _{1,0}+\eta )$ where
\begin{equation}\|(P_*(H(\tau _1)+(\beta _1\tau _1+\beta _2)^2-k^2)P_*)^{-1})\|<12^6C_1(V,\Lambda )\eta ^{-6}. \label{Nov29-13b}
\end{equation}
Therefore,  $\left|\tilde\lambda_{i}(\tau_1)+(\beta _1\tau _1+\beta _2)^2-k^2\right|>12^{-6}C_1(V,\Lambda )^{-1}\eta ^{6}$ for all $i$.
It contradicts to the assumption $(\tau _1, \beta _1\tau _1+\tau _2)\in S_{\q}$, since  $2\varepsilon _0<\eta ^6k^{-1}$.

Let now $|\tilde\lambda_{i_0}(\tau_1)|\geq \Lambda>>1$.  Let us consider the expression:
\begin{equation}\label{highen}
 D_{i_0}:=(\tilde\lambda_{i_0}(\tau_1)+(\beta_1\tau_1+\beta_2)^2-k^2)(\tilde\lambda_{i_0+1}(\tau_1)+(\beta_1\tau_1+\beta_2)^2-k^2). \end{equation}
 We assume here that $|\tilde\lambda_{i_0}-\tilde\lambda_{i_0+1}|\leq \frac{1}{2}p_\q\Lambda ^{1/2}$, otherwise, only one factor is needed, and this makes arguments just simpler.
 Since $i_0$ is big enough, all other eigenvalues are at the distance greater than $p_q\Lambda ^{1/2}$  from this pair.  Unlike each individual factor in the r.h.s. of (\ref{highen}), $D_{i_0}$ is analytic in the neighborhood of $\tau_{1,0}$ with the radius of analyticity $\sigma_{large}=\pi p_\q/4$.
   If $\Lambda=\Lambda(V)$ is large enough, we can apply standard perturbative arguments  to compare $D_{i_0}$ with the same expression for $V=0$ which we denote by $D_{i_0,0}$. Indeed, $D_{i_0,0}$ is the polynomial of order four with respect to $\tau _1$ with the main coefficient $(1+\beta_1^2)^2$. We consider $\sigma_{large}/100$-neighborhood of each zero and denote the connected component(s) of these neighborhoods, intersecting the $\sigma_{large}/2$-neighborhood of $\tau_{1,0}$,  by $T_0$. By definition, $T_0$ consists of $J$, $1\leq J\leq 4$ discs.
   We have
\begin{equation}
|D_{i_0,0}(\tau_1)|>(\sigma_{large}/100)^J\ \ \ \hbox{on}\ \partial T_0. \label{Jan8-14}
\end{equation}
Now, we choose $\Lambda=\Lambda(V)$ sufficiently large, for details see Appendix 5. By perturbation and Rouch\'{e}'s Theorem, $D_{i_0}$ has exactly $J$ zeros in $T_0$ (obviously, even in twice more narrow neighborhood) and
\begin{equation}\label{highen1}
|D_{i_0}(\tau_1)|>(\sigma_{large}/100)^J/2\ \ \ \hbox{on}\ \partial T_0.
\end{equation}
(For a proof see Appendix 5.)  It follows $D_{i_0}(\tau_1)=g(\tau _1)\prod _{j=1}^J(\tau -\tau _j)$ in $T_0$, where, by the minimum principle, $|g(\tau _1)|>2^{-5}$.
This means that there is $\tau _1 \in (\tau _{1,0}, \tau _{1,0}+\eta)$ such that $|D_{i_0}(\tau_1)|>2^{-17}\eta ^4$ and thus $\tau _1\not\in
S_{\q}$. Indeed, if $\tau _1\in
S_{\q}$ then, by \eqref{highen},  $|D_{i_0}(\tau_1)|<2\varepsilon _0<\eta^6k^{-1}$ and we arrive at the contradiction.

Obviously, the estimate \eqref{Nov29-13a} (and even better one, since we have at most four poles) holds in this case too.
\end{proof}

 If $\k \in \RR_{\q}$ and $\k +\p_\m \in \RR_{\q}$, $0<\||\p_\m\||<k^{\delta }$, then $\m =c\q$, $c\in \R$.
This  easily follows from the definition of $\RR_{\q}$. Let $S_\q(\varepsilon _0)$ be defined by $D_{aux}$ as in the previous lemma.
We define $S_\q(\varepsilon _0)^{simple}\subset S_\q(\varepsilon _0)$  and
$S_\q(\varepsilon _0)^{double }\subset S_\q(\varepsilon _0)$ as follows:
\begin{equation} S_\q(\varepsilon _0)^{simple}=\{\k\in S_\q(\varepsilon _0):  \k +\p_\m \not \in S_\q(\varepsilon _0)\  \mbox{when } 0<\||\p_\m\||<k^{\delta }, \m \not \in \cal M_*\}. \label{Ssimple}
\end{equation}
\begin{equation} S_\q(\varepsilon _0)^{double}=\{\k\in S_\q(\varepsilon _0):  \k +\p_{\m _0} \in S_\q(\varepsilon _0),\ \hbox{for some}\ \m_0\not \in \cal M_*, \ \mbox{and}   \label{Sdouble}
\end{equation}
$$ \k +\p_\m \not \in S_\q(\varepsilon _0)\ \mbox{when } 0<\||\p_\m\||<k^{\delta },\ \m\not \in {\cal M_*},\ \m\not= \m _0 \}.$$
It is easy to show that the analogous definition of $S_\q(\varepsilon _0)^{triple}$ gives the empty set; the proof just follows the arguments from the proof of Lemma~\ref{per3} (we recall that we work with  intervals of $\tau_1$ just small enough so that pairs of corresponding eigenvalues $\lambda_{i_0},\,\lambda_{i_0+1}$ stay stable).

Next we define the projection $P_\q$. First, let $\k \in S_\q(\varepsilon _0)^{simple}$. We consider the $k^{\delta }$-neighborhood of ${\cal M _*}$. Next, if it contains any $\m: \k +\p_{\m}\in \RR_{\q}$, then we attach the whole set ${\m}+{\cal M _*}$ to the neighborhood.\footnote{Note that these new points are not in $S_{\q}$. Indeed, suppose
$\k +\vec p_{\m +n\q} \in S_{\q}$, $n\q\in {\cal M _*}(\k +\vec p_{\m})$. Since the definition of $S_\q $ is the same for $\k +\vec p_{\m}$ and  $\k +\vec p_{\m +n\q}$ (see footnote
\ref{f1}) this means
that $\k +\vec p_{\m}\in S_{\q}$, and this is not the case by definition of $S_\q(\varepsilon _0)^{simple}$.} Thus, we obtained the $k^{\delta }$-neighborhood  with ``one dimensional branches" growing out of it. We denote this set by  $\tilde{\cal M} _*$ and call it the extended $k^{\delta }$-neighborhood of ${\cal M _*}(\k)$.
In the case of $\k \in S_\q(\varepsilon _0)^{double}$ the construction of $\tilde{\cal M} _*$ is analogous. Namely, we consider  the union of $\tilde{\cal M} _*(\k)$ and
$\tilde{\cal M} _*(\k+\p_{\m_0})$ being constructed as for the simple case.  With a slight abuse of notations this union we denote by $\tilde{\cal M _*}(\k)$ again. Let $P_{\q}$ be the projection, corresponding to $\tilde{\cal M} _*(\k)$ constructed for the simple and double cases.
We prove the following lemma for the case $S_\q(\varepsilon _0)^{double}$. The case $S_\q(\varepsilon _0)^{simple}$  is completely analogous, just simpler.
By $P_1$, $P_2$ we denote projectors corresponding to ${\cal M _*}(\k)$, ${\cal M _*}(\k+\p_{\m_0})$ respectively.  Let $P=P_1+P_2$ and $P'$ be the projection corresponding to the ``branches", i.e. to
the union of the sets $\m +{\cal M}_*$: $\k +\vec p_{\m}\in \RR _\q $, $\m \neq {\bf 0}, \m_0 $, and $\m \not \in {\cal M _*}$, i.e. it corresponds to ``weakly-resonant" sets $\m +\cal M _*$ (cf. with the corresponding definition for $\MM_2^{j,s}$).

As before, we introduce the coordinates $\tau_1,\,\tau_2$, the operator $P_\q(H(\k)-k^2-\varepsilon_0')P_\q$ and  its determinant $D(\k,k^2+\varepsilon_0')$. Note that
as above, $P_{\q}, P, P'$ are piece-wise constant with respect to $\tau_1$.

\begin{lemma}\label{4.9nn} Let $0<\varepsilon_0\leq k^{-242\mu\delta}$. The sets $D(\k,k^2)=0$ and ${S}_\q$ have the following properties in $\RR_\q $:

1) The equation $D(\k,k^2)=0$ describes at most eight curves in $\RR_\q$ which are represented by $\tau_2=f_i(\tau_1)$ where $\k=(\tau_1,\tau_2)$ and $|f_i'|\leq k^{\delta_*+\mu\delta-1}$, $i=1,\dots,8$ \footnote{In fact, we can still use $4$ instead of $8$ since, as explained at the end of the proof of Lemma~\ref{per3}, we can have two operators in the cluster only when quasi-momentum is {\em not} close to the boundary, i.e. when we have only one eigenvalue and two zeros for $D_{aux}$ corresponding to $P_1$ and $P_2$.}.

2) The set ${S}_\q$ belongs to $\cup_{i}S_i$, $$S_i(k,\varepsilon_0):=\{\k:  |\tau_2-f_i(\tau_1)|<2\varepsilon_0 k^{-1},\ |\tau_1|\leq k^{\delta_*+\mu\delta}\}.$$

3) Every curve $\tau_2=f_i(\tau_1)$ contains no more than $2^{31}k^{8\delta}$ inflection points.

4) Let $\l$ be a segment of a straight line,
\begin{equation}\l=\{\k=(\tau_1, \beta _1 \tau _1+\beta _2),\  \tau _{1,0}<\tau _1<\tau _{1,0}+\eta
\},\ \ |\tau_{1}|<k^{\delta_*+\mu\delta},\label{segment-1nn}\end{equation}
such that both its ends belong to ${S}_{i}(k,\varepsilon _0)$,
    $2\varepsilon _0<\eta^{12}k^{-1}$, $0<\eta<k^{-41\mu\delta}$. Then, there is an
inner part $\l'$ of the segment
 which is not in $S_{i}(k,\varepsilon _0)$. Moreover, there is a point $(\tau _{1*},\tau _{2*})$ in $\l'$ such that
 $f'_i(\tau _{1*})=\beta _1$, i.e., the curve and the segment have
 the same direction when $\tau _{1}=\tau _{1*}$.
\end{lemma}

\begin{proof}
Let us consider the operator
$$ H_{aux}:=PHP+P'HP'+(P_\q-P-P')H_0.
$$
By the previous lemma, the resolvent of this operator as a function of $\tau _2$ has poles at the points $\tau _2=f_i(\tau _1)$, $i\leq 8$. Let $T_{\tau_2}$ be the $k^{-1-40\mu \delta }$-
neighborhood of these points in the complex plane of $\tau _2 $ ($\tau _1$ is fixed).
We consider $P_{\q}HP_{\q}$ as a perturbation of $H_{aux}$. Next we construct the perturbation series. Indeed,
\begin{equation*}
\left(H_{aux}-k^2\right)^{-1/2}\left(P_\q(H-k^2)P_\q\right)\left(H_{aux}-k^2-\varepsilon_0'\right)^{-1/2}=
P_\q+A,
\end{equation*}
where
\begin{equation*}
\begin{split}&
A:=\left(H_{aux}-k^2\right)^{-1/2}W\left(H_{aux}-k^2\right)^{-1/2},\cr &
W:=(P_\q-P-P')V(P+P')+(P+P')V(P_\q-P-P')+(P_\q-P-P')V(P_\q-P-P').
\end{split}
\end{equation*}
Because of separation of variables,  $\tau_2^2$ can be treated as a spectral parameter.
Hence,
\begin{equation}\label{totalen1}\|(H_{aux}-k^2)^{-1}\|=O(k^{40\mu\delta})\ \ \ \hbox{when}\ \tau _2\in \partial T_{\tau_2}.\end{equation} Next,
$$\|(P_\q-P-P')(H_0-k^2)^{-1}\|\leq ck^{-\delta_*},$$
since  $(P_\q-P-P')_{\m\m}\neq 0$ only if $\m \not \in \RR_{\q,nontriv}$, see \eqref{Rqn}. Using the last two estimates, we obtain:
\begin{equation}\label{Anontriv}
\|A\|\leq k^{-\delta_*/4}\ \ \ \hbox{when }\ \tau _2\in \partial T_{\tau_2}.
\end{equation}
The last inequality, in particular, shows that the perturbation series for $(P_\q(H-k^2)P_\q)^{-1}$ with respect to $(P_\q(H_{aux}-k^2-\varepsilon_0')P_\q)^{-1}$ converges and we have (see \eqref{totalen1})
\begin{equation}\label{totalen1'}
\|(P_\q(H-k^2)P_\q)^{-1}\|=O(k^{40\mu\delta})\ \ \ \hbox{when}\ \tau _2\in \partial T_{\tau_2}.
\end{equation}
The perturbative arguments used in the proof of Lemma~\ref{trivial} show that $\left(P_\q(H-k^2)P_\q\right)^{-1}$ has the same number of zeros in $T_{\tau_2}$ as $\left(H_{aux}-k^2\right)^{-1 }$, i.e. not more than $8$.  Thus, solutions of  $D:=\hbox{det}\,\left(P_\q(H-k^2)P_\q\right)$  are described by at most 8 curves $\tau _2=f_i(\tau _1)$. The rest of the proof  of the statements 1-3 is completely the same as in Lemma~\ref{4.9n}.

Let us prove the fourth statement. By Lemma~\ref{4.9n}, the estimate \eqref{Nov29-13a} holds for $P=P_1+P_2$. 
Choosing $r=\frac12 k^{-40\mu \delta }$, 
we obtain
\begin{equation}\label{totalen2}
\|(P(H-k^2)P)^{-1}\|\leq k^{241\mu\delta}\ \ \ \hbox{when}\ \tau _1\in \partial T_{\tau_1}.
\end{equation}
Here, by $T_{\tau_1}$ we denote $\frac12 k^{-40\mu\delta}$-neighborhood of zeros corresponding to both $P_1HP_1$ and $P_2HP_2$. The resolvent $(P(H(\tau_1)-k^2)P)^{-1}$ has no more than 12 poles in $T_{\tau_1}$ (all of them situated in the twice more narrow neighborhood).


By definition, the resolvent $(P'(H-k^2)P')^{-1}$ has no poles in the $k^{-40\mu \delta }$-neighborhood of $\tau _{1,0}$. Hence, the estimate similar to
\eqref{totalen2} holds for $(P'(H-k^2)P')^{-1}$. Further, using the perturbative arguments as above, we obtain:
\begin{equation}\label{totalen2'}
\|(P_\q(H-k^2)P_\q)^{-1}\|\leq 2k^{241\mu\delta}\ \ \ \hbox{when}\ \tau _1\in \partial T_{\tau_1},
\end{equation}
and $(P_\q(H-k^2)P_\q)^{-1}$ has no more than 12 poles in $T_{\tau_1}$. Applying the same scaling as many times before we see that
\begin{equation}\label{totalen2''}
\|(P_\q(H-k^2)P_\q)^{-1}\|\leq 2k^{241\mu\delta}\left(\frac{\frac12 k^{-40\mu\delta}}{r}\right)^{12},
\end{equation}
where $0<r<\frac12 k^{-40\mu\delta}$ is the distance to the nearest pole of the resolvent. Hence, there is a point in $(\tau _{1,0}, \tau _{1,0}+\eta )$ where the norm of the resolvent is smaller than $k^{-238\mu\delta}\eta ^{-12}$. This means that the point is not in $S_{\q}$.  \end{proof}

Let $N_{\q, nontriv}(k ,r_1,\k _0,\varepsilon_0)$ be the number of points $\k _0
+\p_{\n}$, $\||\p_{\n}\||<k^{r_1}$ in $S_{\q}^{single}(k,\varepsilon _0)$, when $\k _0$ is in a $S_{\q}^{single}(k,\varepsilon _0)$, or in $S_{\q}^{double}(k,\varepsilon _0)$, when $\k _0$ is in a $S_{\q}^{double}(k,\varepsilon _0)$, $\k _0$ being fixed. The analogue of Lemma~\ref{4.10} holds.
\begin{lemma} \label{nontriv4.10} Let $\delta_* <r_1<\infty $. If $0<\varepsilon _0<k^{-16\mu r_1}$, then the number
of points $N_{\q, nontriv}(k ,r_1,\k_0,\varepsilon_0)$ admits the estimate
\begin{equation}\label{test4.10}N_{\q, nontriv}(k ,r_1,\k_0,\varepsilon_0)\leq
2^{44}k^{2r_1/3+8\delta_*}.\end{equation}\end{lemma}
\begin{remark} \label{Mar17-14}
Note that the the results of Lemmas \ref{4.9n}, \ref{4.9nn}, and, hence, that of Lemma  \ref{nontriv4.10}, are stable with respect  to variations of $\M_*$ at the ends. Indeed, assume that $\M_*'\subset \M_*$ and $\M_*'$ contains a set similar to $\M_*$ with, say, $\frac 78$ instead of $\frac 98$ in \eqref{M*}. Then, Lemmas \ref{4.9n}, \ref{4.9nn} hold for $\M_*'$ too, since the estimate \eqref{Feb14} and others are  stable with respect to such a perturbation of $\M_*$. The sets $S_\q(\varepsilon _0)$ are essentially the same for $\M_*$ and $\M_*'$  when $\varepsilon _0>
ck^{-\frac{\delta ^*}{C(Q)}k^{\delta ^*/2}}$, see \eqref{Feb14}.
\end{remark}

\subsection{Preparation for  Step III - Analytic Part \label{S:3}}
\subsubsection{Model Operator for Step III\label{MOforStep3}}
Let $r_2>r_1>10^8$. Further we use the notation: \begin{equation}
\label{box} \Omega (r_2)=\{\m:\||\p_\m\||<k^{r_2}\}.
\end{equation}
 We repeat for $r_2$ the construction of
Section \ref{MOforStep2} which was done for an arbitrary
$r_1>2$. It is easy to see that the whole construction is monotonous
with respect to $r_1$. Namely,
$$\M(\varphi _0,r_1) \subseteq\M(\varphi _0,r_2), \ \M'(\varphi _0,r_1)\subseteq \M'(\varphi
_0,r_2),$$ $$
\M_1(\varphi _0,r_1)\subseteq \M_1(\varphi _0,r_2),\ \ \M_2(\varphi _0,r_1)\subseteq \M_2(\varphi _0,r_2).$$
Again, following the procedure in Section \ref{MOforStep2}, we split $\M_2(\varphi _0,r_2)$ into components  $\M_2^j$ and, further, into $\M_2^{j,s}$, see \eqref{Mjs} and the text above.
According to Definition \ref{Nov27-12}, we classify the sets $\M_2^{j,s}$ as weakly resonant and strongly resonant. Let $\M_2^{weak}$ be the union of all weakly resonant sets $\M_2^{j,s}$ and, correspondingly $\M_2^{str}=\M_2\setminus \M_2^{weak}$. By $\MM_{2,tw}$ we denote trivial weakly resonant points. The main difficulty at this step is in treating $\M_2^{str}$.\footnote{It follows from Definition \ref{Nov27-12} that  each trivial weakly resonant $\M_2^{j,s}$ can be treated as points from the complement to $\M$ with $k^{-40\mu\delta}$ instead of $k^{\delta_*}$ in \eqref{resonance1}, \eqref{M} (it does not make a lot of difference). A non-trivial weakly resonant set can be treated similarly up to technical details.}


Next, we introduce an analog of $\M $, see \eqref{M},  for the second step. Indeed, let $\varphi_0\in \omega^{(2)} (k, \delta ,1)$ and
\begin{equation}\label{M^2} {\MM}^{(2)}(\varphi _0, r_2):={\MM}^{(2)}(\varphi _0, r_2)=\{\m\in \M\setminus \M_2^{weak}
:\ \varphi_0\in{\cal O}_\m^{(2)}(10r_1',1)\},\end{equation} where
${\cal O}_\m^{(2)}(10r_1',\tau)$ is the union of the disks of the
radius $\tau k^{-10r_1'}$ with the centers at poles of the resolvent
of $k^\delta$-component containing $\k^{(1)}(\varphi_0)+\p_\m$. More
precisely, for each $\m\in \M(\varphi _0, r_2)\setminus \M_2^{weak}(\varphi _0, r_2)$ we construct
the $k^\delta$-neighborhood  (in $|||\cdot \||$ norm) around it.
 Note that only points from $\MM_1$ generate separated $k^\delta$-boxes around each of them. Points corresponding to strongly resonant sets $\MM_2^{j,s}$ generate $k^\delta$-clusters around no more than two different strongly resonant $\MM_2^{j,s}$, thus their size is $O(k^{\delta_*})$ along $\MM_2^{j,s}$.
The details are provided in Step II (see \eqref{defP}). So, when we say $k^\delta$-component/cluster/box we, in fact, mean one of these sets. 
Thus, ${\cal
O}_\m^{(2)}(10r_1',\tau)$ is the union of the disks of the radius
$\tau k^{-10r_1'}$ with the centers at poles of the operator
$(P({\m})(H(\k^{(1)}(\varphi))-k^{2}I)P({\m}))^{-1}$, where $P({\m})$ is the
projection onto a particular $k^\delta$-component containing
$\k^{(1)}(\varphi_0)+\p_\m$. 
For $\m$ generating the same $k^\delta$-component corresponding
sets ${\cal O}_\m^{(2)}(10r_1',\tau)$ are identical. By construction
of the non-resonant set $\omega^{(2)} (k, \delta ,1)$, we have
${\MM}^{(2)}\cap \Omega (r_1)=\emptyset $.

Further we use the property of the set $\MM^{(2)}$ formulated in the
next lemma.

\begin{lemma}\label{L:2/3-1} Let  $\m _0\in \Omega (r_2)$, $1/20<\gamma '<20$ and $\Pi _{\m_0}$ be the $k^{\gamma 'r_1}$-neighborhood (in $\||\cdot\||$-norm) of $\m_0$.
Then, the set $\Pi
_{\m _0}$ contains less than $ck^{2\gamma'r_1/3+1}$ elements of
$\MM^{(2)}$.
\end{lemma}
\begin{proof}
If $\m\in \MM ^{(2)}$, then there is a $\varphi _*: |\varphi
_0-\varphi _*|<k^{-10r_1'}$ such that
\begin{equation}\label{gulf5}\det \Big(P({\m})\big(H(\k^{(1)}(\varphi _*))-k^{2}I\big)P({\m})\Big)=0.\end{equation}
Therefore, for some $\varepsilon_0':|\varepsilon _0'|<\varepsilon
_0,\ \varepsilon _0:=ck^{1-10r_1'}$, \begin{equation}
\label{Feb21-1}\det \Big(P({\m})\big(H(\k^{(1)}(\varphi
_0))-(k^{2}+\varepsilon _0')I\big)P({\m})\Big)=0.
\end{equation}
Indeed, if \eqref{Feb21-1} holds for no $\varepsilon '$, then $\left\|\Big(P({\m})\big(H(\k^{(1)}(\varphi
_0))-k^{2}I\big)P({\m})\Big)^{-1}\right\|<ck^{-1+10r_1'}$, since $\varphi _0$ is real and, hence, $H(\k^{(1)}(\varphi _0)$ is selfadjoint. Using Hilbert identity, we obtain that
$\Big(P({\m})\big(H(\k^{(1)}(\varphi
_*))-k^{2}I\big)P({\m})\Big)^{-1}$ is bounded. This contradicts to \eqref{gulf5}. Hence, \eqref{Feb21-1} holds for some $\varepsilon_0':|\varepsilon _0'|<\varepsilon
_0$.

Suppose $\m\in \MM_1(\varphi _0, r_2)$. Then, \eqref{Feb21-1} means
that $|\lambda ^{(1)}(\k^{(1)}(\varphi
_0)+\p_{\m})-k^{2}|<\varepsilon _0$. Introducing the notation
$\k_0=\k^{(1)}(\varphi _0)+\p_{\m_0}$, we rewrite the last inequality
in the form: $|\lambda ^{(1)}(\k_0+\p_{\m-\m_0})-k^{2}|<\varepsilon
_0$, where $\||\p_{\m-\m_0}\||<k^{\gamma ' r_1}$. It follows that
$\k_0+\p_{\m-\m_0}$ is in the real $c\varepsilon _0
k^{-1}$-neighborhood of ${\cal D}_1(k^{2})$.  Applying Lemma
\ref{L:number of points-1}, we obtain that the number of such points does not exceed
$ck^{2\gamma ' r_1/3+1}$.

Let $\m \in \MM _2^{str}(\varphi _0,r_2)$.
Namely, let us consider all $\m$ belonging  to a particular component $\MM_{2,str}^{j,s}(\varphi _0, r_2)$.
Then, $\left||\k^{(1)}(\varphi _0)+\p_{\m}|^2_{\R}-k^2\right|<k^{\delta_* }$ and \eqref{Feb21-1} holds, $P({\m})$ being the projection
on $\tilde \MM_{2,str}^{j,s}(\varphi _0, r_2)$.
 Using again the notation
$\k_0=\k^{(1)}(\varphi _0)+\p_{\m_0}$, and the definition of
$\MM_{2,str}^{j,s}(\varphi _0, r_2)$, we see that
in terms of Section \ref{Lattice Points in a Resonant Set},
\eqref{Feb21-1} means $\k_0+\p_{\m-\m_0}\in S_{\q}(k, \varepsilon
_0)$, see \eqref{Sq}.
Applying Lemmas \ref{4.10},\ref{t4.10},\ref{2t4.10},\ref{nontriv4.10}
 and using
\eqref{Feb21-1}, we obtain that the number of such points does not
exceed $ck^{2\gamma ' r_1/3+8 \delta_*}$ for a fixed set $\MM_{2,str}^{j,s}(\varphi _0, r_2)$.\footnote{Note that in Lemmas \ref{4.9n}, \ref{4.9nn} the set $\M_*$ essentially coincides
with $\MM_{2,str}^{j,s}(\varphi _0, r_2)-\m$, see also Remark \ref{Mar17-14}. } Since the number of such sets is bounded by $ck^{4\delta}$, the lemma is proven. 
\end{proof}

Let us split $k^{r_2}$-box into $k^{\gamma r_1}$-boxes as described
below. In the whole construction below we will have $\gamma
=\frac{1}{5}$, but in some cases we will refer to the similar
estimates with other values of $\gamma$. That's why in what follows
we prefer to use implicit notation.  The procedure consists of
several steps. On each step we introduce a new scale of a box.
Further structure will acquire additional scales at each step of
approximation procedure. This is why we call the procedure
Multiscale Construction in the Space of Momenta.
\begin{enumerate} \item {\em Simple region.} \label{simple} Let $\Omega _s^{(2)}(r_2)$ be the  collection of $\m\in \Omega(r_2)$ with small values of $p_\m$, namely,
$\Omega _s^{(2)}(r_2)=\{\m\in \Omega(r_2):0<p_\m\leq k^{- 5r_1'}\}$.  Then, $k^{r_1}$-boxes around such $\m$-s are similar to $\Omega (r_1)$, since $p_\m $ is small. Indeed,
it is easy to see  that $\Omega _s^{(2)}(r_2)\subset \MM (\varphi _0,r_2)$, since $p_{\m}$ is small, see \eqref{M}, \eqref{resonance1}. Next, if  $\m\in \Omega_s^{(2)}(r_2)$, then there are
 no other elements of $\MM (\varphi _0,r_2)$ in
 the $k^{\delta }$-box  around  $\m$. Indeed, let $\k=\k^{(1)}(\varphi _0)+\p_\m$. It is  is a small perturbation of
$\k^{(1)}(\varphi _0)$, hence it satisfies $\left||\k +\p_{\n}|^2-
|\k|^2\right|> \frac \tau 2 k^{1-40\mu \delta }(1+o(1))$ when
$0<\||\p_\n\||<k^{\delta }$, see \eqref{jan28b}. This means $\m +\n
\not \in \MM (\varphi _0,r_2)$. Further, if $\m\in
\Omega_s^{(2)}(r_2)$, then there are no other  elements  of $\Omega
_s^{(2)}(r_2)$ in the surrounding box of the size $k^{r_1}$, see
\eqref{below}. Last, $\m$ itself can belong or do not belong to
$\MM^{(2)}$,  but there are
 no other elements of
$\MM^{(2)}$ in the $k^{r_1}$-box  around such $\m$. Indeed,
$\k^{(1)}(\varphi _0)$ satisfies the conditions of Lemma
\ref{L:geometric2}.  This means that the $k^{\delta }$-cluster
around each $\q$: $0<\||\p_\q\||<k^{r_1}$ is non-resonant in the sense of \eqref{M^2}-\eqref{Feb21-1}. Moreover,
the $k^{\delta }$-box around each $\m+\q$: $0<\||\p_\q\||<k^{r_1}$
is non-resonant too, since $p_{\m}$ is sufficiently small. This
means $\m+\q \not \in \MM^{(2)}(\varphi _0, r_2)$.

For each $\m \in \Omega _s^{(2)}(r_2)$  we consider its $k^{r_1/2}$-neighborhood. The union of
such boxes we call the simple region and denote it by $\Pi _s(r_2)$. The
corresponding projection is $P_s$. Note, that the distance from the simple region to the nearest point of $\MM^{(2)}$ is greater than $\frac12k^{r_1}$.

\item {\em Black region.} Next, we split $\Omega (r_2)\setminus \left(\Omega (r_1)\cup \Pi _s\right)$ into boxes of the size $k^{\gamma r_1}$. All elements $\m \in \MM^{(2)}$ there satisfy
$p_\m>k^{-5 r_1'}$. We call a box black, if  together with its
neighbors it contains more than $k^{\gamma r_1/2+\delta _0r_1}$
elements of $\MM ^{(2)}$, $\delta_0=\gamma /100$ (in particular
$\delta_0r_1>100$). Let us consider all "black" boxes together with
their $k^{\gamma r_1+\delta _0r_1}$-neighborhoods. We call this the
black region.  Note that that the size of the neighborhoods involved is much smaller than the size of the neighborhoods $k^{r_1/2}$ for the simple region, since $\gamma +\delta _0<\frac12$.
The estimates for the size of the black region will be proven in Lemma \ref{L:black}. We denote the black  region by $\Pi _b$. The
corresponding projector is $P_{b}$. Obviously the distance between black and simple regions is greater than $\frac12k^{r_1}$.

\item {\em Grey region.} By a white box we mean a
$k^{\gamma r_1}$-box, which together with its neighbors contains no
more than $k^{\gamma r_1/2+\delta_0r_1}$ elements of $\MM ^{(2)}$.
Every white box we split into "small" boxes of the size $k^{\gamma
r_1/2+2\delta_0r_1}$. We call a small box "grey", if together with
its neighbors it contains more than $k^{\gamma r_1/6-\delta_0r_1}$
elements of $\MM ^{(2)}$.  The grey region is the union of all grey
small boxes together with their $k^{\gamma
r_1/2+2\delta_0r_1}$-neighborhoods.  Note that that the size of the
neighborhoods involved is much smaller than the size of the
neighborhoods  the simple and black regions. The estimates for the
size of the grey region will be proven in Lemma \ref{L:grey}. The
notation for this region is $\Pi _g$. The corresponding projector is
$P_{g}$. The part of the grey region, which is outside the black
region, we denote by $\Pi _g'$ and the corresponding projection by
$P_g'$.   Obviously, the distance between grey and simple regions is
greater than $\frac12k^{r_1}$.

\item {\em White region.} By a white small box we mean a small box, which
together with its neighbors has no more than $k^{\gamma
r_1/6-\delta_0r_1}$ elements of $\MM ^{(2)}$.  In each small white
box we consider $k^{\gamma r_1/6}$-boxes around each point of
$\MM^{(2)}$. The union of such $k^{\gamma r_1/6}$-boxes we call the
white region and denote it by $\Pi _w$. The corresponding projection
is $P_w$.  Note that the size of the neighborhoods involved is much
smaller than the size of the neighborhoods  the simple, black and
grey regions. The estimates for the size of the white region will be
proven in Lemma \ref{L:white}. The part of the white region which is
outside the black and grey regions, we denote $\Pi _w'$ and the
corresponding projection by $P_w'$. Obviously, the distance between
grey and simple regions is greater than $\frac12k^{r_1}$.

\item {\em Non-resonant region.} We also consider $k^{\delta }$-components surrounding  points in the set
$\MM(r_2, \varphi _0)\setminus \left(\MM(r_1, \varphi _0)\cup
\MM^{(2)}\cup \Omega _s^{(2)}(r_2)\cup\MM_{2,tw}(\varphi _0,r_2)\right)$. The union of these components we call
the non-resonant region $\Pi _{nr}$. The corresponding projection
is $P_{nr}$.  The part of the non-resonant region which is outside
$\Pi_s\cup\Pi_b\cup\Pi_g\cup\Pi_w$, we denote $\Pi _{nr}'$ and the
corresponding projection by $P_{nr}'$.


\end{enumerate}


Let
\begin{equation} P_r:=P_s+P_b+P_g'+P_w',\ \ \ P^{(2)}:=P_r+P_{nr}'+P(r_1), \label{P^2} \end{equation}
index $r$ standing for "resonant".


First, we establish $3k^{\gamma r_1+\delta _0r_1}$-equivalence
relation between black boxes. Then the set $\Pi _b$ can be
represented as the union of components (clusters) separated by distance no less
than $k^{\gamma r_1+\delta _0r_1}$. We denote  such a component by
$\Pi _b^j$.
\begin{lemma} \label{L:black} \begin{enumerate} \item
Each $\Pi _b^j$ contains no more than $ck^{\gamma r_1/2-\delta
_0r_1+3}$ black boxes.
\item The size of $\Pi _b^j$ in $\||\cdot \||$ norm is less than
$ck^{3\gamma r_1/2+3}$.
\item Each $\Pi _b^j$ contains no more than $ck^{\gamma r_1+3}$ elements of
$\MM^{(2)}$. Moreover, any box of $\||\cdot \||$-size $ck^{3\gamma r_1/2+3}$ containing $\Pi _b^j$ has no more than $ck^{\gamma r_1+3}$ elements of
$\MM^{(2)}$ inside.
\end{enumerate}
\end{lemma}
\begin{proof} Let $n_b$ be the number of black boxes in $\Pi
_b^j$,
 $L_b$ be the size of $\Pi _b^j$ and $N_b$ the number of elements of
$\MM^{(2)}$ in $\Pi _b^j$.  Obviously, $L_b<n_b3k^{\gamma r_1+\delta
_0r_1}$ and $N_b>cn_bk^{\gamma r_1/2+\delta _0r_1}$. By Lemma
\ref{L:2/3-1}, $N_b<cL_b^{2/3}k$. Solving the last three
inequalities for $n_b$, we get $n_b<ck^{\gamma r_1/2-\delta
_0r_1+3}$. It follows $L_b<ck^{3\gamma r_1/2+3}$. Next, we consider
a box of the size $k^{3\gamma r_1/2+3}$, containing $\Pi _b^j$.
Using again  Lemma \ref{L:2/3-1}, we obtain that the number of
elements of $\MM^{(2)}$ in this box is less than $cL_b^{2/3}k$.
Therefore,  $N_b<ck^{\gamma r_1+3}$.
\end{proof}

Second, we establish $3k^{\gamma r_1/2+2\delta _0r_1}$-equivalence
relation between small grey boxes. Then the set $\Pi _g$ can be
represented as the union of components separated by distance no less
than $k^{\gamma r_1/2+2\delta _0r_1}$. We denote each such component
as $\Pi _g^j$.
\begin{lemma}\label{L:grey} \begin{enumerate} \item
Each $\Pi _g^j$ contains no more than $ck^{\gamma r_1/3+2\delta
_0r_1}$ grey boxes.
\item The size of $\Pi _g^j$ in $\||\cdot \||$ norm is less than
$ck^{5\gamma r_1/6+4\delta _0r_1}$.
\item Each $\Pi _g^j$ contains no more than $k^{\gamma r_1/2+\delta _0r_1}$ elements of
$\MM^{(2)}$.\end{enumerate}
\end{lemma}
\begin{proof} Let us consider a part of $\Pi _g^j$
belonging to one "big" white box. Let $n_g$ be the number of grey
boxes in $\Pi _g^j$.
 Let $L_g$ be the size of $\Pi _g^j$ and $N_g$ be  the number of elements of
$\MM^{(2)}$ in $\Pi _g^j$.  Obviously, $N_g>cn_gk^{\gamma r_1
/6-\delta _0r_1}$. By definition of a big white box $N_g<k^{\gamma
r_1/2+\delta _0r_1}$. Therefore, $n_g<ck^{\gamma r_1/3+2\delta
_0r_1}$. Clearly, $L_g<n_g3k^{\gamma r_1/2+2\delta
_0r_1}<ck^{5\gamma r_1/6+4\delta _0r_1}$. Since $\delta  _0<\gamma
/24$, we obtain that the size of each grey component is much less
than the size $k^{\gamma r_1}$ of a big box. The lemma is proven under condition that
$\Pi _g^j$ is inside one of white boxes. Suppose $\Pi _g^j$ intersects more
than one white box. Considering that the size of $\Pi _g^j$ in each
big white box is much less than the size of this box, we conclude that $\Pi
_g^j$
 fits into neighboring boxes and satisfies the estimates proven
above.
\end{proof}

Third, we consider points of $\MM^{(2)}$ in small white boxes. We
establish $3k^{\gamma r_1/6}$-equivalence relation between them.
Considering $k^{\gamma r_1/6}$-neighborhoods of the points in
$\MM^{(2)}$, we see that this neighborhoods form clusters $\Pi _w^j$
of $\Pi _w $ separated by the distance no less  than $k^{\gamma
r_1/6}$. The number of $\MM^{(2)}$ points in a white cluster we
denote by $N_w^j$.
\begin{lemma}\label{L:white} \begin{enumerate} \item The size of $\Pi _w^j$ in $\||\cdot \||$ norm is less than
$ck^{\gamma r_1/3-\delta _0r_1}$.
\item Each $\Pi _w^j$ contains no more
than $k^{\gamma r_1/6-\delta _0r_1}$ points of $\MM^{(2)}$.
\end{enumerate}
\end{lemma}
\begin{proof} Let us consider points of $\MM^{(2)}$ in a small white
box. By the definition of the white small box, the number of such
points does not exceed $k^{\gamma r_1/6-\delta _0r_1}$. We consider
the $k^{\gamma r_1/6}$-neighborhoods of these points. They can form
clusters. The total contribution from all points of $\MM^{(2)}$ in
the small white box and its neighbors, obviously, does not exceed
$3k^{\gamma r_1/3-\delta _0r_1}$, which is much less than the size
of a small white box. Therefore, each $\Pi _w^j$ can't spread
outside of the small white box and its neighbors. This proves both
statements of the lemma.
\end{proof}

\hskip 1cm

Next, we slightly change definitions of the simple, black, grey and white
areas to adjust their boundary to the structure of clusters. Namely (cf. construction of the $k^\delta$-neighborhood of strongly resonant sets $\MM_2^{j,s}$ above),
if $k^{\delta }$-cluster generated by points of $\MM(\varphi
_0,r_2)\setminus (\MM^{(2)}\cup\MM_{2,tw}(r_2, \varphi _0))$  intersects a $k^\delta$-neighborhood of a simple, white, grey or black
area, then we include it into the corresponding region.
 This ``addition" does not change
formulation of Lemmas \ref{L:black}, \ref{L:grey}, \ref{L:white},
since the size of a $k^{\delta }$-cluster is $O(k^{\delta_*})$ which is much smaller that the
sizes of  $\Pi _b^j$, $\Pi _g^j$, $\Pi _w^j$. If a white cluster has a distance less than $k^{\gamma r_1 /6}$ to a grey or black cluster, we include it into that with the lighter color.
This ``addition" also does not change formulation of Lemmas
\ref{L:black}, \ref{L:grey},  since the size of a white cluster is
much smaller that the characteristic sizes of $\Pi _b^j$, $\Pi _g^j$.  If a grey
cluster has a distance less than $k^{\gamma r_1/2+2\delta _0r_1}$  to a black cluster $\Pi _b^j$, we include it into this $\Pi _b^j$.
This ``addition" does not change formulation of Lemma \ref{L:black},
since the size of a grey cluster is much smaller that the characteristic size of
any $\Pi _b^j$. The new structure has the following properties. 
If the intersection of the $k^{\gamma r_1/6}$-neighborhood of a white cluster with  grey or black
area is not empty, then this cluster is completely in this area. If
  the intersection of the $k^{\gamma r_1/2+2\delta _0r_1}$-neighborhood of a grey cluster with  the black area is
not empty, then this cluster is completely in this area.\footnote{Unlike polyharmonic situation in \cite{KaSh} here we cannot introduce the corresponding $k^\delta$-gap between non-resonant regions or between non-resonant region and all other regions because we can't isolate nontrivial weakly resonant sets. This technical difficulty though can be overcome as in Theorem \ref{Thm2}, see also the proof of Theorem \ref{Thm3} below.}
Recall that $\delta _0<\gamma /24$, $\gamma <1/3$. This means that each
component of the white, grey, black and non-resonance region is much
smaller in $\||\cdot \||$-size than $\Omega (r_1)$. Moreover, there
are no points of $\MM^{(2)}$ inside $\Omega (r_1)$. If the
$k^{\gamma r_1/6}$-neighborhood of a white cluster intersect $\Omega
(r_1)$, we reduce $\Omega (r_1)$ by this neighborhood. This
insignificant reduction does not change Step II. We make a similar
reduction of $\Omega (r_1)$ if it is intersected by neighborhoods of
grey or black clusters.
Sometimes it will be convenient to numerate the projections
$P(r_1)$, $P_b$, $P_g'$, $P_w'$, $P_{nr}'$, $P_s$ by indices
0,1,2,3,4,5 as $P_0,P_1,P_2,P_3,P_4,P_5$. The corresponding sets are
$\Pi _i$. Note that each $\Pi _i$ consists of  components
$\Pi_{ij}$, $j=1, ...,J(i)$ as described in the construction of
sets $\Pi_i$. The distance between closest components $\Pi _{ij}$,
$i=1,2,3,4,5$ with the same first index is greater than $k^{\gamma r_1+\delta_0 r_1}$,  $k^{\gamma r_1/2+2\delta _0 r_1}$, $k^{\gamma r_1/6}$, $k^{\delta }$, $k^{\gamma r_1}$, correspondingly. Then we can rewrite $P^{(2)}$, see \eqref{P^2}, in the form:
\begin{equation}P^{(2)}=\sum _{i=0}^{5}P_i.\label{P(2)}
\end{equation}
We introduce the boundaries $\partial \Omega (r_1)$, $\partial \Pi
_b$, $\partial \Pi _g'$, $\partial \Pi _w'$, $\partial \Pi _{nr}'$,
$\partial \Pi _s$
of the sets $\Omega (r_1)$, $\Pi_b$, $\Pi_g'$, $\Pi_w'$,
$\Pi_{nr}'$, $\Pi_s$ as follows: $\partial \Omega (r_1)$, $\partial
\Pi _b$, $\partial \Pi _g'$, $\partial \Pi _w'$, $\partial \Pi
_{nr}'$, $\partial \Pi _s$ are the sets of points in $\Omega (r_1)$,
$\Pi_b$, $\Pi_g'$, $\Pi_w'$, $\Pi _{nr}'$, $\Pi_s$ which can be connected by $V$ with the complements of $\Omega (r_1)$, $\Pi _b$, $\Pi
_g'$, $\Pi _w'$, $\Pi _{nr}'$, $\Pi _s$, respectively. 
 The corresponding
projectors we denote as $P^{\partial }(r_1)$, $P_b^{\partial}$,
$P_g^{'\partial}$, $P_w^{'\partial}$, $P_{nr}^{'\partial}$,
$P_s^{\partial}$ or $P_i^{\partial}$, $i=0,1,2,3,4,5$.

\begin{lemma} \label{L:boundary} Let $i,i'=0,1,2,3,4,5$, $i\neq i'$. The following relations hold:
\begin{equation}
P_iP_{i'}=0, \label{ort}
\end{equation}
\begin{equation} \label{PVP-2'} P_iVP_{i'}=0,
\end{equation}
\begin{equation}
(I-P^{(2)})VP_i=(I-P^{(2)})VP_{i}^{\partial}. \label{boundary}
\end{equation}\end{lemma}
\begin{corollary} \label{C:PHP-2}Operators $P^{(2)}VP^{(2)}$ and $P^{(2)}HP^{(2)}$ have a block structure. Namely,  \begin{equation}
P^{(2)}VP^{(2)}=\sum _{i=0}^5 P_iVP_i, \   \  \ P^{(2)}HP^{(2)}=\sum
_{i=0}^5 P_iHP_i. \label{PHP-2}
\end{equation}
\begin{equation}
P_iVP_i=\sum _{j}P_{ij}VP_{ij},\  \  \ P_iHP_i=\sum _{j}P_{ij}HP_{ij}, \label{blocks}
\end{equation}\end{corollary}
\begin{proof} The lemma easily follows from the construction of the projectors and Lemma \ref{ortogonal}. The identity \eqref{boundary} is simply the definition of the boundary.\end{proof}


{\it Remark.} Thus, we have constructed a multiscale structure
inside $P^{(2)}HP^{(2)}$,  blocks of different colors having
distinctly different size. Merging blocks of a smaller size
 (a lighter color) with neighboring blocks of a bigger size (a darker color), we made the blocks to be
 separated by the $\|||\cdot \||$-distance greater than the size of the corresponding lighter block. This property will be important for the proof of main theorem of the Step III (see Theorem \ref{Thm3} below). Absence of the similar property for nonresonant region is partially compensated by the fact that potential $V$ still does not connect different nontrivial weakly resonant sets, while strongly resonant sets are still $k^\delta$-separated. Corresponding details were provided in the proof of Theorem \ref{Thm2} and will be used again in the proof of Theorem \ref{Thm3}.

\begin{lemma}\label{L:Pnr}Let $\varphi _0\in \omega^{(2)}(k,\delta ,\tau )$,
$|\varphi-\varphi _0|<k^{-2-40r_1'-\delta }$. Then,
\begin{equation}\label{Pnr}
\left\|\Bigl(P_{nr}\bigl(H(\k^{(2)}(\varphi
))-k^{2}I\bigr)P_{nr}\Bigr)^{-1}\right\|<ck^{40r_1'}.
\end{equation} \end{lemma}
\begin{proof} The set $\Pi _{nr}$ can be presented as $\cup _{j} \Pi _{nr}
^j$,  each $\Pi _{nr}
^j$ being a $k^{\delta }$-component. Let $P_{nr}=\sum _j
P_{nr}^{j}$, where $P_{nr}^{j}$ are projections corresponding
to $\Pi _{nr} ^j$. Then by Corollary \ref{C:PHP-2}, $P_{nr}HP_{nr}=\sum
_jP_{nr}^{j}HP_{nr}^{j}$. Hence, it is enough to prove
\begin{equation}\label{Pnrj}
\left\|\Bigl(P_{nr}^j\bigl(H(\k^{(2)}(\varphi
))-k^{2}I\bigr)P_{nr}^{j}\Bigr)^{-1}\right\|<ck^{40r_1'}.
\end{equation}
By construction, each $\Pi _{nr} ^j$ contains $\m \in \M(\varphi , r_2)\setminus \Omega _s^{(2)}$, $\M(\varphi , r_2)=\M _1(\varphi , r_2)\cup \M _2(\varphi , r_2)$. We can apply Lemmas \ref{L:estnonres1}, \ref{estnonres} and Corollary
\ref{estnonres1}, since they were proven for any $r_1$ (no restrictions from above). We take $\varepsilon _0=k^{-10r_1'}$ in Lemma \ref{L:estnonres1}, since the distance from $\varphi _0$ to the
nearest pole of the operator
$\Bigl(P_{nr}^{j}\bigl(H(\k^{(1)}(\varphi _0
))-k^{2}I\bigr)P_{nr}^{j}\Bigr)^{-1}$ is greater than
$k^{-10r_1'}$, see \eqref{M^2}, and $p_\m>k^{-5r_1'}$ because $\m \not \in \Omega _s^{(2)}$.
By analogy with
Corollary \ref{estnonres1}, we obtain:
\begin{equation}\label{Pnrj-1}
\left\|\Bigl(P_{nr}^{j}\bigl(H(\k^{(1)}(\varphi _0
))-k^{2}I\bigr)P_{nr}^{j}\Bigr)^{-1}\right\|<ck^{40r_1'}.
\end{equation} Taking into account that
$\varkappa^{(2)}(\varphi _0)-\varkappa^{(1)}(\varphi
_0)=o(k^{-1-40 r_1'})$ and $\k^{(2)}(\varphi )-\k^{(2)}(\varphi
_0)=o(k^{-1-40r_1'})$, we arrive at \eqref{Pnrj}.
\end{proof}
\begin{lemma}\label{L:Pr} Let $\varphi _0\in \omega^{(2)}(k,\delta ,\tau )$,
and $|\varphi-\varphi _0|<k^{-44r_1'-2-\delta}$, i=1,2,3. Then, \begin{enumerate}
\item The number of poles of the resolvent $\Bigl(P_i\bigl(H(\k^{(2)}(\varphi
))-k^{2}I\bigr)P_i\Bigr)^{-1}$ in the disc  $|\varphi-\varphi _0|<k^{-44r_1'-2-\delta}$ is no greater than $N_i^{(1)}$, where $N_1^{(1)}=k^{\gamma r_1+3}$, $N_2^{(1)}=k^{\gamma r_1/2+\delta _0r_1}$,
$N_3^{(1)}=k^{\gamma r_1/6-\delta _0r_1}$.
\item Let  $\varepsilon$ be the
distance to the nearest pole of the resolvent in ${\cal W}^{(2)}$
and let $\varepsilon_0:=\min\{\varepsilon,\,k^{-11r_1'}\}$. Then,
the following estimates hold:
\begin{equation}\label{Pr-1}
\left\|\Bigl(P_i\bigl(H(\k^{(2)}(\varphi
))-k^{2}I\bigr)P_i\Bigr)^{-1}\right\|<ck^{44r_1'}\left(\frac{k^{-11r_1'}}{\varepsilon
_0}\right)^{N_{i}^{(1)}},
\end{equation}
\begin{equation}\label{Pr-2}
\left\|\Bigl(P_i\bigl(H(\k^{(2)}(\varphi
))-k^{2}I\bigr)P_i\Bigr)^{-1}\right\|_1<ck^{44r_1'+8\gamma
r_1}\left(\frac{k^{-11r_1'}}{\varepsilon _0}\right)^{N_{i }^{(1)}}.
\end{equation}
\end{enumerate}
\end{lemma}
\begin{proof}  Let $\Pi $ be a component  $\Pi _b^j$, $\Pi _g^j$ or $\Pi _w^j$ and $P_\Pi$ be the corresponding projection. By Lemmas \ref{L:black}, \ref{L:grey}, \ref{L:white} the number $N$ of elements  $\MM^{(2)}\cap \Pi$ does not exceed $ck^{\gamma
r_1+3}$.  Let us recall that the set ${\MM}^{(2)}$ is defined by the
formula \eqref{M^2}, where ${\cal O}_\m^{(2)}$ is the union of open
disks of the radius $k^{-10r_1'}$ with the centers at poles of the
resolvent of $k^\delta$-components containing
$\k^{(1)}(\varphi_0)+\p_\m$. Let us consider ${\cal O}^{(2)}_{\Pi
}=\cup_{\m\in \Pi \cap {\MM}^{(2)}}{\cal O}_\m^{(2)}$ and an
analogous set consisting of smaller discs: $\tilde{\cal
O}^{(2)}_{\Pi}=\cup_{\m\in \Pi \cap {\MM}^{(2)}}\tilde{\cal
O}_\m^{(2)}$, where $\tilde{\cal O}_\m^{(2)}$ have the radius
$k^{-11r_1'}$. Since $N<ck^{\gamma r_1+3}$, the total size of
$\tilde{\cal O}^{(2)}_{\Pi}$ is less than $k^{-11r_1'+\gamma
r_1+3}=o(k^{-10r_1'})$.

First, assume  $\varphi _0\not \in \tilde{\cal O}^{(2)}_{\Pi}$.
Then, we can prove an estimate analogous to \eqref{Pnr}. Indeed, let us consider a $k^{\delta }$-component in $\Pi$. We denote the corresponding
projection by $P({\m})$. By the definitions of $\OO _{\m}^{(2)}$, $\tilde
\OO _{\m}^{(2)}$, the distance from $\varphi _0$ to the nearest pole
of $\left(P({\m})(H(\k^{(1)})-k^{2}I)P({\m})\right)^{-1}$ is greater than
$k^{-11r_1'}$. Applying Lemmas \ref{L:estnonres1}, \ref{estnonres} to these
resolvents, we obtain (recall that now $p_\m>k^{-5r'_1}$):
\begin{equation} \label{March3-1a}\left\|\left(P({\m})(H(\k
^{(1)}(\varphi
_0))-k^{2}I)P({\m})\right)^{-1}\right\|<ck^{44r_1'},\end{equation}
\begin{equation}
\label{March3-1b}\left\|\left(P({\m})(H(\k ^{(1)}(\varphi
_0))-k^{2}I)P({\m})\right)^{-1}\right\|_1<ck^{44r_1'+4\delta_*}.\end{equation}
By analogy with Corollary \ref{estnonres1},
$$\left\|\left(P(H(\k^{(1)}(\varphi _0))-k^{2}I)P\right)^{-1}\right\|<ck^{44r_1'},$$
$$\left\|\left(P(H(\k^{(1)}(\varphi _0))-k^{2}I)P\right)^{-1}\right\|_1<ck^{44r_1'}L^4,$$ where
$P$ is the projection onto all $k^{\delta }$-components in $\Pi $,
$L$ is the size of $\Pi $.  Next, arguing as in the proof of Theorem
\ref{Thm2}, we show that the perturbation series for the resolvent
$\left(P_{\Pi }(H(\k^{(1)}(\varphi _0))-k^{2}I)P_{\Pi
}\right)^{-1}$ converges when we take $PH(\k^{(1)}(\varphi
_0))P+(P_{\Pi }-P)H_0$ as the unperturbed operator. Therefore,
$$\left\|\left(P_{\Pi }(H(\k^{(1)}(\varphi
_0))-k^{2}I)P_{\Pi }\right)^{-1}\right\|<ck^{44 r_1'},$$ no poles being inside of the disc. Taking into
account that $|\varphi -\varphi _0|<k^{-44 r_1'-2-\delta}$ and
$|\k^{(2)}-\k^{(1)}|=o(k^{-44 r_1'-2})$,  we obtain
$$\left\|\left(P_{\Pi }(H(\k^{(2)}(\varphi
))-k^{2}I)P_{\Pi }\right)^{-1}\right\|<ck^{44 r_1'}.$$ By Lemmas
\ref{L:black}, \ref{L:grey}, \ref{L:white}, $L<k^{2\gamma r_1}$,
$$\left\|\left(P_{\Pi }(H(\k^{(2)}(\varphi
))-k^{2}I)P_{\Pi }\right)^{-1}\right\|_1<ck^{8\gamma r_1+44r_1'}.$$
Thus, the resolvent has no poles inside the disk around $\varphi _0$ and the estimates \eqref{Pr-1}, \eqref{Pr-2} hold with $\varepsilon_0:=k^{-11r_1'}$.
Second, if $\varphi \not \in \tilde{\OO }_{\Pi }^{(2)}$, then $\varphi
_0 \not \in \tilde{\OO}_{\Pi }^{(2)}(11r_1',\frac{1}{2})$, the latter set  consisting of discs of the radius $\frac{1}{2}k^{-11r_1'}$. Therefore
the estimates analogous to the last two  hold. Now estimates
\eqref{Pr-1}, \eqref{Pr-2} easily follow.

It remains to consider the case $\varphi _0, \varphi \in \tilde{\cal
O}_{\Pi }^{(2)}$.   Obviously, $\varphi _0, \varphi $ belong to the same connected component of $\tilde{\cal
O}_{\Pi }^{(2)}$ or to different components being at the distance less than $k^{-44r_1'-2-\delta}$ from each other.  We consider a $\varphi _*\in \partial \tilde{\cal O}_{\Pi
}^{(2)}$, where $\partial \tilde{\cal O}_{\Pi
}^{(2)}$ is the boundary of the component(s) containing  $\varphi _0, \varphi $.
Note that  $\varphi _*\not \in \OO _{\m}^{(2)}(11r_1', 1)$ for all $\m \in \Pi$. Indeed, for $\m \in  \MM^{(2)}$, it just follows from the relation $\varphi _*\in \partial \tilde{\cal O}_{\Pi }^{(2)}$ and the definition of an open disk
${\cal
O}_{\m}^{(2)}(11r_1', 1)$.  If $\m\in \Pi \setminus \MM^{(2)}$, then $\varphi _0$ is not in $\OO _{\m}^{(2)}(10r_1', 1)$ by the definition of $\MM^{(2)}$.  Since  $\varphi _0, \varphi \in \tilde{\cal
O}_{\Pi }^{(2)}$ and
 the length of $\tilde{\cal O}_{\Pi }^{(2)}$ is $o(k^{-10r_1'})$,
we have $\varphi _*\not \in \OO _{\m}^{(2)}(10r_1', \frac{1}{2})$.

  Now, considering as in the case
$\varphi _0\not \in \tilde{\cal O}_{\Pi }^{(2)}$, we obtain that the
perturbation series for the resolvent $\left(P_{\Pi
}(H(\k^{(1)}(\varphi _*))-k^{2}I)P_{\Pi }\right)^{-1}$ converges
when we take $PH(\k^{(1)}(\varphi _*))P+(P_{\Pi }-P)H_0$ as the
unperturbed operator. Therefore,
$$\left\|\left(P_{\Pi }(H(\k^{(2)}(\varphi
_*))-k^{2}I)P_{\Pi }\right)^{-1}\right\|<ck^{44 r_1'}.$$ The number
of poles of the resolvent $\left(P_{\Pi }(H(\k^{(2)}(\varphi
))-k^{2}I)P_{\Pi }\right)^{-1}$ in $\tilde \OO_{\Pi }^{(2)}$ is the same
as the number of poles of the resolvent of unperturbed operator.
Hence, it is $N$. Using the Maximum principle, we get \eqref{Pr-1}
for the case $\varepsilon \leq k^{-11r_1'}$, where $N_i=N$ and
depends on the color of $\Pi $. Considering that the dimension of
$P_{\Pi }$  does not exceed $k^{8\gamma r_1}$, we obtain
\eqref{Pr-2}
\end{proof}

At last, let $\Pi_s^j$ be a particular $k^{r_1/2}$-box around $\m \in \Omega _s^{(2)}$.
 Let $P_s^j$ be corresponding projection.
\begin{lemma}\label{L:Ps} Let $\varphi _0\in \omega^{(2)}(k,\delta ,\tau )$. Then, the operator
$\Bigl(P_s^j\bigl(H(\k^{(2)}(\varphi
))-k^{2}I\bigr)P_s^j\Bigr)^{-1}$ has no more than one pole in the
disk $|\varphi-\varphi _0|<k^{-r_1'-\delta}$. Moreover,
\begin{equation}\label{Ps-1}
\left\|\Bigl(P_s^j\bigl(H(\k^{(2)}(\varphi
))-k^{2}I\bigr)P_s^j\Bigr)^{-1}\right\|<\frac{8k^{-1}}{p_{\m}\varepsilon
_0},
\end{equation}
\begin{equation}\label{Ps-2}
\left\|\Bigl(P_s^j\bigl(H(\k^{(2)}(\varphi
))-k^{2}I\bigr)P_s^j\Bigr)^{-1}\right\|_1<\frac{8k^{-1+
4r_1}}{p_{\m}\varepsilon_0},
\end{equation}
$\varepsilon _0=\min \{\varepsilon , k^{-r_1'-\delta }\}$, where
$\varepsilon $ is the distance to the pole of the operator.
\end{lemma}
\begin{proof}  The proof is similar to that of Lemma \ref{L:estnonres1} (part 3).  Indeed, when $p_\m<k^{-5r_1'}$, the series for $\lambda ^{(2)}(\k^{(2)}(\varphi )+\p_{\m})$ converges
in the complex $k^{-r_1'-\delta }$ neighborhood of $\omega^{(2)}(k,\delta ,\tau )$ and $\lambda ^{(2)}(\k^{(2)}(\varphi )+\p_{\m})=\lambda ^{(1)}(\k^{(1)}(\varphi )+\p_{\m})+
o(k^{-100r_1})$, see \eqref{perturbation-2c}. By Lemma \ref {L: Appendix1} (Appendix 4), the equation $\lambda^{(1)} \big(\k^{(1)}(\varphi
)+\p_\m\big)=k^{2}+\varepsilon _0,\
|\varepsilon _0|\leq p_\m k^{\delta }$ has no more than two solutions in this neighborhood  of $\omega^{(2)}(k,\delta ,\tau )$. Using Lemma \ref{L:3.7.1} and Rouch\'{e}'s Theorem, we obtain the same
 fact for $\lambda ^{(2)}(\k^{(2)}(\varphi )+\p_{\m})=\varepsilon _0$. It is easy to show that the analogues of Lemmas \ref{L:Appendix 2}, \ref{L:3.7.1} and
  \ref{L:July5} hold for  $\lambda ^{(2)}(\k^{(2)}(\varphi )+\p_{\m})$. Thus, we obtain  \eqref{Ps-1}, \eqref{Ps-2}. \end{proof}


\subsubsection{Resonant and Nonresonant Sets for Step III \label{GSIII}}

We divide $[0,2\pi )$ into $[2\pi k^{44r_1'+2+\delta}]+1$ intervals
$\Delta_l^{(2)}$ with the length not bigger than $k^{-44r_1'-2-\delta }$.
If a particular interval belongs to $\OO^{(2)}$, we ignore it;
otherwise, let $\varphi_0^{(l)}$ be a point inside the
$\Delta_l^{(2)}$, $\varphi_0^{(l)}\not\in\OO^{(2)}$. Let
\begin{equation}\W_l^{(2)}=\{\varphi \in \W^{(2)}:\ | \varphi -\varphi
_0^{(l)}|<4k^{-44r_1'-2-\delta }\}. \label{W2m} \end{equation} Clearly, neighboring sets
$\W_l^{(2)}$  overlap (because of the multiplier 4 in the
inequality), they cover $\hat \W^{(2)}$ , which is
the restriction of $\W^{(2)}$ to the $2k^{-44r_1'-2-\delta
}$-neighborhood of $[0,2\pi )$. For each $\varphi \in \hat \W^{(2)}$
there is an $l$ such that $|\varphi -\varphi
_{0}^{(l)}|<4k^{-44r_1'-2-\delta }$. For each $\varphi
_{0}^{(l)}$ we construct the projection $P^{(2)}$, see \eqref{P^2}. Further, we consider the poles of the
resolvent  $\left(P^{(2)}
(H(\k^{(2)}(\varphi))-k^{2})P^{(2)}\right)^{-1}$ in $\hat
\W_l^{(2)}$ and denote them by $\varphi^{(2)} _{lm}$, $m=1,...,M_l$, $P^{(2)}$ being constructed for $\varphi _0=\varphi _0^{(l)}$.
By Corollary \ref{C:PHP-2}, the resolvent has a block structure. The
number of blocks clearly cannot exceed the number of elements in
$\Omega (r_2)$, i.e. $k^{4r_2}$. Using the estimates for the number
of poles for each block, the estimate being provided by Lemma
\ref{L:Pr} Part 1, we can roughly estimate the number of poles of
the resolvent by $k^{4r_2+r_1}$.


 Next, let $r_2'>11r_1'$ and $\OO^{(3)}_{lm}$ be the disc of the radius
$k^{-r_2'}$ around $\varphi ^{(2)}_{ml}$.
\begin{definition} The set
\begin{equation}\OO^{(3)}=\cup _{lm}\OO^{(3)}_{lm} \label{O3}
\end{equation}
we call the third resonant set. The set
\begin{equation}\W^{(3)}= \hat  \W^{(2)}\setminus \OO^{(3)}\label{W3}
\end{equation}
is called the third non-resonant set. The set
\begin{equation}\omega^{(3)}= \W^{(3)}\cap [0,2\pi) \label{w3}
\end{equation}
is called the third real non-resonant set. \end{definition}
\begin{lemma}\label{L:geometric3}Let  $r_2'>\mu r_2>44r_1'$, $\varphi \in \W^{(3)}$, $\varphi
_0^{(l)}$ corresponds  to an interval $\Delta _l^{(2)}$ containing $\Re
\varphi $. Let $\Pi $ be one of the components $\Pi _s^j(\varphi
_0^{(l)})$, $\Pi _b^j(\varphi _0^{(l)})$, $\Pi _g^j(\varphi _0^{(l)})$, $\Pi
_w^j(\varphi _0^{(l)})$ and
 $P(\Pi )$ be the projection corresponding to $\Pi $. Let also
 $\varkappa \in \C: |\varkappa-\varkappa^{(2)}(\varphi )|<k^{-r_2'k^{2\gamma
r_1}}$. Then, for $\k(\varphi )=\varkappa (\cos \varphi ,\sin \varphi )$ the following inequality holds:
\begin{equation} \label{March3-2} \left\|\left(P(\Pi )\left(H\big(\k(\varphi
)\big)-k^{2}I\right)P(\Pi )\right)^{-1}\right\|<ck^{2\mu
r_2+r_2'N^{(1)}},\end{equation}
\begin{equation} \label{March3-3}
\left\|\left(P(\Pi )\left(H\big(\k(\varphi
)\big)-k^{2}I\right)P(\Pi )\right)^{-1}\right\|_1<ck^{(2\mu+1)
r_2+r_2'N^{(1)}},\end{equation} $N^{(1)}$ corresponding to the color
of $\Pi $: $N^{(1)}=1,\ k^{\gamma r_1+3 },\ k^{\gamma r_1/2+\delta
_0r_1 },\ k^{\gamma r_1/6-\delta _0r_1 }$ for simple, black, grey
and white clusters, correspondingly.
\end{lemma}
\begin{proof} For $\k=\k^{(2)}(\varphi )$ the lemma follows
immediately from the definition of $\W^{(3)}$ and Lemmas \ref{L:Pr}
and \ref{L:Ps} ($p_{\m}>k^{-2\mu r_2}$). It is easy to see that
estimates \eqref{March3-2} and \eqref{March3-3} are stable with
respect to perturbation of $\varkappa^{(2)}$ of order
$k^{-r_2'k^{2\gamma r_1}}$.
\end{proof}

 By total size of the set $\OO^{(3)}$ we mean the sum of
the sizes of its connected components.
\begin{lemma}\label{4.16} Let $r_2>45r_1'+2$, $r_2'\geq (\mu+10)r_2$
. Then, the size of each connected component of
 $\OO^{(3)}$ is less
than $k^{5r_2-r_2'}$. The total size of $\OO^{(3)}$ is less than
$k^{-r_2'/2}$.
\end{lemma}
\begin{proof} Indeed, each set $\W_l^{(2)}$ contains no more than
$k^{4r_2+r_1}$ discs $\OO_{lm}^{(3)}$. Therefore, the size of
$\OO^{(3)}\cap \W_l^{(2)}$ is less than $k^{-r_2'+5r_2}$.
Considering that $k^{-r_2'+5r_2}$ is much smaller that the length
of $\Delta _l^{(2)}$, we obtain that there is no connected components
which go across the whole set $\W_l^{(2)}$ and the size of each
connected component of $\OO^{(3)}$ is less than $k^{5r_2-r_2'}$.
Considering that the number of intervals $\Delta _l^{(2)}$ is less than
$k^{45r_1'+2+\delta }$, we obtain the required estimate for the total size
of $\OO^{(3)}$.
\end{proof}


\begin{lemma}\label{estnonres0-1} Let $\varphi\in\W^{(2)}$ and
$C_3$ be the circle $|z-k^{2}|=k^{-2r_2'k^{2\gamma r_1}}$. Then
\begin{equation}
\left\|\left(P(r_1)(H(\k^{(2)}(\varphi))-z)P(r_1)\right)^{-1}\right\|\leq
4^2k^{2r_2'k^{2\gamma r_1}}. \label{4.30} \end{equation} \end{lemma}
\begin{proof} The proof is similar to the proof of Lemma 3.21 in \cite{KaSh}, when we take into account
\eqref{step2raz*} and \eqref{2.70-2}, \eqref{dk0-2}.  In
the proof  we use the estimates from the previous step
along with some perturbation arguments.  First, we use the series
decomposition  \eqref{step2raz} --\eqref{step2raz*}) for $\k=\k^{(2)}(\varphi)$ and $z\in C_2$. Then, considering that  the resolvent has a single pole in  $ C_2$ located at $z=k^2$, and using the Maximum principle, we obtain \eqref{4.30}.
\end{proof}


\section{Step III}
Let $k_*$ be sufficiently large to satisfy the estimates:
$$ k_*\geq k_{1}(V,\delta ,\tau), \ \ \  \ k^{\delta /8}_*>10^{8}+\|V\|+\mu +2,$$
$k_{1}(V,\delta ,\tau)$ being introduced in the formulation of Theorem \ref{Thm2}. We also assume that $k_*$ is such that all constants $c$ in previous estimates (e.g. \eqref{March3-2}, \eqref{March3-3}) satisfy $c<k^{\delta /8}_*$. Since now on we consider $k>k_*$. This restriction on $k$ won't change in all consecutive steps.
We introduce a new notation $O_T(\cdot )$: let $f(k)=O_T(k^{-\gamma })$ mean that $|f(k)|<Tk^{-\gamma }$ when $k>k_*$.
\subsection{Operator $H^{(3)}$. Perturbation Formulas}
Let $P(r_2)$ be an orthogonal projector onto $\Omega(r_2):=\{\m:\
|\|\p_\m\||\leq k^{r_2}\}$ and $H^{(3)}=P(r_2)HP(r_2) $. From now on
we assume \begin{equation}
\label{r_2}k^\delta<r_2<k^{\gamma10^{-7}
r_1}.\end{equation} Note that $45r_1'+2<k^\delta<k^{\gamma10^{-7}
r_1}$ for all $k>k_*$, since $10^{8}<r_1<k^{\delta /8}$, see \eqref{Aug13-1} and above \eqref{box}. Let  $\beta:=\frac{\delta_*}{100}$. Let us introduce  $r_2'$, satisfying the inequality:
 \begin{equation}\label{r_2'}5\mu r_2<r_2'<k^{\delta_0r_1-4}. \end{equation} We consider
$H^{(3)}(\k^{(2)}(\varphi ))$ as a perturbation of $\tilde
H^{(2)}(\k^{(2)}(\varphi ))$:
$$\tilde H^{(2)}:=\tilde
P_l^{(2)}H\tilde
P_l^{(2)}+\left(P(r_2)-\tilde
P_l^{(2)}\right)H_0,$$
where $H=H(\k^{(2)}(\varphi ))$, $H_0=H_0(\k^{(2)}(\varphi ))$ and $\tilde
P_l^{(2)}$
is the projection $P^{(2)}$, see \eqref{P(2)}, corresponding to $\varphi _{0}^{(l)}$ in
the interval $\Delta _l^{(2)}$ containing $\varphi $.
Note that the operator $\tilde H^{(2)}$ has a block structure, each block $\tilde
P_l^{(2)}H\tilde
P_l^{(2)}$ being composed of smaller blocks $P_iHP_i$, $i=0,...,5$, see \eqref{PHP-2}, \eqref{blocks}.
 Let
\begin{equation}W^{(2)}=H^{(3)}-\tilde H^{(2)}=P(r_2)VP(r_2)-\tilde P_l^{(2)}V\tilde P_l^{(2)}, \label{W2*}\end{equation}
\begin{equation}\label{g3} g^{(3)}_r({\k}):=\frac{(-1)^r}{2\pi
ir}\hbox{Tr}\oint_{C_3}\left(W^{(2)}(\tilde
H^{(2)}({\k})-zI)^{-1}\right)^rdz,
\end{equation} \begin{equation}\label{G3}
G^{(3)}_r({\k}):=\frac{(-1)^{r+1}}{2\pi i}\oint_{C_3}(\tilde
H^{(2)}({\k})-zI)^{-1}\left(W^{(2)}(\tilde
H^{(2)}({\k})-zI)^{-1}\right)^rdz,
\end{equation}
where $C_3$ is the circle $|z-k^{2}|=\varepsilon _0^{(3)}$,
$\varepsilon _0^{(3)}=k^{-2r_2'k^{2\gamma r_1}}.$
\begin{theorem} \label{Thm3} Suppose  $k>k_*$, $\varphi $ is in
the real  $k^{-r_2'-\delta }$-neighborhood of $\omega
^{(3)}(k,\delta,\tau )$ and $\varkappa\in\R$,
$|\varkappa-\varkappa^{(2)}(\varphi )|\leq \varepsilon
^{(3)}_0k^{-1-\delta }$, $\k=\varkappa(\cos \varphi ,\sin \varphi
)$. Then,  there exists a single eigenvalue of $H^{(3)}({\k})$ in
the interval $\varepsilon _3( k,\delta,\tau )=\left(
k^{2}-\varepsilon _0^{(3)}, k^{2}+\varepsilon _0^{(3)}\right)$. It
is given by the absolutely converging series:
\begin{equation}\label{eigenvalue-3}\lambda^{(3)}({\k})=\lambda^{(2)}({\k})+
\sum\limits_{r=2}^\infty g^{(3)}_r({\k}).\end{equation} For
coefficients $g^{(3)}_r({\k})$ the following estimates hold:
\begin{equation}\label{estg3} |g^{(3)}_r({\k})|<k^{-\frac{\beta}{5} k^{r_1-\delta_*
}-\beta (r-1)}.
\end{equation}
The corresponding spectral projection is given by the series:
\begin{equation}\label{sprojector-3}
\E ^{(3)}({\k})=\E^{(2)}({\k})+\sum\limits_{r=1}^\infty
G^{(3)}_r({\k}), \end{equation} $\E^{(2)}({\k})$ being the spectral
projection of $H^{(2)}(\k)$. The operators $G^{(3)}_r({\k})$ satisfy
the estimates:
\begin{equation}
\label{Feb1a-3}
\left\|G^{(3)}_r({\k})\right\|_1<k^{-\frac{\beta}{10} k^{r_1-\delta_*
} -\beta r},
\end{equation}
\begin{equation}G^{(3)}_r({\k})_{\s\s'}=0,\ \mbox{when}\ \ \
2rk^{\gamma
r_1+3}+3k^{r_1}<\||\p_\s\||+\||\p_{\s'}\||.\label{Feb6a-3}
\end{equation}
\end{theorem}
\begin{corollary}\label{corthm3} For the perturbed eigenvalue and its spectral
projection the following estimates hold:
 \begin{equation}\label{perturbation-3}
\lambda^{(3)}({\k})=\lambda^{(2)}({\k})+ O_2\left(k^{-\frac{\beta}{5}
k^{r_1-\delta_* } }\right),
\end{equation}
\begin{equation}\label{perturbation*-3}
\left\|\E^{(3)}({\k})-\E^{(2)}({\k})\right\|_1<k^{-\frac{\beta}{10}
k^{r_1-\delta_* }},
\end{equation}
\begin{equation}
\left|\E^{(3)}({\k})_{\s\s'}\right|<k^{-d^{(3)}(\s,\s')},\ \
\mbox{when}\ \||\p_\s\||>4k^{r_1} \mbox{\ or }
\||\p_{\s'}\||>4k^{r_1 },\label{Feb6b-3}
\end{equation}
$$d^{(3)}(\s,\s')=\frac{1}{16}(\||\p_\s\||+\||\p_{\s'}\||)k^{-\gamma r_1-3}\beta +\frac{\beta}{10} k^{r_1-\delta_*
}.$$
\end{corollary}
Formulas \eqref{perturbation-3} and \eqref{perturbation*-3} easily
follow from \eqref{eigenvalue-3}, \eqref{estg3} and \eqref{sprojector-3}, \eqref{Feb1a-3}. The estimate \eqref{Feb6b-3}
follows from \eqref{sprojector-3}, \eqref{Feb1a-3} and \eqref{Feb6a-3}.
\begin{proof}The proof will follow the constructions from the proof of Theorem \ref{Thm2}. Let us consider the perturbation series
\begin{equation}\label{step3dva}
(H^{(3)}-z)^{-1}=\sum_{r=0}^\infty (\tilde
H^{(2)}-z)^{-1}\left(-W^{(2)} (\tilde H^{(2)}-z)^{-1}\right)^r,
\end{equation} here and below all the operators are computed at
$\k $. Further, we consider $\k$ and, therefore, the operators, as
analytic functions of $\varphi $ in $\W_l^{(2)}$, assuming
$\varkappa$ is fixed. By \eqref{W2*} and \eqref{PHP-2},
$W^{(2)}=P(r_2)\left( V-\sum _{i=0}^5 P_iVP_i \right) P(r_2)$. By assumption on $\varkappa$ and
Lemmas \ref{weak}, \ref{L:geometric3} and \ref{estnonres0-1},
\begin{equation}
\left\|(\tilde
H^{(2)}(\k)-z)^{-1}\right\|<2\cdot4^2k^{2r_2'k^{2\gamma r_1}}.
\label{tik-tak}
\end{equation}
To check the convergence it is enough to show that
 \begin{equation} \label{March5}\left\| \left(\tilde
H^{(2)}-z\right)^{-1}W^{(2)}\left(\tilde
H^{(2)}-z\right)^{-1}\right\|<k^{-2\beta } .\end{equation}
Then,
\begin{equation}
\left\|(H^{(3)}(\k)-z)^{-1}\right\|<4^3k^{2r_2'k^{2\gamma r_1}}.
\label{tik-tak*}
\end{equation}
Let us prove \eqref{March5}. Operator $\tilde H^{(2)}$ has a block structure. Identities \eqref{PVP-2'} imply that not only the blocks themselves, but also the blocks multiplied by $W^{(2)}$ have zero action on orthogonal subspaces. The operator $\tilde H^{(2)}$ acts as $H_0$ ``outside" the blocks. Because of the block structure, Corollary \ref{C:PHP-2} and \eqref{boundary}, it suffices to prove:
\begin{equation} \label{5-0} \left\|\left(
H_0-z\right)^{-1}\left(P(r_2)-\tilde P^{(2)}_l\right)V\left(P(r_2)-\tilde P^{(2)}_l\right)\left(
H_0-z\right)^{-1}\right\|\leq k^{-\delta_*/8}<k^{-3\beta
}, \end{equation}
\begin{equation} \label{5-}\left\|\left(
H_0-z\right)^{-1}\left(P(r_2)-\tilde P^{(2)}_l\right)VP^\partial(r_1)\left(\tilde
H^{(2)}-z\right)^{-1}\right\|\leq k^{-\delta_*/8+215\mu\delta}<k^{-3\beta },\end{equation}
\begin{equation} \label{5-1} \left\|\left(
H_0-z\right)^{-1}\left(P(r_2)-\tilde P^{(2)}_l\right)VP^\partial_{nr}\left(\tilde
H^{(2)}-z\right)^{-1}\right\|\leq k^{-\delta_*/8}<k^{-3\beta
},\end{equation}
\begin{equation} \label{5-s} \left\|\left(
H_0-z\right)^{-1}\left(P(r_2)-\tilde P^{(2)}_l\right)VP^\partial_{s}\left(\tilde
H^{(2)}-z\right)^{-1}\right\|\leq k^{-\delta_*/8+215\mu\delta}<k^{-3\beta
},
 \end{equation}
\begin{equation} \label{5-2} \left\|\left(
H_0-z\right)^{-1}\left(P(r_2)-\tilde P^{(2)}_l\right)VP^\partial_{w}\left(\tilde
H^{(2)}-z\right)^{-1}\right\|\leq k^{-\delta_*/8+215\mu\delta}<k^{-3\beta
},
 \end{equation}
\begin{equation} \label{5-3} \left\|\left(
H_0-z\right)^{-1}\left(P(r_2)-\tilde P^{(2)}_l\right)VP^\partial_{g}\left(\tilde
H^{(2)}-z\right)^{-1}\right\|\leq k^{-\delta_*/8+430\mu\delta}<k^{-3\beta
},
\end{equation}
\begin{equation} \label{5-4} \left\|\left(
H_0-z\right)^{-1}\left(P(r_2)-\tilde P^{(2)}_l\right)VP^\partial_{b}\left(\tilde
H^{(2)}-z\right)^{-1}\right\|\leq k^{-\delta_*/8+645\mu\delta}<k^{-3\beta
}. \end{equation}

By definition of $\tilde P^{(2)}_l$, a matrix element $(P(r_2)-\tilde P^{(2)}_l)_{\m \m}$ can differ from zero only if $\m$ is not in $\MM$ or in a trivial weakly resonant set. Considering as in the proof of \eqref{||A||2}, we obtain \eqref{5-0}.

Let us prove \eqref{5-}. We notice that $P(r_1)\tilde H^{(2)}=H^{(2)}$. As in the proof of Theorem \ref{Thm2}, we consider $H^{(2)}$ as a perturbation
of $\tilde H^{(1)}$, $H^{(2)}=\tilde H^{(1)}+W$.
Taking into account that $\tilde H^{(1)}$ has a block structure and $V$ is a trigonometric polynomial, we obtain $$P^{\partial}(r_1)\left( \left(\tilde
H^{(1)}-z\right)^{-1}W\right)^sP(\delta )=0,\ \  \mbox{when  } 1\leq s\leq S,\  \    S:=[k^{r_1-\delta_*-\delta}].$$ Hence,
\begin{equation} \label{hence}
\begin{split}
& P^{\partial }(r_1)(H^{(2)}-z)^{-1}P(r_1)=\sum _{s=0}^{S-1} P^{\partial }(r_1)\left(-\left(\tilde
H^{(1)}-z\right)^{-1}W\hat P\right)^s \left(\tilde
H^{(1)}-z\right)^{-1}\hat P \cr &
+P^{\partial }(r_1)\left( -\left(\tilde
H^{(1)}-z\right)^{-1}W\hat P\right)^S \left(
H^{(2)}-z\right)^{-1},
\end{split}
\end{equation}
where $\hat P=P(r_1)-P(\delta)$.
Considering as in the  proof of \eqref{||A||2} (here, we
replace $P(r_1)$ by $\hat P$, this compensates  for the smallness of
$C_3$), we obtain:
\begin{equation}\label{A_*}
\left\|\hat P\left(\tilde
H^{(1)}-z\right)^{-1}W\hat P\left(\tilde
H^{(1)}-z\right)^{-1}\right\|<ck^{-\delta_*/8 }.
\end{equation}
Next, by Theorem \ref{Thm2} and the definition of
$C_3$,  $\left\|\left( \tilde
H^{(2)}-z\right)^{-1}\right\|<(\varepsilon_0^{(3)}) ^{-1}< k^{\delta_* S/20}$.
Substituting the last two estimates into \eqref{hence}, we obtain:
$$
P^{\partial }(r_1)(H^{(2)}-z)^{-1}P(r_1)=P^{\partial }(r_1)\left(\tilde
H^{(1)}-z\right)^{-1}(I+O(k^{-\delta_*/8 }))+O(k^{-\delta_*/8 }).
$$
Again, considering as in the proof of \eqref{||A||2}, we get:
$$\left\|\left(
H_0-z\right)^{-1}\left(P(r_2)-\tilde P^{(2)}_l\right)VP^{\partial }(r_1)\left(\tilde
H^{(1)}-z\right)^{-1}\right\|<ck^{-\delta_*/8 }.$$ The estimate \eqref{5-} follows from the last two estimates.

The proof of \eqref{5-1} is analogous to the proof of \eqref{||A||2},
 where one uses Lemma \ref{L:Pnr} rather than Corollary \ref{estnonres1}.

Next, we prove \eqref{5-s}. Denote by $\hat{H}$ the reduction of the operator $H$ onto a
particular simple cluster i.e. $\hat{H}=P_sHP_s$ where $(P_s)_{\m\m}=1$, if
$\m$ belongs to this simple cluster and $(P_s)_{\m\m}=0$ otherwise. By
Lemma \ref{L:geometric3} (and obvious perturbation arguments to replace $k^2$ by $z$),
\begin{equation} \label{Apr6-s}
\|(\hat{H}-z)^{-1}\|\leq ck^{2\mu r_2+r_2'}.\end{equation}
We are going to construct a perturbation formula for
$P_s^{\partial }(\hat{H}-z)^{-1}P_s$. Let
$\hat{H}_0=P_{s,nr}HP_{s,nr}+(P_s-P_{s,nr})H_0$, where
$P_{s,nr}=P_sP_{nr}=P_{nr}P_{s}$ .  Operator $\hat H_0$ has a block structure. It is analogous to  operator $\tilde H^{(1)}$ in the proof of
Theorem \ref{Thm2}. The perturbation formula for
$P_s^{\partial }(\hat{H}-z)^{-1}P_s$ has the form:
\begin{equation}\label{dva-s}
\begin{split}&
P_s^{\partial
}(\hat{H}-z)^{-1}P_s=\sum_{r=0}^{R_s}P_s^{\partial
}\left[-(\hat{H}_0-z)^{-1}W_s\right]^r(\hat{H}_0-z)^{-1}P_s\cr &
+P_s^{\partial
}\left[-(\hat{H}_0-z)^{-1}W_s\right]^{R_s+1}(\hat{H}-z)^{-1}P_s ,\cr & W_s=\hat{H}-\hat{H}_0=P_sVP_s-P_{s,nr}VP_{s,nr},\  \ \  R_s=[k^{\frac{
r_1}{2}-\delta_*-\delta}].
\end{split}
\end{equation}
 Recall that when $\p_{\m'}$ belongs to the boundary of  a simple cluster, the $\||\cdot\||$-distance from $\p_{\m'}$ to the
point $\p_{\m}:\,0<p_\m< k^{-5r_1'}$ is at least $k^{ r_1/2}-Q$.
Since $\p_{\m}$ is small, the series \eqref{dva-s} is analogous to \eqref{hence}, $r_1$ being replaced by $r_1/2$.  We also notice that (see \eqref{Apr6-s}) $k^{2\mu r_2+r_2'}<<k^{\delta_*R_s/20}$ by \eqref{r_2}, \eqref{r_2'}. Now \eqref{5-s} easily follows.

Now, we prove \eqref{5-2}. Here and in what follows
we  often use the same notation for objects formally different,
but playing similar roles in different parts of the proof. We hope
it will not lead to confusion, but rather make it easier to keep the
whole construction and further inductive arguments in mind. Denote by $\hat{H}$ the reduction of the operator $H$ onto a
particular white cluster i.e. $\hat{H}=PHP$ where $P_{\m\m}=1$ if
$\m$ belongs to the white cluster and $P_{\m\m}=0$ otherwise. By
Lemma \ref{L:geometric3} and obvious perturbative  arguments to replace $k^2$ by $z$,
\begin{equation} \label{Apr6}
\|(\hat{H}-z)^{-1}\|\leq ck^{2\mu r_2+r_2'k^{\frac{\gamma
r_1}{6}-\delta _0r_1}}.
\end{equation}

We are going to construct a perturbation formula for
$P_w(\hat{H}-z)^{-1}P_w^{\partial }$. Let
$\hat{H}_0=P_{w,nr}HP_{w,nr}+(P_w-P_{w,nr})H_0$, where
$P_{w,nr}=P_wP_{nr}=P_{nr}P_{w}$ . Again,  operator $\hat H_0$ has a block structure analogous to the one of the operator $\tilde H^{(1)}$. The perturbation formula for
$P_w(\hat{H}-z)^{-1}P_w^{\partial }$ has the form :
\begin{equation}\label{dva}
\begin{split}&
P_w^{\partial
}(\hat{H}-z)^{-1}P_w=\sum_{r=0}^{R_w}P_w^{\partial
}\left[-(\hat{H}_0-z)^{-1}W\right]^r(\hat{H}_0-z)^{-1}P_w\cr &
+P_w^{\partial
}\left[-(\hat{H}_0-z)^{-1}W\right]^{R_w+1}(\hat{H}-z)^{-1}P_w ,\cr & W=\hat{H}-\hat{H}_0,\  \ \  R_w=[k^{\frac{\gamma
r_1}{6}-\delta_*-\delta}].
\end{split}
\end{equation}
 When $\p_{\m'}$ belongs to the boundary of the
white cluster, the $\||\cdot\||$-distance from $\p_{\m'}$ to the
closest  point in $\MM^{(2)}$ is $k^{\gamma r_1/6}$. Notice that
$(\hat{H}_0-z)^{-1}_{\m\m'}=0$ if
$\||\p_\m-\p_{\m'}\||>k^{\delta_*+\delta}$. Considering that $R_w<
k^{\frac{\gamma r_1}{6}-\delta_*-\delta }$ (so, we never reach the
points in $\MM^{(2)}$), we obtain that the finite sum in \eqref{dva}
can be estimated as those in  \eqref{hence} and \eqref{dva-s}.  Arguing as in the proof of \eqref{5-} and taking into account that $k^{2\mu r_2+r_2'k^{\frac{\gamma
r_1}{6}-\delta _0r_1}}<<k^{\delta_*R_w/20}$, we arrive at \eqref{5-2}.

Now, we prove \eqref{5-3}. Denote a component of the grey region  by $\Pi $ and its
boundary (see convention above) by $\partial\Pi $. Corresponding
projectors are denoted by $P$ and $P^{\partial }$
respectively. Denote by $\hat{H}$ the reduction of the operator $H$ onto a
particular grey cluster i.e. $\hat{H}=PHP$. We are going to construct the perturbation formula for
$P^{\partial }(\hat{H}-z)^{-1}P$. Recall, that the size of the
neighborhood of grey boxes is $D=k^{\frac{\gamma r_1}{2}+2\delta
_0r_1}$.  Let $P_i$ be a projector corresponding to a white or
non-resonant cluster laying inside $\frac D2$-neighborhood of
$\partial\Pi$, the size of these clusters being much smaller than
the size of the neighborhood. For definiteness, let
$i=1,\dots,\hat{I}$. Let $P^{(int)}$ be the projector onto all
points in $\Pi$ which are at least $D/2$ away of the boundary
(internal points). Note that $P_iP^{(int)}=0$. At last, put
\begin{equation} \label{opP'}
P_0:=P-P^{(int)}-\sum_{i=1}^{\hat{I}}P_i.
\end{equation}
Denote (cf. the case of a white cluster)
\begin{equation}
\hat{H}_0:=\sum_{i=1}^{\hat{I}}P_iHP_i+P^{(int)}HP^{(int)}+H_0P_0,
\label{opH_0} \end{equation}
\begin{equation} \label{opW} W=\hat H-\hat H_0=PVP-\sum_{i=1}^{\hat{I}}P_iVP_i-P^{(int)}VP^{(int)}.\end{equation}
We consider $\hat{H}$ as a perturbation of
$\hat{H}_0$.
 Let $R$ be
the smallest natural number for which
\begin{equation}
{\cal
A}_{R}:=P^{\partial}\left[(\hat{H}_0-z)^{-1}W\right]^{R+1}P^{(int)}\not=0,\
\ \ W:=\hat{H}-\hat{H}_0. \label{r_0} \end{equation} It is proven in
Appendix 6 that $R>
k^{(\frac\gamma2+2\delta_0)r_1-\delta_*-2\delta}$. Therefore,
\begin{equation}\label{chetire*g}
\begin{split}
P^{\partial}(\hat{H}-z)^{-1}P=\sum_{r=0}^{R-1}
P^{\partial}\left[-(\hat{H}_0-z)^{-1}W(P-P^{(int)})\right]^r(\hat{H}_0-z)^{-1}\cr
+P^{\partial}\left[-(\hat{H}_0-z)^{-1}W(P-P^{(int)})\right]^{R}(\hat{H}-z)^{-1}P.
\end{split}
\end{equation}
Again, we plan to proceed as in the proof of \eqref{5-}. However, here the operator $\hat H_0$ contains not only nonresonant clusters, but also white clusters. Hence, to obtain the estimate
$$\left\|(P-P^{(int)})(\hat{H}_0-z)^{-1}W(P-P^{(int)})(\hat{H}_0-z)^{-1}\right\|\leq k^{-\delta_*/8+215\mu\delta}$$ we use \eqref{A_*} (or just \eqref{5-1}) and already proven estimate \eqref{5-2} for white clusters. Now, we notice that by Lemma \ref{L:geometric3} and perturbation arguments,
\begin{equation}\label{vosem'}
\|(\hat{H}-z)^{-1}\|\leq ck^{2\mu
r_2+r_2'k^{(\frac\gamma2+\delta_0)r_1}}.
\end{equation}
Considering that $2r_2'<
k^{\delta_0r_1-2\delta_*}$  and combining the estimates above, we
obtain \eqref{5-3} in the same way as in the proof of \eqref{5-}.

We prove \eqref{5-4} in the analogous way. Indeed, denote a component of the black region  by $\Pi $ and its
boundary (see convention above) by $\partial\Pi $. Corresponding
projectors are denoted by $P$ and $P^{\partial }$
respectively, $\hat H:=PHP$. We and construct the perturbation formula for
$P^{\partial }(\hat{H}-z)^{-1}P$. Recall, that the size of the
neighborhood of black boxes is $D=k^{\gamma r_1+\delta _0r_1}$. Let $P_i$
be a projector corresponding to a grey, white or non-resonant
cluster laying inside $\frac D2$-neighborhood of $\partial\Pi$, the
size of these clusters being much smaller than the size of the
neighborhood. For definiteness, let $i=1,\dots,\hat{I}$. Let
$P^{(int)}$ be the projector onto all points in $\Pi$ which are at
least $D/2$ away of the boundary (internal points).  Again, we
define $P_0$, $\hat H_0$ and $W$ by formulas \eqref{opP'},
\eqref{opH_0} and \eqref{opW}.
We are going to use perturbation arguments between $\hat{H}_0$ and
$\hat{H}$.  Let $R$ be
the smallest positive integer for which \eqref{r_0} holds in the case of a black cluster.
 It is proven
in Appendix 7 that $R>k^{(\gamma r_1+\delta
_0r_1-\delta_*-2\delta)}$. Next, we use \eqref{chetire*g}.
The first term in the RHS of \eqref{chetire*g} contains only non-resonant, white and grey
clusters. Thus, we can use the estimates \eqref{5-1}, \eqref{5-2} and \eqref{5-3} obtained before in the case
of non-resonant, white and grey clusters. We get
\begin{equation}\label{sem'-b}
\left\|(P-P^{(int)})(\hat{H}_0-z)^{-1}W(P-P^{(int)})(\hat{H}_0-z)^{-1}\right\|\leq k^{-\delta_*/8+430\mu\delta}.
\end{equation}
By Lemma \ref{L:geometric3} and perturbation arguments,
\begin{equation}\label{vosem'-b}
\|(\hat{H}-z)^{-1}\|\leq ck^{2\mu r_2+r_2'k^{\gamma r_1+3}}.
\end{equation}
Considering that $2r_2'<k^{\delta _0r_1-3-2\delta_*}$
and combining the estimates above we obtain \eqref{5-4} in the same way as above.

Estimates \eqref{5-0} -- \eqref{5-4} provide convergence of the series for the resolvent. Integrating the resolvent over the contour we get \eqref{eigenvalue-3} and \eqref{sprojector-3}.

Proof of  \eqref{Feb1a-3} is analogous to that of \eqref{Feb1a} in
Theorem \ref{Thm2}. Indeed, we consider the operator $A=
W^{(2)}\left(\tilde H^{(2)}-z\right)^{-1}$ and represent it as
$A=A_0+A_1+A_2$, where $A_0=\left(P(r_2)-\E^{(2)}({\k})\right)A
\left(P(r_2) -\E^{(2)}({\k})\right)$, $A_1=\left(P(r_2
)-\E^{(2)}({\k})\right)A \E^{(2)}({\k})$, $A_2= \E^{(2)}({\k})A
\left(P(r_2)-\E^{(2)}({\k})\right)$. Note that
$\E^{(2)}({\k})W^{(2)}\E^{(2)}({\k})=0$, because of \eqref{W2*}. We
see that $$\oint _{C_3}\left(\tilde H^{(2)}-z\right)^{-1}A_0^r
dz=0,$$ since the integrand is a holomorphic function inside $C_3$.
Therefore,
\begin{equation} \label{Feb1-3} G^{(3)}_r({\k})=\frac{(-1)^r}{2\pi i}\sum
_{j_1,...j_r=0,1,2,\ j_1^2+...+j_r^2\neq 0}\oint _{C_3}\left(\tilde
H^{(2)}-z\right)^{-1}A_{j_1}.....A_{j_r} dz.
\end{equation} At least one of indices in each term is equal to 1 or 2.
Let us show that
\begin{equation} \label{A_2-3}
\|A_2\|_1<k^{-{\beta} k^{r_1-\delta_*}+215\mu\delta }.
\end{equation}
Using \eqref{W2*} and the obvious relation $\E^{(2)}=\E^{(2)}P(r_1)$, we obtain:
\begin{align}\nonumber &\E^{(2)}W^{(2)}(P(r_2)-\E^{(2)})=\E^{(2)}P(r_1)W^{(2)}(P(r_2)-\tilde P^{(2)}_l)=\cr &
\E^{(2)}P^{\partial}(r_1)W^{(2)}(P(r_2)-\tilde P^{(2)}_l). \end{align}
 Hence, $A_2=\E^{(2)}P^{\partial}(r_1)A(P(r_2)-\tilde P^{(2)}_l)$. Using \eqref{Feb6b}, we obtain
 $\|\E^{(2)}P^{\partial}(r_1)\|<k^{-{\beta} k^{r_1-\delta_*}}$.
Considering that $\E^{(2)}$ is a one-dimensional projection, we
obtain the same estimate for $\bf S_1 $-norm. Using \eqref{single1} and \eqref{weak1}, we obtain
$$
\|A(P(r_2)-\tilde P^{(2)}_l)\|\leq k^{215\mu\delta}.
$$
Now \eqref{A_2-3}
easily follows. Applying the same considerations as in the proof of
\eqref{Feb1a}  with the estimates \eqref{tik-tak}, \eqref{March5} and \eqref{r_2'}, we obtain \eqref{Feb1a-3}.

Let us obtain the estimate for $g_r ^{(3)}({\k})$.
Obviously,\begin{equation} \label{Feb1'*}
g^{(3)}_r({\k})=\frac{(-1)^r}{2\pi ir}\sum _{j_1,...j_r=0,1,2,\
j_1^2+...+j_r^2\neq 0}Tr\oint _{C_3}A_{j_1}.....A_{j_r} dz.
\end{equation}
Note that each term contains both $A_1$ and $A_2$, since we compute
the trace of the integral. Using \eqref{A_2-3} and repeating
arguments from the proof of \eqref{estg2}, we obtain \eqref{estg3}.

 The estimate \eqref{Feb6a-3} follows from the fact that the biggest white, grey or black component
 has the size not greater than $k^{\gamma r_1 +3}$. Therefore the biggest block of $\tilde H^{(2)}$ not coinciding
 with $P(r_1)HP(r_1)$ has the size not greater than $k^{\gamma r_1 +3}$.

\end{proof}

It is easy to see that coefficients $g^{(3)}_r({\k})$ and operators
$G^{(3)}_r({\k})$ can be analytically extended into the complex
$k^{-r_2'-\delta}$ neighborhood of $\omega ^{(3)}$ (in fact, into
$k^{-r_2'-\delta}$-neighborhood of $\W^{(3)}$) as functions of
$\varphi $ and to the complex $(\varepsilon
^{(3)}_0k^{-1-\delta})-$neighborhood of
$\varkappa=\varkappa^{(2)}(\varphi )$ as functions of $\varkappa$,
estimates \eqref{estg3}, \eqref{Feb1a-3} being preserved. Now, we
use formulae \eqref{g3}, \eqref{eigenvalue-3} to extend
$\lambda^{(3)}\left({\k}\right)=\lambda^{(3)}\left(\varkappa,\varphi\right)$
as an analytic function. Obviously, series \eqref{eigenvalue-3} is
differentiable. Using Cauchy integral and Lemma
\ref{L:derivatives-2} we get the following lemma.

\begin{lemma} \label{L:derivatives-3}Under conditions of Theorem \ref{Thm3} the following
estimates hold when $\varphi \in \omega ^{(3)}(k,\delta )$ or its
complex $k^{-r_2'-\delta}$-neighborhood and $\varkappa\in \C:$
$|\varkappa-\varkappa^{(2)}(\varphi )|<\varepsilon
^{(3)}_0k^{-1-\delta}$.
\begin{equation}\label{perturbation-3c}
\lambda^{(3)}({\k})=\lambda^{(2)}({\k})+O_2\left(k^{-\frac{\beta}{5}
k^{r_1-\delta_*} }\right),
\end{equation}
\begin{equation}\label{estgder1-3k}
\frac{\partial\lambda^{(3)}}{\partial\varkappa}=
\frac{\partial\lambda^{(2)}}{\partial\varkappa} +O_2\left(k^{-\frac
{\beta}{5} k^{r_1-\delta_*}  }M_1\right), \  \ \ \
M_1:=\frac{k^{1+\delta}}{\varepsilon ^{(3)}_0},\end{equation}
\begin{equation}\label{estgder1-3phi}\frac{\partial\lambda^{(3)}}{\partial \varphi }=\frac{\partial\lambda^{(2)}}{\partial \varphi }+
O_2\left(k^{-\frac {\beta}{5} k^{r_1-\delta_*}+r_2'+\delta }\right),
 \end{equation}
\begin{equation}\label{estgder2-3}
\frac{\partial^2\lambda^{(3)}}{\partial\varkappa^2}=
\frac{\partial^2\lambda^{(2)}}{\partial\varkappa^2}+
O_2\left(k^{-\frac {\beta}{5} k^{r_1-\delta_*} }M_1^2\right),
\end{equation}
\begin{equation} \label{gulf2-3}
\frac{\partial^2\lambda^{(3)}}{\partial\varkappa\partial \varphi
}=\frac{\partial^2\lambda^{(2)}}{\partial\varkappa\partial \varphi
}+ O_2\left(k^{-\frac{\beta}{5} k^{r_1-\delta_*}+r_2'+\delta
}M_1\right),
\end{equation}
\begin{equation} \label{gulf3-3}
\frac{\partial^2\lambda^{(3)}}{\partial\varphi
^2}=\frac{\partial^2\lambda^{(2)}}{\partial\varphi ^2}+
O_2\left(k^{-\frac {\beta}{5} k^{r_1-\delta_*}+2r_2'+2\delta
}\right).
\end{equation}\end{lemma}

\begin{corollary} \label{"O"} All ``$O_2$"-s on the right hand sides of \eqref{perturbation-3c}-\eqref{gulf3-3} can be written as $O_1\left(k^{-\frac {\beta}{10} k^{r_1-\delta_*}}\right)$.
\end{corollary}

\subsection{\label{IS3}Isoenergetic Surface for  Operator $H^{(3)}$}

\begin{lemma}\label{ldk-3} \begin{enumerate}
\item For every $\lambda :=k^{2}$,  $k>k_*$, and $\varphi $ in the real  $\frac{1}{2} k^{-r_2'-\delta }$-neighborhood
of $\omega^{(3)}(k,\delta, \tau )$, there is a unique
$\varkappa^{(3)}(\lambda, \varphi )$ in the interval
$$I_2:=[\varkappa^{(2)}(\lambda, \varphi )-\varepsilon
^{(3)}_0k^{-1-\delta},\varkappa^{(2)}(\lambda, \varphi
)+\varepsilon ^{(3)}_0k^{-1-\delta}],$$ such that
    \begin{equation}\label{2.70-3}
    \lambda^{(3)} \left(\k
^{(3)}(\lambda ,\varphi )\right)=\lambda ,\ \ \k ^{(3)}(\lambda
,\varphi ):=\varkappa^{(3)}(\lambda ,\varphi )\vec \nu(\varphi).
    \end{equation}
\item  Furthermore, there exists an analytic in $ \varphi $ continuation  of
$\varkappa^{(3)}(\lambda ,\varphi )$ to the complex  $\frac{1}{2}
k^{-r_2'-\delta }$-neighborhood of $\omega^{(3)}(k,\delta, \tau )$
such that $\lambda^{(3)} (\k ^{(3)}(\lambda, \varphi ))=\lambda $.
Function $\varkappa^{(3)}(\lambda, \varphi )$ can be represented as
$\varkappa^{(3)}(\lambda, \varphi )=\varkappa^{(2)}(\lambda, \varphi
)+h^{(3)}(\lambda, \varphi )$, where
\begin{equation}\label{dk0-3} |h^{(3)}(\varphi )|=O_1\left(k^{-\frac {\beta}{5} k^{r_1-\delta_*}-1
}\right),
\end{equation}
\begin{equation}\label{dk-3}
\frac{\partial{h}^{(3)}}{\partial\varphi}= O_2\left(k^{-\frac {\beta}{5}
k^{r_1-\delta_*}-1 +r_2'+\delta }\right),\ \ \ \ \
\frac{\partial^2{h}^{(3)}}{\partial\varphi^2}= O_4\left(k^{-\frac {\beta}{5} k^{r_1-\delta_*}-1 +2r_2'+2\delta }\right).
\end{equation} \end{enumerate}\end{lemma}
\begin{proof}  The proof is completely analogous to that of Lemma \ref{ldk}, estimates \eqref{perturbation-3c} --\eqref{gulf3-3} being used. \end{proof}


Let us consider the set of points in $\R^2$ given by the formula:
$\k=\k^{(3)} (\varphi), \ \ \varphi \in \omega ^{(3)}(k,\delta, \tau
)$. By Lemma \ref{ldk-3} this set of points is a slight distortion
of ${\cal D}_{2}$. All the points of this curve satisfy the equation
$\lambda^{(3)}(\k ^{(3)}(\varphi ))=k^{2}$. We call it isoenergetic
surface of the operator $H^{(3)}$ and denote by ${\cal D}_{3}$.

\subsection{Preparation for Step IV \label{S:4}}
\subsubsection{Properties of the Quasiperiodic Lattice. Continuation}\label{Quasiperiodicgeomcont}
Let \begin{equation}\label{triind}
\SS ^{(2)}(k,\xi):=\{\k\in \R^{2}:\
\|(H^{(2)}(\k)-k^{2})^{-1}\|>k^{\xi}\}.
\end{equation}
The main purpose of this section is to estimate the number of
 points $\k _0+\p_{\m}$, $\||\p_{\m}\||< k^{r_2}$ in $\SS ^{(2)}(k,\xi)$, $\k _0$ being fixed. In fact, we prove a more subtle result, see
Lemma~\ref{norm2/3}.

 We consider $\p_\m=2\pi(\s_1+\alpha\s_2)$ with
integer vectors $\s_j$ such that $|\s_j|\leq 4k^{r_2}$. We repeat
the arguments from the beginning of Section~\ref{geomIII}. Namely,
let $(q,p)\in\Z^2$ be a pair such that $0<q\leq 4k^{r_2}$ and
\begin{equation}\label{qind}
|\alpha q+p|\leq \frac14 k^{-r_2}.
\end{equation}
We choose a pair $(p,q)$ which gives the best approximation. In
particular, $p$ and $q$ are mutually prime. Put
$\epsilon_q:=\alpha+\frac{p}{q}$. We have
\begin{equation}k^{-2r_2\mu}\leq|\epsilon_q|\leq \frac14 q^{-1}k^{-r_2}.\label{epsilon_qind}\end{equation}
 We  write $\s_2$ in the form\begin{equation}\s_2=q\s_2'+\s_2'' \label{s1ind}
\end{equation}
with  integer vectors $\s_2'$ and $\s_2''$, $0\leq (\s_2'')_j< q$
for $j=1,2$. Hence, $|(\s_2')_j|\leq 4k^{r_2}/q+1$. It follows
$$
(2\pi)^{-1}\p_\m=(\s_1-p\s_2')+(-\frac{p}{q}\s_2''+\epsilon_q\s_2'')+\epsilon_q
q\s_2'.
$$
Denote $\s:=\s_1-p\s_2'$. Then $|\s|\leq 8k^{r_2}$. The number of
different vectors $\tilde{\s}:=-\frac{p}{q}\s_2''+\epsilon_q\s_2''$
is not greater than $(2q)^2$. For each fixed pair $\tilde \s,\ \s$
we obtain a lattice parameterized by $\s_2'$. We call this lattice a
cluster corresponding to given $\tilde \s,\ \s$. Each cluster,
obviously, is a square lattice with the step $\epsilon _qq$. It
contains no more than $\left(9k^{r_2}q^{-1}\right)^2$ elements,
since $|(\s_2')_j|\leq 4k^{r_2}q^{-1}+1$, $j=1,2$. The size of each
cluster is less than $5|\epsilon _q|k^{r_2}$. As before we have the
following statements.

     \begin{lemma}\label{Lattice-1ind}Suppose that $\epsilon _q$ satisfies the
     inequality
     \begin{equation}|\epsilon_q|\leq \frac{1}{64}q^{-1}k^{-r_2}.\label{epsilon_q'ind}\end{equation}
     Then, the size of each cluster is less that $\frac{1}{8q}$. The distance between clusters is greater than
     $\frac{1}{2q}$. \end{lemma}

     \begin{lemma}\label{Lattice-2ind} The number of vectors $\p_\m$,
      satisfying the inequalities $\||\p_{\m}\||<2k^{r_2}$,
     $p_{\m}<|\epsilon _q|qk^{r_2/3}$, does not exceed
     $k^{2r_2/3}$.\end{lemma}

     \begin{lemma} \label{Lattice-3ind} Suppose $q$ in the inequality
     \eqref{q} satisfies the estimate $q>k^{2r_2/3}$. Then, the
     number of vectors  $\p_\m$,
     $\||\p_{\m}\||<2k^{r_2}$, satisfying the inequality
     $p_{\m}<k^{-2r_2/3}$ does not exceed
     $2^{12}\cdot k^{2r_2/3}$.\end{lemma}

We consider the matrix $H^{(2)}(\k)=P(\gamma r_1)H(\k )P(\gamma
r_1)$ where $\k \in \R^2$, $P(\gamma r_1)$ is the orthogonal
projection corresponding to $\Omega (\gamma r_1)$ (it is a
slight abuse of notations, since $H^{(2)}$ in Step II was defined
for $\gamma =1$). Let  $\k(\tau _1)=\b+\a \tau_1,\  \ \ |\a|=1,\  |\b|<4k^{\gamma r_1}$. We
consider $H^{(2)}(\k)$ as a function of $\tau _1$ in the complex $
3k^{-\rho _1}$-neighborhood of zero, $\rho _1=4\gamma
r_1'+1+\delta $. We construct the block structure in  $H^{(2)}(\k)$
analogous to that in Step II.  The difference is that now we
consider any $|\b|<4k^{\gamma r_1}$, not only $\b$ being close to
$\k^{(1)}(\varphi )$. Here is the construction of the block operator. We call $\m \in \Omega (\gamma r_1)$
resonant in the sense of  \eqref{resonance1} if
\begin{equation}
\left||\b+\p_{\m }|_\R^{2}-k^{2}\right|<k^{\delta_*}.
\label{Aug29a}
\end{equation}
Next, we introduce $\MM_1$, $\MM_2$ by analogy with Section \ref{MOforStep2}, here $\MM_1=\MM_1(\b)$, $\MM_2=\MM_2(\b)$.
 Around each resonant $\m\in \MM_1$
 we construct $k^\delta$-blocks. 
 Next, we  split $\MM _2$ into components $\MM _2^{j}$ and  introduce trivial and non-trivial $\MM_2^{j}$. Further we split each $\MM_2^{j}$ into $\MM_2^{j,s}$. Slightly abusing the original Definition~\ref{Nov27-12} we say that $\MM _2^{j,s}$
 is weakly resonant if \eqref{weak1} holds with $\b$ instead of $\vec \varkappa ^{(1)} (\varphi )$. Otherwise,  $\MM_2^{j,s}$ is called strongly resonant.
Around each strongly resonant $\MM _2^{j,s}$  
 we construct $k^\delta$-components and introduce the corresponding projectors (see \eqref{defP}).
 Next, we construct a block operator $\tilde
H^{(1)}(\k)$:
\begin{equation}\tilde H^{(1)}(\k)=PH(\k)P+H_0(\k)(I-P),\label{Aug20-14}\end{equation}  where $P$ is defined in complete analogy with \eqref{defP}. 

Next, let $P(\m)$ be the projection on the $k^{\delta }$-component containing $\m$. 
\begin{definition}\label{resdef}If
\begin{equation}
\left\|(P({\m})(H(\b)-k^{2})P({\m}))^{-1}\right\|<k^{4\gamma r_1'},
\label{Aug29b} \end{equation} then we call the $k^{\delta }$-cluster
effectively non-resonant (cf. \eqref{Mon3a**}) for a given $\b$. Otherwise, it is called effectively resonant.
\end{definition}
\begin{definition} \label{D-J} We denote   by
$J(\b)$ the number of effectively resonant clusters.\end{definition}


\begin{lemma} \label{May24-1}The resolvent $( H^{(2)}(\k)-k^{2})^{-1}$, $\k=\k(\tau _1)$, has no more than $12J(\b)$ poles  in the the complex $ 2k^{-\rho _1}$-neighborhood of $\tau _1=0$,
$\rho _1=4\gamma  r'_1+1+\delta $. It satisfies the following estimate
in the complex $k^{-\rho _1}$-neighborhood of zero:
\begin{equation}\label{Sept3a}
\|(H^{(2)}(\k)-k^{2})^{-1}\|<k^{24\rho _1}\left(\frac{2k^{ -\rho
_{1}}}{\varepsilon _0}\right)^{12J(\b)}, \ \ \ \k=\k(\tau _1),
\end{equation}
where $\varepsilon _0=\min \{k^{-2\rho _1},\varepsilon\}$,
$\varepsilon $ being the distance to the nearest pole.
\end{lemma}
\begin{proof}
We note that the proof of the lemma is quite similar to that of Theorem \ref{Thm2}.  In both cases we  prove the estimates for the resolvent of $H^{(2)}$, considering it as a perturbation of a block operator. There are some technical differences though.  Here we have $k^{\gamma r_1}$-block around arbitrary $\k$, while in  Theorem \ref{Thm2} we had $k^{r_1}$ -block around $\vec \varkappa ^{(1)} (\varphi )$.  Here we consider the resolvent  as an analytic function of $\tau _1$, not $\varphi $. However, the proofs of convergence of perturbation series for the resolvents are quite similar. Though we can't use properties of the non-resonant set $\omega^{(1)}$ now,  dependence of the operator on parameter $\tau _1$ is simpler than it was on $\varphi $, because diagonal elements of the matrix $H^{(2)}(\k)$, $\k=\k(\tau _1)$, are quadratic polynomials of   $\tau_1$, see \eqref{Aug20-14}.

 Let us consider blocks of $PHP$. Let $\m\in \MM _1(\b)$. We consider the corresponding   $k^{\delta }$-block $P_{\m}H(\k )P_{\m}$.
Our first goal is to show that each block $(P_\m(H(\k)-k^2)P_\m)^{-1}$, $\m\in \MM _1(\b)$  has no more than $2$ poles in a $k^{-\rho_1}$-neighborhood of $\tau_1=0$ and the estimate
\begin{equation}\label{star} \|(P_\m(H(\k)-k^2)P_\m)^{-1}\|<k^{4\rho _1}
\end{equation}
holds at the distance greater than $k^{-2\rho _1 }$ from the poles. Indeed, let $\m\in\MM_1(\b)$, i.e. $||\k(0)+\p_\m|_\R^2-k^2|<k^{\delta_*}$ and $||\k(0)+\p_{\m+\q}|_\R^2-k^2|>k^{\delta_*}$ for $\||\p_\q\||<k^\delta$ ($\k(0)=\b$). We will assume (otherwise the situation becomes trivial) that $||\k(0)+\p_\m|_\R^2-k^2|<k^{\delta_*/2}$.
Let $\D_0$ be such a neighborhood of the corresponding zeros of $|\k(\tau_1)+\p_\m|_\R^2-k^2$ that $||\k(\tau_1)+\p_\m|_\R^2-k^2|=k^{\delta_*-\delta} \ \ \hbox{on}\ \partial\D_0$.
Since $|\k(\tau_1)+\p_\m|_\R^2$ is a quadratic polynomial in $\tau_1$ with the main coefficient $1$, the set $\D_0$, obviously, belongs to the $k^{(\delta_*-\delta)/2}$-neighborhood of the zeros. We consider only a connected component of  $\D_0$, which intersects $k^{\delta_*-3\delta }$-neighborhood of $\tau_1=0$ (with a slight abuse of notations we still denote it  $\D_0$). Clearly,
$$
\left(|\k(\tau_1)+\p_{\m+\q}|_\R^2-|\k(\tau_1)+\p_{\m}|_\R^2\right)-\left(|\k(0)+\p_{\m+\q}|_\R^2-|\k(0)+\p_{\m}|_\R^2\right)=2\tau_1(\a,\p_\q).
$$
Considering that $|\tau _1|<k^{\delta _*-2\delta }$, we obtain:
$$
||\k(\tau_1)+\p_{\m+\q}|_\R^2-k^2|\geq\frac12 k^{\delta_*}\ \ \hbox{on}\ \partial\D_0.
$$
Using simple perturbation arguments, we see that $(P_\m(H-k^2)P_\m)^{-1}$ has no more than $2$ poles in $\D_0$ and
$$
\|(P_\m(H(\k)-k^2)P_\m)^{-1}\|<2k^{-\delta_*+\delta }\ \ \hbox{on}\ \partial\D_0.
$$
Now, shrinking the neighborhood around poles we prove \eqref{star}.

Next, we consider the case of a trivial $\MM_2^j$. Each such $\MM_2^{j,s}$ consists just from one point $\m$. Assume that $\MM_2^{j,s}$ is strongly resonant, i.e.,  $\m:\ ||\k(0)+\p_\m|_\R^2-k^2|<k^{-214\mu\delta}$. Then, (cf. Lemma \ref{per3}) we have at most $2$ such $\m$-s.  
Choosing a neighborhood of the size $k^{-300\mu\delta}$ around  zeros of the the corresponding quadratic polynomials (see the case $\m \in \MM_1$), we can prove the result similar to Lemma \ref{trivial}. In particular, the resolvent of the corresponding $k^{\delta }$-cluster has no more than four poles $\tau _1$ in the $4k^{-300\mu\delta}$ neighborhood of zero and estimate analogous to \eqref{star} holds (with $8\rho _1$ instead of $4\rho _1$ in the r.h.s). If $\MM_2^{j,s}$ is just weakly resonant,
than there are no poles of the resolvent in the $k^{-300\mu\delta}$ of zero.

It remains to consider a non-trivial case. We start with the model operator $H^{j,s}(\tau_1)$. and assume that it is strongly resonant:
$$
\|(H^{j,s}(\b)-k^2)^{-1}\|>k^{214\mu\delta}.
$$
Considering as in the proof of Lemma \ref{per2}, we introduce eigenvalues $\tilde{\lambda}_n(t_\q)={\lambda}_n(t_\q)-(t_\q^\perp)^2(\tau_1)$ and corresponding eigenvalues of the periodic one-dimensional operator $\tilde{\lambda}_n^{per}(t_\q)$,  here and below, $t_\q=(\b+\p_\m,\vec\nu_\q)+\tau_1(\a,\vec\nu_\q)$, $t^\perp_\q=(\b+\p_\m,\vec\nu^\perp_\q)+\tau_1(\a,\vec\nu^\perp_\q)$, $\m$ being the central point of $\MM _2^{j,s}$, and $|\tau _1|<\sigma _{large}$, $\sigma _{large}=\pi p_\q/2$. We have $|\tilde{\lambda}_n(t_\q)|\leq\frac18 k^{\delta_*}$, 
$
\tilde{\lambda}_n^{per}(t_\q)=\tilde{\lambda}_n(t_\q)+O(k^{-\frac{\delta_*}{C(Q)}k^{\delta_*/2}}).
$
Assume, first, that $|\tilde{\lambda}_n^{per}(t_\q)|>\Lambda$ with sufficiently large $\Lambda$ to be fixed later. Only two eigenvalues, say, with indices $i_0$ and $i_0+1$ can be close to each other. Now, the arguments are somewhat similar to those from the proof of Lemma~\ref{4.9n}, see also Appendix 5. Indeed, let us consider the expression:
\begin{equation}\label{highennew}
 D_{i_0}:=(\tilde\lambda_{i_0}(\tau_1)+(t_\q^\perp)^2-k^2)(\tilde\lambda_{i_0+1}(\tau_1)+(t_\q^\perp)^2-k^2). \end{equation}
 Since $i_0$ is big enough, all other eigenvalues are sufficiently far away and perturbation arguments work. 
 Unlike each individual factor in the r.h.s. of (\ref{highennew}), $D_{i_0}$ is analytic in the neighborhood of $\tau_1=0$ with the radius of analyticity $\sigma_{large}$.
  We assume that both factors in $D_{i_0}$ are close to zero, i.e. $|\tilde\lambda_{i_0}-\tilde\lambda_{i_0+1}|\leq \Lambda ^{-1/2}$, otherwise, only one factor is needed which makes arguments even simpler. If $\Lambda=\Lambda(V)$ is large enough we can apply standard perturbative arguments (see Appendix 5 for details) to compare $D_{i_0}$ with the same expression for $V=0$ which we denote by $D_{i_0,0}$. Obviously, $D_{i_0,0}$ is a polynomial of order four with respect to $\tau_1$ with the main coefficient $1$. We consider $\sigma_{large}/100$-neighborhood of each zero and denote the union of these neighborhoods (intersecting the $k^{-41\mu \delta }$-neighborhood of $\tau_1=0$) by $T_0$. By definition, $T_0$ consists of no more than $4$ discs. Without loss of generality we may assume that $T_0$ is in the domain of analyticity of $D_{i_0}$ (otherwise, it means that a particular zero of $D_{i_0,0}$ and $D_{i_0}$ is far away from $\tau_1=0$ and is not of our interest). We have
\begin{equation}
|D_{i_0,0}(\tau_1)|>(\sigma_{large}/100)^4\ \ \ \hbox{on}\ \partial T_0. \label{Jan8-14new}
\end{equation}
Now, we choose $\Lambda=\Lambda(V)$ as described in Appendix 5. By perturbation and Rouch\'{e}'s Theorem, $D_{i_0}$ has no more than $4$ zeros in $T_0$ (obviously, even in twice more narrow neighborhood) and
\begin{equation}\label{star1}
\|(H^{j,s}-k^2)^{-1}\|\leq2(100/\sigma_{large})^4\ \ \ \hbox{on}\ \partial T_0.
\end{equation}
Let $|\tilde{\lambda}_n^{per}(t_\q)|\leq\Lambda(V)$. Put 
\begin{equation}\label{star111}C_*=C_*(V,\Lambda):=\Lambda+\max\limits_{s,n:\,|\lambda_n(s)|\leq\Lambda}\{|\lambda'_n(s)|+|\lambda''_n(s)|+|\lambda'''_n(s)|\}.\end{equation}
If $|(\a,\vec\nu_\q^\perp)(\b+\p_\m,\vec\nu_\q^\perp)|\geq 2C_*$, then
$$
\left|\frac{d}{d\tau_1}\left(\tilde{\lambda}_n^{per}(t_\q)+(t_\q^\perp)^2(\tau_1)\right)\right|\geq   C_*$$
in a complex neighborhood $\sigma $, $\sigma =\sigma (V, \Lambda)$, of $\tau _1=0$. This means that $\tilde{\lambda}_n^{per}(t_\q)+(t_\q^\perp)^2(\tau_1)-k^2$ has no more than one zero in this neighborhood and it can happen only for a real $\tau _1$. Estimate of the type \eqref{star1} (2 instead of 4 and $\sigma $ instead of $\sigma _{large}$ in the r.h.s) holds on the boundary of the neighborhood. 
If $|(\a,\vec\nu_\q^\perp)(\b+\p_\m,\vec\nu_\q^\perp)|< 2C_*$, then
$$
\frac{d^2}{d\tau_1^2}\left(\tilde{\lambda}_n^{per}(t_\q)+(t_\q^\perp)^2(\tau_1)\right)=\frac{d^2}{dt_\q^2}\tilde{\lambda}_n^{per}(t_\q)+O_{\Lambda,V}(k^{-2}), $$ $$
\left|\frac{d^3}{d\tau_1^3}\left(\tilde{\lambda}_n^{per}(t_\q)+(t_\q^\perp)^2(\tau_1)\right)\right|=\left|\frac{d^3}{dt_\q^3}\tilde{\lambda}_n^{per}(t_\q)\right|+O_{\Lambda,V}(k^{-2}).
$$
It follows from the properties of the one-dimensional periodic Schr\"odinger operator (see e.g. \cite{Kor}) that
$$
\left|\frac{d^2}{dt_\q^2}\tilde{\lambda}_n^{per}(t_\q)\right|+\left|\frac{d^3}{dt_\q^3}\tilde{\lambda}_n^{per}(t_\q)\right|>c(V,\Lambda)>0.
$$
This means that $\tilde{\lambda}_n^{per}(t_\q)+(t_\q^\perp)^2-k^2$ has no more than three zeros $\tau _1$ in a $\sigma $-neighborhood (for details see Appendix 10).
Thus we can apply again the same arguments as above to prove an estimate similar to \eqref{star1} with  6 instead of 4. 
Next, as in Lemma \ref{per3} there can be no more than $2$ strongly resonant sets $\MM_2^{j,s}$ in a cluster of $k^{\delta }$-blocks. Using \eqref{star1} we can proceed as in the proof of Lemma \ref{nontrivial} and, in particular, prove the analog of \eqref{star} for an effectively strongly resonant non-trivial cluster (with $24\rho _1$ instead of $4\rho _1$).
If  $\MM_2^{j,s}$ is a non-trivial weakly resonant cluster disjoint from any strongly resonant cluster, it does not generate any poles.

We proved that the resolvent $\left(\tilde H^{(1)}(\k)-k^{2}\right)^{-1}$
has no more than $12J(\b)$ poles $\tau _{1j}$ in the complex
$k^{-1-40\mu\delta }$-neighborhood of $\tau _{1}=0$ and
satisfies the estimate \begin{equation} \label{May25} \|(\tilde
H^{(1)}(\k)-k^{2})^{-1}\|<\frac{1}{4}k^{24\rho _1}
\end{equation}
at the distance greater than $k^{-2\rho _1 }$ from the poles.
 Let us consider the union of
$k^{-2\rho _1}$-neighborhoods of these poles. It may consist from
several connected components. We are interested only in those
intersecting with the $2k^{-\rho _1}$-disk around $\tau _1=0$. We
denote their union by ${\D}$. Using a rough estimate
$J(\k)<k^{4\gamma r_1}$ gives that $\D$ belongs to the $3k^{-\rho_1}$-neighborhood of zero. Thus, \eqref{May25} holds outside $\D$.
Now, considering as in the proof of Theorem \ref{Thm2} (see \eqref{step2raz}--\eqref{step25}),
we can show that the perturbation series for the resolvent
$(H^{(2)}(\k)-k^{2})^{-1}$ with respect to $(\tilde
H^{(1)}(\k)-k^{2})^{-1}$ converges on the boundary of $\D$ and
$$
\|(H^{(2)}(\k)-k^{2})^{-1}\|<k^{24\rho _1}
$$
outside $\D$; the resolvent has no more than $12J(\k)$ poles in $\D$.
Using again the maximum principle we obtain \eqref{Sept3a}.\end{proof}

Note that each connected component of $\SS ^{(2)}(k,\xi)$, see \eqref{triind}, is
bounded  by the curves $D(\k, k^{2}\pm k^{-\xi})=0$, where $ D(\k,
\lambda)=\hbox{det}\ (H^{(2)}(\k)-\lambda ).$
\begin{lemma}\label{L:curves-2ind}  Let $\l$ be a segment of a straight line in $\R^{2}$,
\begin{equation}\l:=\{\k=\a \tau_1+\b,\ \tau _1\in
(0,\eta)\}, \ \ |\a|=1,\  |\b|<4k^{\gamma r_1},\  \ 0<\eta <
k^{-5\gamma r_1'}. \label{segment}\end{equation} Suppose both ends
of $\l$ belong to a connected component of $\SS ^{(2)}(k,\xi)$. If $\xi $
is sufficiently large, namely, $\xi\geq 48J(\b)\ln
_k\frac{1}{\eta}$, then, there is an inner part $\l'$ of the
segment,
 which is not in $\SS^{(2)}(k,\xi)$.
 \end{lemma}
 \begin{corollary} \label{C:curves-2ind} Let $\k\in \SS ^{(2)}(k,\xi)$ and
 $\frac{\xi}{48J(\k)}>10\gamma r_1'$. Then the distance from $\k$ to the boundary of
 $\SS^{(2)}(k,\xi)$ is less than $k^{-\frac{\xi}{48J(\k)}}$. \end{corollary}
 {\em Proof of the corollary.}  Let us consider a segment of the length $\eta =k^{-\frac{\xi}{48J(\k)}}$ starting at $\k $.
 By the statement of the lemma it intersects a boundary $D(\k, k^{2}\pm k^{-\xi})=0$.

 \begin{proof}
 Choose
$\varepsilon =\eta ^2$ in \eqref{Sept3a}. Using the hypothesis of
the lemma, we obtain that the right-hand side of \eqref{Sept3a} is
less than $k^{{\xi }}$ outside the discs of radius $\varepsilon $ around the poles of the resolvent.  Let us estimate the total
size (sum of the sizes) of the discs. Indeed, the size of each disc
is $2\varepsilon $ and the number of discs is, obviously, less
$16k^{4\gamma r_1}$. Therefore, the total size admits the estimate
from above: $32\varepsilon k^{4\gamma r_1} <<\eta$, since $\eta
<k^{-5\gamma r_1'}$. This means there is a part $\l'$ of $\l$
outside these discs. By \eqref{Sept3a}, this part is
 outside $\SS ^{(2)}(k,\xi)$, when $\xi $ is as described in the statement of the lemma.
\end{proof}

Let $\k  _0\in \R^2$ be fixed and ${\cal N}(k,r_2,\k  _0,J_0)$ be the
following subset of the lattice $\k _0+\p_{\n}$, $\n \in \Omega (r_2)$:
$${\cal N}(k,r_2,\k _0,J_0)=\left\{\k _0+\p_{\n}:\n \in \Omega (r_2):\ J(\k  _0+\p_{\n})\leq J_0\right\},$$
$J$ being defined by Definition \ref{D-J}. Thus, ${\cal N}$ includes
only such $\n$ that the surrounding $k^{\gamma r_1}$- block contains
less than $J_0 $ effectively resonant $k^\delta$-clusters. Let
$N(k,r_2,\k _0,J_0, \xi)$ be the number of points $\k _0+\p_{\n}$ in $\SS ^{(2)}
(k,\xi )\cap {\cal N}(k,r_2,\k _0,J_0)$.
\begin{lemma}\label{norm2/3}
If  $r_2>10\gamma r_1'$ and $\xi>96\mu r_2J_0$, then
\begin{equation}\label{eqnorm2/3}
N(k,r_2,\k _0,J_0, \xi)\leq k^{\frac{2}{3}r_2+43\gamma r_1}.
\end{equation}
\end{lemma}
\begin{proof} Let us call a subset $\tilde \SS $ of $\SS ^{(2)}(k,\xi )$ elementary if it can be described by a formula of the type:
$$ \tilde \SS:=\{\k: a<\varkappa_1<b, f_1(\varkappa_1)<\varkappa_2<f_2(\varkappa_1)\},$$
where the curves $\varkappa_2=f_i(\varkappa_1)$, $i=1,2$, belong to
the boundary of $\SS ^{(2)}(k,\xi )$,   have the lengths less than 1,
functions $f_i(\varkappa_1)$ are monotone, continuously
differentiable and have no inflection points. We assume that the
boundaries $\varkappa_2=f_i(\varkappa_1)$ are parameterized by
$\varkappa_1$ for definiteness. The set where
$\varkappa_1=f_i(\varkappa_2)$, $a<\varkappa_2<b$, is completely
analogous.

Next, we show that the number of points in $\tilde \SS \cap {\cal
N}(k,r_2,\k _0,J_0)$ does not exceed $2^{14}k^{\frac{2}{3}r_2}$.
Indeed, let us consider a segment $\p_{\n-\n'}$ between
     two points $\k _0 +\p_{\n}$ and $\k _0 +\p_{\n'}$ in $\tilde \SS$.
     Obviously, $\||\p_{\n-\n'}\||<2k^{r_2}$ and $p_{\n-\n'}>k^{-\mu
     r_2}$. The direction of the
     segment cannot be parallel to the axis $\varkappa_2$ by Corollary \ref{C:curves-2ind}.
     We
     enumerate the points $ \k _0 +\p_{\n}\in \tilde \SS \cap
     {\cal N}(k,r_2,\k _0,J_0)$ in the order of the increasing first
    coordinate  and connect subsequent points by  segments. Consider all segments with the length greater or equal to
     $\frac{1}{64}k^{-\frac{2r_2}{3}}$. The number of such segments does not
     exceed $128k^{\frac{2r_2}{3}}$, since they are much longer than the width of $\tilde \SS $ (Corollary \ref{C:curves-2ind}). It remains to estimate the
     number of segments with the length less than
     $\frac{1}{64}k^{-\frac{2r_2}{3}}$.

      First, we prove that no more than two
     segments $\p_{\n_1-\n'_1}$, $\p_{\n_2-\n'_2}$ can be parallel to
     each other, if they are in the same elementary component $\tilde \SS $. Indeed, both ends of $\p_{\n_1-\n'_1}$ are in
     $\tilde \SS$. By the previous lemma there is a piece of the segment which is not in $\tilde \SS$ (we notice that now we use the lemma for $k^{-\mu r_2}<\eta <\frac{1}{64}k^{-\frac{2r_2}{3}}$). Hence, the segment intersects one of the curves $\varkappa_2=f_i(\varkappa_1)$ twice. It follows, that there is a point on the curve,
     where the curve is parallel to the segment. Suppose another segment $\p_{\n_2-\n'_2}$  intersects the same curve. Then, there is a point on the curve,
     where the curve is parallel to the second segment. Since the curve is concave, it can not be the case. Therefore, $\p_{\n_2-\n'_2}$  intersects another  curve. It follows that no more than two
     segments $\p_{\n_1-\n'_1}$, $\p_{\n_2-\n'_2}$ can be parallel to
     each other, if they are in the same elementary component $\tilde \SS $.

     To finish the proof of the lemma  we consider two cases. Suppose $q$ in the inequality
     \eqref{qind} satisfies the estimate $q>k^{2r_2/3}$. Then, by Lemma \ref{Lattice-3ind}, the
     number of vectors  $\p_\n$,
     $\||\p_{\n}\||<2k^{r_2}$, satisfying the inequality
     $p_{\n}<\frac{1}{64}k^{-2r_2/3}$ does not exceed
     $2^{12}k^{2r_2/3}$. Since each of them can be used only twice, the
     total number of short segments  does not exceed $2^{13}k^{2r_2/3}$.

     Let  $q\leq k^{2r_2/3}$.
     If $|\epsilon _q|>\frac{1}{64}q^{-1}k^{-r_2}$. Then, obviously,
     $\frac{1}{64}k^{-2r_2/3}<|\epsilon _q|qk^{r_2/3}$. Applying Lemma
     \ref{Lattice-2ind}, we obtain that the number of segments with the
     length less than $\frac{1}{64}k^{-2r_2/3}$ is less than
     $k^{2r_2/3}$. Since each of them can be used only twice, the
     total number of short segments  does not exceed $2k^{2r_2/3}$.
     It remains to consider the case $q\leq
     k^{2r_2/3}$, $|\epsilon _q|\leq \frac{1}{64}q^{-1}k^{-r_2}$. By
     Lemma \ref{Lattice-1ind}, clusters are well separated. Considering
     that the distance between clusters is greater than
     $\frac{1}{2q}$ and the size of each cluster is less than
     $\frac{1}{8q}$, we obtain that no more than $8q$ clusters can intersect $\tilde \SS $. Indeed, the distance between two clusters is greater than $\frac{1}{2q}$.
    By Corollary  \ref{C:curves-2ind},  the set $\tilde \SS $ belongs to the $k^{-\frac{\xi}{48J_0}}$-neighborhood of
    each curve $\varkappa_2=f_i(\varkappa_1)$, $i=1,2$. Using the hypothesis of the lemma we easily get that the size of the neighborhood is $o(q^{-1})$.
    If a cluster intersects $ \tilde \SS$, its $\frac{1}{4q}$-neighborhood intersects both curves $\varkappa_2=f_i(\varkappa_1)$, $i=1,2$. Since the distance between clusters is greater than
     $\frac{1}{2q}$, the distance along the curve between its intersection with $\frac{1}{4q}$-neighborhoods of different clusters is  greater than $\frac{1}{4q}$. Considering that the lengths of the curves is
     less than 1, we obtain that no more than $8q$ clusters
     can intersect  $\tilde \SS$.
     Next, the segments with the length less than
     $\frac{1}{2}k^{-2r_2/3}$ cannot connect different clusters, since the
     distance between clusters is greater than $\frac{1}{2q}\geq
     \frac{1}{2}k^{-2r_2/3}$. Therefore, any segment of the length
     less than $\frac{1}{2}k^{-2r_2/3}$ is inside one
     cluster.   The part of the shorter curve inside the clusters has the length $L_{in}$ which is less
     than the double size of a cluster $10|\epsilon _q|k^{r_2}$ (the curve is concave) multiplied by
     the number of clusters $8q$, i.e., $L_{in}<80|\epsilon
     _q|qk^{r_2}$. If we consider the segments with the length greater
     than $|\epsilon _q|qk^{r_2/3}$, then the number of such segments
     is less than $L_{in}/|\epsilon _q|qk^{r_2/3}$, i.e., it is less
     than $80k^{2r_2/3}$. By Lemma \ref{Lattice-2ind}, the total
     number of segments of the length less than $|\epsilon
     _q|qk^{r_2/3}$ is less than $k^{2r_2/3}$. Each of them can be used only twice. Thus, the total
     number of segments is less than $162k^{2r_2/3}$.

     We proved that the number of segments in $ \tilde \SS$ does not exceed
     $2^{14}k^{2r_2/3}$. Therefore, the number of points in
     $\tilde \SS \cap {\cal N}(k,r_2,\k _0,J_0)$ does not exceed
     $2^{14}k^{\frac{2}{3}r_2}+1$. Considering that
     $k^{\gamma r_1}>2^{15}$, we obtain that the number of points
     inside $\tilde \SS \cap {\cal N}(k,r_2,\k _0,J_0)$ does not
     exceed $k^{\frac{2}{3}r_2+\gamma r_1}$.

     If we show that  $\SS ^{(2)}(k,\xi )$ is the union of no more than $k^{42\gamma r_1}$ elementary components $\tilde \SS$, then  estimate
     \eqref{eqnorm2/3} easily follows. Indeed, let us consider the boundary of $\SS ^{(2)}(k,\xi)$. It is described by  curves
$D(\k , k^{2}\pm k^{-\xi})=0$, $\k \in \R^2$. We break each curve
into elementary components as described in Appendix 8. By Lemma
\ref{elpieces} the number of such pieces is less than $k^{17\gamma
r_1}$. With each elementary piece of the boundary we associate the
part of the adjacent connected component of $\SS ^{(2)}(k ,\xi) $,
which is in the   $k^{-\frac{\xi}{48J_0}}$-neighborhood of the elementary
piece. By Corollary \ref{C:curves-2ind}, every point in $\SS
^{(2)}(k,\xi)$ belongs to such a component,
some  components  overlapping. Let us consider one of these
components $\hat \SS $. By construction, it is adjacent to a
boundary elementary component, which can be described in the form
$\varkappa_1=f_1(\varkappa_2)$ or $\varkappa_2=f_1(\varkappa_1)$.
Let us assume for definiteness that it is described by the formula
$\varkappa_2=f_1(\varkappa_1)$. By Corollary \ref{C:curves-2ind},
there is another boundary (described by
$\varkappa_2=f_2(\varkappa_1)$) of $\hat \SS $ in the
$k^{-\frac{\xi}{48J_0}}$-neighborhood of $\varkappa_2=f_1(\varkappa_1)$. It
also can be split into  no more than $k^{17\gamma r_1}$ elementary
components. Further, each elementary component contains no more than
$k^{8\gamma r_1}$  points $\k: D(\k , k^{2}+ k^{-\xi})=D(\k ,
k^{2}- k^{-\xi})=0$, unless the last equality is an identity on
this component (Bezout Theorem). We use these points to break each
elementary component into at most $k^{8\gamma r_1}$ parts.
Correspondingly, we split the set $\hat \SS$ by lines
$\varkappa_1=C$ into at most $k^{25\gamma r_1}$ components $\tilde
\SS $. The second boundary of $\tilde \SS $ also can be
parameterized by $\varkappa_2$, since $D_{\varkappa_2}\neq 0$ on an
elementary component of the boundary. By the definition of an
elementary component of the boundary (Appendix 8), both functions
$\varkappa_2=f_i(\varkappa_1)$ are monotone, continuously
differentiable and don't have inflection points, the length of the
corresponding curves being less than 1. Moreover, neither boundary
contains intersections with other pieces of the boundary of $\SS
^{(2)}(k,\xi )$. Thus, $\SS ^{(2)}(k,\xi )$ is the union of at most
$k^{42\gamma r_1}$ elementary components $\tilde \SS$, each being
bounded by lines $\varkappa_i=a,b$ and elementary pieces of the
boundary of $\SS  ^{(2)}(k,\xi )$, which do not intersect with other
pieces of the boundary of $\SS ^{(2)}(k,\xi )$.
\end{proof}

\subsubsection{Model Operator for Step IV \label{MOforStep4}}
Let $r_3>r_2$.
We repeat for $r_3$ the construction from the
subsection~\ref{MOforStep3}, which was done for arbitrary $r_2>r_1$. In particular, we introduce the corresponding sets $\MM_2^{weak}(r_3,\varphi_0)$, $\MM_2^{str}(r_3,\varphi_0)$, $\MM_{2,tw}(r_3,\varphi_0)$. Let $\m \in \Omega (r_3)$. We denote the $k^{\gamma r_1}$-component
containing $\m$ by $\tilde \Pi (\m)$ and the corresponding projector
by $\tilde P(\m)$. For $\m$ belonging to the same $k^{\gamma
r_1}$-component, $\tilde \Pi (\m)$ and  $\tilde P(\m)$ are the same.
 Put
\begin{equation}\label{M^3} {\MM}^{(3)}:={\MM}^{(3)}(\varphi _0, r_3)=\{\m\in \MM^{(2)}(\varphi _0, r_3)\cup \Omega _{s}^{(2)}(r_3)\setminus\MM_2^{weak}(r_3,\varphi_0)
:\ \varphi_0\in{\cal O}_\m^{(3)}(r_2',1)\},\end{equation} where $\Omega _{s}^{(2)}(r_3)$ is the extension of $\Omega _s^{(2)}(r_2)$ to $\Omega (r_3)$,
\begin{equation} \label{se}\Omega _{s}^{(2)}(r_3)=\{\m\in \Omega (r_3),\ 0<p_\m\leq
k^{-5r_1'}\},\end{equation}
${\cal O}_\m^{(3)}(r_2',\tau)$ is the union of the disks of the
radius $\tau k^{-r_2'}$ with the centers at poles of the resolvent
$(\tilde P(\m)(H(\k^{(2)}(\varphi ))-k^{2}I)\tilde P(\m))^{-1}$ in the $k^{-44r_1'-2-\delta }$-neighborhood of $\varphi _0$.
(Here $\MM^{(2)}(\varphi _0, r_3)$ is defined as in \eqref{M^2} with
$r_3$ instead of $r_2$). For $\m$ belonging to the same $k^{\gamma
r_1}$-component, the sets ${\cal O}_\m^{(3)}(r_2',\tau)$ are
identical. We say that $\m \in {\MM}^{(3)}$ is $k^{\gamma
r_1}$-resonant. The corresponding $k^{\gamma r_1}$-clusters we call
resonant too.

 Let $\varphi_0\in \omega^{(3)} (k,
\delta ,1)$. By
construction of the non-resonant set $\omega^{(3)} (k, \delta ,1)$,
we have ${\MM}^{(3)}\cap \Omega (r_2)=\emptyset $.

Further we use the property of the set $\MM^{(3)}$ formulated in the
next lemma which is an analogue of the Lemma~\ref{L:2/3-1}.

\begin{lemma}\label{L:2/3-1ind} Let $r_2'>2k^{(\gamma +\delta_0)10^{-4}r_1-2\delta_*}$. Let  $1/20<\gamma '<20$, $\m _0\in\Omega
(r_3)$ and $\Pi _{\m_0}$ be the $k^{\gamma
'r_2}$-neighborhood (in $\||\cdot\||$-norm) of $\m_0$. Then the set
$\Pi _{\m _0}$ contains less than $k^{\frac 23 \gamma'r_2 +50\gamma r_1}$
elements of $\MM^{(3)}$.
\end{lemma}

\begin{proof}
First, we notice that the condition $r_2'>2k^{(\gamma +\delta_0)10^{-4}r_1-2\delta_*}$ is consistent with the restriction
$r_2'<k^{\delta_0 r_1-4}$ in \eqref{r_2'}. If $\m\in\MM^{(3)}$, then there is a $\varphi_*$ such that
$|\varphi_*-\varphi_0|<k^{-r_2'}$ and
$$
\hbox{det}\,\left(\tilde P(\m)(H(\k^{(2)}(\varphi_*))-k^{2}I)\tilde P(\m)\right)=0,
$$
where $\tilde P(\m)$ is the projection corresponding to the $k^{\gamma r_1}$-cluster $\tilde \Pi (\m)$, which includes $\m$. The cluster $\tilde \Pi (\m)$ can be simple, white, grey or black.
Since $\varphi _0$ is close to $\varphi_*$, perturbation arguments give:
\begin{equation}\label{1ind}
\left\|(\tilde P(\m)(H(\k^{(2)}(\varphi_0))-k^{2}I)\tilde P(\m))^{-1}\right\|\geq \frac{1}{4}k^{\xi },\ \ \ \xi \geq r_2'-1.
\end{equation}
 We will apply
Lemma~\ref{norm2/3} to $\tilde \Pi(\m) $ with $\xi=r_2'-1$ in
order to prove the lemma in hand in the same way we proved
Lemma~\ref{L:2/3-1}, using Lemmas \ref{L:number of points-1}, \ref{4.10}, \ref{t4.10}, \ref{2t4.10}, \ref{nontriv4.10}.  There are some technical complications though. Here is a detailed proof.

 We start with considering simple boxes $\tilde \Pi (\m)$,  $\m\in \MM^{(3)}\cap \Omega^{(2)} _{s}(r_3)$. Each box has the $\||\cdot \||$-size $2k^{r_1/2}$ and contains no other than $\m$ elements of $\MM^{(2)}(\varphi _0, r_3)\cup \Omega^{(2)} _s(r_3)$.  Indeed,
$\k^{(2)}(\varphi_0)$ satisfies the conditions of Lemma
\ref{L:geometric2}. This means that the $k^{\delta }$-cluster around
each $\q$: $\||\p_\q\||<k^{r_{1}}$ is non-resonant. Since
$\k=\k^{(2)}(\varphi_0)+\p_\m$ is a small perturbation of
$\k^{(2)}(\varphi_0)$, the $k^{\delta }$-cluster around each $\m+\q$:
$\||\p_\q\||<k^{r_{1}}$ is non-resonant too. This means $\m+\q \not
\in \MM^{(2)}$. Further, $\m+\q \not \in  \Omega^{(2)} _s(r_3)$ by
\eqref{below}, since $\m \in  \Omega^{(2)} _s(r_3)$ and
$\||\p_\q\||<k^{r_1}$. Thus, $\m+\q \not \in  \MM^{(2)}(\varphi _0,
r_3)\cup\Omega^{(2)} _s(r_3)$. Next, we apply Lemma \ref{norm2/3} with
$\k _0=\k^{(2)}(\varphi_0)+\p_{\m_0}$, $J_0=1$, $\xi ={r_2'-1}$,
to conclude that the number of simple
 boxes $\tilde \Pi (\m)$, $\m\in \MM^{(3)}\cap \Omega^{(2)} _{s}(r_3)$ does not exceed $k^{\frac 23 \gamma'r_2 +43\gamma r_1}$. Indeed, we rewrite $\k^{(2)}(\varphi_0) +\p_{\m}$ in the form:
 $\k^{(2)}(\varphi_0)+\p_{\m} =\k  _0+\p_{\n}$,  $\n=\m -\m_0\in \Omega(\gamma 'r_2)$.
 By \eqref{triind},
 $\k _0 +\p_{\n} \in \SS ^{(2)}(k,\xi )$ (the operator in formula \eqref{triind} having the size
$2k^{\frac{1}{2}r_1}$ and $\xi =r_2'-1$, see \eqref{1ind}).
Since, $\tilde \Pi (\m)$ is simple, $\k _0 +\p_{\n}\in {\cal
N}(k,\gamma'r_2,\k _0,1)$ (here, $\gamma $ is taken to be equal to
$1/2$ in the definition of $\SS ^{(2)}(k,\xi)$). Thus, $\k _0 +\p_{\n}\in
\SS^{(2)}(k,\xi )\cap {\cal N}(k,\gamma'r_2,\k _0,1)$. By  Lemma
\ref{norm2/3}, the number of such $\k _0 +\p_{\n}$ does not exceed
 $k^{2\gamma'r_2/3+\frac{43}{2} r_1}$. Therefore, the number of
 $\MM^{(3)}$ elements in simple boxes also does not exceed
 $k^{2\gamma'r_2/3+\frac{43}{2} r_1}$.

 Next, let us consider white clusters $\tilde \Pi (\m)$, such that
 $\xi \geq k^{\frac 16 \gamma r_1-2\delta_* }$. Generally speaking,
 $\tilde \Pi (\m)$ has a shape (in $\Z^4$) more complicated than a
 rectangular. However, each such cluster can be put in a box of the
 size $3k^{\gamma r_1/2+2\delta _0 }$, the box containing less than
 $k^{\frac 16 \gamma r_1 -\delta _0 r_1}$ elements of $\MM ^{(2)}$
 and the box resolvent satisfying \eqref{1ind} with
$\xi =k^{\frac 16 \gamma r_1-2\delta_* }$ (Lemma \ref{propw}).
Applying Lemma~\ref{norm2/3} to such boxes\footnote{Here we slightly abuse the notation of $J$ as Definition \ref{D-J} deals with the estimates of the resolvent of a cluster while the points in $\MM^{(2)}$ used in the definition of the colored boxes are defined by the distance to the poles of the cluster. Though the result of the form of Corollary \ref{estnonres1} and simple perturbative arguments show that these two definitions are equivalent upto an insignificant factor.}
($\k _0=\k^{(2)}(\varphi_0)+\p_{\m_0}$, $J_0=k^{\frac 16 \gamma r_1
-\delta _0 r_1}$, $\xi =k^{\frac 16 \gamma r_1-2\delta_* }$), we
obtain that the number of $\Pi _{\m _0}$ points $\m$ in  such boxes
does not exceed $k^{2\gamma'r_2/3+43\gamma r_1}$.  Similarly, we
can treat grey boxes when $\xi \geq k^{\frac 12 \gamma r_1+2\delta
_0r_1-2\delta_* }$ (Lemma \ref{propg}), black boxes when $\xi \geq
k^{\gamma r_1+\delta _0r_1-2\delta_* }$ (Lemma \ref{propb}).
 However, in some cases $\xi $ does not satisfy the previous estimates from below. For such $\xi $ a somewhat more complicated construction is needed.
Indeed, let us consider $(\tilde P(H(\k^{(2)}(\varphi_0))-k^{2}I)\tilde
P)^{-1}$ for $\tilde \Pi$ being white, grey or black cluster
containing a point(s) of ${\MM}^{(3)}$. A cluster $\tilde \Pi$
consists of blocks with the minimal size $k^{\gamma r_1/6}$. Let us
create a substructure inside $\tilde \Pi$. Namely, we construct
 white, grey and black clusters corresponding to a smaller
$\gamma$ which we denote by $\tilde \gamma $,
$\tilde{\gamma}=10^{-4}\gamma$. Note, that there are no simple small clusters inside $\tilde \Pi$, since $\tilde \Pi$ is not simple. The size of these new clusters is much
smaller than $k^{\gamma r_1/6}$. However, they have  properties
analogous to those of the bigger clusters ($\gamma $). These new clusters we call
subclusters.
We assert that at least  one subcluster satisfies one of the following estimates (depending on whether this subcluster is white, grey or black):

\begin{equation}\label{estwind}
\left\|(P_{w,sub}(H(\k^{(2)}(\varphi_0))-k^{2}I)P_{w,sub})^{-1}\right\|>k^{k^{\frac{\tilde{\gamma}r_1}{6}-2\delta_*}},
\end{equation}

\begin{equation}\label{estgind}
\left\|(P_{g,sub}(H(\k^{(2)}(\varphi_0))-k^{2}I)P_{g,sub})^{-1}\right\|>
k^{k^{(\frac{\tilde{\gamma}}{2}+2\tilde{\delta}_0)r_1-2\delta_*}},
\end{equation}

\begin{equation}\label{estbind}
\left\|(P_{b,sub}(H(\k^{(2)}(\varphi_0))-k^{2}I)P_{b,sub})^{-1}\right\|>
k^{k^{(\tilde{\gamma}+\tilde{\delta}_0)r_1-2\delta_*}},
\end{equation}
where $\tilde{\delta}_0=\tilde{\gamma}/100$ (cf. definition of
$\delta_0$). Indeed, if all subclusters satisfy the inequalities
opposite to the inequalities above, then the perturbation series for the resolvent of the
bigger cluster  ($\gamma $) (with respect to the block operator consisting of subclusters) converges, see the proof of Theorem \ref{Thm3}, in particular the proof of \eqref{5-s} -- \eqref{5-4}.
Hence, we have
$$
\left\|(\tilde P(H(\k^{(2)}(\varphi_0))-k^{2}I)\tilde
P)^{-1}\right\|\leq
k^{k^{(\tilde{\gamma}+\tilde{\delta}_0)r_1-2\delta_*}},
$$
which contradicts to \eqref{1ind} under the hypothesis of the lemma
$r_2'>2k^{(\gamma +\delta_0)10^{-4}r_1-2\delta_*}$.

Now, let us consider a resonant $k^{\gamma r_1}$-cluster $\tilde
\Pi$, see \eqref{1ind}, and the substructure inside. Note that each
subcluster satisfying \eqref{estwind}-\eqref{estbind} can be treated
the same way we treated the bigger clusters for large $\xi $.
Namely, let us  consider all $k^{\gamma r_1}$-clusters $\tilde \Pi $
for which there exists a white subcluster satisfying
\eqref{estwind}. By Lemma \ref{propg} each such subcluster can be
put in a box of the size $3k^{\tilde \gamma r_1/2+2\tilde \delta _0
}$, the box resolvent satisfying \eqref{estwind}. Such box has less
than $k^{(\frac{\tilde{\gamma}}{6}-\tilde{\delta}_0)r_1}$ points $\m
$ of $ {\MM }^{(2)}$. Now, applying Lemma~\ref{norm2/3} with
$\k _0=\k^{(2)}(\varphi_0)+\p_{\m_0}$,
$J_0=k^{(\frac{\tilde{\gamma}}{6}-\tilde{\delta}_0)r_1}$,
$\xi=k^{\frac{\tilde{\gamma}r_1}{6}-2\delta_*}$, we obtain that the
number of  points $\m$ in white subclusters  \eqref{estwind}  does
not exceed $k^{\frac{2\gamma'r_2}{3}+43\tilde \gamma r_1}$. Here we
notice that condition of Lemma~\ref{norm2/3} holds, since
$r_2<k^{\tilde{\delta}_0r_1-8\delta_*}$ by \eqref{r_2}. It follows
that the number of $k^{\gamma r_1}$-clusters $\tilde \Pi (\m)$,
containing at least one white subcluster \eqref{estwind}, does not
exceed $k^{\frac{2\gamma'r_2}{3}+43\tilde \gamma r_1}$.

Next, we consider all $k^{\gamma r_1}$-clusters $\tilde \Pi (\m )$ for which there
exists a grey subcluster, satisfying \eqref{estgind}, but no white subclusters satisfying \eqref{estwind} . Applying
Lemma \ref{propg} and Lemma~\ref{norm2/3} with
$J_0=k^{(\frac{\tilde{\gamma}}{2}+\tilde{\delta}_0)r_1}$ and $\xi=k^{(\frac{\tilde{\gamma}}{2}+2\tilde{\delta}_0)r_1-2\delta_*}$, we obtain that the number of such $k^{\gamma
r_1}$-clusters $\tilde \Pi (\m )$ in $\Pi _{\m _0}$ does not exceed
$k^{\frac{2\gamma'r_2}{3}+43\gamma r_1}$.

Similarly, applying Lemma \ref{propb} and Lemma~\ref{norm2/3} with
$\xi=k^{(\tilde{\gamma}+\tilde{\delta}_0)r_1-2\delta_*}$ and
$J_0 =ck^{\tilde{\gamma}r_1+3}$, we obtain
that the number of \ $k^{\gamma r_1}$-clusters $\tilde \Pi (\m )$, containing a
black subcluster \eqref{estbind} (and no grey or white subclusters, satisfying   \eqref{estwind},  \eqref{estgind}), does not exceed
$k^{\frac{2\gamma'r_2}{3}+43\gamma r_1}$. Here, we also used
$r_2<k^{\tilde{\delta}_0r_1-3-8\delta_*}$.


Combining these estimates, we see that the number of clusters $\tilde \Pi $, containing at least one point of $\MM ^{(3)}$ does not exceed $k^{\frac{2\gamma'r_2}{3}+43\gamma r_1}$. Taking into account that each
$k^{\gamma r_1}$-cluster has a size not greater than
$k^{\frac{3\gamma r_1}{2}+3}$ and, hence,  contains less than
$k^{6\gamma r_1+12}$ elements, we obtain that the total number of
elements of $\MM^{(3)}$ in $\Pi _{\m _0}$, does not exceed
$k^{\frac{2\gamma'r_2}{3}+50\gamma r_1}$.

\end{proof}

We continue with constructing $k^{\gamma r_1}$-clusters in $\Omega
(r_3)$, $r_3>r_2$, the same way we did it for $\Omega (r_2)$ in
Section \ref{MOforStep3}. We call a $k^{\gamma r_1}$-cluster
resonant if it contains $\m \in {\MM}^{(3)}$, see \eqref{M^3}. Next,
we repeat the construction after Lemma \ref{L:2/3-1}. More
precisely, let us split $\Omega (r_3)\setminus \Omega (r_2)$ into
$k^{\gamma r_2}$-boxes, $\gamma =\frac{1}{5}$.

\begin{enumerate} \item {\em Simple region.} \label{simple-2} Let $\Omega _s^{(3)}(r_3)\subset \Omega _{s}^{(2)}(r_3)$ be defined by the formula:
\begin{equation} \label{Omega-s}\Omega _{s}^{(3)}(r_3)=\{\m\in \Omega (r_3),\ 0<p_\m\leq k^{- r_2'k^{2\gamma
r_1}}\}.\end{equation}
It is easy to see  that $\Omega _s^{(3)}(r_3)\subset \MM (\varphi _0,r_3)$, since $p_{\m}$ is small, see \eqref{M}, \eqref{resonance1}. Next, if  $\m\in \Omega_s^{(3)}(r_3)$, then there are
 no other elements of $\MM (\varphi _0,r_3)$ in the $k^{\delta }$-box  around  $\m$.  Further,
$\m$ itself can belong
or do not belong to $\MM^{(2)}(\varphi _0, r_3)$,  but there are
 no other elements of
$\MM^{(2)}(\varphi _0, r_3)$ in the $k^{r_1}$-box  around such $\m$.
The proof of these facts is analogous to that in Step III, see
``Simple region", page \pageref{simple}. Next, if $\m\in\Omega
_s^{(3)}(r_3)$, then  there are no other elements of $\Omega
_s^{(3)}(r_3)$ in the surrounding $\||\cdot\||$-box of the size
$k^{r_2}$, see \eqref{below}. Last, $\m$ can belong or do not belong
to $\MM^{(3)}(\varphi _0, r_3)$, but there are no other elements
from $\MM^{(3)}(\varphi _0, r_3)$ in the $k^{r_2}$-box around such
$\m$. Indeed, $\k^{(2)}(\varphi_0)$ satisfies the conditions of
Lemma \ref{L:geometric3}. This means that the $k^{\gamma
r_1}$-cluster around each $\q$: $0<\||\p_\q\||<k^{r_2}$ is
non-resonant. Since $\k=\k^{(2)}(\varphi_0)+\p_\m$ is a small
perturbation of $\k^{(2)}(\varphi_0)$, the $k^{\gamma r_1}$-box
around each $\m+\q$: $0<\||\p_\q\||<k^{r_2}$ is non-resonant too.
This means $\m+\q \not \in \MM^{(3)}(\varphi _0, r_3)$.

For each $\m \in \Omega _s^{(3)}(r_3)$ we consider its $k^{r_2/2}$-neighborhood.
The union of such boxes we call the simple region and denote it by
$\Pi _s(r_3)$. The corresponding projection is $P_s(r_3)$. Note that the
distance from the simple region to the nearest point of $\MM^{(3)}$
is greater than $\frac12 k^{r_2}$.
\item {\em Black, grey and white regions} are
defined in the same way as in the construction after Lemma
\ref{L:2/3-1} with $r_3$ instead of $r_2$, $r_2$ instead of $r_1$, $\MM^{(3)}$ instead of $\MM^{(2)}$ and the restriction $p_\m >k^{- r_2'k^{2\gamma
r_1}}$ instead of $p_\m >k^{- 5r_1'}$. We continue to use
notation $P_b, P_g, P_g', P_w, P_w'$ and $\Pi_b, \Pi_g, \Pi_g',
\Pi_w, \Pi_w'$. Sometimes, where it can lead to confusion we will
write $P_b(r_3)$ etc. to distinguish these objects from the ones
introduced in Step II.

\item {\em Non-resonant region.} \label{NRR} Now, the non-resonant region consists of two parts: the simpler part which was non-resonant
already in the previous step and the part which is new for the
current step. Namely, first we consider $k^{\delta}$-components corresponding to points in the set $\MM(r_3, \varphi _0)\setminus
\left(\MM(r_2, \varphi _0)\cup
\MM^{(2)}(r_3,\varphi_0)\cup \Omega ^{(2)}_s(r_3)\cup\MM_{2,tw}(r_3,\varphi_0)\right)$. The union of this
neighborhoods we denote $\Pi _{nr,\delta}$. The corresponding
projection is $P_{nr,\delta}$. These $k^{\delta }$-clusters can be treated by means of the second step. We also consider
 all points in the set $\MM^{(2)}(r_3,
\varphi _0)\cup \Omega ^{(2)}_s(r_3)\setminus \left(\MM ^{(2)}(r_2, \varphi _0)\cup
\MM^{(3)}(r_3,\varphi_0)\cup \Omega _{s}^{(3)}(r_3)\right)$.
 We construct simple, white, grey and black clusters around them exactly as in preparation to Step III. The union of these
clusters we denote $\Pi _{nr,r_1}$. The corresponding
projection is $P_{nr,r_1}$. The set $\Pi _{nr}:=\Pi
_{nr,\delta}\cup\Pi _{nr,r_1}$ is called the non-resonant set with
$P_{nr}$ being the corresponding projection. The part of the
non-resonant region which is outside
$\Pi_s\cup\Pi_b\cup\Pi_g\cup\Pi_w$, we denote $\Pi _{nr}'$ and the
corresponding projection by $P_{nr}'$.
\end{enumerate}

We put as before
\begin{equation}P_r(r_3):=P_s(r_3)+P_b(r_3)+P_g'(r_3)+P_w'(r_3),\ \ \ P^{(3)}:=P_r(r_3)+P_{nr}'(r_3)+P(r_2). \label{P(3)} \end{equation}
We also continue to use the similar agreement in the notation which
we set in Step II. We just note that now we use $r_2$ rather than
$r_1$ to establish equivalence between the boxes.

We continue construction from Step II. Repeating the arguments from
the proofs of Lemmas~\ref{L:black},~\ref{L:grey},~\ref{L:white} with
obvious changes (in particular, using Lemma~\ref{L:2/3-1ind} instead
of Lemma~\ref{L:2/3-1}) we obtain the following results.
\begin{lemma} \label{L:blackind} \begin{enumerate} \item
Each $\Pi _b^j$ contains no more than $k^{\gamma r_2/2-\delta
_0r_2+150\gamma r_1}$ black boxes.
\item The size of $\Pi _b^j$ in $\||\cdot \||$ norm is less than
$k^{3\gamma r_2/2+150\gamma r_1}$.
\item Each $\Pi _b^j$ contains no more than $k^{\gamma r_2+150\gamma r_1}$ elements of
$\MM^{(3)}$. Moreover, any box of $\||\cdot \||$-size $k^{3\gamma
r_2/2+150\gamma r_1}$ containing $\Pi _b^j$ has no more than
$k^{\gamma r_2+150\gamma r_1}$ elements of $\MM^{(3)}$ inside.
\end{enumerate}
\end{lemma}

\begin{lemma}\label{L:greyind} \begin{enumerate} \item
Each $\Pi _g^j$ contains no more than $k^{\gamma r_2/3+2\delta
_0r_2}$ grey boxes.
\item The size of $\Pi _g^j$ in $\||\cdot \||$ norm is less than
$k^{5\gamma r_2/6+4\delta _0r_2}$.
\item Each $\Pi _g^j$ contains no more than $k^{\gamma r_2/2+\delta _0r_2}$ elements of
$\MM^{(3)}$.\end{enumerate}
\end{lemma}

\begin{lemma}\label{L:whiteind} \begin{enumerate} \item The size of $\Pi _w^j$ in $\||\cdot \||$ norm is less than
$k^{\gamma r_2/3-\delta _0r_2}$.
\item Each $\Pi _w^j$ contains no more
than $k^{\gamma r_2/6-\delta _0r_2}$ points of $\MM^{(3)}$.
\end{enumerate}
\end{lemma}

The construction of the rest of Section~\ref{MOforStep3} stays unchanged. Let us
introduce corresponding notation, formulate the results and provide
some comments.

Next lemmas are the analogues of
Lemmas~\ref{L:Pnr},~\ref{L:Pr},~\ref{L:Ps}.

\begin{lemma}\label{L:Pnrn}Let $\varphi _0\in \omega^{(3)}(k,\delta ,\tau )$,
$|\varphi-\varphi _0|<k^{-k^{r_1}}$. Then,
\begin{equation}\label{Pnrn}
\left\|\Bigl(P_{nr}\bigl(H(\k^{(3)}(\varphi
))-k^{2}I\bigr)P_{nr}\Bigr)^{-1}\right\|<k^{r_2'k^{2\gamma
r_1}}k^{r_2'}\leq k^{k^{3\gamma r_1}}.
\end{equation} \end{lemma}
\begin{proof} Construction in Section \ref{S:3} is made for an arbitrary large $r_2$. Here we repeat it for $r_3$. We use Lemma \ref{L:Pnr} for $\Pi _{nr, \delta }$, Lemma \ref{L:Pr} for white, grey and black clusters ($\varepsilon _0=k^{-r_2'}$).  We also use Lemma \ref{L:Ps} ($p_{\m}>k^{-r_2'k^{2\gamma r_1}}$, $\varepsilon _0=k^{-r_2'}$), for simple clusters in $\Pi _{nr, r_1}$.  We also use \eqref{dk0-3}. All together the estimates for the clusters resolvents yield \eqref{Pnrn}.  The estimate \eqref{Pnrn} is stable when $|\varphi-\varphi _0|<k^{-k^{r_1}}$, since
$k^{-k^{r_1}+2}=o(k^{-k^{3\gamma r_1}})$.
\end{proof}

\begin{lemma}\label{L:Prn} Let $\varphi _0\in \omega^{(3)}(k,\delta ,\tau )$,
and $|\varphi-\varphi _0|<k^{-k^{r_1}}$, $i=1,2,3$. Then,
\begin{enumerate}
\item The number of poles of the resolvent $\Bigl(P_i\bigl(H(\k^{(3)}(\varphi
))-k^{2}I\bigr)P_i\Bigr)^{-1}$ in the disc $|\varphi-\varphi
_0|<k^{-k^{r_1}}$ is no greater than $N_i^{(2)}$, where $N_1^{(2)}=k^{\gamma
r_2+150\gamma r_1}$, $N_2^{(2)}=k^{\gamma r_2/2+\delta _0r_2}$,
$N_3^{(2)}=k^{\gamma r_2/6-\delta _0r_2}$.
\item Let $\varepsilon$ be the distance to the nearest pole of
the resolvent in ${\cal W}^{(3)}$ and
$\varepsilon_0=\min\{\varepsilon,\,k^{-r_2'}\}$. Then the following
estimates hold:

\begin{align}\label{Pr-1n}
\left\|\Bigl(P_i\bigl(H(\k^{(3)}(\varphi
))-k^{2}I\bigr)P_i\Bigr)^{-1}\right\|<k^{2r_2'k^{2\gamma
r_1}}k^{r_2'}\left(\frac{k^{-r_2'}}{\varepsilon
_0}\right)^{N_{i}^{(2)}}\leq \cr k^{k^{3\gamma
r_1}}\left(\frac{k^{-r_2'}}{\varepsilon _0}\right)^{N_{i}^{(2)}},
\end{align}
\begin{align}\label{Pr-2n}
\left\|\Bigl(P_i\bigl(H(\k^{(3)}(\varphi
))-k^{2}I\bigr)P_i\Bigr)^{-1}\right\|_1<k^{2r_2'k^{2\gamma
r_1}}k^{r_2'+8\gamma r_2}\left(\frac{k^{-r_2'}}{\varepsilon
_0}\right)^{N_{i}^{(2)}}\leq \cr k^{k^{3\gamma
r_1}}\left(\frac{k^{-r_2'}}{\varepsilon _0}\right)^{N_{i}^{(2)}}.
\end{align}
\end{enumerate}
\end{lemma}
\begin{proof} The proof of this lemma is analogous to that of Lemma \ref{L:Pr} up to the replacement of ${\MM}^{(2)}$ by ${\MM}^{(3)}$, $\OO^{(2)}_{\m}$ by $\OO^{(3)}_{\m}$, and the shift of indices: $\delta $ to $r_1$, $r_1$ to
$r_2$, etc. We apply Lemmas \ref{L:blackind}--\ref{L:whiteind} instead of \ref{L:black}--\ref{L:white} and Lemmas \ref{L:Pr}, \ref{L:Ps} with $\varepsilon _0=k^{-r_2'}$ and $p_\m>k^{-r_2'k^{2\gamma r_1}}$ instead of Lemma
\ref{L:estnonres1}. We also note that $N_i^{(1)}<k^{2\gamma r_1}$ in \eqref{Pr-1}, \eqref{Pr-2}. \end{proof}

\begin{lemma}\label{L:Psn} Let $\varphi _0\in \omega^{(3)}(k,\delta ,\tau )$. Then, the operator
$\Bigl(P_s^j\bigl(H(\k^{(3)}(\varphi
))-k^{2}I\bigr)P_s^j\Bigr)^{-1}$ has no more than one pole in the
disk $|\varphi-\varphi _0|<k^{-k^{r_1}}$. Moreover,
\begin{equation}\label{Ps-1n}
\left\|\Bigl(P_s^j\bigl(H(\k^{(3)}(\varphi
))-k^{2}I\bigr)P_s^j\Bigr)^{-1}\right\|<\frac{8k^{-1}}
{p_\m\varepsilon _0},
\end{equation}
\begin{equation}\label{Ps-2n}
\left\|\Bigl(P_s^j\bigl(H(\k^{(3)}(\varphi
))-k^{2}I\bigr)P_s^j\Bigr)^{-1}\right\|_1<\frac{8k^{-1+
4r_2}}{p_\m\varepsilon_0},
\end{equation}
$\varepsilon _0=\min\{\varepsilon,\,k^{-r_2'}\}$, where
$\varepsilon$ is the distance to the pole of the operator.
\end{lemma}

Note that $p_{\m}>k^{-\mu r_3}$ when $\m \in \Omega (r_3)$. The
analogues of Lemma~\ref{L:boundary} and Corollary~\ref{C:PHP-2} also
hold.

\subsubsection{Resonant and Nonresonant Sets for Step IV \label{GSIV}}

We divide $[0,2\pi )$ into $[2\pi k^{k^{r_1}}]+1$ intervals
$\Delta_l^{(3)}$ with the length not bigger than $k^{-k^{r_1}}$. If a
particular interval belongs to $\OO^{(3)}$ we ignore it; otherwise,
let $\varphi_0^{(l)}\not\in\OO^{(3)}$ be a point inside the $\Delta_l^{(3)}$.
Let
\begin{equation}\W_l^{(3)}=\{\varphi \in \W^{(3)}:\ | \varphi -\varphi
_0^{(l)}|<4k^{-k^{r_1}}\}. \label{W2mind} \end{equation} Clearly,
neighboring sets $\W_l^{(3)}$  overlap (because of the multiplier 4
in the inequality), they cover $\hat \W^{(3)}$ , which is
the restriction of $\W^{(3)}$ to the $2k^{-k^{r_1}}$-neighborhood of
$[0,2\pi )$. For each $\varphi \in \hat \W^{(3)}$ there is an $l$
such that $|\varphi -\varphi _{0}^{(l)}|<4k^{-k^{r_1}}$. We consider
the poles of the resolvent  $\left(P^{(3)}
(H(\k^{(3)}(\varphi))-k^{2})P^{(3)}\right)^{-1}$ in $\hat
\W_l^{(3)}$ and denote them by $\varphi^{(3)} _{lm}$, $m=1,...,M_m$.
As before, the resolvent has a block structure. The number of blocks
clearly cannot exceed the number of elements in $\Omega (r_3)$, i.e.
$k^{4r_3}$. Using the estimates for the number of poles for each
block, the estimate being provided by Lemma \ref{L:Prn}, Part 1, we
can roughly estimate the number of poles of the resolvent by
$k^{4r_3+r_2}$. Next, let $r_3'>k^{r_1}$ and $\OO^{(4)}_{lm}$ be the
disc of the radius $k^{-r_3'}$ around $\varphi ^{(3)}_{lm}$.
\begin{definition} The set
\begin{equation}\OO^{(4)}=\cup _{lm}\OO^{(4)}_{lm} \label{O4}
\end{equation}
we call the fourth resonant set. The set
\begin{equation}\W^{(4)}= \hat\W^{(3)}\setminus \OO^{(4)}\label{W4}
\end{equation}
is called the fourth non-resonant set. The set
\begin{equation}\omega^{(4)}= \W^{(4)}\cap [0,2\pi) \label{w4}
\end{equation}
is called the fourth real non-resonant set. \end{definition} The
following statements can be proven in the same way as
Lemmas~\ref{L:geometric3}, \ref{4.16} and \ref{estnonres0-1}.
\begin{lemma}\label{L:geometric4n}Let  $r_3'>\mu r_3>k^{r_1}$, $\varphi \in \W^{(4)}$, $\varphi
_0^{(l)}$ corresponds  to an interval $\Delta _l^{(3)}$ containing $\Re
\varphi $. Let $\Pi $ be one of the components $\Pi _s^j(\varphi
_0^{(l)})$, $\Pi _b^j(\varphi _0^{(l)})$, $\Pi _g^j(\varphi _0^{(l)})$, $\Pi
_w^j(\varphi _0^{(l)})$ and
 $P(\Pi )$ be the projection corresponding to $\Pi $. Let also
 $\varkappa \in \C: |\varkappa-\varkappa^{(3)}(\varphi )|<k^{-r_3'k^{2\gamma
r_2}}$. Then,
\begin{equation} \label{March3-2n4} \left\|\left(P(\Pi )
\left(H\big(\k(\varphi )\big)-k^{2}I\right)P(\Pi
)\right)^{-1}\right\|<k^{2\mu r_3+r_3'N^{(2)}},\end{equation}
\begin{equation} \label{March3-3n4}
\left\|\left(P(\Pi )\left(H\big(\k(\varphi
)\big)-k^{2}I\right)P(\Pi )\right)^{-1}\right\|_1<k^{(2\mu+1)
r_3+r_3'N^{(2)}},\end{equation} $N^{(2)}$ corresponding to the color
of $\Pi $ ($N^{(2)}=1,\ k^{\gamma r_2+150\gamma r_1},\ k^{\gamma
r_2/2+\delta _0r_2},\ k^{\gamma r_2/6-\delta _0r_2}$ for simple,
black, grey and white clusters, correspondingly).
\end{lemma}
By total size of the set $\OO^{(4)}$ we mean the sum of the sizes of
its connected components.
\begin{lemma}\label{7.10} Let $r_3'\geq (\mu+10)r_3$, $r_3>k^{r_1}$. Then, the size of each connected component of
 $\OO^{(4)}$ is less
than $32k^{4r_3-r_3'}$. The total size of $\OO^{(4)}$ is less than
$k^{-r_3'/2}$.
\end{lemma}


\begin{lemma}\label{estnonres0-1n} Let $\varphi\in\W^{(3)}$ and
$C_4$ be the circle $|z-k^{2}|=k^{-2r_3'k^{2\gamma r_2}}$. Then
$$
\left\|\left(P(r_2)(H(\k^{(3)}(\varphi))-z)P(r_2)\right)^{-1}\right\|\leq
4^3k^{2r_3'k^{2\gamma r_2}}. $$ \end{lemma} We prove this lemma
using \eqref{tik-tak*}.

\section{STEP IV}

\subsection{Operator $H^{(4)}$. Perturbation Formulas}
Let $P(r_3)$ be an orthogonal projector onto $\Omega(r_3):=\{\m:\
|\|\p_\m\||\leq k^{r_3}\}$ and $H^{(4)}=P(r_3)HP(r_3) $. From now on,
we assume \begin{equation} \label{r_2IV}k^{r_1}<r_3<k^{\gamma
10^{-7}r_2},\ \ \ \ \ k^{2\gamma 10^{-4}r_2}<r_3'<k^{\delta_0
r_2/2}.\end{equation} We consider $H^{(4)}(\k^{(3)}(\varphi ))$ as a
perturbation of $\tilde H^{(3)}(\k^{(3)}(\varphi ))$:
$$\tilde H^{(3)}:=\tilde
P_l^{(3)}H\tilde P_l^{(3)}+\left(P(r_3)-\tilde
P_l^{(3)}\right)H_0,$$ where $H=H(\k^{(3)}(\varphi ))$,
$H_0=H_0(\k^{(3)}(\varphi ))$ and $\tilde P_l^{(3)}$ is the
projection $P^{(3)}$ corresponding to $\varphi _{0}^{(l)}$ in the
interval $\Delta _l^{(3)}$ containing $\varphi $, see \eqref{P(3)}.  Note that the operator
$\tilde H^{(3)}$ has a block structure, the block $\tilde
P_l^{(3)}H\tilde P_l^{(3)}$ being composed of smaller blocks
$P_iHP_i$, $i=0,...,5$.
By analogy with \eqref{W2*}--\eqref{G3},
\begin{equation}W^{(3)}=H^{(4)}-\tilde H^{(3)}=P(r_3)VP(r_3)-\tilde P_l^{(3)}V\tilde P_l^{(3)}, \label{W2IV}\end{equation}
\begin{equation}\label{g3IV} g^{(4)}_r({\k}):=\frac{(-1)^r}{2\pi
ir}\hbox{Tr}\oint_{C_4}\left(W^{(3)}(\tilde
H^{(3)}({\k})-zI)^{-1}\right)^rdz,
\end{equation} \begin{equation}\label{G3IV}
G^{(4)}_r({\k}):=\frac{(-1)^{r+1}}{2\pi i}\oint_{C_4}(\tilde
H^{(3)}({\k})-zI)^{-1}\left(W^{(3)}(\tilde
H^{(3)}({\k})-zI)^{-1}\right)^rdz,
\end{equation}
where $C_4$ is the circle $|z-k^{2}|=\varepsilon _0^{(4)}$,
$\varepsilon _0^{(4)}=k^{-2r_3'k^{2\gamma r_2}}.$

The proof of the following statements is analogous to the one in the
previous step (see Theorem~\ref{Thm3}, Corollary~\ref{corthm3} and
Lemma~\ref{L:derivatives-3}) up to the replacement of  $r_3$ by $r_4$, $r_2$ by $r_3$, $r_1$ by $r_2$, etc.

\begin{theorem} \label{Thm3IV} Suppose  $k>k_*$, $\varphi $ is in
the real  $k^{-r_3'-\delta }$-neighborhood of $\omega
^{(4)}(k,\delta,\tau )$ and $\varkappa\in\R$,
$|\varkappa-\varkappa^{(3)}(\varphi )|\leq \varepsilon
^{(4)}_0k^{-1-\delta }$, $\k=\varkappa(\cos \varphi ,\sin \varphi
)$. Then,  there exists a single eigenvalue of $H^{(4)}({\k})$ in
the interval $\varepsilon _4( k,\delta,\tau )=\left(
k^{2}-\varepsilon _0^{(4)}, k^{2}+\varepsilon _0^{(4)}\right)$. It
is given by the absolutely converging series:
\begin{equation}\label{eigenvalue-3IV}\lambda^{(4)}({\k})=\lambda^{(3)}({\k})+
\sum\limits_{r=2}^\infty g^{(4)}_r({\k}).\end{equation} For
coefficients $g^{(4)}_r({\k})$ the following estimates hold:
\begin{equation}\label{estg3IV} |g^{(4)}_r({\k})|<k^{-\frac{\beta}{5}
k^{r_2-r_1 }-\beta (r-1)}.
\end{equation}
The corresponding spectral projection is given by the series:
\begin{equation}\label{sprojector-3IV}
\E ^{(4)}({\k})=\E^{(3)}({\k})+\sum\limits_{r=1}^\infty
G^{(4)}_r({\k}), \end{equation} $\E^{(3)}({\k})$ being the spectral
projection of $H^{(3)}(\k)$. The operators $G^{(4)}_r({\k})$ satisfy
the estimates:
\begin{equation}
\label{Feb1a-3IV}
\left\|G^{(4)}_r({\k})\right\|_1<k^{-\frac{\beta}{10} k^{r_2-r_1 }
-\beta r},
\end{equation}
\begin{equation}G^{(4)}_r({\k})_{\s\s'}=0,\ \mbox{when}\ \ \
2rk^{\gamma r_2+150\gamma
r_1}+3k^{r_2}<\||\p_\s\||+\||\p_{\s'}\||.\label{Feb6a-3IV}
\end{equation}
\end{theorem}

\begin{corollary}\label{corthm3IV} For the perturbed eigenvalue and its spectral
projection the following estimates hold:
 \begin{equation}\label{perturbation-3IV}
\lambda^{(4)}({\k})=\lambda^{(3)}({\k})+ O_2\left(k^{-\frac15 \beta
k^{r_2-r_1 }}\right),
\end{equation}
\begin{equation}\label{perturbation*-3IV}
\left\|\E^{(4)}({\k})-\E^{(3)}({\k})\right\|_1<k^{-\frac{\beta}{10}
k^{r_2-r_1 }}.
\end{equation}
\begin{equation}
\left|\E^{(4)}({\k})_{\s\s'}\right|<k^{-d^{(4)}(\s,\s')},\ \
\mbox{when}\ \||\p_\s\||>4k^{r_2} \mbox{\ or }
\||\p_{\s'}\||>4k^{r_2 },\label{Feb6b-3IV}
\end{equation}
$$d^{(4)}(\s,\s')=\frac{1}{16}(\||\p_\s\||+\||\p_{\s'}\||)k^{-\gamma r_2-150\gamma r_1}\beta +\frac{1}{10}\beta
k^{r_2-r_1 }.$$
\end{corollary}

\begin{lemma} \label{L:derivatives-3IV}Under conditions of Theorem \ref{Thm3IV} the following
estimates hold when $\varphi \in \omega ^{(4)}(k,\delta )$ or its
complex $k^{-r_3'-\delta}$-neighborhood and $\varkappa\in \C:$
$|\varkappa-\varkappa^{(3)}(\varphi )|<\varepsilon
^{(4)}_0k^{-1-\delta}$.
\begin{equation}\label{perturbation-3cIV}
\lambda^{(4)}({\k})=\lambda^{(3)}({\k})+O_2\left(k^{-\frac 15 \beta
k^{r_2-r_1} }\right),
\end{equation}
\begin{equation}\label{estgder1-3kIV}
\frac{\partial\lambda^{(4)}}{\partial\varkappa}=\frac{\partial\lambda^{(3)}}{\partial\varkappa}
+O_2\left(k^{-\frac 15 \beta k^{r_2-r_1} }M_2\right), \  \ \
\ M_2:=\frac{k^{1+\delta}}{\varepsilon ^{(4)}_0},\end{equation}
\begin{equation}\label{estgder1-3phiIV}\frac{\partial\lambda^{(4)}}{\partial \varphi }=\frac{\partial\lambda^{(3)}}{\partial \varphi }+
O_2\left(k^{-\frac 15 \beta k^{r_2-r_1}+r_3'+\delta }\right),
 \end{equation}
\begin{equation}\label{estgder2-3IV}
\frac{\partial^2\lambda^{(4)}}{\partial\varkappa^2}=
\frac{\partial^2\lambda^{(3)}}{\partial\varkappa^2}+
O_2\left(k^{-\frac 15 \beta k^{r_2-r_1} }M_2^2\right),
\end{equation}
\begin{equation} \label{gulf2-3IV}
\frac{\partial^2\lambda^{(4)}}{\partial\varkappa\partial \varphi
}=\frac{\partial^2\lambda^{(3)}}{\partial\varkappa\partial \varphi
}+ O_2\left(k^{-\frac 15 \beta k^{r_2-r_1}+r_3'+\delta
}M_2\right),
\end{equation}
\begin{equation} \label{gulf3-3IV}
\frac{\partial^2\lambda^{(4)}}{\partial\varphi
^2}=\frac{\partial^2\lambda^{(3)}}{\partial\varphi
^2}+O_2\left(k^{-\frac 15 \beta k^{r_2-r_1}+2r_3'+2\delta
}\right).
\end{equation}\end{lemma}

\begin{corollary} \label{"O"*} All ``$O_2$"-s on the right hand sides of \eqref{perturbation-3cIV}-\eqref{gulf3-3IV} can be written as $O_1\left(k^{-\frac {1}{10} \beta k^{r_2-r_1}}\right)$.
\end{corollary}

\begin{remark}
In the proof of Theorem~\ref{Thm3IV} and similar statements in every
further step of the induction we obtain the estimate of the form
\eqref{March5}. It is important to notice that the right hand side
of these estimates is always $k^{-2\beta}$. It can't become
better since it comes, in particular, from the estimate of the free resolvent on the
set of points satisfying $\left||\k+\p_\m|^{2}_\R-k^{2}\right|\geq
k^{\delta_*}$. What changes is the first term in
the perturbation series, see  e.g. \eqref{estg3}, \eqref{Feb1a-3} vs \eqref{estg3IV}, \eqref{Feb1a-3IV}.
\end{remark}

\subsection{\label{IS3IV}Isoenergetic Surface for Operator $H^{(4)}$}

The following statement is an analogue of Lemma~\ref{ldk-3}.

\begin{lemma}\label{ldk-3IV} \begin{enumerate}
\item For every $\lambda :=k^{2}$,  $k>k_*$, and $\varphi $ in the real  $\frac{1}{2} k^{-r_3'-\delta }$-neighborhood
of $\omega^{(4)}(k,\delta, \tau )$,$\ $ there is a unique
$\varkappa^{(4)}(\lambda, \varphi )$ in the interval
$I_3:=[\varkappa^{(3)}(\lambda, \varphi )-\varepsilon
^{(4)}_0k^{-1-\delta},\varkappa^{(3)}(\lambda, \varphi
)+\varepsilon ^{(4)}_0k^{-1-\delta}]$, such that
    \begin{equation}\label{2.70-3IV}
    \lambda^{(4)} \left(\k
^{(4)}(\lambda ,\varphi )\right)=\lambda ,\ \ \k ^{(4)}(\lambda
,\varphi ):=\varkappa^{(4)}(\lambda ,\varphi )\vec \nu(\varphi).
    \end{equation}
\item  Furthermore, there exists an analytic in $ \varphi $ continuation  of
$\varkappa^{(4)}(\lambda ,\varphi )$ to the complex  $\frac{1}{2}
k^{-r_3'-\delta }$-neighborhood of $\omega^{(4)}(k,\delta, \tau )$
such that $\lambda^{(4)} (\k ^{(4)}(\lambda, \varphi ))=\lambda $.
Function $\varkappa^{(4)}(\lambda, \varphi )$ can be represented as
$\varkappa^{(4)}(\lambda, \varphi )=\varkappa^{(3)}(\lambda, \varphi
)+h^{(4)}(\lambda, \varphi )$, where
\begin{equation}\label{dk0-3IV} |h^{(4)}(\varphi )|=O_1\left(k^{-\frac 15 \beta k^{r_2-r_1}-1
}\right),
\end{equation}
\begin{equation}\label{dk-3IV}
\frac{\partial{h}^{(4)}}{\partial\varphi}= O_2\left(k^{-\frac 15 \beta
k^{r_2-r_1}-1 +r_3'+\delta }\right),\ \ \ \ \
\frac{\partial^2{h}^{(4)}}{\partial\varphi^2}= O_4\left(k^{-\frac 15
\beta k^{r_2-r_1}-1 +2r_3'+2\delta }\right).
\end{equation} \end{enumerate}\end{lemma}

Let us consider the set of points in $\R^2$ given by the formula:
$\k=\k^{(4)} (\varphi), \ \ \varphi \in \omega ^{(4)}(k,\delta, \tau
)$. By Lemma \ref{ldk-3IV} this set of points is a slight distortion
of ${\cal D}_{3}$. All the points of this curve satisfy the equation
$\lambda^{(4)}(\k ^{(4)}(\varphi ))=k^{2}$. We call it isoenergetic
surface of the operator $H^{(4)}$ and denote by ${\cal D}_{4}$.

\section{Induction}

\subsection{Inductive formulas for $r_n$} Now, we are ready to introduce the induction. In fact, STEP IV has
been the first inductive step. Here, for the sake of convenience, we
reformulate the main statements from the previous step in terms of
$r_n$, $n\geq3$, and provide necessary comments. First, we choose
\begin{equation}\label{indrn}
k^{r_{n-2}}<r_n<k^{\gamma10^{-7}r_{n-1}},\ \ \
k^{2\gamma10^{-4}r_{n-1}}<r_n'<k^{\delta_0r_{n-1}/2},\ \ \ n\geq3.
\end{equation}
\subsection{Preparation for Step $n+1$, $n\geq 4$}
\subsubsection{Properties of the Quasiperiodic Lattice. Induction}\label{Lattice-Induction}
Here we prove the inductive version of the results from
Section~\ref{Quasiperiodicgeomcont}. We consider
$\p_\m=2\pi(\s_1+\alpha\s_2)$ with integer vectors $\s_j$ such that
$|\s_j|\leq 4k^{r_{n-1}}$. We repeat the arguments from the
beginning of Section~\ref{geomIII}. Namely, let $(q,p)\in\Z^2$ be a
pair such that $0<q\leq 4k^{r_{n-1}}$ and
\begin{equation}\label{qind-last}
|\alpha q+p|\leq \frac14 k^{-r_{n-1}}.
\end{equation}
We choose a pair $(p,q)$ which gives the best approximation. In
particular, $p$ and $q$ are mutually prime. Put
$\epsilon_q:=\alpha+\frac{p}{q}$. We have
\begin{equation}k^{-2r_{n-1}\mu}\leq|\epsilon_q|\leq \frac14 q^{-1}k^{-r_{n-1}}.\label{epsilon_qind-last}\end{equation}
The analogs of Lemmas \ref{Lattice-1ind}--\ref{Lattice-3ind} hold
with $n-1$ instead of $2$.

We consider the matrix $H^{(n-1)}(\k)=P(\gamma r_{n-2})H(\k )(\gamma
r_{n-2})$ where $\k \in \R^2$, $P(\gamma r_{n-2})$ is the orthogonal
projection corresponding to $\Omega (\gamma r_{n-2})$ (it
is a slight abuse of notations, since $H^{(n-1)}$ in Step $n-1$ was
defined for $\gamma =1$). We construct the block structure in
$H^{(n-1)}(\k)$ analogous to that in Step $n-1$.  The difference is
that now we consider any $\k \in \R^2$, not only $\k$ being close to
$\k^{(n-2)}(\varphi )$. Correspondingly, we define non-resonant $\m$
not in terms of $\varphi $, but in more general terms of
inequalities providing convergence of perturbation series. Indeed,
we call $\m \in \Omega (\gamma r_{n-2})$ non-resonant if  (cf.
\eqref{resonance1})
\begin{equation}
\left||\k+\p_{\m }|^{2}-k^{2}\right|>k^{\delta_*}.
\label{Aug29a-last}
\end{equation}
Obviously, this estimate is stable in the $k^{-\delta_*
-1-\delta}$-neighborhood of a given $\k$. Hence, the definition of a
non-resonant $\m$ is stable in this neighborhood up to a multiplier
$1+o(1)$ in the r.h.s. of \eqref{Aug29a-last}.  Around each resonant
$\m$ (which is also not trivial weakly resonant in the sense of Step II) we construct $k^\delta$-boxes/clusters (see \eqref{defP}). Let
$P({\m})$ be the projection on the $k^{\delta }$-cluster containing
$\m $. If
\begin{equation}
\left\|(P({\m})(H(\k)-k^{2})P({\m}))^{-1}\right\|<k^{4\gamma r_1'}
\label{Aug29b-last} \end{equation} (cf. Definition \ref{resdef}), then we
call the $k^{\delta }$-cluster effectively non-resonant  for a given
$\k$.
 Note, that the above
estimate and, therefore, the definition of an effectively
non-resonant $k^{\delta}$-cluster is  stable in the $k^{-4\gamma
r_1'-1-\delta }$-neighborhood of a given $\k$. The
$k^{\delta}$-clusters, where \eqref{Aug29b-last} is not valid, are
called  effectively resonant $k^{\delta}$-clusters.  Around each
effectively resonant $k^{\delta }$-cluster, we construct $k^{\gamma
r_1}$-clusters. We sort these clusters into four types: simple,
white, grey and black clusters as in Section~\ref{MOforStep3}, using the
term ``$\m $ is effectively resonant" instead of ``$\m \in \MM
^{(2)}$". There is no need to consider a special case of simple
clusters here. Note that Lemmas \ref{L:black} -- \ref{L:white} are
valid for an arbitrary $\k$, since they are based on Lemmas
\ref{L:number of points-1}, \ref{4.10}, \ref{t4.10}, \ref{2t4.10}, \ref{nontriv4.10} proven for an
arbitrary $\k$. Be analogy with \eqref{March3-2}, a $k^{\gamma
r_1}$-cluster is called effectively non-resonant if
\begin{equation}
\left\|(P({\m})(H(\k)-k^{2})P({\m}))^{-1}\right\|<k^{2\mu
r_2+r_2'N^{(1)}_i}, \label{Aug29c-last} \end{equation} where
$N^{(1)}_i$ corresponds to the color of a $k^{\gamma r_1}$-cluster,
$N^{(1)}_i=k^{\gamma r_1+3}, k^{\gamma r_1/2+\delta _0r_1}$ or
$k^{\gamma r_1/6-\delta _0r_1}$. If $n=4$ we stop here. If $n>4$, we
surround effectively resonant $k^{\gamma r_1}$-clusters by blocks of
the next size, etc.  The analogues of Lemmas \ref{L:black} --
\ref{L:white} are valid, see Lemmas \ref{L:blackind} --
\ref{L:whiteind}, \ref{L:blackindlast} -- \ref{L:whiteindlast}.
Eventually, the $k^{\gamma r_{n-3}}$-cluster  is effectively
non-resonant if
\begin{equation}
\left\|(P({\m})(H(\k)-k^{2})P({\m}))^{-1}\right\|<k^{2\mu
r_{n-2}+r_{n-2}'N^{(n-3)}_i}, \label{Aug29d-last} \end{equation}
where $N^{(n-3)}_i$ is $N^{(n-3)}_i= k^{\gamma r_{n-3}+150\gamma
r_{n-4} }, k^{\gamma r_{n-3}/2+\delta _0r_{n-3} }, k^{\gamma
r_{n-3}/6-\delta _0r_{n-3} }$, depending on the color of the cluster
(cf. \eqref{Aug29c-last}, \eqref{March3-2n4}). Further we put
$150\gamma r_{0}=3$. This will make \eqref{Aug29c-last} to be a
special case of \eqref{Aug29d-last} ($n=4$). Thus, we have
constructed a block structure in $H^{(n-1)}(\k)$, which is stable in
the $k^{-\rho _{n-2}}$-neighborhood of a given $\k$, where $\rho
_1=4\gamma r_1'+1+\delta $ and $$\rho _{n-2}=\mu
r_{n-2}+r_{n-2}'k^{\gamma r_{n-3}+150\gamma r_{n-4} }+1+\delta,\
\mbox{when } n\geq 4.$$ It is not difficult to see that $\rho
_{n-2}<r_{n-1}$.
\begin{definition} \label{D-J-last}We denote   by
$J(\k)$ the number of the effectively resonant $k^{\gamma
r_{n-3}}$-clusters  in $H^{(n-1)}(\k)$ for a given $\k$.
 Further (with a slight abuse of notations) we
consider $J(\k)$ to be constant in the $k^{-\rho
_{n-2}}$-neighborhood of a given $\k$. \end{definition}
Let  $\k=\a \tau_1+\b,\  \ \ |\a|=1,\  |\b|<4k^{\gamma r_{n-2}}$. We
consider $H^{(n-1)}(\k)$ as a function of $\tau _1$ in the complex $
k^{-\rho _{n-2}}$-neighbothood of zero.
\begin{lemma} \label{May24-2}The resolvent $( H^{(n-1)}(\k)-k^{2})^{-1}$ has no more than $k^{2\gamma r_{n-3}}J(\b)$ poles $\tau _{1j}$ in the the complex $ 2k^{-\rho _{n-2}}$-neighborhood of zero. It satisfies the following estimate in the the complex $k^{-\rho _{n-2}}$-neighborhood of zero:
\begin{equation}\label{Sept3a-last}
\|(H^{(n-1)}(\k)-k^{2})^{-1}\|<k^{\rho _{n-2}k^{2\gamma
r_{n-3}}}\left(\frac{k^{ -\rho _{n-2}}}{\varepsilon
_0}\right)^{J(\k)k^{2\gamma r_{n-3}}},
\end{equation}
where $\varepsilon _0=\min \{k^{-2\rho _{n-2}},\varepsilon\}$,
$\varepsilon $ being the distance to the nearest pole $\tau _{1,j}$.
\end{lemma}
\begin{proof} The lemma is proved by induction. For $n=3$, see Lemma \ref{May24-1}. Let us consider the case $n\geq 4$.  Recall (Definition \ref{D-J-last}) that
$J(\k)$ may be considered to be constant in $2k^{-\rho
_{n-2}}$-neighborhood  of $\tau _1=0$. Hence, $J(\k)=J(\b)$ for such
$\k$-s.

Let us consider the collection of all  $k^{\delta }$,...,$k^{\gamma
r_{n-3}}$-clusters $P({\m })$ for $H^{(n-1)}(\k)$. Note that the
collection is the same for all such $\k$. We construct the
corresponding block operator $\tilde H^{(n-2)}(\k)$:
$$\tilde H^{(n-2)}(\k)=\sum P({\m})HP({\m})+H_0(I-\sum P ({\m })).$$  If a $k^{\gamma r_{n-3} }$-cluster $P({\m})H(\k )P({\m})$ is effectively non-resonant, then  its resolvent, obviously,  has no poles $\tau _1$ in the $2k^{- \rho _{n-2} }$-neighborhood of $\tau _{1}=0$. The resolvent of each effectively resonant $k^{\gamma r_{n-3} }$-cluster $P({\m})H(\k )P({\m})$ has no more than $N_i^{(n-3)}k^{2\gamma r_{n-4}}$  ($k^{2\gamma r_{0}}$ is taken to be equal to $12$ for $n=4$) poles $\tau _{1j}$   in the
$k^{-\rho _{n-3}}$-neighborhood of $\tau _{1}=0$. It follows from
this lemma for the previous step and also
Lemmas~\ref{L:black}--\ref{L:white},
\ref{L:blackind}--\ref{L:whiteind} and
\ref{L:blackindlast}-\ref{L:whiteindlast} below for previous steps, which
give the estimates for $J(\k)$ in the previous steps, based on the
color of clusters. Let us consider the union of $k^{-2\rho _{n-2}}$
neighborhoods of these poles and denote it by ${\D}_{\m}$. By this
lemma for $n-1$, instead of $n$, each $k^{\gamma r_{n-3}}$ cluster
satisfies the estimate
$$
\|(P({\m}) (H^{(n-1)}(\k)-k^{2})P({\m}))^{-1}\|<k^{\rho
_{n-3}k^{2\gamma r_{n-4}}}k^{2\rho _{n-2}N_i^{(n-3)}k^{2\gamma
r_{n-4}}}
$$
outside $\D_{\m}$, $N_i^{(n-3)}$ corresponding to the color of the
cluster. Note that $\max _i N_i^{(n-3)}=N_1^{(n-3)}<k^{\gamma
r_{n-3} +150\gamma r_{n-4}}$.
Therefore, the resolvent $\left(\tilde
H^{(n-2)}(\k)-k^{2}\right)^{-1}$ has no more than
$J(\k)N_1^{(n-3)}k^{2\gamma r_{n-4}}$ poles $\tau _{1j}$ in the
complex $k^{-\rho _{n-3}}$-neighborhood of $\tau _{1}=0$. Let $\D
=\cup _{\m}\D_{\m}$, the union being taken over all $\m$
corresponding to all  resonant clusters. The number of $\m$-s in the
union, obviously, does not exceed $k^{4\gamma r _{n-2}}$, which is
the number of different $\m$ in $H^{(n-1)}(\k)$. Therefore, the size
of each connected component of $\D$ is less than $k^{-2\rho
_{n-2}}k^{4\gamma r _{n-2}}=o\left(k^{-\rho _{n-2}}\right)$. We are
interested only in those components of $\D$, which are completely in
the disk of the radius $2k^{-\rho _{n-2}}$ around $\tau _1=0$.
 Considering as before (see the proof of Theorem \ref{Thm3} with  $r_{n-2}$ instead of $r_2$, $r_{n-3}$ instead of $r_1$ and $k^{\gamma r_{n-3} +150\gamma r_{n-4}}$ instead of $k^{\gamma r_1+3}$, when one considers black clusters),
we can show that the perturbation series for the resolvent
$(H^{(n-1)}(\k)-k^{2})^{-1}$ with respect to $(\tilde
H^{(n-2)}(\k)-k^{2})^{-1}$ converges on the boundary of $\D$. The
resolvents have the same number of poles inside each component of
$\D$. Hence, $(H^{(n-1)}(\k)-k^{2})^{-1}$ has no more than
$J(\k)N_1^{(n-3)}k^{2\gamma r_{n-4}}$ poles in $\D$. It is easy to
see that $J(\k)N_1^{(n-3)}k^{2\gamma r_{n-4}}<J(\k)k^{2\gamma
r_{n-3}}$. The resolvent satisfies the following estimate outside
$\D$:
$$
\|(H^{(n-1)}(\k)-k^{2})^{-1}\|<k^{\rho _{n-3}k^{2\gamma
r_{n-4}}}k^{2\rho _{n-2}N_1^{(n-3)}k^{2\gamma r_{n-4}}}<k^{\rho
_{n-2}k^{2\gamma r_{n-3}}}.
$$
 Using the maximum principle we obtain \eqref{Sept3a-last}.
 \end{proof}

 Next, we introduce
\begin{equation}\label{triind-last}
\SS ^{(n-1)}(k,\xi):=\{\k\in \R^{2}:\
\|(H^{(n-1)}(\k)-k^{2})^{-1}\|>k^{\xi}\}.
\end{equation}
It is easy to see that each connected component of $\SS^{(n-1)}(k,\xi)$ is
bounded  by the curves $D(\k, k^{2}\pm k^{-\xi})=0$, where $ D(\k,
\lambda)=\hbox{det}\,(H^{(n-1)}(\k)-\lambda ).$
\begin{lemma}\label{L:curves-2ind-last}  Let $\l$ be a segment of a straight line in $\R^{2}$,
\begin{equation}\l:=\left\{\k=\a \tau_1+\b,\ \tau _1\in
(0,\eta)\}, \ \ |\a|=1,\  |\b|<4k^{\gamma r_{n-2}},\  \ 0<\eta <
k^{-\rho _{n-2}}\right\}. \label{segment-last}\end{equation} Suppose
both ends of $\l$ belong to a connected component of $\SS ^{(n-1)}(k,\xi)$.
If $\xi $ is sufficiently large, namely, $\xi\geq 4k^{2\gamma
r_{n-3}}J(\b)\log _k\frac{1}{\eta}$, then, there is an inner part
$\l'$ of the segment,
 which is not in $\SS^{(n-1)}(k,\xi)$.
 \end{lemma}
 \begin{corollary} \label{C:curves-2ind-last} Let $\k\in \SS^{(n-1)}(k,\xi)$ and
 ${\xi }>8k^{2\gamma r_{n-3}}J(\k )\rho _{n-2}$. Then the distance from $\k$ to the boundary of
 $\SS^{(n-1)}(k,\xi)$ is less than $k^{-\tilde \xi}$, $\tilde \xi =\xi \frac14 k^{-2\gamma r_{n-3} }J(\k )^{-1}$. \end{corollary}
 {\em Proof of the corollary.}  Let us consider a segment of the length $\eta =k^{-\tilde \xi }$ starting at $\k $.
 By the statement of the lemma it intersects a boundary $D(\k, k^{2}\pm k^{-\xi})=0$.

 \begin{proof}
 Choose $\varepsilon =\eta ^2$. Using the hypothesis of the lemma, we
obtain that the right-hand side of \eqref{Sept3a-last} is less than
$k^{{\xi }}$ outside the discs of radius $\varepsilon $ around the poles of the resolvent.  Let us estimate the total size (sum
of the sizes) of the discs. Indeed, the size of each disc is $2\eta
^2$ and the number of discs is, obviously, less  $16k^{4\gamma
r_{n-2}}$. Therefore, the total size admits the estimate from above:
$32\eta ^2 k^{4\gamma r_{n-2}}=o(\eta)$, since $\eta <k^{-\rho
_{n-2}}$. This means there is a part $\l'$ of $\l$ outside these
discs. By \eqref{Sept3a-last}, this part is
 outside $\SS ^{(n-1)}(k,\xi)$, when $\xi $ is as described in the statement of the lemma.
\end{proof}

Let $\k _0 \in \R^2$ be fixed and ${\cal N}(k,r_{n-1},\k _0,J_0)$ be the
following subset of the lattice $\k _0+\p_{\n}$, $\n \in \Omega
(r_{n-1})$:
$${\cal N}(k,r_{n-1},\k _0,J_0)=\left\{\k _0+\p_{\n}:\n \in \Omega (r_{n-1}):\
J(\k _0 +\p_{\n})\leq J_0\right\},$$ $J$ being defined by Definition
\ref{D-J-last}. Thus, ${\cal N}$ includes only such $\n$ that the
surrounding $k^{\gamma r_{n-2}}$- block contains less than $J_0 $ of
effectively resonant $k^{\gamma r_{n-3}}$-clusters. Let $N(k,r_{n-1},\k _0,J_0, \xi)$ be the
number of points $\k _0+\p_{\n}$ in $\SS ^{(n-1)}(k,\xi )\cap {\cal
N}(k,r_{n-1},\k _0,J_0)$.
\begin{lemma}\label{norm2/3-last}
If   $\xi>4\mu r_{n-1}J_0k^{2\gamma r_{n-3}}$, then
\begin{equation}\label{eqnorm2/3-last}
N(k,r_{n-1},\k _0,J_0, \xi)\leq k^{\frac{2}{3}r_{n-1}+43\gamma
r_{n-2}}.
\end{equation}
\end{lemma}
\begin{proof} The proof of the lemma is completely analogous to that of \ref{norm2/3} up to replacement of 2 by $n-1$. Instead of  Corollary \ref{C:curves-2ind} we use Corollary \ref{C:curves-2ind-last} and the inequality $\rho _{n-2}<r_{n-1}$.\end{proof}

\subsubsection{Model Operator for Step $n+1$}
We make for $r_n$ the construction analogous to those from 
subsections~\ref{MOforStep3}, \ref{MOforStep4}.
We start with introducing a new notation by analogy with \eqref{se} and \eqref{Omega-s}:
\begin{equation} \Omega _s^{(j)}(r_n)=\{\m \in \Omega (r_n),\ 0<p_{\m}<k^{-r_{j-1}'k^{2\gamma r_{j-2}}}\}, \ \ j\geq 2,\label{Omega-j} \end{equation}
where $k^{2\gamma r_{j-2}}$ is  taken to be just $5$ when $j=2$. Note that $ \Omega _s^{(j+1)}\subset  \Omega _s^{(j)}$ and  $\Omega _s^{(j)}=\emptyset $ when $j>n$.
Next, let $\m \in \Omega (r_n)$. We denote the $k^{\gamma r_{n-2}}$-component
containing $\m$ by $\tilde \Pi (\m)$ and the corresponding projector
by $\tilde P(\m)$ (we assume they are defined by the previous steps of the procedure). For $\m$ belonging to the same $k^{\gamma
r_{n-2}}$-component, $\tilde \Pi (\m)$ and  $\tilde P(\m)$ are the same.
We define ${\MM}^{(n)}$ by a recurrent formula, which starts with ${\MM}^{(3)}$, see \eqref{M^3}:
\begin{equation}\label{M^n} {\MM}^{(n)}:={\MM}^{(n)}(\varphi _0, r_n)=\{\m\in \MM^{(n-1)}(\varphi _0, r_n)\cup \Omega _s^{(n-1)}(r_n)\setminus\MM_2^{weak}(\varphi _0, r_n)
:\ \varphi_0\in{\cal O}_\m^{(n)}(r_{n-1}',1)\},\end{equation} where
${\cal O}_\m^{(n)}(r_{n-1}',\tau)$ is the union of the disks of the
radius $\tau k^{-r_{n-1}'}$ with the centers at poles of the resolvent
$(\tilde P(\m)(H(\k^{({n-1})}(\varphi ))-k^{2}I)\tilde P(\m))^{-1}$ in the $k^{-44r_{n-2}'-2-\delta }$-neighborhood of $\varphi _0$.
For $\m$ belonging to the same $k^{\gamma
r_{n-2}}$-component, the sets ${\cal O}_\m^{(n)}(r_{n-1}',\tau)$ are
identical. We say that $\m \in {\MM}^{(n)}$ is $k^{\gamma
r_{n-2}}$-resonant. The corresponding $k^{\gamma r_{n-2}}$-clusters we call
resonant too.

 Let $\varphi_0\in \omega^{(n)} (k,
\delta ,1)$. By
construction of the non-resonant set $\omega^{(n)} (k, \delta ,1)$,
we have ${\MM}^{(n)}\cap \Omega (r_{n-1})=\emptyset $.

Further we use the property of the set $\MM^{(n)}$ formulated in the
next lemma which is an analogue of the Lemmas~\ref{L:2/3-1}, \ref{L:2/3-1ind}.

\begin{lemma}\label{L:2/3-1ind-last} Let $r_{n-1}'>2k^{(\gamma +\delta_0)10^{-4}r_{n-2}-2\delta}$. Let  $1/20<\gamma '<20$, $\m _0\in\Omega
(r_n)$ and $\Pi _{\m_0}$ be the $k^{\gamma
'r_{n-1}}$-neighborhood (in $\||\cdot\||$-norm) of $\m_0$. Then the set
$\Pi _{\m _0}$ contains less than $k^{\frac 23 \gamma'r_{n-1} +50\gamma r_{n-2}}$
elements of $\MM^{(n)}$.
\end{lemma}

\begin{proof} The proof is similar to that of Lemma \ref{L:2/3-1ind} up to the replacement of $3$ by $n$. First, we
notice that the condition $r_{n-1}'>2k^{(\gamma +\delta_0)10^{-4}r_{n-2}-2\delta}$ is consistent with the restriction
\eqref{indrn}.  Instead of Lemma \ref{L:geometric2} we use Lemma~\ref{L:geometric3} ($n=4$), Lemma~\ref{L:geometric4n} ($n=5$) and Lemma~\ref{L:geometric4nlast} with  $n-2$ instead of $n$ when $n>5$. We also use Lemma \ref{norm2/3-last} instead of Lemma \ref{norm2/3}. We use \eqref{indrn} to show that the hypothesis
of Lemma \ref{norm2/3-last} holds. In particular, we use the inequality $r_{n-1}'>>4\mu  r_{n-1}k ^{2\gamma r_{n-3}}$, following from \eqref{indrn}. \end{proof}

We continue with constructing $k^{\gamma r_{n-2}}$-clusters in
$\Omega (r_n)$, $r_n>r_{n-1}$, the same way we did it for $\Omega
(r_2)$ in Section \ref{MOforStep3}. We call a $k^{\gamma
r_{n-2}}$-cluster resonant if it contains $\m \in {\MM}^{(n)}$, see
\eqref{M^3}, \eqref{M^n}. Next, we repeat the construction after
Lemma \ref{L:2/3-1} up to the replacement of $r_1$ by $r_{n-1}$ and
$\delta $ be $\gamma r_{n-2}$. Indeed, let us split $\Omega
(r_n)\setminus \Omega (r_{n-1})$ into $k^{\gamma r_{n-1}}$-boxes,
$\gamma =\frac{1}{5}$.

First, let's consider  $\m \in \Omega _s^{(n)}(r_n)$. As before (see
``Simple region", page \pageref{simple}) one can prove that $\Omega
_s^{(n)}(r_n)\subset \MM (r_n)$; there are no other elements of $\MM
(r_n)$ in the $k^{\delta}$-box around $\m$; $\m$ itself can belong
or do not belong to $\MM^{(j)}(r_n)$, but there are
 no other elements of
$\MM^{(j)}(r_n)$ in the $k^{r_{j-1}}$-box around such $\m$,
$j=2,\dots,n$; and there are no other elements of $\Omega
_s^{(n)}(r_n)$ in the $k^{r_{n-1}}$-box around $\m$.

For each $\m \in \Omega _s^{(n)}(r_n)$ we consider its $k^{
r_{n-1}/2}$-neighborhood in $\||\cdot \||$ norm. The union of such
boxes we call the simple region and denote $\Pi _s$. The
corresponding projection is $P_s$.

Now, consider all other boxes (all elements $\p_\m$ there satisfy
$p_\m>k^{- r_{n-1}'k^{2\gamma r_{n-2}}}$). We call a box black if it
together with its neighbors contains more than $k^{\gamma
r_{n-1}/2+\delta _0r_{n-1}}$ elements of $\MM ^{(n)}$,
$\delta_0=\gamma /100$. Let us consider "black" boxes together with
their $k^{\gamma r_{n-1}+\delta _0r_{n-1}}$-neighborhoods and call
this the black region. We denote the black region by $\Pi _b$. The
corresponding projector is $P_{b}$. By white boxes we mean
$k^{\gamma r_{n-1}}$-boxes which together with its neighbors contain
no more than $k^{\gamma r_{n-1}/2+\delta_0r_{n-1}}$ elements of $\MM
^{(n)}$. Every white box we split into "small" boxes of the size
$k^{\gamma r_{n-1}/2+2\delta_0r_{n-1}}$. We call a small box "grey"
if it together with its neighbors contains more than $k^{\gamma
r_{n-1}/6-\delta_0r_{n-1}}$ elements of $\MM ^{(n)}$. Grey small
boxes together with its $k^{\gamma
r_{n-1}/2+2\delta_0r_{n-1}}$-neighborhoods we call the grey region.
The notation for this region is $\Pi _g$. The corresponding
projector is $P_{g}$. The part of the grey region which is outside
the black region, we denote $\Pi _g'$ and the corresponding
projection by $P_g'$. By a white small box we call a small box which
has no more than $k^{\gamma r_{n-1}/6-\delta_0r_{n-1}}$ elements of
$\MM ^{(n)}$. In each small white box we consider $k^{\gamma
r_{n-1}/6}$-boxes around each point of $\MM^{(n)}$. The union of
such $k^{\gamma r_{n-1}/6}$-boxes we call the white region and
denote $\Pi _w$. The corresponding projection is $P_w$. The part of
the white region which is outside the black and grey regions, we
denote $\Pi _w'$ and the corresponding projection by $P_w'$.

We put as before
$$P_r^{(n)}:=P_s^{(n)}+P_b^{(n)}+P_g^{(n)'}+P_w^{(n)'}.$$
The construction of the non-resonant region is the inductive extension of that for Step IV, see Section \ref{S:4}, page \pageref{NRR}. Indeed, we start with construction of $k^{\delta }$ clusters in $\Omega (r_n)$. Those of them, which are resonant, we extend to $k^{\gamma r_1}$ clusters, those of them, which are resonant we extend to  $k^{\gamma r_2}$ clusters, and so on until we reach the size $k^{\gamma r_{n-2}}$. On each step we construct a colored structure (simple, black, grey, white). If $k^{\gamma r_j}$-cluster happens to intersect $k^{\gamma r_{j+1}}$-cluster, we consider it to be a part of $k^{\gamma r_{j+1}}$-cluster. Thus, $k^\delta$-clusters are built around $\MM(r_n, \varphi _0)\setminus
\left(\MM(r_{n-1}, \varphi _0)\cup
\MM^{(2)}(r_n,\varphi_0)\cup \Omega ^{(2)}_s(r_n)\cup\MM_{2,tw}(r_n,\varphi_0)\right)$. The union of this
neighborhoods we denote $\Pi _{nr,\delta}$. Next, $k^{\gamma r_{j}}$-clusters ($j\geq1$) are built around the points of
${\MM}^{(j+1)}(r_n,\varphi _0)\cup \Omega _{s}^{(j+1)}(r_n)\setminus \left({\MM}^{(j+1)}(r_{n-1},\varphi _0)\cup {\MM}^{(j+2)}(r_n,\varphi _0)\cup \Omega _{s}^{(j+2)}(r_n)\right)$.
 The set of all non-resonant
$k^{\gamma r_{j}}$-clusters we denote by
$\Pi_{nr,r_{j}}^{(n)}$. It is convenient to identify $\Pi_{nr,r_{0}}^{(n)}:=\Pi_{nr,\delta}^{(n)}$. Then
$$
\Pi_{nr}^{(n)}:=\cup_{j=0}^{n-2}\Pi_{nr,r_j}^{(n)},$$
 Those $\Pi_{nr,r_j}^{(n)}$, which intersect with
$\Pi_r^{(n)}$ we attach to $\Pi_r^{(n)}$ just slightly abusing the
notation (cf. Section \ref{S:4}). The part of $\Pi_{nr,r_j}^{(n)}$ which does not intersect with $\Pi_r^{(n)}$ we denote by $\Pi_{nr,r_j}^{(n)'}$. Correspondingly, the part of $\Pi_{nr}^{(n)}$ which does not intersect $\Pi_{r}^{(n)}$ is denoted by $\Pi_{nr}^{(n)'}$.
Further, \begin{equation} \label{P(n)}
P^{(n)}:=P^{(n)}_r+P^{(n)'}_{nr}+P(r_{n-1}).
\end{equation}
We continue construction from Section \ref{S:3}. Repeating the arguments
from the proofs of Lemmas~\ref{L:black},~\ref{L:grey},~\ref{L:white}
with obvious changes (in particular, using Lemma~\ref{L:2/3-1ind-last}
instead of Lemmas~\ref{L:2/3-1}, \ref{L:2/3-1ind}) we obtain the following results.
(Here and in what follows we will omit superscript $(n)$ when it
does not lead to a confusion.)

\begin{lemma} \label{L:blackindlast} \begin{enumerate} \item
Each $\Pi _b^j$ contains no more than $k^{\gamma r_{n-1}/2-\delta
_0r_{n-1}+150\gamma r_{n-2}}$ black boxes.
\item The size of $\Pi _b^j$ in $\||\cdot \||$ norm is less than
$k^{3\gamma r_{n-1}/2+150\gamma r_{n-2}}$.
\item Each $\Pi _b^j$ contains no more than $k^{\gamma r_{n-1}+150\gamma r_{n-2}}$ elements of
$\MM^{(n)}$. Moreover, any box of $\||\cdot \||$-size $k^{3\gamma
r_{n-1}/2+150\gamma r_{n-2}}$ containing $\Pi _b^j$ has no more
than $k^{\gamma r_{n-1}+150\gamma r_{n-2}}$ elements of
$\MM^{(n)}$ inside.
\end{enumerate}
\end{lemma}

\begin{lemma}\label{L:greyindlast} \begin{enumerate} \item
Each $\Pi _g^j$ contains no more than $k^{\gamma r_{n-1}/3+2\delta
_0r_{n-1}}$ grey boxes.
\item The size of $\Pi _g^j$ in $\||\cdot \||$ norm is less than
$k^{5\gamma r_{n-1}/6+4\delta _0r_{n-1}}$.
\item Each $\Pi _g^j$ contains no more than $k^{\gamma r_{n-1}/2+\delta _0r_{n-1}}$ elements of
$\MM^{(n)}$.\end{enumerate}
\end{lemma}

\begin{lemma}\label{L:whiteindlast} \begin{enumerate} \item The size of $\Pi _w^j$ in $\||\cdot \||$ norm is less than
$k^{\gamma r_{n-1}/3-\delta _0r_{n-1}}$.
\item Each $\Pi _w^j$ contains no more
than $k^{\gamma r_{n-1}/6-\delta _0r_{n-1}}$ points of $\MM^{(n)}$.
\end{enumerate}
\end{lemma}

The construction of the rest of Section~\ref{MOforStep3} stays unchanged. Let us
introduce corresponding notation, formulate the results and provide
some comments.

Next lemmas are the analogues of
Lemmas~\ref{L:Pnr},~\ref{L:Pr},~\ref{L:Ps}.

\begin{lemma}\label{L:Pnrnlast}Let $\varphi _0\in \omega^{(n)}(k,\delta ,\tau )$,
$|\varphi-\varphi _0|<k^{-k^{r_{n-2}}}$. Then,
\begin{equation}\label{Pnrnlast}
\left\|\Bigl(P_{nr}\bigl(H(\k^{(n)}(\varphi
))-k^{2}I\bigr)P_{nr}\Bigr)^{-1}\right\|<k^{r_{n-1}'k^{2\gamma
r_{n-2}}}k^{r_{n-1}'}\leq k^{k^{3\gamma r_{n-2}}}.
\end{equation} \end{lemma}

\begin{lemma}\label{L:Prnlast} Let $\varphi _0\in \omega^{(n)}(k,\delta ,\tau )$,
and $|\varphi-\varphi _0|<k^{-k^{r_{n-2}}}$, $i=1,2,3$. Then,
\begin{enumerate}
\item The number of poles of the resolvent $\Bigl(P_i\bigl(H(\k^{(n)}(\varphi
))-k^{2}I\bigr)P_i\Bigr)^{-1}$ in the disc $|\varphi-\varphi
_0|<k^{-k^{r_{n-2}}}$ is no greater than $N_i^{(n-1)}$, where $N_1^{(n-1)}=k^{\gamma
r_{n-1}+150\gamma r_{n-2}}$, $N_2^{(n-1)}=k^{\gamma r_{n-1}/2+\delta
_0r_{n-1}}$, $N_3^{(n-1)}=k^{\gamma r_{n-1}/6-\delta _0r_{n-1}}$.
\item Let $\varepsilon$ be the distance to the nearest pole of
the resolvent in ${\cal W}^{(n)}$ and let
$\varepsilon_0=\min\{\varepsilon,\,k^{-r_{n-1}'}\}$. Then the
following estimates hold:
\begin{equation}\label{Pr-1nlast}
\begin{split}
\left\|\Bigl(P_i\bigl(H(\k^{(n)}(\varphi
))-k^{2}I\bigr)P_i\Bigr)^{-1}\right\|<k^{2r_{n-1}'k^{2\gamma
r_{n-2}}}k^{r_{n-1}'}\left(\frac{k^{-r_{n-1}'}}{\varepsilon
_0}\right)^{N_{i}^{(n-1)}}\leq \cr k^{k^{3\gamma
r_{n-2}}}\left(\frac{k^{-r_{n-1}'}}{\varepsilon _0}\right)^{N_{i}^{(n-1)}},
\end{split}
\end{equation}
\begin{equation}\label{Pr-2nlast}
\begin{split}
\left\|\Bigl(P_i\bigl(H(\k^{(n)}(\varphi
))-k^{2}I\bigr)P_i\Bigr)^{-1}\right\|_1<k^{2r_{n-1}'k^{2\gamma
r_{n-2}}}k^{r_{n-1}'+8\gamma
r_{n-1}}\left(\frac{k^{-r_{n-1}'}}{\varepsilon
_0}\right)^{N_{i}^{(n-1)}}\leq \cr k^{k^{3\gamma
r_{n-2}}}\left(\frac{k^{-r_{n-1}'}}{\varepsilon _0}\right)^{N_{i}^{(n-1)}}.
\end{split}
\end{equation}
\end{enumerate}
\end{lemma}
\begin{proof} The proof of this lemma is analogous to that of Lemma \ref{L:Pr} up to the replacement of ${\MM}^{(2)}$ by ${\MM}^{(n)}$, $\OO^{(2)}_{\m}$ by $\OO^{(n)}_{\m}$, and the shift of indices: $\delta $ to $r_{n-2}$, $r_1$ to
$r_{n-1}$, etc. We apply Lemmas \ref{L:blackindlast}--\ref{L:whiteindlast} instead of \ref{L:black}--\ref{L:white}. We apply Lemmas \ref{L:Prn}, \ref{L:Psn} with $\varepsilon _0=k^{-r_{3}'}$ and $p_\m>k^{-r_{3}'k^{2\gamma r_{2}}}$ instead of Lemma \ref{L:Pr}, \ref{L:Ps} for $n=4$ and Lemmas \ref{L:Prnlast}, \ref{L:Psnlast} with  inductively (with $n-1$ instead of $n$ and $\varepsilon _0=k^{-r'_{n-1}}$, $p_\m>k^{-r'_{n-1}k^{2\gamma r'_{n-2}}}$) for further steps.
\end{proof}
Let $\Pi_s^j$ be a particular $k^{r_{n-1}/2}$-box around $\m\in\Omega_s^{(n)}(r_n)$ and let $P_s^j$ be the corresponding projection.
\begin{lemma}\label{L:Psnlast} Let $\varphi _0\in \omega^{(n)}(k,\delta ,\tau )$. Then, the operator
$\Bigl(P_s^j\bigl(H(\k^{(n)}(\varphi
))-k^{2}I\bigr)P_s^j\Bigr)^{-1}$ has no more than one pole in the
disk $|\varphi-\varphi _0|<k^{-k^{r_{n-2}}}$. Moreover,
\begin{equation}\label{Ps-1nlast}
\left\|\Bigl(P_s^j\bigl(H(\k^{(n)}(\varphi
))-k^{2}I\bigr)P_s^j\Bigr)^{-1}\right\|<\frac{8k^{-1}}
{p_\m\varepsilon _0},
\end{equation}
\begin{equation}\label{Ps-2nlast}
\left\|\Bigl(P_s^j\bigl(H(\k^{(n)}(\varphi
))-k^{2}I\bigr)P_s^j\Bigr)^{-1}\right\|_1<\frac{8k^{-1+
4r_{n-1}}}{p_\m\varepsilon_0},
\end{equation}
$\varepsilon _0=\min\{\varepsilon,\,k^{-r_{n-1}'}\}$, where
$\varepsilon$ is the distance to the pole of the operator.
\end{lemma}

Note that $p_{\m}>k^{-2\mu r_n}$ when $\m \in \Omega (r_n)$. The
analogues of Lemma~\ref{L:boundary} and Corollary~\ref{C:PHP-2} also
hold.

\subsubsection{Resonant and Nonresonant Sets for Step $n+1$ \label{GSlast}}

We divide $[0,2\pi )$ into $[2\pi k^{k^{r_{n-2}}}]+1$ intervals
$\Delta_l^{(n)}$ with the length not bigger than $k^{-k^{r_{n-2}}}$.
If a particular interval belongs to $\OO^{(n)}$ we ignore it;
otherwise, let $\varphi_0^{(l)}\not\in\OO^{(n)}$ be a point inside the
$\Delta_l^{(n)}$. Let
\begin{equation}\W_l^{(n)}=\{\varphi \in \W^{(n)}:\ | \varphi -\varphi
_0^{(l)}|<4k^{-k^{r_{n-2}}}\}. \label{W2mind-last} \end{equation} Clearly,
neighboring sets $\W_l^{(n)}$  overlap (because of the multiplier 4
in the inequality), they cover $\hat \W^{(n)}$ , which is
the restriction of $\W^{(n)}$ to the
$2k^{-k^{r_{n-2}}}$-neighborhood of $[0,2\pi )$. For each $\varphi
\in \hat \W^{(n)}$ there is an $l$ such that $|\varphi -\varphi
_{0}^{(l)}|<4k^{-k^{r_{n-2}}}$. We consider the poles of the resolvent
$\left(P^{(n)} (H(\k^{(n)}(\varphi))-k^{2})P^{(n)}\right)^{-1}$ in
$\hat \W_l^{(n)}$ and denote them by $\varphi^{(n)} _{lm}$,
$l=1,...,M_m$. As before, the resolvent has a block structure. The
number of blocks clearly cannot exceed the number of elements in
$\Omega (r_n)$, i.e. $k^{4r_n}$. Using the estimates for the number
of poles for each block, the estimate being provided by Lemma
\ref{L:Prnlast} Part 1, we can roughly estimate the number of poles
of the resolvent by $k^{4r_n+r_{n-1}}$. Next, let  $\OO^{(n+1)}_{lm}$ be the disc of the radius $k^{-r_n'}$ around
$\varphi ^{(n)}_{lm}$.
\begin{definition} The set
\begin{equation}\OO^{(n+1)}=\cup _{lm}\OO^{(n+1)}_{lm} \label{Olast}
\end{equation}
we call the $n+1$-th resonant set. The set
\begin{equation}\W^{(n+1)}= \W^{(n)}\setminus \OO^{(n+1)}\label{Wlast}
\end{equation}
is called the $n+1$-th non-resonant set. The set
\begin{equation}\omega^{(n+1)}= \W^{(n+1)}\cap [0,2\pi) \label{wlast}
\end{equation}
is called the $n+1$-th real non-resonant set. \end{definition}
The
following statements can be proven in the same way as
Lemmas~\ref{L:geometric3}, \ref{4.16} and \ref{estnonres0-1}.
\begin{lemma}\label{L:geometric4nlast} Let  $r_n'>\mu r_n>k^{r_{n-2}}$, $\varphi \in \W^{(n+1)}$, $\varphi
_0(m)$ corresponds  to an interval $\Delta _l^{(n)}$ containing $\Re
\varphi $. Let $\Pi $ be one of the components $\Pi _s^j(\varphi
_0^{(l)})$, $\Pi _b^j(\varphi _0^{(l)})$, $\Pi _g^j(\varphi _0^{(l)})$, $\Pi
_w^j(\varphi _0^{(l)})$ and
 $P(\Pi )$ be the projection corresponding to $\Pi $. Let also
 $\varkappa \in \C: |\varkappa-\varkappa^{(n)}(\varphi )|<k^{-r_n'k^{2\gamma
r_{n-1}}}$. Then,
\begin{equation} \label{March3-2nlast} \left\|\left(P(\Pi )
\left(H\big(\k(\varphi )\big)-k^{2}I\right)P(\Pi
)\right)^{-1}\right\|<k^{2\mu r_n+r_n'N^{(n-1)}},\end{equation}
\begin{equation} \label{March3-3nlast}
\left\|\left(P(\Pi )\left(H\big(\k(\varphi
)\big)-k^{2}I\right)P(\Pi )\right)^{-1}\right\|_1<k^{(2\mu+1)
r_n+r_n'N^{(n-1)}},\end{equation}$N^{(n-1)}$ corresponding to the
color of $\Pi $ ($N^{(n-1)}=1$, $k^{\gamma r_{n-1}+150\gamma
r_{n-2}}$, $k^{\gamma r_{n-1}/2+\delta _0r_{n-1}}$, $k^{\gamma
r_{n-1}/6-\delta _0r_{n-1}}$ for simple, black, grey and white
clusters, correspondingly).
\end{lemma}
\begin{proof} The lemma follows from Lemmas \ref{L:Prnlast}, \ref{L:Psnlast} and the definition of $\W^{(n+1)}$.\end{proof}
We also notice that inequalities  $r_n'>\mu r_n>k^{r_{n-2}}$ follow from \eqref{indrn}.
By total size of the set $\OO^{(n+1)}$ we mean the sum of
the sizes of its connected components.
\begin{lemma}\label{10.9} Let $r_n'\geq (\mu+10)r_n$, $r_n>k^{r_{n-2}}$. Then, the size of each connected component of
 $\OO^{(n+1)}$ is less
than $32k^{4r_n-r_n'}$. The total size of $\OO^{(n+1)}$ is less than
$k^{-r_n'/2}$.
\end{lemma}


\begin{lemma}\label{estnonres0-1last} Let $\varphi\in\W^{(n)}$ and
$C_{n+1}$ be the circle $|z-k^{2}|=k^{-2r_n'k^{2\gamma r_{n-1}}}$.
Then
$$
\left\|\left(P(r_{n-1})(H(\k^{(n)}(\varphi))-z)P(r_{n-1})\right)^{-1}\right\|\leq
4^nk^{2r_n'k^{2\gamma r_{n-1}}}. $$ \end{lemma}

\subsection{Operator $H^{(n+1)}$. Perturbation Formulas}
Let $P(r_{n})$ be an orthogonal projector onto
$\Omega(r_{n}):=\{\m:\ |\|\p_\m\||\leq k^{r_{n}}\}$ and
$H^{(n+1)}=P(r_{n})HP(r_{n}) $. We consider
$H^{(n+1)}(\k^{(n)}(\varphi ))$ as a perturbation of
\begin{equation}
\begin{split}
&\tilde H^{(n)}=\tilde
P_l^{(n)}H\tilde
P_l^{(n)}+\left(P(r_n)-\tilde
P_l^{(n)}\right)H_0,
\end{split}
\end{equation}
where $H=H(\k^{(n)}(\varphi ))$, $H_0=H_0(\k^{(n)}(\varphi ))$, and $\tilde
P_l^{(n)}$
is the projection $P^{(n)}$, see \eqref{P(n)}, corresponding to $\varphi _{0}^{(l)}$ in
the interval $\Delta _l^{(n)}$ containing $\varphi $. Let
\begin{equation}W^{(n)}=H^{(n+1)}-\tilde H^{(n)}=P(r_{n})VP(r_{n})-\tilde
P_l^{(n)}V\tilde
P_l^{(n)}, \label{W2last}\end{equation}
\begin{equation}\label{g3last} g^{(n+1)}_r({\k}):=\frac{(-1)^r}{2\pi
ir}\hbox{Tr}\oint_{C_{n+1}}\left(W^{(n)}(\tilde
H^{(n)}({\k})-zI)^{-1}\right)^rdz,
\end{equation} \begin{equation}\label{G3last}
G^{(n)}_r({\k}):=\frac{(-1)^{r+1}}{2\pi i}\oint_{C_{n+1}}(\tilde
H^{(n)}({\k})-zI)^{-1}\left(W^{(n)}(\tilde
H^{(n)}({\k})-zI)^{-1}\right)^rdz,
\end{equation}
where $C_{n+1}$ is the circle $|z-k^{2}|=\varepsilon _0^{(n+1)}$,
$\varepsilon _0^{(n+1)}=k^{-2r_n'k^{2\gamma r_{n-1}}}.$

Recall that $\beta:=\delta_*/100$. The proof of the following
statements is analogous to the one in Step III (see
Theorem~\ref{Thm3}, Corollary~\ref{corthm3},
Lemma~\ref{L:derivatives-3} and Lemma~\ref{ldk-3}).

\begin{theorem} \label{Thm3last} Suppose $k>k_*$, $\varphi $ is in
the real  $k^{-r_n'-\delta }$-neighborhood of $\omega
^{(n+1)}(k,\delta,\tau )$ and $\varkappa\in\R$,
$|\varkappa-\varkappa^{(n)}(\varphi )|\leq \varepsilon
^{(n+1)}_0k^{-1-\delta }$, $\k=\varkappa(\cos \varphi ,\sin
\varphi )$.  Then, there exists a single eigenvalue of
$H^{(n+1)}({\k})$ in the interval $\varepsilon _{n+1}( k,\delta,\tau
)=\left( k^{2}-\varepsilon _0^{(n+1)}, k^{2}+\varepsilon
_0^{(n+1)}\right)$. It is given by the absolutely converging series
series:
\begin{equation}\label{eigenvalue-3last}
\lambda^{(n+1)}({\k})=\lambda^{(n)}({\k})+ \sum\limits_{r=2}^\infty
g^{(n+1)}_r({\k}).\end{equation} For coefficients
$g^{(n+1)}_r({\k})$ the following estimates hold:
\begin{equation}\label{estg3last} |g^{(n+1)}_r({\k})|<
k^{-\frac{\beta}{5}
k^{r_{n-1}-r_{n-2} }-\beta (r-1)}.
\end{equation}
The corresponding spectral projection is given by the series:
\begin{equation}\label{sprojector-3last}
\E ^{(n+1)}({\k})=\E^{(n)}({\k})+\sum\limits_{r=1}^\infty
G^{(n+1)}_r({\k}), \end{equation} $\E^{(n)}({\k})$ being the
spectral projection of $H^{(n)}$. The operators $G^{(n+1)}_r({\k})$
satisfy the estimates:
\begin{equation}
\label{Feb1a-3last}
\left\|G^{(n+1)}_r({\k})\right\|_1<k^{-\frac{\beta}{10}
k^{r_{n-1}-r_{n-2} } -\beta r},
\end{equation}
\begin{equation}G^{(n+1)}_r({\k})_{\s\s'}=0,\ \mbox{when}\ \ \
2rk^{\gamma r_{n-1}+150\gamma
r_{n-2}}+3k^{r_{n-1}}<\||\p_\s\||+\||\p_{\s'}\||.\label{Feb6a-3last}
\end{equation}
\end{theorem}
\begin{corollary}\label{corthm3last} For the perturbed eigenvalue and its spectral
projection the following estimates hold:
 \begin{equation}\label{perturbation-3last}
\lambda^{(n+1)}({\k})=\lambda^{(n)}({\k})+ O_2\left(k^{-\frac15
\beta k^{r_{n-1}-r_{n-2} }}\right),
\end{equation}
\begin{equation}\label{perturbation*-3last}
\left\|\E^{(n+1)}({\k})-\E^{(n)}({\k})\right\|_1<k^{-\frac{\beta}{10}
k^{r_{n-1}-r_{n-2} }},
\end{equation}
\begin{equation}
\left|\E^{(n+1)}({\k})_{\s\s'}\right|<k^{-d^{(n+1)}(\s,\s')},\ \
\mbox{when}\ \||\p_\s\||>4k^{r_{n-1}} \mbox{\ or }
\||\p_{\s'}\||>4k^{r_{n-1} },\label{Feb6b-3last}
\end{equation}
$$d^{(n+1)}(\s,\s')=\frac{1}{16}(\||\p_\s\||+\||\p_{\s'}\||)k^{-\gamma r_{n-1}-150\gamma
r_{n-2}}\beta +\frac{1}{10}\beta k^{r_{n-1}-r_{n-2} }.$$
\end{corollary}

\begin{lemma} \label{L:derivatives-3last}Under conditions of Theorem \ref{Thm3last} the following
estimates hold when $\varphi \in \omega ^{(n+1)}(k,\delta )$ or its
complex $k^{-r_n'-\delta}$ neighborhood and $\varkappa\in \C$,
$|\varkappa-\varkappa^{(n)}(\varphi
)|<\varepsilon_0^{(n+1)}k^{-1-\delta}$.
\begin{equation}\label{perturbation-3cIVlast}
\lambda^{(n+1)}({\k})=\lambda^{(n)}({\k})+O_2\left(k^{-\frac 15
\beta k^{r_{n-1}-r_{n-2}}}\right),
\end{equation}
\begin{equation}\label{estgder1-3kIVlast}
\frac{\partial\lambda^{(n+1)}}{\partial\varkappa}=
\frac{\partial\lambda^{(n)}}{\partial\varkappa} +O_2\left(k^{-\frac
15 \beta k^{r_{n-1}-r_{n-2}} }M_{n-1}\right), \ \ \ \
M_{n-1}:=\frac{k^{1+\delta}}{\varepsilon
^{(n+1)}_0},\end{equation}
\begin{equation}\label{estgder1-3phiIVlast}\frac{\partial\lambda^{(n+1)}}{\partial \varphi }=\frac{\partial\lambda^{(n)}}{\partial \varphi }+
O_2\left(k^{-\frac 15 \beta k^{r_{n-1}-r_{n-2}}+r_n'+\delta
}\right),
 \end{equation}
\begin{equation}\label{estgder2-3IVlast}
\frac{\partial^2\lambda^{(n+1)}}{\partial\varkappa^2}=
\frac{\partial^2\lambda^{(n)}}{\partial\varkappa^2}+
O_2\left(k^{-\frac 15 \beta k^{r_{n-1}-r_{n-2}}
}M_{n-1}^2\right),
\end{equation}
\begin{equation} \label{gulf2-3IVlast}
\frac{\partial^2\lambda^{(n+1)}}{\partial\varkappa\partial \varphi
}=\frac{\partial^2\lambda^{(n)}}{\partial\varkappa\partial \varphi
}+ O_2\left(k^{-\frac 15 \beta k^{r_{n-1}-r_{n-2}}+r_n'+\delta
}M_{n-1}\right),
\end{equation}
\begin{equation} \label{gulf3-3IVlast}
\frac{\partial^2\lambda^{(n+1)}}{\partial\varphi
^2}=\frac{\partial^2\lambda^{(n)}}{\partial\varphi
^2}+O_2\left(k^{-\frac 15 \beta
k^{r_{n-1}-r_{n-2}}+2r_n'+2\delta }\right).
\end{equation}\end{lemma}

\begin{corollary} \label{"O"last} All ``$O_2$"-s on the right hand sides of \eqref{perturbation-3cIVlast}-\eqref{gulf3-3IVlast}
can be written as $O_1\left(k^{-\frac {1}{10} \beta
k^{r_{n-1}-r_{n-2}}}\right)$.
\end{corollary}

\subsection{\label{IS3IVlast}Isoenergetic Surface for Operator $H^{(n+1)}$}

The following statement is an analogue of Lemma~\ref{ldk-3}.

\begin{lemma}\label{ldk-3IVlast} \begin{enumerate}
\item For every $\lambda :=k^{2}$,  $k>k_*$, and $\varphi $ in the real  $\frac{1}{2} k^{-r_n'-\delta }$-neighborhood
of $\omega^{(n+1)}(k,\delta, \tau )$ , there is a unique
$\varkappa^{(n+1)}(\lambda, \varphi )$ in the interval
$I_n:=[\varkappa^{(n)}(\lambda, \varphi )-\varepsilon
^{(n+1)}_0k^{-1-\delta},\varkappa^{(n)}(\lambda, \varphi
)+\varepsilon ^{(n+1)}_0k^{-1-\delta}]$, such that
    \begin{equation}\label{2.70-3IVlast}
    \lambda^{(n+1)} \left(\k
^{(n+1)}(\lambda ,\varphi )\right)=\lambda ,\ \ \k ^{(n+1)}(\lambda
,\varphi ):=\varkappa^{(n+1)}(\lambda ,\varphi )\vec \nu(\varphi).
    \end{equation}
\item  Furthermore, there exists an analytic in $ \varphi $ continuation  of
$\varkappa^{(n+1)}(\lambda ,\varphi )$ to the complex  $\frac{1}{2}
k^{-r_n'-\delta }$-neighborhood of $\omega^{(n+1)}(k,\delta, \tau )$
such that $\lambda^{(n+1)} (\k ^{(n+1)}(\lambda, \varphi ))=\lambda
$. Function $\varkappa^{(n+1)}(\lambda, \varphi )$ can be
represented as $\varkappa^{(n+1)}(\lambda, \varphi
)=\varkappa^{(n)}(\lambda, \varphi )+h^{(n+1)}(\lambda, \varphi )$,
where
\begin{equation}\label{dk0-3IVlast} |h^{(n+1)}(\varphi )|=O_1\left(k^{-\frac 15 \beta k^{r_{n-1}-r_{n-2}}-1
}\right),
\end{equation}
\begin{equation}\label{dk-3IVlast}
\frac{\partial{h}^{(n+1)}}{\partial\varphi}= O_2\left(k^{-\frac 15
\beta k^{r_{n-1}-r_{n-2}}-1 +r_n'+\delta }\right),
\end{equation}
\begin{equation}
\frac{\partial^2{h}^{(n+1)}}{\partial\varphi^2}= O_4\left(k^{-\frac
15 \beta k^{r_{n-1}-r_{n-2}}-1 +2r_n'+2\delta }\right).
\end{equation} \end{enumerate}\end{lemma}

Let us consider the set of points in $\R^2$ given by the formula:
$\k=\k^{(n+1)} (\varphi), \ \ \varphi \in \omega ^{(n+1)}(k,\delta,
\tau )$. By Lemma \ref{ldk-3IVlast} this set of points is a slight
distortion of ${\cal D}_{n}$. All the points of this curve satisfy
the equation $\lambda^{(n+1)}(\k^{(n+1)}(\varphi ))=k^{2}$. We call
it isoenergetic surface of the operator $H^{(n+1)}$ and denote by
${\cal D}_{n+1}$.

\section{Isoenergetic Sets. Generalized Eigenfunctions of $H$}
\subsection{Construction of Limit-Isoenergetic Set} At every step $n$ we constructed  a set
$\B _n(\lambda)$, ${\cal B}_{n}(\lambda)\subset {\cal
B}_{n-1}(\lambda)\subset S_1$,  and a function
$\varkappa^{(n)}(\lambda,\vec \nu)$, $\vec \nu \in {\cal
B}_n(\lambda)$, with the following properties. The set ${\cal
D}_{n}(\lambda )$ of vectors $\k=\varkappa^{(n)}(\lambda ,\vec
\nu)\vec \nu$,
   $\vec \nu \in {\cal B}_{n}(\lambda )$,
    is a slightly distorted circle with holes, see formula (\ref{Dn})
    and Lemmas \ref{ldk}, \ref{ldk-2},
\ref{ldk-3}, \ref{ldk-3IV}, \ref{ldk-3IVlast}. For any $\k^{(n)}(\lambda,\vec \nu)\in {\cal D}_{n}(\lambda )$ there is a
single eigenvalue of
 $H^{(n)}(\k^{(n)})$
equal to $\lambda $ and  given by the perturbation series.
Let
    ${\cal B}_{\infty}(\lambda)=\bigcap_{n=1}^{\infty}{\cal B}_n(\lambda).$
Since ${\cal B}_{n+1} \subset {\cal B}_n$ for every $n$, ${\cal
B}_{\infty}(\lambda)$ is a unit circle with infinite number of
holes, more and more holes of smaller and smaller size appearing at
each step. \begin{lemma} \label{L:Dec9} The length of ${\cal
B}_{\infty}(\lambda)$ satisfies estimate (\ref{B}) with $\gamma
_4=37\mu\delta $.
\end{lemma}
\begin{proof}
 Using   Lemmas \ref{L:G1} (part 3),  \ref{L:O2size}, \ref{4.16},
 \ref{7.10} and \ref{10.9} and considering that $r_n>>37\delta \mu $, we easily
conclude that $L\left({\cal B}_n\right)=\left(2\pi+O(k^{-37\mu
\delta })\right)$, $k=\lambda ^{1/2}$ uniformly in $n$. Since
${\cal B}_n$ is a decreasing sequence of sets, (\ref{B}) holds.
\end{proof} Let us consider
    $\varkappa_{\infty}(\lambda, \vec \nu)=\lim_{n \to \infty}\varkappa^{(n)}(\lambda,\vec \nu),\quad
    \vec \nu \in {\cal B}_{\infty}(\lambda ).$
    \begin{lemma} The limit $\varkappa_{\infty}(\lambda,
    \vec \nu)$ exists for any $\vec \nu \in {\cal B}_{\infty}(\lambda
    )$ and the following estimates hold\footnote{Here and below to make the formulations shorter we will use the notation $r_0$ which may be equal either $\delta$ or $\delta_*$. In any case $\Omega(r_0):=\Omega(\delta)$. In all other situations where $r_0$ appears we will specify what it is equal to.}:
    \begin{align}\label{6.1}& \left|\varkappa_{\infty}(\lambda,
\vec \nu)-\lambda^{1/2}\right|<ck^{-3+(80\mu +6)\delta},
     \cr &\left|\varkappa_{\infty}(\lambda, \vec \nu)-\varkappa
    ^{(1)}(\lambda,\vec \nu)\right|<ck^{-k^{\delta }(2Q)^{-1}}k^{-1 },\cr &\left|\varkappa_{\infty}(\lambda, \vec \nu)-\varkappa
    ^{(n)}(\lambda,\vec \nu)\right|<k^{-\frac15\beta k^{r _{n-1}-r_{n-2}}}
    ,\ \ \ n\geq 2,\ \ r_0:=\delta_*.
    \end{align}
    \end{lemma}
    \begin{corollary}\label{Dec18}
    For every $\vec \nu \in {\cal B}_{\infty}(\lambda)$ estimate (\ref{h}) holds, where\\ $\gamma
_5=(3-(80\mu +6)\delta)/2
    >0$.
    \end{corollary}
    The lemma easily follows from
  the estimates (\ref{dk0}),  (\ref{dk0-2}), (\ref{dk0-3}), \eqref{dk0-3IV} and (\ref{dk0-3IVlast}).

    Estimates (\ref{dk}), (\ref{dk-2}) (\ref{dk-3}), \eqref{dk-3IV} and (\ref{dk-3IVlast}) justify convergence of the
    series $\sum_{n=1}^{\infty}
    \frac{\partial h_n}{\partial \varphi },$ and hence,
    of the sequence $\frac{\partial \varkappa^{(n)}}{\partial \varphi }.$
    We denote the limit of this sequence by $\frac{\partial \varkappa_{\infty}}{\partial \varphi }.$
    \begin{lemma} The  estimate (\ref{Dec9a}) with $\gamma
_{11}=(3-(120\mu+7)\delta) /2>0,$ holds for any $\vec \nu \in
    {\cal B}_{\infty}(\lambda)$.
   \end{lemma}
We define ${\cal D}_{\infty}(\lambda )$ by (\ref{Dinfty}). Clearly,
${\cal D}_{\infty}(\lambda )$ is a slightly distorted circle of
radius $k$ with infinite number of holes. We can assign a tangent
vector $\frac{\partial \varkappa_{\infty }}{\partial \varphi }\vec
\nu +\varkappa_{\infty } \vec \mu $, $\vec \mu =(-\sin \varphi ,\cos
\varphi )$ to the curve ${\cal D}_{\infty}(\lambda )$, this tangent
vector being the limit of corresponding tangent vectors for  curves
${\cal D}_{n}(\lambda )$ at points $\k^{(n)}(\lambda ,\vec \nu )$ as
$n\to \infty $.

Next we show that ${\cal D}_{\infty}(\lambda )$ is an isoenergetic
curve for $H$. Namely for every $\k \in {\cal
D}_{\infty}(\lambda )$ there is a generalized eigenfunction $\Psi
_{\infty }(\k ,\x)$ such that $H\Psi _{\infty }=\lambda
\Psi _{\infty }$.

\subsection{Generalized Eigenfunctions of $H$}

 We show that for
    every  $\k $ in
a set $${\cal G} _{\infty }=\cup _{\lambda
>\lambda _*}{\cal D}_{\infty}(\lambda ),\ \ \lambda _*= k _*^{2} ,$$ there is a solution
$\Psi _{\infty }(\k , \x)$ of the equation for
eigenfunctions:
    \begin{equation} -\Delta\Psi _{\infty}(\k , \x)+V(\x)\Psi _{\infty }(\k ,
    \x)=\lambda _{\infty}(\k )\Psi _{\infty }(\k , \x),
    \label{6.2.1}
    \end{equation}
which can be represented in the form
    \begin{equation}
    \Psi _{\infty }(\k , \x)=e^{i\langle \k , \x
    \rangle}\Bigl(1+u_{\infty}(\k , \x)\Bigr),\ \ \ \
    \bigl\|u_{\infty}(\k , \x))\bigr\| _{L_{\infty }(\R^2)}=O(|\k |^{-\gamma_1}),
   \label{6.2.1a}
    \end{equation}
where $u_{\infty}(\k , \x)$ is a quasi-periodic
function,
 $\gamma _1=1-45\mu \delta >0$; the eigenvalue $\lambda _{\infty}(\k )$ satisfies the asymptotic
formula:
 \begin{equation}\lambda _{\infty}(\k )=|\k |^{2}+O(|\k |^{-\gamma _2}), \ \ \ \gamma _2=2-(80\mu +6)\delta
>0.\label{6.2.4}
\end{equation}
We also show that the set $\cal{G} _{\infty }$ satisfies
(\ref{full}).

 In fact,   by (\ref{6.1}), any $\k \in {\cal D}_{\infty}(\lambda
)$ belongs to the $k^{-\frac15\beta k^{r
_{n-1}-r_{n-2}}}$-neighborhood of ${\cal D}_n(\lambda ),\ \ n\geq3$.
Let us consider spectral projectors $\E^{(n)}$, each of them being
defined in a finite dimensional space of sequences with indices in
$\Omega (r_{n-1})$. We extend each of them to the
whole space $\ell^{2}(\Z^4)$ by putting it to be zero into the
orthogonal complement of $\ell^{2}\left(\Omega (r_{n-1})\right)$.
This way they all act in space $\ell^{2}(\Z^4)$.  Applying the
perturbation formulae proved in the previous sections (see
Corollaries \ref{corthm1}, \ref{corthm2}, \ref{corthm3},
\ref{corthm3IV}, \ref{corthm3last}), we obtain the following
inequalities:
    \begin{equation}
    \begin{split}&\bigl\|\E^{(1)}(\k)-{\E}_0(\k)\bigr\|_1<ck^{-\gamma_0},\ \
    \gamma_0:=1-44\mu\delta,
    \cr &\bigl\|\E^{(2)}(\k)-{\E}^{(1)}(\k)\bigr\|_1<k^{-k^\delta (4Q)^{-1}},
    \cr &\bigl\|\E^{(n)}(\k)-{\E}^{(n-1)}(\k)\bigr\|_1<
    2k^{-\frac1{10}\beta k^{r_{n-2}-r_{n-3}}}, \quad n \geq 3,\ \ r_0:=\delta_*,\label{6.2.2}
    \end{split}
    \end{equation}
    \begin{equation}
    \begin{split}&\bigl|\lambda ^{(1)}(\k)-|\k |^{2}
     \bigr|
     <ck^{-\gamma _2}
    , \ \ \bigl|\lambda ^{(2)}(\k)-\lambda ^{(1)}(\k)
     \bigr|
     <ck^{-k^\delta (2Q)^{-1}},
     \cr &\bigl|\lambda ^{(n)}(\k )-\lambda ^{(n-1)}(\k )\bigr|<
     2k^{-\frac15\beta k^{r_{n-2}-{r_{n-3}}}},
     \quad n \geq 3,\ \ r_0:=\delta_*,
    \label{6.2.3}
    \end{split}
    \end{equation}
where  $\lambda ^{(n+1)}(\k )$ is the eigenvalue
corresponding to $\E^{(n+1)}(\k)$, $\
{\E}_0(\k)$ corresponds to $V=0$.

 \begin{remark} \label{R:Dec9} We see from (\ref{6.1}), that any $\k
\in {\cal D}_{\infty}(\lambda )$ belongs to the $k^{-\frac15\beta
k^{r _{n-1}-r_{n-2}}}$-neighborhood of ${\cal D}_n(\lambda ),\
n\geq3$. Applying perturbation formulae for $n$-th step, we easily
obtain that  there is an
 eigenvalue  $\lambda^{(n)}(\k )$ of $H^{(n)}(\k )$
satisfying the estimate $\lambda^{(n)}(\k )=\lambda
+\delta _n$, $\delta _n=o(1)$ as $n\to \infty $, the eigenvalue
$\lambda^{(n)}(\k )$
     being given by a perturbation
series of the type (\ref{eigenvalue-3last}). Hence, for every
$\k \in {\cal D}_{\infty}(\lambda)$ there is a limit:
\begin{equation} \lim _{n\to \infty }\lambda^{(n)}(\k
)=\lambda.\label{6.2}
\end{equation}
\end{remark}
Let $\v^{(n)}$ be a unit vector corresponding to  the projection
${\E}^{(n)}(\k )$, ${\E}^{(n)}(\k)=(\cdot ,\v^{(n)})\v^{(n)}$,
$\v^{(n)}=\{ v^{(n)}_{\s}\}_{\s \in \Z^4}\in\ell^2(\Z^4)$. By
construction, $v^{(n)}_{\s}=0$ when $\s\not \in \Omega (r_{n-1})$.
Let us consider the linear combination of exponents corresponding to
this vector:
$$\Psi _n(\k,\x)=\sum _{\s \in \Omega (r_{n-1})} v^{(n)}_{\s}e^{i\langle\k+\p_{\s},\x\rangle}.$$
\begin{lemma}\label{prelimit} Function $\Psi _n(\k,\x),\ n\geq4,$ satisfies the equation:
$$-\Delta\Psi _n(\k ,\x)+V(\x)\Psi _n(\k,\x)=
\lambda _n(\k)\Psi _n(\k,\x)+g_n(\k,\x),$$
 the vector $\g _n$ of the Fourier coefficients of $g_n(\k,\x)$ satisfying the estimate:
\begin{equation}\|\g_n\|_{\ell
^1(\Z^4)}<k^{-k^{\frac 12 r_{n-1}}}.
\label{Aug13-2} \end{equation} Coefficients $(\g_n)_{\s}$
can differ from zero only when  $k^{r_{n-1}}<\||\p_\s\||\leq k^{r_{n-1}}+Q$. Function
$g_n(\k,\x)$ obeys the estimate:
\begin{equation}\label{g_n}\|g_n\|_{L_{\infty }(\R^2 )}<k^{-
k^{\frac 12 r_{n-1}}}.\end{equation} \end{lemma}
\begin{proof} Let $P(r_{n-1})$ be the projection in $\ell
^2(\Z^4)$ on the subspace corresponding to $\Omega (r_{n-1})$.
By  construction, $P(r_{n-1})\v^{(n)}=\v^{(n)}$ and
$$H_0\v^{(n)}+P(r_{n-1})VP(r_{n-1})\v^{(n)}=\lambda _n(\k )\v^{(n)}.$$
Since $V$ is a trigonometric polynomial,
$$(I-P(r_{n-1}))VP(r_{n-1})=(I-P(r_{n-1}))VP_{\partial }(r_{n-1}),$$
where $P_{\partial }(r_{n-1})_{\m \m}=1$ only if $\m$ is in the $Q$-vicinity of the boundary. Using \eqref{Feb6b-3last} with $n$ instead of $n+1$, we obtain:
 $\|P_{\partial
}(r_{n-1}){\E}^{(n)}\|<k^{- k^{r_{n-1}(1-\gamma)}}$ and, hence,
$\|P_{\partial }(r_{n-1}){\v^{(n)}}\|<k^{- k^{r_{n-1}(1-\gamma)}}$.
Therefore, $\|(I-P(r_{n-1}))VP(r_{n-1})\v^{(n)}\|<\|V\|k^{-
k^{r_{n-1}(1-\gamma)} }$. It follows that $H(\k)\v^{(n)}=\lambda
_n(\k)\v^{(n)}+\g_{n}$, where $\|\g_{n}\|_{\ell ^2(\Z^4)}<\|V\|k^{-
k^{ r_{n-1}(1-\gamma)} }$. Note that elements $(\g_n)_{\s}$ are
equal to zero when $ \||\p_\s\||\leq k^{r_{n-1}}$ or $ \||\p_\s\||>
k^{r_{n-1}}+Q$. Therefore, \eqref{Aug13-2} holds. Estimate
\eqref{g_n} follows.\end{proof}
\begin{lemma} Functions $\Psi _n(\k,\x)$ satisfy the inequalities:
\begin{equation}\label{psi1}
    \|\Psi _{1}-\Psi _0\|_{L_{\infty }(\R^2)}<ck^{-\gamma _0+2\delta },\ \ \
    \|\Delta (\Psi _{1}- \Psi _0)\|_{L_{\infty }(\R^2)}<ck^{-\gamma _0+2\delta +2},\ \ \
    \Psi _0(\x)=e^{i \langle \k ,\x
    \rangle},
    \end{equation}
    \begin{equation}\label{psi2}
    \begin{split}
&\|\Psi _{2}-\Psi _1\|_{L_{\infty
}(\R^2)}<ck^{-k^{\delta}(4Q)^{-1}+2r_1
    },\cr &\|\Delta (\Psi _{2}- \Psi _1)\|_{L_{\infty
}(\R^2)}<ck^{-k^{\delta}(4Q)^{-1}+4r_1
    },
    \end{split}
    \end{equation}
    \begin{equation}
    \begin{split}
    &\|\Psi _{n}-\Psi _{n-1}\|_{L_{\infty}(\R^2)}
    < ck^{-\frac1{10}\beta k^{r_{n-2}-r_{n-3}}+2r_{n-1} }, \ \cr &\|\Delta (\Psi _{n}- \Psi _{n-1})\|_{L_{\infty}(\R^2)}
    < ck^{-\frac1{10}\beta k^{r_{n-2}-r_{n-3}}+4r_{n-1} },\ \ n \geq 3,\ \ r_0:=\delta_*. \label{Dec10}
\end{split}
    \end{equation}
    \end{lemma}
        \begin{corollary} \label{C:Psi}All functions $\Psi _{n}$, $n=0,1,...,$ obey the estimate $ \|\Psi _{n}\|_{L_{\infty
}(\R^2)}<1+Ck^{-\gamma _0+2\delta }$ uniformly in $n$. \end{corollary}
    \begin{proof} Using \eqref{6.2.2} and considering that
    $v^{(n)}_{\s}$ are equal to zero when $\||\p_\s \||>k^{r_{n-1}}$, we obtain
    \begin{equation}
    \begin{split}&\|\v^{(1)}-\v^{(0)}\|_{\ell ^{1}(\Z^4)}<ck^{-\gamma _0+2\delta },\
    \ \|\v^{(2)}-\v^{(1)}\|_{\ell ^{1}(\Z^4)}<ck^{-k^{\delta}(4Q)^{-1}+2r_1
    },\cr &
    \|\v^{(n)}-\v^{(n-1)}\|_{\ell ^{1}(\Z^4)}<ck^{-\frac1{10}\beta k^{r_{n-2}-r_{n-3}}+2r_{n-1}},\ \ n\geq3.
    \end{split}\label{Cauchy-1} \end{equation}
    \begin{equation}
    \begin{split}&\|H_0(\v^{(1)}-\v^{(0)})\|_{\ell ^{1}(\Z^4)}<ck^{-\gamma _0+2\delta +2},\ \
    \|H_0(\v^{(2)}-\v^{(1)})\|_{\ell ^{1}(\Z^4)}<ck^{-k^{\delta}(4Q)^{-1}+4r_1
    },\cr &\|H_0(\v^{(n)}-\v^{(n-1)})\|_{\ell ^{1}(\Z^4)}<ck^{-\frac1{10}\beta k^{r_{n-2}-r_{n-3}}+4r_{n-1}},\ \
    n\geq3,\ \ r_0:=\delta_*.\label{Cauchy-2}\end{split} \end{equation}
Now \eqref{psi1} -- \eqref{Dec10} easily follow. \end{proof}

\begin{theorem} \label{T:Dec10}
For every  $\lambda >k_*^{2}$ and $\k \in {\cal
D}_{\infty}(\lambda)$ the sequence of functions $\Psi _n(\k,\x)$
converges  in $L_{\infty}(\R^2)$ and $W_{2,loc}^{2}(\R^2)$. The
limit function $\Psi _{\infty }(\k,\x):=\lim_{n\to \infty }
    \Psi _n(\k,\x)$, is a quasi-periodic function:
  \begin{equation} \label{quasi-periodic}
  \Psi _{\infty }(\k,\x)=\sum _{\s \in \Z^4}  (v_{\infty})_{\s}e^{i\langle\k +\p_\s,\x\rangle}, \end{equation}
  where $\v_{\infty }=\{(v_\infty)_\s\}_{\s\in\Z^4}\in \ell^1(\Z^4)$ and $\|\v_{\infty }\|_{\ell^2 (\Z^4)}=1$.
  The function $\Psi _{\infty }(\k,\x)$   satisfies the equation
    \begin{equation}\label{6.7}
     -\Delta\Psi _{\infty }(\k, \x)+V(\x)\Psi _{\infty }(\k,
    \x)= \lambda \Psi _{\infty }(\k, \x).
    \end{equation}
 It can be represented in the form
   \begin{equation}\label{6.4}
    \Psi _{\infty }(\k,\x)=e^{i\langle \k, \x
    \rangle}\bigl(1+u_{\infty}(\k, \x)\bigr),
    \end{equation}
where $u_{\infty}(\k, \x)$ is a quasi-periodic
function:
   \begin{equation}\label{6.5}
    u_{\infty}(\k, \x)=\sum_{n=1}^{\infty} u_n(\k,
    \x),
    \end{equation}
 \begin{equation}u_n(\k, \x)=\sum _{\s \in \Omega (r_{n-1})}
 (v^{(n)}_{\s}-v^{(n-1)}_{\s})e^{i\langle\p_\s,\x\rangle},\label{u_n} \end{equation}
 functions $u_n$ satisfying the estimates:
\begin{equation}
    \|u_1\|_{L_{\infty}(\R^2)} <ck^{-\gamma_0+2\delta }, \ \
    \|u_2\|_{L_{\infty}(\R^2)} <ck^{-k^{\delta}(4Q)^{-1}+2r_1
    },\label{6.6a}
    \end{equation}
    \begin{equation}\label{6.6}
    \|{u}_n\|_{L_{\infty}(\R^2)} <ck^{-\frac1{10}\beta k^{r_{n-2}-r_{n-3}}+2r_{n-1} }, \ \ \ n\geq
    3,\ \ r_0:=\delta_*.\end{equation}
\end{theorem}
\begin{corollary} Function $u_{\infty}(\k, \x)$ obeys the
 estimate (\ref{6.2.1a}).
\end{corollary}

\begin{proof} Using \eqref{Cauchy-1},\eqref{Cauchy-2}, we obtain that the sequence $\v^{(n)}$ has the limit in $\ell ^1(\Z^4)$. We denote this limit by $\v_{\infty }$. Since, vectors $\v^{(n)}$ are normalized in $\ell^ {2}(\Z ^4)$,
\begin{equation}
\|\v_{\infty }\|_{\ell ^2(\Z^4)}=1.\label{norm} \end{equation} By
\eqref{Dec10}, we obtain that $\Psi _n(\k,\x)$ is a Cauchy sequence
in $L_{\infty }(\R^2)$ and $W_{2,loc}^{2}(\R^2)$. Let $\Psi _{\infty }(\k
,\x)=\lim_{n \to \infty}\Psi _n(\k, \x).$ This limit is
defined pointwise uniformly in $\x$ and in $W_{2,loc}^{2}(\R^2)$.
Noting also that $\lim \lambda _n(\k )=\lambda $, and
taking into account Lemma~\ref{prelimit} we obtain that \eqref{6.7}
holds.

 Let us show that $\Psi _{\infty }$ is a
quasi-periodic function. Obviously,
    $$ \Psi _{\infty } =\Psi _0+\sum_{n=1}^{\infty}(\Psi_{n}-\Psi
    _{n-1}),
$$ the series converging in $L_{\infty}(\R^2)$ by (\ref{Dec10}).
Introducing $u_{n}:=e^{-i \langle \k ,\x\rangle}(\Psi_{n}-\Psi
_{n-1}),$ we arrive at (\ref{6.4}), (\ref{6.5}). Note that
 $u_n$ has a form \eqref{u_n} Estimates (\ref{6.6a}), (\ref{6.6})  follow from (\ref{psi1}), (\ref{Dec10}).
 \end{proof}

\begin{theorem} Formulae (\ref{6.2.1}), (\ref{6.2.1a}) and (\ref{6.2.4}) hold for every $\k \in \cal{G} _{\infty }$.
The set $\cal{G} _{\infty }$ is Lebesgue measurable and satisfies
  (\ref{full})  with $\gamma _3=37\mu\delta $.
\end{theorem}
\begin{proof}
By Theorem~\ref{T:Dec10}, (\ref{6.2.1}), (\ref{6.2.1a}) hold, where
$\lambda _{\infty}(\k )=\lambda $ for $\k \in {\cal
D}_{\infty}(\lambda )$. Using (\ref{h}), which is proven in
Corollary~\ref{Dec18}, with $\varkappa_{\infty }=|\k |$, we easily
obtain (\ref{6.2.4}). It remains to prove (\ref{full}). Let us
consider a small region $U_n(\lambda _0)$ around an isoenergetic
surface ${\cal D}_n(\lambda _0)$, $\lambda _0>k_*^{2}$. Namely,
$U_n(\lambda _0)=\cup _{|\lambda -\lambda _0|<\varepsilon
_0^{(n+1)}} {\cal D}_{n}(\lambda )$, $k=\lambda _0^{1/2}$. By
Theorem \ref{Thm3last} the construction of the $n$-th non-resonant
set is stable in $\varepsilon _0^{(n)}$ -neighborhood of $\lambda
_0$. Therefore, in fact, we can (and for the sake of convenience
will) assume that the sets $\omega^{(n)}(\lambda)$ are chosen to be
equal to $\omega^{(n)}(\lambda_0)$ for
$|\lambda-\lambda_0|<\varepsilon _0^{(n)}$ . Thus, $U_n(\lambda _0)$
is an open set (a distorted ring with holes) and $|U_n(\lambda
_0)|=\varepsilon _0^{(n)} 2\pi
\left(1+O(k^{-37\mu\delta})\right)$. It easily follows from
\eqref{dk1} and \eqref{dk0-3IVlast} that $U_{n+1}\subset U_n$.
Definition of ${\cal D}_{\infty }(\lambda _0)$ yield: ${\cal
D}_{\infty }(\lambda _0)=\cap _{n=1}^{\infty }U_n(\lambda _0)$.
Hence, ${\cal G}_{\infty }=\cap _{n=1}^{\infty }{\cal G}_n$, where
 \begin{equation}{\cal G}_n=\cup _{\lambda
>\lambda _*}{\cal D}_n(\lambda ), \ \ \lambda _*=k_*^{2}.\label{Gn} \end{equation} Considering that $U_{n+1}\subset U_n$
for every $\lambda _0>\lambda _*$, we obtain  ${\cal G}_{n+1}\subset
{\cal G}_n$. Hence, $\left|{\cal G} _{\infty }\cap
        \bf B_R\right|=\lim _{n\to \infty }\left|{\cal G} _{n}\cap
        \bf B_R\right|$. Calculating the volume of the region
        $\cup_{\lambda_*<\lambda<R^{2}}U_{n}(\lambda)$, we easily conclude $\left|{\cal G} _{n}\cap
        \bf B_R\right|=|{\bf B_R}|\left(1+O(R^{-37\mu\delta})\right)$ uniformly in $n$. Thus, we have
        obtained (\ref{full}) with $\gamma _3=37\mu\delta $.
\end{proof}

\begin{theorem}[Bethe-Sommerfeld Conjecture]
The spectrum of  operator $H$ contains a semi-axis.
\end{theorem}
\begin{proof}
    The theorem immediately follows from the fact that the equation \eqref{6.7} has a bounded solution $\Psi _{\infty }(\k ,\x)$ for every  sufficiently large $\lambda $.
\end{proof}

\section{Proof of Absolute Continuity of the Spectrum}\label{chapt8}
The proof is somewhat analogous to that for the case of limit-periodic
potentials  \cite{KL2}. We will just refer to  \cite{KL2} in some places. We also note that proofs
 of absolutely continuous spectrum through establishing localization in the momentum space have been done for 1D operators in the past (see \cite{AvJit}, \cite{BouJit}).
\subsection{Operators $E_n(\mathcal{G}_n')$,
$\mathcal{G}_n'\subset \mathcal{G}_n$}

Let us consider the open sets $\mathcal{G}_n$ given by (\ref{Gn}).
 There is a family of  eigenfunctions $\Psi _n(\k ,\x)$, $\k \in \mathcal{G}_n$, of the operator
    $H^{(n)}$, which are described by the perturbation formulas
    (\ref{na}), (\ref{na-n}). Let, $\mathcal{G}_n'\subset \mathcal{G}_n$, where $\mathcal{G}_n'$ is Lebesgue measurable and bounded.
    Let
    \begin{equation} E_n\left( \mathcal{G}'_n\right)F=\frac{1}{4\pi ^2}\int
    _{ \mathcal{G}'_n}\bigl( F,\Psi _n(\k )\bigr) \Psi _n(\k) d\k \label{s}
    \end{equation}
    for any $F\in C_0^{\infty}(\R^2)$, here and below $\bigl( \cdot ,\cdot \bigr)$
    is the canonical scalar product in $L_2(\R^2)$, i.e.,
    $$\bigl( F,\Psi _n(\k )\bigr)=\int _{\R^2}F(\x)\overline{\Psi _n(\k ,\x)}d\x.$$
    We will show that $ E_n\left( \mathcal{G}'_n\right)$ is almost a projector in $L_2(\R^2)$ in the sense: $ E_n\left( \mathcal{G}'_n\right)= E_n^*\left( \mathcal{G}'_n\right)$, $ E_n^2\left( \mathcal{G}'_n\right)= E_n\left( \mathcal{G}'_n\right)+o(1)$, where $o(1)$ is in the class of bounded operators as $n\to \infty $.
First, we note that \eqref{s} can be rewritten in the form:
    \begin{equation} E_n\left(\mathcal{G}'_n\right)=S_n\left(\mathcal{G}'_n\right)T_n \left(
    \mathcal{G}'_n\right), \label{ST}
    \end{equation}
    $$T_n:L_2(\R^2) \to L_2\left(  \mathcal{G}'_n\right), \ \
    \ \ S_n:L_2\left( \mathcal{G}'_n\right)\to L_2(\R^2),$$
    \begin{equation}
    T_nF=\frac{1}{2\pi }\bigl( F,\Psi _n(\k )\bigr) \mbox{\ \ for any $F\in C_0^{\infty}(\R^2)$},
    \label{eq2}
    \end{equation}
    $T_nF$ being in $L_{\infty }\left(  \mathcal{G}'_n\right)$, and,
    \begin{equation}S_nf = \frac{1}{2\pi }\int _{  \mathcal{G}'_n}f (\k)\Psi _n(\k ,\x)d\k  \mbox{\ \ for any $f \in L_{\infty }\left(
    \mathcal{G}'_n\right)$.} \label{ev}
    \end{equation}
    Note that $S_nf \in L_2(\R^2)$, since $\Psi _n$ is a finite combination of exponentials $e^{i\langle\k +\p_{\q},\x\rangle}$.

   \begin{lemma}\label{9.1} Let $\mathcal{G}'_n$ be  bounded and $f(\k), g(\k)\in L_{\infty }\left(
    \mathcal{G}'_n\right)$. Then,
    \begin{equation} \left( S_n f, S_n g \right)_{L_2(\R^2)}=_{n\to \infty}
    \left(f, g \right)_{L_2(\mathcal{G}'_n)}+o(1)\|f\|_{L_2(\mathcal{G}'_n)}\|g\|_{L^2(\mathcal{G}'_n)}. \label{May14a-12}\end{equation}
    where $o(1)$ goes to zero uniformly in $f$, $g$  and $\mathcal{G}'_n$ as $n\to \infty $; namely, $|o(1)|<\xi_*^{-{r_{n-3}(\xi_*)}}$, where  $\xi _*=\inf _{\vec \xi \in \mathcal{G}'_n}|\vec \xi |$.
\end{lemma}
\begin{corollary} The following relation holds:
     \begin{equation} \left|\left( S_n f, S_n g \right)_{L_2(\R^2)}\right|<(1+o(1))
     \left\|f\right\|_{L_{\infty}(\mathcal{G}'_n)}\left\|g\right\|_{L_{\infty}(\mathcal{G}'_n)}|\mathcal{G}_n'| \label{Vienna-1},
     \end{equation}
     where $|\mathcal{G}_n'|$ is the Lebesgue measure of $\mathcal{G}_n'$.
     \end{corollary}
\begin{corollary} The operator $S_n$ is bounded and $\|S_n\|=_{n\to \infty} 1+o(1)$. \end{corollary}
\begin{proof}
     The function $\Psi _n(\k,\x)$  can be
    represented as a  combination of plane waves:
    \begin{equation}\Psi _n(\k,\x)=
    \sum _{\m\in \Z^4}v_{\m}^{(n)}(\k)
    \exp\{ i\langle \k+\p _{\m},\x\rangle\},\label{++}
    \end{equation}
where $v_{\m}^{(n)}(\k)$ are Fourier coefficients. By construction,
$v_{\m}^{(n)}(\k)=0$, when $\m \not \in \Omega (r_{n-1}) $. Let
$\v^{(n)}(\k)$ be the vector in $\ell^2(\Z^2)$ with  components
equal to $v_{\m}^{(n)}(\k)$. Note that the size of $\Omega
(r_{n-1})$ depend on $\varkappa=|\k|$; to stress this fact we will
use here the notations $\Omega (r_{n-1},\varkappa)$ and
$r_{n-1}(\varkappa)$.
The
Fourier transform  $\widehat \Psi_n$ is a combination of $\delta
$-functions:
    $$\widehat \Psi _n(\k,\vec \xi)=2\pi \sum _{\m \in \Z^4}v_{\m}^{(n)}(\k)
    \delta \bigl(\vec \xi +\k+\p_{\m}\bigr)$$
From this, we easily compute  the Fourier
    transform of $(S_nf )(\x)$:
    $$ (\widehat{S_nf})(\vec \xi)=\sum _{\m \in \Z^4}v_{\m}^{(n)}\bigl(-\vec \xi-
    \p_{\m}\bigr)f \bigl(-\vec \xi-
    \p_{\m}\bigr)\chi \bigl({\cal G}_{n}',-\vec \xi-
    \p_{\m}\bigr),$$
where $\chi ({\cal G}_{n}',\cdot ) $ is the
characteristic function of ${\cal G}_{n}'$. Note that $v_{\m}^{(n)}\bigl(-\vec \xi-
    \p_{\m}\bigr)\chi \bigl({\cal G}_{n}',-\vec \xi-
    \p_{\m}\bigr)$ can differ from zero only when $\m \in \Omega (r_{n-1}, |\vec \xi+
    \p_{\m}|)\subset  \Omega (r_{n-1}, \xi _{**})$, $\xi _{**}=\sup _{\vec\xi \in {\cal G}_{n}'}|\vec\xi |$.
    By Parseval's identity,
    $$
    \left(S_nf,S_ng\right)_{L_2(\R^2)}=\left(\widehat{S_nf},\widehat{S_ng}\right)_{L_2(\R^2)}=$$ $$\int _{\R^2 }\sum _{\m , \m'\in \Z^4} T_{\m,\m'}(\vec \xi )f \bigl(-\vec \xi-
    \p_{\m}\bigr)\bar g \bigl(-\vec \xi-
    \p_{\m'}\bigr)\chi \bigl({\cal G}_{n}',-\vec \xi-
    \p_{\m}\bigr)\chi \bigl({\cal G}_{n}',-\vec \xi-
    \p_{\m'}\bigr)d\vec \xi ,$$ $$
   T_{\m,\m'}(\vec \xi ):=v_{\m}^{(n)}\bigl(-\vec \xi-
    \p_{\m}\bigr)\overline{v_{\m '}^{(n)}}\bigl(-\vec \xi-
    \p_{\m '}\bigr).$$
    Note that, in fact, the summation here is over the finite set $\m,\m' \in \Omega (r_{n-1}, \xi _{**})$. Hence we can exchange summation and integration in the above formula.
    Next, shifting the variable $\vec \xi+
    \p_{\m} \to \vec \xi $, denoting $\m'-\m$ by $\m''$ and considering that $\langle\v^{(n)},\v^{(n)}\rangle=1$, we obtain:
    $$\left(\widehat{S_nf},\widehat{S_ng}\right)_{L_2(\R^2)}=\left(f,g\right)_{L_2({\cal G}_{n}')}+$$
    \begin{equation}\sum _{\m''\in \Z^4\setminus\{\bf 0\}}\int _{\R^2 }B_{\m''}(\vec \xi )f \bigl(-\vec \xi\bigr)\bar g \bigl(-\vec \xi-
    \p_{\m''}\bigr)\chi \bigl({\cal G}_{n}',-\vec \xi \bigr)\chi \bigl({\cal G}_{n}',-\vec \xi-
    \p_{\m''}\bigr)\ d\vec \xi , \label{May14-12}
    \end{equation}
     $$B_{\m''}(\vec \xi )= \sum _{\m\in \Z^4}v_{\m}^{(n)}\bigl(-\vec \xi\bigr)\overline{v_{\m +\m''}^{(n)}}\bigl(-\vec \xi-
    \p_{\m ''}\bigr).$$
    Obviously,
    \begin{equation}B_{\m ''}=\langle\v^{(n)}(-\vec \xi ), \v^{(n)}_*(-\vec \xi-\p_{\m''})\rangle ,\label{Aug11-2} \end{equation}
    where $\v_*^{(n)}\bigl(-\vec \xi-\p_{\m''}\bigr)$ the ``shifted" eigenvector:
    $\v_*^{(n)}\bigl(-\vec \xi-\p_{\m''}\bigr)$: $\left(\v_*^{(n)}\bigl(-\vec \xi-\p_{\m''}\bigr)\right)_{\m}=v^{(n)}_{\m+\m''}\bigl(-\vec \xi-\p_{\m''}\bigr).$
   To obtain \eqref{May14a-12}, it is enough to prove two estimates for $n\geq 4$:
   \begin{equation} \label{Aug9-2}\sum _{\||\p_{\m''}\||>\xi_*^{ r_{n-3}(\xi _*)}}\ \ \sup _{\vec \xi \in {\cal G}_{n}'}\left| B_{\m''}(\vec \xi )\right| <\frac 12 \xi _*^{- r_{n-3}(\xi _*)}, \end{equation}
\begin{equation} \label{Aug9-2*}\sum _{0<\||\p_{\m''}\||\leq {\xi _*^{ r_{n-3}(\xi _*)}}}\ \ \sup _{\vec \xi \in {\cal G}_{n}'}\left| B_{\m''}(\vec \xi )\right| <\frac 12 \xi _*^{- r_{n-3}(\xi _*)}. \end{equation}
To prove \eqref{Aug9-2} we first check that
    \begin{equation}
    \sup _{\vec \xi \in {\cal G}_{n}'}\left| B_{\m''}(\vec \xi )\right| <\||\p_{\m''}\||^{-8}
    \ \ \ \mbox{when} \ \  \||\p_{\m''}\||>\xi _*^{ r_{n-3}(\xi _*)}.\label{Aug8-1}\end{equation}
    Indeed, for every $\m''$ we break ${\cal G}_{n}'$ into several parts, partition being dependent on $\m''$:
    $ {\cal G}_{n}'=\cup _{s,s'=0}^{n}{\cal G}_{ss'},$
\begin{equation}\label{Aug8-2}
\begin{split}&
    {\cal G}_{ss'}=\cr &
    \left\{\vec \xi \in {\cal G}_{n}':\ |\vec \xi|^{r_{s-1}(|\vec \xi | )}\leq \frac 12
    \||\p_{\m''}\||<\gamma _s |\vec \xi|^{r_{s}(|\vec \xi | )
    }\right\}\cap\cr &
    \left\{\vec \xi \in {\cal G}_{n}':\ |\vec \xi+\p_{\m''}|^{r_{s'-1}(|\vec \xi +\p_{\m''}|)}\leq \frac 12
    \||\p_{\m''}\||< \gamma _{s'} |\vec \xi+\p_{\m''}|^{r_{s'}(|\vec \xi +\p_{\m''}|)}\right\},
\end{split}
\end{equation}
     where $r_{-1}:=0$, $r_0:=\delta $, $\gamma _s=1$ when $s<n$, $\gamma _n=\infty $. To prove \eqref{Aug8-1}, it is enough to show
     \begin{equation}
    \sup _{\vec \xi \in {\cal G}_{ss'}}\left| B_{\m''}(\vec \xi )\right| <\||\p_{\m''}\||^{-8}\label{Aug8-3}\end{equation}
    for all $s,s'$. Assume $s,s'=n$. It follows from \eqref{Aug8-2} that for any $\m \in \Z^4$ either
     $v_{\m }^{(n)}(\vec \xi )$ or  $v_{\m +\m''}^{(n)}(\vec \xi +\p_{\m''})$ is zero.   Hence, $\langle\v^{(n)}(-\vec \xi ), \v^{(n)}_*(-\vec \xi-\p_{\m''})\rangle =0$, i.e.,
    $B_{\m''}(\vec \xi )=0$. Next, let $0<s<n$, $s'=n$. By \eqref{perturbation*-3last},
    \begin{equation}\|\v^{(n)}(\vec \xi )-\v^{(s)}(\vec \xi )\|<|\vec \xi |^{-|\vec \xi |^{\frac 12r_{s-1}(|\vec \xi |)}}. \label{Aug8-4} \end{equation} It follows from the definition of
    ${\cal G}_{sn}$ that $\langle \v^{(s)}(\vec \xi ),\v^{(n)}_*(\vec \xi +\p_{\m''})\rangle =0$. Therefore,
    \begin{equation}
    \left| B_{\m''}(\vec \xi )\right| \leq \|\v^{(n)}(\vec \xi )-\v^{(s)}(\vec \xi )\|<|\vec \xi |^{-|\vec \xi |^{\frac 12 r_{s-1}(|\vec \xi |)}}\ \ \mbox{when  } \vec \xi \in {\cal G}_{sn}.
    \label{Aug8-5} \end{equation}
    Using \eqref{Aug13-1}, \eqref{r_2} and  \eqref{indrn}, we obtain
    $\left| B_{\m''}(\vec \xi )\right| \leq |\vec \xi |^{-10 r_{s}(|\vec \xi| )}$. Considering again the definition of
    ${\cal G}_{sn}$ , we get
    \eqref{Aug8-3}. Next, we consider ${\cal G}_{0n}$. By \eqref{perturbation*},
    $\v^{(1)}=\v^{(0)}+O(|\vec \xi |^{-1+44\mu \delta })$, where $v^{(0)}_{\m}=\delta _{\m,{\bf 0}}$. By \eqref{Aug8-2},
  $\langle \v^{(0)}(\vec \xi ),\v^{(n)}(\vec \xi +\p_{\m''})\rangle =0$. Hence,  $\left| B_{\m''}(\vec \xi )\right| \leq C|\vec \xi |^{-1+44\mu \delta }$. Using again the definition of ${\cal G}_{0n}$ and the inequality $1-44\mu \delta >8\delta $, we obtain \eqref{Aug8-3}. The case $s'<n$ is considered in the analogous way. Thus, \eqref{Aug8-3} is proved. Summarizing \eqref{Aug8-3} over $\m''$, we obtain \eqref{Aug9-2}.

    Suppose $0<\||\p_{\m''}\||\leq \xi _*^{r_{n-3}(\xi _*)}$. Let us estimate $B_{\m''}(\vec \xi )$. Assume for definiteness that $|\vec \xi +\p_{\m''}|\leq |\vec \xi |$. The case of the opposite inequality is analogous up to the change of the notation $\vec \xi \to \vec \xi +\p_{-\m''}$, since $B_{\m''}(\vec \xi )=B_{-\m''}(\vec \xi +\p_{\m''})$.
    By \eqref{++},
\begin{equation}H^{(n)}(-\vec \xi )\v^{(n)}\bigl(-\vec \xi\bigr)=\lambda ^{(n)}\bigl(-\vec \xi\bigr)\v^{(n)}\bigl(-\vec \xi).
\label{Vienna-4+} \end{equation}
The analogous relation holds for $\v^{(n)}\bigl(-\vec \xi-
    \p_{\m ''}\bigr) $ up to the replacement of $H^{(n)}(-\vec \xi )$ by $ H^{(n)}\left(-\vec \xi - \p_{\m ''}\right)$ and $\lambda ^{(n)}\bigl(-\vec \xi\bigr)$ by $\lambda ^{(n)}\bigl(-\vec \xi-\p_{\m''}\bigr)$:
      \begin{equation}H^{(n)}(-\vec \xi -\p_{\m''})\v^{(n)}\bigl(-\vec \xi-\p_{\m''}\bigr)=\lambda ^{(n)}\bigl(-\vec \xi-\p_{\m''}\bigr)\v^{(n)}\bigl(-\vec \xi-\p_{\m''}).\label{Vienna-4"+} \end{equation}
 Note that $ H^{(n)}\left(-\vec \xi - \p_{\m ''}\right)$ up to the shift of indices by $-\m''$ is equivalent to
    the operator $P_{\m''}H(-\vec \xi )P_{\m''}$, where $P_{\m''}$ is the projection onto the box of the size
    $|\vec \xi +\p_{\m ''}|^{r_{n-1}(|\vec \xi +\p_{\m ''}|)}$ around $-\m''$. Using  the shifted eigenvector
    $\v_*^{(n)}\bigl(-\vec \xi-\p_{\m''}\bigr)$,
     we can rewrite \eqref{Vienna-4"+}  in the form:
    \begin{equation}P_{\m ''}H(-\vec \xi )P_{\m ''}\v_*^{(n)}\bigl(-\vec \xi-\p_{\m''}\bigr)=
    \lambda ^{(n)}\bigl(-\vec \xi-\p_{\m''}\bigr)\v_*^{(n)}\bigl(-\vec \xi-\p_{\m''}),\label{Vienna-4'+} \end{equation}
    where $P_{\m ''}\v_*^{(n)}=\v_*^{(n)}$.
    By \eqref{Aug13-2},
    \begin{equation}H(-\vec \xi )\v^{(n)}\bigl(-\vec \xi\bigr)=\lambda ^{(n)}\bigl(-\vec \xi\bigr)\v^{(n)}\bigl(-\vec \xi)+O\left(|\xi |^{-
|\xi |^{\frac{1}{2}r_{n-1}(|\vec \xi
|)}}\right).\label{Vienna-4BHM+}
\end{equation} Similarly,
\begin{equation}\begin{split}& H(-\vec \xi )\v_*^{(n)}\bigl(-\vec
\xi-\p_{\m''}\bigr)=\cr &
    \lambda ^{(n)}\bigl(-\vec \xi-\p_{\m''}\bigr)\v_*^{(n)}\bigl(-\vec \xi-\p_{\m''})+O\left(|\xi +\p_{\m''}|^{-
|\xi +\p_{\m''}|^{\frac{1}{2}r_{n-1}(|\vec \xi+\p_{\m''}
|)}}\right).\end{split}\label{Aug10-2+}
\end{equation} Assume first $|\vec \xi +\p_{\m''}|\leq \frac 12 |\vec \xi |$. Then
$|\lambda ^{(n)}\bigl(-\vec \xi\bigr)-\lambda ^{(n)}\bigl(-\vec \xi
-\p_{\m''}\bigr)|> \frac{1}{2}|\vec \xi |^{2}$. Using
\eqref{Aug11-2}, \eqref{Vienna-4BHM+} and
     \eqref{Aug10-2+}, we obtain:
     \begin{equation} \label{Aug11-1}B_{\m ''}=O\left(\xi_*^{-
\xi _*^{\frac{1}{2}r_{n-1}(\xi _*) }}\right)\ \ \mbox{when  } |\vec
\xi +\p_{\m''}| \leq \frac 12 |\vec \xi |. \end{equation}
     Similar, but somewhat more subtle considerations are required when $|\vec \xi +\p_{\m''}|>\frac 12 |\vec \xi |$. We start with introducing a parameter $s$. We will use it to
      cut $\Omega (r_{n-1}, |\vec \xi |)$ to approximately the same size as $\Omega (r_{n-1}, |\vec \xi +\p_{\m''}|)$.
  If the boxes are of approximately  the same size, then $s=n-1$. Indeed, for each $\vec \xi $ one of the following relations holds:
\begin{equation}
  \label{Aug10-5}|\vec \xi |^{r_{s-1}(|\vec \xi |)}\leq
 |\vec \xi +\p_{\m''} |^{r_{n-1}(|\vec \xi +\p_{\m''} |)}
 <|\vec \xi |^{r_{s}(|\vec \xi |)},
  \end{equation}
  where $1\leq s\leq n-1$ and $s$ is defined by $\m ''$ and $\vec \xi $. Note that $s<n-1$ when $\Omega (r_{n-1}, |\vec \xi |)$ is essentially bigger than $\Omega (r_{n-1}, |\vec \xi +\p_{\m''}|)$.
  Using the second inequality in \eqref{Aug10-5} and \eqref{indrn}, we get
  \begin{equation}|\vec \xi +\p_{\m''} |^{r_{n-3}(|\vec \xi +\p_{\m''} |)}<\frac18 |\vec \xi |^{r_{s-1}(|\vec \xi |)}. \label{Aug11-3} \end{equation}
  Let  $P_{{\bf 0}s}$ be the projecting corresponding to $\Omega (r_s, |\vec \xi |)$.
 By \eqref{perturbation*-3last} with $s$ instead of $n$, \begin{equation}
 (I-P_{{\bf 0}s})\v^{(n)}(-\vec \xi )=O\left(|\vec \xi | ^{-
|\vec \xi | ^{\frac 12 r_{s-1}({|\vec \xi |})}}\right).\label{Spb-1+} \end{equation}
    Let us prove the analogous estimate for $\v^{(n)}_*$:
   \begin{equation} \label{Spb-1}(I-P_{{\bf 0}s})\v^{(n)}_*(-\vec \xi-
    \p_{\m ''})=O\left(|\vec \xi | ^{-
|\vec \xi | ^{\frac 12 r_{s-1}({|\vec \xi
|})}}\right).\end{equation} Indeed, if $(P_{{\bf 0}s})_{\m \m}=0$,
then $\||\p_\m\||> |\vec \xi |^{r_{s-1}(|\vec \xi|)}$.  Using
\eqref{Aug11-3} and  the bound on $\||\p_{\m ''}\||$, we obtain
$\||\p_{\m+\m''}\||> \frac 12 |\vec \xi |^{r_{s-1}(|\vec \xi|)}$.
Using \eqref{perturbation*-3last}, \eqref{Feb6b-3last}, we obtain
\eqref{Spb-1}. From \eqref{Vienna-4BHM+},\eqref{Aug10-2+},
considering that $\|P_{{\bf 0}s}H\|= O\left( |\xi |^{2r_{s-1}(|\vec
\xi|)}\right)$ and using \eqref{Spb-1+},\eqref{Spb-1}, we get:
\begin{equation}H^{(s)}(-\vec \xi )P_{{\bf 0}s}\v^{(n)}\bigl(-\vec \xi\bigr)=\lambda ^{(n)}\bigl(-\vec \xi\bigr)P_{{\bf 0}s}\v^{(n)}\bigl(-\vec \xi)+O\left(|\xi |^{-
|\xi |^{\frac{1}{4}r_{s-1}}}\right),\label{Vienna-4BHM} \end{equation}
     \begin{equation} \label{Vienna-5}H^{(s)}(-\vec \xi )P_{{\bf 0}s}\v^{(n)}_*\bigl(-\vec \xi-
    \p_{\m ''}\bigr)=\lambda ^{(n)}\bigl(-\vec \xi-
    \p_{\m ''}\bigr)P_{{\bf 0}s}\v^{(n)}_*\bigl(-\vec \xi-
    \p_{\m ''}\bigr)+O\left(|\vec \xi | ^{-
|\vec \xi | ^{\frac {1}{4}r_{s-1}}}
\right).
\end{equation}
Next, by Theorem \ref{Thm3last} for step $s$, $\left|\lambda
^{(n)}\bigl(-\vec \xi \bigr)-\lambda ^{(n)}\bigl(-\vec \xi-
    \p_{\m ''}\bigr)\right|>\varepsilon _0^{(s)}/2$, where $\varepsilon _0^{(s)}=|\vec \xi | ^{-2r_{s-1}'|\vec \xi |^{2\gamma r_{s-2}}}.$
    Indeed, $\v^{(n)}(-\vec \xi )$ and $\v^{(n)}_*\bigl(-\vec \xi-
    \p_{\m ''}\bigr)$ are almost orthogonal since they are concentrated around $\m={\bf 0}$ and $\m=\m''\not={\bf 0}$ respectively; thus
    $\lambda ^{(n)}\bigl(-\vec \xi-
    \p_{\m ''}\bigr)$ must be outside of the interval described in Theorem \ref{Thm3last}, while $\lambda ^{(n)}\bigl(-\vec \xi \bigr)$
    is inside twice shorter interval. Now, using \eqref{Vienna-4BHM} and \eqref{Vienna-5}, we obtain:
    $$\langle P_{{\bf 0}s}\v^{(n)}(-\vec \xi ), P_{{\bf 0}s}\v^{(n)}_*(-\vec \xi-\p_{\m''})\rangle= O\left(|\vec \xi | ^{-
|\vec \xi |^{\frac{1}{4}r_{s-1}}}\right)
|\vec \xi | ^{2r_{s-1}'|\vec \xi |^{2\gamma r_{s-2}}}=O\left(|\vec \xi |^{-
|\vec \xi |^{\frac{1}{8}r_{s-1} }}\right),$$
see \eqref{indrn}.  Using one more time \eqref{Spb-1+},
\eqref{Spb-1}, and considering \eqref{Aug11-3}, we obtain $B_{\m
''}=O\left( \xi _*^{-\xi _*^{\frac18 r_{n-3}(\xi _*) }}\right)$  for
the case  $|\vec \xi +\p_{\m''}|>\frac 12 |\vec \xi |$. Using this
estimate together with \eqref{Aug11-1} and considering that the
number of $\m''$ satisfying $0<\||\p_{\m''}\||\leq \xi
_*^{r_{n-3}(\xi _*)}$ does not exceed $16\xi _*^{4r_{n-3}(\xi _*)}$,
we obtain \eqref{Aug9-2*}. Substituting the estimates for $B_{\m
''}$ into \eqref{May14-12}, we obtain  \eqref{May14a-12}.


\end{proof}

It is easy to see that $T_n\subset S_n^*$. Therefore, $\|T_n\|\leq 1+o(1)$ and can be extended to the whole space $L_2({\cal G}_n)$. We still denote the extended operator by $T_n$, $T_n=S_n^*$. Therefore, $E_n$ is a self-adjoint operator.
\begin{lemma} \label{L:10.4}  Let $\mathcal{G}_n'\subset \mathcal{G}_n''\subset \mathcal{G}_n$. The following relation holds as $n\to \infty $:
\begin{equation} E_n(\mathcal{G}_n')E_n(\mathcal{G}_n'')=E_n(\mathcal{G}_n')+o(1), \label{Vienna-6} \end{equation}
where $o(1)$ is taken in the space of bounded operators and uniform in $\mathcal{G}_n'$, $\mathcal{G}_n''$. \end{lemma}
\begin{corollary} \label{C:10.4-1} $E_n(\mathcal{G}_n'')E_n(\mathcal{G}_n')=E_n(\mathcal{G}_n')+o(1).$\end{corollary}
This corollary is valid, since $E_n$ is selfajoint.
\begin{corollary} \label{C:10.4-2}  $E_n^2(\mathcal{G}_n')=E_n(\mathcal{G}_n')+o(1)$ for any $\mathcal{G}_n'\subset \mathcal{G}_n$. \end{corollary}

\begin{proof}  Let $I_n(\mathcal{G}_n')$ be the projection from $L_2(\mathcal{G}_n'')$ to $L_2(\mathcal{G}_n')$. It is easy to see that $T_n(\mathcal{G}_n')=I_n(\mathcal{G}_n')T_n(\mathcal{G}_n'')$. Hence, $T_n(\mathcal{G}_n')S_n(\mathcal{G}_n'')=I_n(\mathcal{G}_n')T_n(\mathcal{G}_n'')S_n(\mathcal{G}_n'')$. By \eqref{May14a-12} for set $\mathcal{G}_n''$,
$T_n(\mathcal{G}_n'')S_n(\mathcal{G}_n'')=id(\mathcal{G}_n'')+o(1)$, where $id(\mathcal{G}_n'')$ is the identity in $L_2(\mathcal{G}_n'')$. It immediately follows $T_n(\mathcal{G}_n')S_n(\mathcal{G}_n'')=I_n(\mathcal{G}_n')+o(1)$. Substituting the last relation into the formula $ E_n(\mathcal{G}_n')E_n(\mathcal{G}_n'')= S_n(\mathcal{G}_n')T_n(\mathcal{G}_n')
 S_n(\mathcal{G}_n'')T_n(\mathcal{G}_n'')$, we obtain \eqref{Vienna-6}.
\end{proof}

Let \begin{equation}
\mathcal{G}_{n, \lambda}=\{ \k \in {\mathcal{G}}_n:
\lambda ^{(n)}(\k) < \lambda\}. \label{d} \end{equation}
 This set is Lebesgue measurable, since ${\mathcal{G}}_n $ is
open and $\lambda ^{(n)}(\k)$ is continuous on $
{\mathcal{G}}_n$.

\begin{lemma}\label{L:abs.6}
$\left|{\mathcal{G}}_{n,\lambda+\varepsilon} \setminus
{\mathcal{G}}_{n,\lambda}\right| \leq 2\pi \varepsilon $ when $0\leq
\varepsilon \leq 1$. \end{lemma}
The proof is based on Lemma \ref{ldk-3IVlast} and completely analogous to that of Lemma 2 in \cite{KL2}.

 By (\ref{s}),
$E_n\left({\mathcal{G}}_{n,\lambda+\varepsilon}\right)-E_n\left({\mathcal{G}}_{n,\lambda}\right)
=E_n\left({\mathcal{G}}_{n,\lambda+\varepsilon}\setminus
    {\mathcal{G}}_{n,\lambda}\right)$. Let us obtain an estimate for
    this projection.
\begin{lemma}\label{L:abs.7} For any $F \in C_0^{\infty}(\R^2)$ and
$0\leq \varepsilon \leq 1$, \begin{equation}
\left\|\bigl(E_n({\mathcal{G}}_{n,\lambda+\varepsilon})-E_n({\mathcal{G}}_{n,\lambda})\bigr)F\right\|^2_{L_2(\R^2)}
 \leq C( F) \epsilon , \label{tootoo1}
 \end{equation}
 where $C(F)$ is uniform with respect to $n$ and $\lambda$.
\end{lemma}
\begin{proof} Let $\mathcal{G}_n'=\mathcal{G}_{n,\lambda+\varepsilon}\setminus \mathcal{G}_{n,\lambda}$. Using the definition \eqref{ST} of $E_n$ and formula \eqref{Vienna-1}  with $f=g=T_nF$, we obtain
\begin{equation}
\|E_n(\mathcal{G}_n')F\|^2_{L_2(\R^2)}<(1+o(1))\|T_nF\|^2_{L_{\infty}(\mathcal{G}_n')}|\mathcal{G}_n'|. \label{Vienna-7} \end{equation}
Using \eqref{eq2} and Corollary \ref{C:Psi}  we easily get $ \|T_nF\|_{L_{\infty}(\mathcal{G}_n')}<2\|F\|_{L_1(\R^2)}$. Substituting this estimate into \eqref{Vienna-7} and using Lemma \ref{L:abs.6}, we obtain \eqref{tootoo1}.

\end{proof}

\subsection{Sets ${\mathcal{G}}_{\infty}$ and
${\mathcal{G}}_{\infty ,\lambda }$} \label{S:8.1}
By construction,
$
    \mathcal{G}_{n+1}\subset \mathcal{G}_n,$
    $\mathcal{G}_{\infty}=\bigcap_{n=1}^{\infty}{\mathcal{G}}_n. $
Therefore, the perturbation formulas for $\lambda ^{(n)}(\k )$ and $\Psi _n(\k )$ hold in
$\mathcal{G}_{\infty}$ for all $n$.
 Let
     \begin{equation}
     \mathcal{G}_{\infty, \lambda }=\left\{\k \in
    \mathcal{G}_{\infty }: \lambda _{\infty }(\k )<\lambda
    \right\}. \label{dd}
    \end{equation}
The function $\lambda _{\infty }(\k )$ is a Lebesgue
measurable function, since it is a limit of the sequence of
measurable functions. Hence, the set  $\mathcal{G}_{\infty, \lambda
}$ is measurable.

\begin{lemma}\label{add6*} The measure of the symmetric difference of
two sets $\mathcal{G}_{\infty, \lambda }$ and $\mathcal{G}_{n,
\lambda}$ converges
 to zero as $n \to
\infty$ uniformly in $\lambda$ in every bounded interval:
    $$\lim _{n\to \infty }\left|\mathcal{G}_{\infty, \lambda }\Delta \mathcal{G}_{n, \lambda
    }\right|=0.$$
 \end{lemma}
 The proof is completely analogous to the proof of Lemma 4 in \cite{KL2}.

     \subsection{Projections $E(\mathcal{G}_{\infty , \lambda })$}

In this section, we show that the operators
$E_n(\mathcal{G}_{\infty , \lambda })$ have a strong limit
$E_{\infty }(\mathcal{G}_{\infty , \lambda })$ in $L_2(\R^2)$ as $n$
tends to infinity. The operator $E_{\infty }(\mathcal{G}_{\infty ,
\lambda })$ is a spectral projection of $H$. It can be represented
in the form $E_{\infty }(\mathcal{G}_{\infty , \lambda })=S_\infty
T_{\infty }$, where $S_{\infty }$ and $T_{\infty }$ are strong
limits of $S_n(\mathcal{G}_{\infty , \lambda })$ and
$T_n(\mathcal{G}_{\infty , \lambda })$, respectively.  For any $F\in
C_0^{\infty }(\R^2)$, we show: \begin{equation} E_{\infty }\left(
\mathcal{G}_{\infty , \lambda }\right)F=\frac{1}{4\pi ^2}\int
    _{ \mathcal{G}_{\infty , \lambda }}\bigl( F,\Psi _{\infty }(\k )\bigr) \Psi _{\infty }(\k,\x)\,d\k ,\label{s1u}
    \end{equation}
    \begin{equation} HE_{\infty }\left(
\mathcal{G}_{\infty , \lambda }\right)F=\frac{1}{4\pi ^2}\int
    _{ \mathcal{G}_{\infty , \lambda }}\lambda _{\infty }(\k )\bigl( F,\Psi _{\infty }(\k )\bigr) \Psi _{\infty }(\k,\x)\,d\k .\label{s1uu}
    \end{equation}
Using properties of $E_{\infty }\left( \mathcal{G}_{\infty , \lambda
}\right)$, we  prove absolute continuity of the branch of the
spectrum corresponding to functions $\Psi _{\infty }(\k )$.

We consider the sequence of operators $S_n(\mathcal{G}_{\infty ,
\lambda })$ which are given by (\ref{ev}) with $\mathcal
G_n'=\mathcal{G}_{\infty , \lambda }$.
\begin{lemma} We have
\begin{equation} \label{Vienna-2}\left\|(S_n(\mathcal{G}_{\infty , \lambda})-S_{n-1}(\mathcal{G}_{\infty , \lambda}))
f\right\|_{L_2(\R^2)}< C\|f\|_{L_2(\mathcal{G}_{\infty , \lambda
})}\xi_*^{-\frac14\xi _*^{r_{n-2}(\xi_*)}},\ \ n\geq3.
\end{equation}
\end{lemma}
\begin{proof} Considering as in the proof of Lemma \ref{9.1}, we obtain $$
    \left\|(S_n-S_{n-1})f\right\|_{L_2(\R^2)}^2=\left\|\widehat{S_nf}-\widehat{S_{n-1}f}\right\|_{L_2(\R^2)}^2=$$
  \begin{equation}     \int _{\R^2 }\sum _{\m''}\tilde B_{\m''}(\vec \xi )f \bigl(-\vec \xi\bigr)\bar f \bigl(-\vec \xi-
    \p_{\m''}\bigr)\chi \bigl(\mathcal{G}_{\infty , \lambda },-\vec \xi \bigr)\chi \bigl(\mathcal{G}_{\infty , \lambda },-\vec \xi-
    \p_{\m''}\bigr)\ d\vec \xi , \label{Vienna-3} \end{equation}
    \begin{align}\nonumber &\tilde B_{\m''}(\vec \xi )=\cr & \sum _{\m \in \Omega (r_{n-1},|\vec \xi |):\,\m+\m''\in \Omega (r_{n-1},|\vec \xi +\p_{\m''}|)}\left(v_{\m}^{(n)}-v_{\m}^{(n-1)}\right)\bigl(-\vec \xi\bigr)\overline{\left(v_{\m +\m''}^{(n)}-v_{\m +\m''}^{(n-1)}\right)}\bigl(-\vec \xi-
    \p_{\m ''}\bigr).\end{align}
    Assume for definiteness that $|\vec \xi +\p_{\m''}|\leq |\vec \xi |$.
    If $\||\p_{\m''}\||>2|\vec \xi |^{r_{n-1}(|\vec \xi |) }$, then $\tilde B_{\m''}(\vec \xi )=0$.
    Let $\vec \xi :\||\p_{\m''}\||\leq 2|\vec \xi |^{r_{n-1}(|\vec \xi |) }$.
   Using \eqref{perturbation*-3last} with $n$ instead of $n+1$, we easily obtain:
    \begin{equation}\left|\tilde B_{\m''}(\vec \xi )\right|=O\left(|\vec \xi |^{-
|\vec \xi |^{\frac{1}{2}r_{n-2}(|\vec \xi |) }}\right) O\left(|\vec
\xi +\p_{\m''}|^{- |\vec \xi +\p_{\m''}|^{\frac{1}{2}r_{n-2}(|\vec
\xi +\p_{\m''}|) }}\right).\label{Aug13-3}
\end{equation} Considering  \eqref{indrn} with $n-1$ instead of and
$n$ and taking into account that $\||\p_{\m''}\||\leq 2|\vec \xi
|^{r_{n-1}(|\vec \xi |) }$, we easily get:
$$\left|\tilde B_{\m''}(\vec \xi )\right|=\||\p_{\m''}|\|^{-8} O\left(\xi _*^{-
\xi _*^{\frac{1}{2}r_{n-2}(\xi_*) }}\right).$$
 Summarizing the last estimate for $\m ''\neq {\bf 0}$ and using \eqref{Aug13-3} for $\m''={\bf 0}$,
 we arrive at \eqref{Vienna-2}.
 \end{proof}

 By \eqref{Vienna-2}, the sequence of operators $S_n(\mathcal{G}_{\infty ,
\lambda })$ is a Cauchy sequence in the space of bounded operators.
We denote its limit by $S_{\infty }(\mathcal{G}_{\infty , \lambda
    })$. Note that the convergence of $S_n(\mathcal{G}_{\infty , \lambda
    })$ to $S_{\infty }(\mathcal{G}_{\infty , \lambda
    }) $  is uniform in $\lambda $ when $\lambda >\lambda _*$.
\begin{lemma} \label{Lem2}  The operator $S_{\infty }(\mathcal{G}_{\infty , \lambda
    })$
    can be described by the
    formula
    \begin{equation}
    (S_{\infty }f) (\x)= \frac{1}{2\pi }\int _{\mathcal{G}_{\infty , \lambda
    }}f (\k)\Psi _{\infty }(\k ,\x) d\k  \label{ev1}
    \end{equation}
for any $f \in L_{\infty }\left( \mathcal{G}_{\infty , \lambda
}\right)$.
\end{lemma}
\begin{proof} From Theorem \ref{T:Dec10} it follows that for every $f \in L_{\infty }\left( \mathcal{G}_{\infty , \lambda
}\right)$
\begin{equation}
     \lim _{n\to \infty }\int _{\mathcal{G}_{\infty , \lambda
    }}f (\k)\Psi _{n}(\k ,\x) d\k =\int _{\mathcal{G}_{\infty , \lambda
    }}f (\k)\Psi _{\infty }(\k ,\x) d\k
 \label{Vienna-9}
    \end{equation}
    for all $\x$. Hence, \eqref{ev1} holds. \end{proof}

Now we consider the sequence of operators $T_n(\mathcal{G}_{\infty ,
\lambda
    })$ which are given by (\ref{eq2}) and act from $L_2(\R^2)$ to $L_2(\mathcal{G}_{\infty , \lambda
    })$. Since, $T_n=S_n^*$, the sequence has a limit  $T_{\infty }$ in the class of bounded operators,  $T_{\infty }=S^*_{\infty }$.
     Note that the convergence of $T_n(\mathcal{G}_{\infty , \lambda
    })$ to $T_{\infty }(\mathcal{G}_{\infty , \lambda
    }) $  is uniform in $\lambda $ when $\lambda >\lambda _*$.

\begin{lemma} \label{Lem1} The operator $T_{\infty }(\mathcal{G}_{\infty , \lambda
    })$  can be described by the
    formula $T_{\infty }(\mathcal{G}_{\infty , \lambda
    })F=\frac{1}{2\pi }\bigl( F,\Psi _{\infty }(\k )\bigr) $
for any $F\in C_0^{\infty }(\R^2)$.
\end{lemma}
\begin{proof} The lemma easily follows from  Theorem \ref{T:Dec10} and formula \eqref{eq2}. \end{proof}

\begin{lemma}\label{May8} Operators
$E_n(\mathcal{G}_{\infty , \lambda })$ have a limit
$E_{\infty }(\mathcal{G}_{\infty , \lambda })$ in the class of bounded operators in $L_2(\R^2)$, the
convergence being uniform for $\lambda >\lambda _*$. The
operator $E_{\infty }(\mathcal{G}_{\infty , \lambda })$ is a
projection. For any $F\in C_0^{\infty }(\R^2)$ it is given by
(\ref{s1u}).
\end{lemma}
\begin{proof} The lemma immediately follows from convergence of sequences $S_n$, $T_n$ and Lemmas \ref{L:10.4},  \ref{Lem2}, \ref{Lem1}. \end{proof}

\begin{lemma}\label{onemore} There is a strong limit
$E_\infty(\mathcal{G}_{\infty})$ of the projections $E_\infty
(\mathcal{G}_{\infty,\lambda })$ as $\lambda $ goes to infinity.
\end{lemma} \begin{corollary}\label{onemore1} The operator $E_{\infty
}(\mathcal{G}_{\infty})$ is a projection. \end{corollary}
\begin{proof} It can be  easily seen from \eqref{s1u} that the sequence of $E_{\infty}
(\mathcal{G}_{\infty,\lambda })$ is monotonuos in $\lambda $. It is well known that a monotone sequence of projectors has a limit.\end{proof}

The proofs of the next two lemmas are completely analogous to the proofs of Lemmas 10, 11 in \cite{KL2}.

\begin{lemma}\label{L:abs.9} Projections
$E_\infty(\mathcal{G}_{\infty,\lambda })$, $\lambda \in \R$, and
$E_\infty(\mathcal{G}_{\infty})$ reduce the operator $H$.
\end{lemma}

\begin{lemma} The family of projections
$E_\infty(\mathcal{G}_{{\infty},\lambda} )$ is the resolution of the
identity of the operator $HE_\infty(\mathcal{G}_{\infty})$ acting in
$E_\infty(\mathcal{G}_{\infty})L_2(\R^2)$. \end{lemma}

\begin{lemma} Formula \eqref{s1uu} holds, when $F\in C_0^{\infty}(\R^2)$. \end{lemma}
\begin{proof} By the previous lemma, $E_{\infty }\left(
\mathcal{G}_{\infty , \lambda }\right)F \in D(H)$. It is easy to see that the r.h.s. of \eqref{s1u} can be differentiated with respect to $\x$ under the integral sign. Now, considering
\eqref{6.7}, we get \eqref{s1uu}.

\end{proof}

\subsection{Proof of Absolute Continuity}

Now we show that the  branch of spectrum (semi-axis) corresponding
to $\mathcal G_{\infty }$  is absolutely continuous.

\begin{thm}\label{T:abs} For any $F\in C_0^{\infty }(\R^2)$ and
$0\leq \varepsilon \leq 1$,
    \begin{equation}
    \left| \left(E_\infty(\mathcal{G}_{\infty,\lambda+\varepsilon})F,F\right)-
    \left(E_\infty(\mathcal{G}_{\infty,\lambda})F,F
    \right) \right| \leq C_F\varepsilon .\label{May10*}
    \end{equation}

\end{thm} \begin{corollary} The spectrum of the operator
$HE_\infty(\mathcal G_{\infty})$ is absolutely continuous.
\end{corollary}
\begin{proof} By formula (\ref{s1u}),
    $$ | \left(E_\infty(\mathcal{G}_{\infty,\lambda+\varepsilon})F,F\right)-\left(E(\mathcal{G}_{\infty,\lambda})F,F
    \right) | \leq C_F\left| \mathcal{G}_{\infty , \lambda +\varepsilon
    }\setminus \mathcal{G}_{\infty , \lambda } \right| .$$
Applying Lemmas \ref{L:abs.6} and \ref{add6*}, we immediately get
(\ref{May10*}).

\end{proof}

\section{Appendices}

\subsection{Appendix 1. Proof of Lemma \ref{L:3.1}}
\begin{proof}
\begin{enumerate}
\item The case $p_\m>4k$. From \eqref{Jan23a} it immediately follows
that $|\Im \varphi _{\m}^{\pm }|>(\cosh )^{-1}2>1$. Hence, $\W_0\cap
\OO _\m (k, \tau )=\emptyset $.

Further we use the Taylor series for $|\vec k(\varphi
)+\p_\m|_\R^2-k^{2}$ near its zeros: Noting that
  \begin{equation}
  |\vec k(\varphi)+\p_\m|_\R^2-k^2 = 2kp_\m\cos (\varphi-\varphi_\m)+p_\m^2
  \label{08.20}
  \end{equation}
 and recalling that $\varphi_\m^\pm$ are the solutions of
 $|\vec k(\varphi)+\p_\m|_\R^2 = k^2 $, we see:
  \begin{equation}\label{E:sin-cos}
  \cos (\varphi_\m^\pm-\varphi_\m) = -\frac{p_\m}{2k},\quad
   |\sin (\varphi_\m^\pm-\varphi_\m)| = \sqrt{\left|1-\frac{p_\m^2}{4k^2}\right|}.
  \end{equation}
 Expanding (\ref{08.20}) around $\varphi_\m^\pm$, we get:
  \begin{equation}\label{E:first-factor}
  \begin{split}
  &|\vec k(\varphi)+\p_\m|_\R^2-k^2
  =\cr & -2k p_\m \sin(\varphi_\m^\pm-\varphi_\m)r_\m\left(1+O(r_\m^2)\right)
   + k p_\m
   \cos(\varphi_\m^\pm-\varphi_\m)r_\m^2\left(1+O(r_\m^2)\right),
  \end{split}
  \end{equation}
where $r_\m=|\varphi -\varphi _\m^{\pm}|.$

 \item In the second case we put $r_\m = \frac{\tau k^{-1+\delta_*}}{p_\m  \sqrt{\left|1-\frac{p_\m^2}{4k^2}\right|}}(1+o(1))$
 when $k^{-1+2\delta_* } < p_\m < 4k  \text{ and }
 \left|1-\frac{p_\m^2}{4k^2}\right|>\tau k^{-2+\delta_*}$. Substituting $r_\m$ into \eqref{E:first-factor}, we get that
 the modulus of the first term is $2\tau k^{\delta_* }(1+o(1))$ and that of the second term is
 $\frac{\tau ^2k^{-2+2\delta_* }}{2\left|1-\frac{p_\m^2}{4k^2}\right|}(1+o(1))$.
 Using the condition $\left|1-\frac{p_\m^2}{4k^2}\right|>\tau k^{-2+\delta_*}$, one can
  easily see that the former is at least twice greater than the latter. Thus,  we get
  \begin{equation}
  \left||\vec k(\varphi)+\p_\m|_\R^2-k^2\right|>\tau k^{\delta_* }\ \  \mbox{when  }|\varphi -\varphi _\m^{\pm }|=r_\m.
  \label{Lan23b} \end{equation}
  Now, the minimum principle yield that this inequality holds everywhere
  outside the discs $\cup_{\pm,j\in\Z}(\Phi^{\pm}_{\m}+2\pi
j)$. Hence, $\OO_{\m}\subset \cup_{\pm,j\in\Z}(\Phi^{\pm}_{\m}+2\pi
j)$.

\item In the third case we put $r_\m = 32\tau k^{-1+\delta_*/2 }(1+o(1))$ and $\left|1-\frac{p_\m^2}{4k^2}\right| <
 \tau k^{-2+\delta_* }$. This time the modulus of the second term in \eqref{E:first-factor} is
 $64\cdot32\tau ^2 k^{\delta_*}\left(1+o(1)\right)$ and
 that of the first is smaller than $128\tau ^{3/2} k^{\delta_*}(1+o(1))$. Therefore we again have
 \eqref{Lan23b} and $\OO_{\m}\subset \cup_{\pm,j\in\Z}(\Phi^{\pm}_{\m}+2\pi
j)$.\end{enumerate}\end{proof}

\subsection{Appendix 2. Proof of Lemma \ref{2.12}}
\begin{proof}
By definition of $\MM'(\varphi _0)$, $\varphi _0\in \cap _{j=0}^2
\OO_{\m_j}(k,1)$. By \eqref{resonance1},
$$\left||\vec k(\varphi _0)+\p_{\m _j}|^2-k^2\right|\leq k^{\delta_*},\ j=0,1,2.$$It
follows:
\begin{equation}\left|2\left(\vec k(\varphi _0
)+\p_{\m _0},\p_{\q_j}\right)+p_{\q_j}^2\right|<2 k^{\delta_*},\
\q_j=\m_j-\m_{0},\ \ \ j=1,2,\label{I5}\end{equation} where
$\||\p_{\q_j}\||\leq 3jk^{\delta}$. We will complete the proof by
the way of contradiction. Assume that $\p_{\q_1}$ and $\p_{\q_2}$
are not colinear. We represent every $\p_{\q_j}$ in the form:
$\p_{\q_j}=2\pi(\s_j+\alpha\s_j')$, $\s_j,\s_j'\in \Z^2,
|\s_j|,|\s_j'|<8k^{\delta}$ and denote $[\a,\b ]=a_1b_2-a_2b_1$. It
follows from \eqref{I5} that the angle (modulo $\pi $) between
$\p_{\q_1},\p_{\q_2}$ is less than $k^{-1+\delta _* +2\delta \mu }$. Hence,
 \begin{equation} \label{I10} [\p_{\q_1},\p_{\q_2}]=O(k^{-1+\delta _* +2(1+\mu)\delta }).\end{equation}
Substituting $\p_{\q_j}=2\pi(\s_j +\alpha \s_j')$, $j=1,2$, we
obtain:
\begin{equation}
n_1+\alpha p_1+\alpha^2m_1 =O(k^{-1+\delta _* +2(1+\mu)\delta }), \label{I6a}
\end{equation}
where $n_1,p_1,m_1$ are integers, $n_1=[\s_1,\s_2]$,
$p_1=[\s_1,\s_2']+[\s_1',\s_2]$, $m_1=[\s_1',\s_2']$. Obviously,
$n_1,p_1,m_1=O(k^{2\delta })$.  Next, we use the condition \eqref{condition} on $\alpha $ at the  beginning of the paper. Assume first,  $0\leq |n_1|+|p_1|+|m_1| \leq N_1$. If $k$ is sufficiently large,   $k>k_0(N_0,N_1,\delta, \delta _*;\alpha)$, the inequality  \eqref{I6a} yields  $n_1+\alpha p_1+\alpha^2m_1=0$. This means vectors $\p_{\q_1}, \p_{\q_2}$ are
 colinear and the lemma is proved. Let  $|n_1|+|p_1|+|m_1| >N_1$. Then either, again, $n_1+\alpha p_1+\alpha^2m_1=0$ and we are done, or
\eqref{I6a} contradicts the initial condition \eqref{condition} on
$\alpha$ when $2\delta N_0\leq 1/2$ and $k$ is  sufficiently large: $k>k_0(N_0,N_1,\delta, \delta _*;\alpha)$.
\end{proof}

\subsection{Appendix 3}

\begin{lemma}If the $\||\cdot \||$-distance between $\M^{j,s}_2$ and $\M^{j,s'}_2$ is less than $L\,(\geq1)$, then  the $\||\cdot \||$-distance between central points of $\M^{j,s}_2$ and $\M^{j,s'}_2$ is less than $C(Q)L$. Moreover, $\M^{j,s}_2$ belongs to $(C(Q)L+\||\p_{\q}\||)$-neighborhood of $\M^{j,s'}_2$ and vise versa. \end{lemma}
\begin{proof} Let $\m +\hat n\q\in \M^{j,s}_2$,  $\m '+\hat n'\q \in \M^{j,s'}_2$  and $\||\p_{\m +\hat n\q-\m '-\hat n'\q}\||\leq L$ (in particular, $p_{\m +\hat n\q-\m '-\hat n'\q}\leq 2\pi L$).
Note that $\p_{\m-\m'}$ is colinear with $\p_\q$, since $\m,\m'\in \M^{j}_2$, and $p_{\m-\m'}<p_\q$, since $\m,\m'$ are central points, see \eqref{Mjs}. Considering the last inequality, we obtain, $|\hat n-
\hat n'|\leq C_0(Q) L$ and, thus, $\||\p_{\m-\m'}\||\leq C(Q)L$.

Let $\m +n\q\in \M^{j,s}_2$. We prove that there is a point of $\M^{j,s'}_2$ in the $(C(Q)L+\||\p_{\q}\||)$-neighborhood of $\m +n\q$. Using \eqref{Mjs*}, it is easy to see that  $\m '+n\q $ is
in the $\||\p_{\q}\||$-neighborhood of $\M^{j,s'}_2$. Considering that $\||\p_{\m-\m'}\||\leq C(Q)L$, we finish the proof of the lemma.
\end{proof}

\subsection{Appendix 4}\label{A:5} \begin{lemma}\label{L: Appendix1} The equation \begin{equation}\lambda^{(1)} \big(\k^{(1)}(\varphi
)+\p_\m\big)=k^{2}+\varepsilon _0,\ \ \ 0<p_\m\leq4k^\delta,\
|\varepsilon _0|\leq p_\m k^{\delta },\label{25a}
\end{equation} has no more than two solutions
$\varphi^\pm(\varepsilon_0) $ in $\tilde{\cal W}^{(1)}(k,\frac18)\cap \OO
_\m(k,\frac{1}{2})$. They satisfy the estimates:
\begin{equation}\big|\varphi^{\pm }(\varepsilon_0)-\varphi _\m ^{\pm
}\big|<k^{-1+2\delta }.\label{st}\end{equation}\end{lemma}
 \begin{proof} Let $\varphi \in \tilde{\cal W}^{(1)}(k,\frac14)\cap \OO
_\m(k,\frac{1}{2})$. The equation (\ref{25a}) is equivalent to
$$\lambda^{(1)}(\y(\varphi))=\lambda ^{(1)}(\y(\varphi)-\p
_\m)+\varepsilon_0,\ \ \y(\varphi) =\k^{(1)}(\varphi)+\p _\m.$$ We
use perturbation formula (\ref{eigenvalue}): $$|\y
(\varphi)|_{\R}^{2}+f_1(\y (\varphi))
    =|\y (\varphi)-\p _\m|_{\R}^{2}+f_1(\y (\varphi)-\p_\m)+\varepsilon_0,$$
    where $f_1$ is the series in the right-hand side of (\ref{eigenvalue}). This equation can be rewritten as
    \begin{equation}\label{3.7.1.5}
    \Bigl(2(\k^{(1)} (\varphi),\p _\m)_{\R}+p_\m^2\Bigr)
    +f_1(\y (\varphi))
    -f_1(\y (\varphi)-\p _\m)=\varepsilon_0.
    \end{equation}
 Using the
notation $\p_\m=p_\m(\cos \varphi_\m,\sin \varphi_\m)$, dividing
both sides of the equation (\ref{3.7.1.5}) by $2p_\m k$, and considering
that $\y(\varphi )=\k^{(1)}(\varphi)+\p _\m=
    (k+h^{(1)})\vec \nu +\p_\m$, we obtain:
    \begin{equation}\label{3.7.1.6}
    \cos (\varphi-\varphi_\m)+\dfrac{p_\m}{2k}-\varepsilon_0g_1(\varphi)+g_2(\varphi)+g_3(\varphi)=0,
    \end{equation}
where $g_1(\varphi)
    =(2p_\m k)^{-1}$
    and
    $$g_2(\varphi)=\dfrac{( \h^{(1)}(\varphi),\p_\m)}{p_\m k} , \  \  g_3(\varphi)=\Bigl(f_1(\y(\varphi))
    -f_1(\y(\varphi)-\p_\m)\Bigr)g_1(\varphi),\ \ \h^{(1)}(\varphi)=h^{(1)}(\varphi)\vec \nu.$$
 Using Lemma \ref{ldk}  and
considering that $0<p_\m\leq 4k^{\delta }$, we easily obtain:
    $$
   |g_2(\varphi)|= \left|\frac{(\h^{(1)}(\varphi),\p_\m)}{p_\m k}\right|\leq
    \frac{2h^{(1)}}{k}=O(k^{-4+(80\mu +6)\delta}).
    $$
Let us show $g_3(\varphi)=O(k^{-1+\delta})$.  If $p_\m\geq
k^{-2+\delta (80\mu +6)}$, then the estimate easily follows from
\eqref{perturbation} and the estimate for $g_1$. Let $p_\m
<k^{-2+\delta (80\mu +6)}$. It can be easily shown that the series
$f_1(\y)$, $\nabla f_1(\y)$ converge for all $\y$: $|\y-\k^{(1)}(\varphi)|_{\C^2}=O(k^{-\delta (40\mu +1)})$ or $|\y+\p_\m-\k^{(1)}(\varphi)|_{\C^2}=O(k^{-\delta (40\mu +1)})$ (see Lemma~\ref{L:derivatives-1}), the series being holomorphic with respect to $y_1$, $y_2$ in these neighborhoods.
Using \eqref{perturbation} (cf.  Lemma~\ref{L:derivatives-1}), we get $\nabla f_1(\y) =O(k^{-2+\delta
(120\mu +7)})$. Hence,
    $$ \left|f_1(\y(\varphi))-f_1(\y(\varphi)-\p_\m)\right| \leq \sup |\nabla
    f_1|p_\m=o(p_\m ),$$
and therefore, $g_3(\varphi)=O(k^{-1+\delta})$.
Since  $|\varepsilon_0|<p_\m k^{\delta}$, we obtain
    $\varepsilon_0g_1(\varphi)=O(k^{-1+\delta }).$ Thus,
    \begin{equation}
    g_2(\varphi)+ g_3(\varphi)
    -\varepsilon_0g_1(\varphi)=O(k^{-1+\delta}) \ \ \mbox{when}\ \ \varphi \in \tilde{\cal W}^{(1)} (k,\frac14)\cap \OO
_\m(k,\frac{1}{2}).\label{Nov14}
    \end{equation}
    By definition $\varphi_\m^{\pm }$ satisfy the equation $\cos
    (\varphi-\varphi_\m)+\dfrac{p_\m}{2k}$=0.

    Suppose both $\varphi_\m^{\pm
    }$ are in $\tilde{\cal W}^{(1)}(k,\frac{3}{16})$. We draw two circles $C_{\pm}$ centered at
    $\varphi_\m^{\pm }
$ with the radius $k^{-1+2\delta }$. They are both inside
$\tilde{\cal W}^{(1)}(k,\frac14)\cap \OO _\m(k,\frac{1}{2})$, the
perturbation series converging and the estimate (\ref{Nov14}) holds.
For any $\varphi$ on $C_{\pm}$, $|\varphi-\varphi_\m ^{\pm
}|=k^{-1+2\delta}$ and, therefore,
    $|\cos(\varphi-\varphi_\m)+\frac{p_\m}{2k}|>\frac{1}{2}k^{-1+2\delta }>|g_2(\varphi)+g_3(\varphi)-\varepsilon_0g_1(\varphi)|$
    for any $\varphi \in C_{\pm}$.
By Rouch\'{e}'s Theorem, there is only one solution of the equation
(\ref{3.7.1.6}) inside each $C_{\pm}$. Obviously, \eqref{3.7.1.6}
does not have solutions in $\tilde{\cal W}^{(1)}(k,\frac14)\cap \OO
_\m(k,\frac{1}{2})$ outside $C_{\pm}$.

If both $\varphi_\m^{\pm }$ are not in $\tilde{\cal W}^{(1)}(k,\frac{3}{16})$,
then their distance to $\tilde{\cal W}^{(1)}(k,\frac{1}{8})$ is at least
$\frac{1}{16}k^{-\delta (40\mu +1)}$, hence
$|\cos(\varphi-\varphi_\m)+\frac{p_\m}{2k}|>\frac{1}{4}k^{-1+2\delta}$
in $\tilde{\cal W}^{(1)}(k,\frac18)$. Therefore,  equation (\ref{3.7.1.6})
has no solution in $\tilde{\cal W}^{(1)}(k,\frac18)\cap \OO
_\m(k,\frac{1}{2})$. The case, when only one $\varphi_\m^{\pm }$ is
not in $\tilde{\cal W}^{(1)}(k,\frac{3}{16})$ is the obvious combination of the
two previous situations.
 Thus, there are at most two solutions in $\tilde{\cal W}^{(1)}(k,\frac18)\cap \OO
_\m(k,\frac{1}{2})$ and
    $|\varphi^{\pm }({\varepsilon _0})-\varphi_\m ^{\pm }|<k^{-1+2\delta}.$
\end{proof}

    \begin{lemma} \label{L:Appendix 2} For any $\varphi \in \tilde{\cal W}^{(1)}(k,\frac14)\cap \OO _\m(k,1)$ satisfying the estimate $\big|\varphi
-\varphi _\m ^{\pm  }\big|<k^{-\delta },$
\begin{equation}\frac{\partial }{\partial \varphi }\lambda^{(1)}
\big(\k^{(1)}(\varphi )+\p_\m\big)=\pm 2p_\m k(1+o(1)),
\label{26a}
\end{equation} \end{lemma} \begin{proof} First, assume $\varphi $ is real. Let $\y(\varphi )= \k^{(1)}(\varphi )+\p_\m$. Using the perturbation
formula (\ref{eigenvalue}) and Lemma \ref{ldk}, we obtain:
    \begin{align}\label{3.7.1.7}
    \lefteqn{
    \frac{\partial}{\partial \varphi}\lambda^{(1)}\bigl(\y(\varphi)\bigr)
    =\frac{\partial}{\partial \varphi}\left[\lambda^{(1)}\bigl(\y(\varphi)\bigr)-k^{2l}\right]=
    \frac{\partial}{\partial \varphi}\left[\lambda^{(1)}\bigl(\y(\varphi)\bigr)-\lambda^{(1)}\bigl(\y(\varphi)-\p_\m\bigr)\right]=
    }& \notag \\
    &\hspace{1cm}
    \left ( \nabla_{\y}\lambda^{(1)}\bigl(\y(\varphi)\bigr)-\nabla_{\y}\lambda^{(1)}
    \bigl(\y(\varphi)-\p_\m\bigr),\frac{\partial}{\partial \varphi}\y(\varphi)\right )
    _\R=
    \notag \\
    &\hspace{1cm}
    \left ( \nabla \bigl|\y(\varphi)\bigr|_{\R}^{2}-\nabla \bigl|\y(\varphi)-\p_\m\bigr|_{\R}^{2},(k+h^{(1)})\vec t+(h^{(1)})^{\prime}\vec \nu\right )
_\R +\notag\\
    &\hspace{2cm}
    \left ( \nabla f_1\bigl(\y(\varphi)\bigr)-\nabla f_1\bigl(\y(\varphi)-\p_\m\bigr),(k+h^{(1)})\vec t+(h^{(1)})^{\prime}\vec \nu\right ) _\R,
    \end{align}
where $\vec \nu=(\cos \varphi, \sin \varphi)$ and $\vec t=\vec \nu^{\prime}=(-\sin
\varphi,\cos \varphi)$, $f_1$ is the series in the right-hand side
of (\ref{eigenvalue}). Note that
    \begin{equation}
    \nabla \bigl|\y(\varphi)\bigr|_{\R}^{2}-\nabla
    \bigl|\y(\varphi)-\p_\m\bigr|_{\R}^{2}=2\p_\m.
 \label{3.7.1.8}   \end{equation}
Substituting (\ref{3.7.1.8}) into (\ref{3.7.1.7}), we get $
    \frac{\partial}{\partial
    \varphi}\lambda^{(1)}\bigl(\y(\varphi)\bigr)=T_1+T_2,$
    \begin{align*}
    T_1 &
    =2\left ( \p_\m,(k+h^{(1)})\vec t+(h^{(1)})^{\prime} \vec \nu \right ) _\R ,\\
    T_2 &
=\left ( \nabla f_1\bigl( \y(\varphi) \bigr) - \nabla
    f_1\bigl( \y(\varphi)-\p_\m\bigr),(k+h^{(1)})\vec t+(h^{(1)})^{\prime}\vec \nu\right ) _\R .
    \end{align*}
We see that $\varphi$ is close to $\varphi_\m \pm \pi /2$, since
$|\varphi -\varphi _\m^{\pm}|<k^{-\delta }$ by the hypothesis of the
lemma and
\begin{equation}\varphi _\m^{\pm}=\varphi_\m \pm \pi /2 +O(k^{-1+\delta
})\ \mbox{when }p_\m<4k^{\delta }. \label{ts} \end{equation}  Now we readily obtain:
    $( \p_\m,\vec \nu) _\R=o(p_\m),$
 $( \p_\m,\vec t) _\R=\pm p_\m(1+o(1))$.  Using also estimates \eqref{dk0}, \eqref{dk} for $h^{(1)}$, we get
  $T_1=\pm 2p_\m k(1+o(1))$.
Let us estimate $T_2$. As above, the series $f_1(\y)$, $\nabla f_1(\y)$, $D^2 f_1(\y)$ converge for all $\y$: $|\y-\k^{(1)}(\varphi)|_{\C^2}=O(k^{-\delta (40\mu +1)})$ or $|\y+\p_\m-\k^{(1)}(\varphi)|_{\C^2}=O(k^{-\delta (40\mu +1)})$
the series being holomorphic with respect to $y_1$, $y_2$ in these neighborhoods, and
we have $\nabla f_1(\y) =O(k^{-2+\delta
(120\mu +7)})$, $D^2 f_1(\y) =O(k^{-2+\delta (160\mu +8)})$. Let
$p_\m\geq k^{-2+\delta (120\mu +8)}$. Then, using the
estimate for $\nabla f_1(\y) $, we easily obtain
$T_2=O(k^{-1+\delta (120\mu +7)})=o(kp_\m)$. Let $p_\m<
k^{-2+\delta (120\mu +8)}$. Then, using the estimate for the
second derivative in the direction of $\p_\m$, we get $\Bigl|\nabla
f_1\bigl(\y(\varphi)\bigr)-\nabla
f_1\bigl(\y(\varphi)-\p_\m\bigr)\Bigr|=O(p_\m
k^{-2+(160\mu+8)\delta}).$
    Therefore,
$T_2=O\left(p_\m k^{-1+(160\mu +8)\delta}\right)$. Thus,
$T_2=o(kp_\m)$ for all $\p_\m$. Adding the estimates for
$T_1,T_2$, we get (\ref{26a}).

Since all formulas can be analytically extended to the area  of non-real $\varphi $, the estimates being preserved, (\ref{26a}) holds for any $\varphi \in \tilde{\cal W}^{(1)}(k,\frac14)\cap \OO _\m(k,1)$. \end{proof}

 \begin{lemma}\label{L:3.7.1} Let $\tilde {\cal O}_{\m,\varepsilon}^\pm$ be the open discs of
 the radius $\varepsilon $ centered at $\varphi ^{\pm }(0)$ defined in Lemma \ref{L: Appendix1}. For any
$ \varphi \in \tilde{\cal W}^{(1)}(k,\frac18)\cap {\cal
O}_{\m}(k,\frac{1}{2})$, $\varphi \not \in \tilde {\cal
O}_{\m,\varepsilon}^\pm$, and $0\leq \varepsilon<k^{-1+\delta}$,
    \begin{equation}\label{3.7.1.10}
    |\lambda^{(1)}(\y(\varphi))-k^{2}|\geq k
    p_\m\varepsilon.
    \end{equation}
\end{lemma} \begin{proof} Suppose (\ref{3.7.1.10}) does not hold for
some $ \varphi \in \tilde{\cal W}^{(1)}(k,\frac18)\cap {\cal
O}_{\m}(k,\frac{1}{2})$, $\varphi \not \in \tilde {\cal
O}_{\m,\varepsilon}^\pm$. This means that $\varphi$ satisfies
equation (\ref{25a}) with some $\varepsilon _0$: $|\varepsilon _0|
<kp_\m\varepsilon\,(<p_\m k^\delta)$. By Lemma \ref{L: Appendix1}, $\varphi$ could
be either $\varphi^{+}(\varepsilon_0)$ or
$\varphi^{-}(\varepsilon_0)$. Without loss of generality, assume
$\varphi =\varphi^{+}(\varepsilon_0)$. By Lemma \ref{L: Appendix1},
$|\varphi^{+}(\varepsilon_0)-\varphi _{\m}^{\pm}|<k^{-1+2\delta
}$ for $\varphi _{\m}^{+}$ or $\varphi _{\m}^{-}$.  By the same lemma and \eqref{ts}, there is a single $\varphi^{+}(0)$ in the  $2k^{-1+2\delta}$-neighborhood of $\varphi^{+}(\varepsilon_0)$. Obviously, the
$2k^{-1+2\delta}$-neighborhood of $\varphi^{+}(\varepsilon_0)$
satisfies conditions of Lemma \ref{L:Appendix 2}.
Considering (\ref{26a}), we obtain
$\left|\varphi^{+}(\varepsilon_0)-\varphi^{+}(0)
    \right|<\varepsilon$, i.e., $\varphi \in \tilde {\cal O}_{\m,\varepsilon}$. This contradicts
    the hypothesis of the lemma.
    \end{proof}

\begin{lemma} \label{L:July5} If $0<p_\m\leq 4k^{\delta }$ and
$\varphi \in \tilde{\cal W}^{(1)}(k,\frac18)\cap {\cal
O}_\m(k,\frac{1}{2})$, then
\begin{equation} \left \|(\lambda ^{(1)}(\y(\varphi
))-k^{2})\left(P_\m(H(\k^{(1)}(\varphi
))-k^{2})P_\m\right)^{-1}\right\|\leq 8, \ \ \ \ \y(\varphi ):=
\k^{(1)}(\varphi )+\p _\m. \label{July3a}
\end{equation} \end{lemma} \begin{proof} Let $\tilde C_1$ be the
circle in $\C$ of the radius $\frac14 k^{1-40\mu \delta }$ centered at
$z_0=|\y(\varphi )|_{\R}^{2}$. Using \eqref{metka1} and \eqref{metka2}, we easily get:
$$\left||\y(\varphi )+\p_\q|_{\R}^{2}-z\right|\gtrsim
\frac{1}{4}k^{1-40\mu \delta },\ \ \ \mbox{when}\ \ \q\in
\Omega(\delta),\ \ z\in \tilde C_1.$$ Therefore,
\begin{equation}\left\|\left(P_\m(H_0(\k^{(1)}(\varphi ))-z)P_\m\right)^{-1}\right\|\lesssim 4k^{-1+40\mu \delta
},\label{2.8} \end{equation}
\begin{equation}\left\|\left(P_\m(H_0(\k^{(1)}(\varphi ))-z)P_\m\right)^{-1}\right\|_1\lesssim 4k^{-1+40\delta(\mu
+1)}.\label{2.8'} \end{equation} Next, by
 (\ref{determinants}),
 \begin{equation}\Bigl|\det \frac{P_\m(H(\k^{(1)}(\varphi ))-z)P_\m}{P_\m(H_0(\k^{(1)}(\varphi ))-z)P_\m}-1 \Bigr|<
 2 \|V\|\left \|\left(P_\m(H_0(\k^{(1)}(\varphi ))-z)P_\m\right)^{-1}\right\|_1 \label{r1}
 \end{equation}
 for every $z$ on the contour $\tilde C_1$.
 Using the estimate (\ref{2.8'}), we obtain  that
 the right-hand part of (\ref{r1}) is less than 1.
 Applying  Rouch\'{e}'s Theorem, we conclude that the
 determinant has the same number of zeros and poles inside $\tilde C_1$.
 Considering that the resolvent $\left(P_\m(H_0(\k^{(1)}(\varphi ))-z)P_\m\right)^{-1}$ has a single pole,
 $z=|\y(\varphi )|^{2}_\R$,
 we obtain that  $\left(P_\m(H(\k^{(1)}(\varphi ))-z)P_\m\right)^{-1}$ has a single pole inside $\tilde C_1$ too.
  Obviously, the pole is at the point
 $z=\lambda
^{(1)}(\y(\varphi ))$. Therefore $$(\lambda ^{(1)}(\y(\varphi
))-z)\left(P_\m(H(\k^{(1)}(\varphi ))-z)P_\m\right)^{-1}$$ is
a holomorphic function of $z$
 inside $\tilde C_1$.

 Let $z\in \tilde C_1$. Using (\ref{perturbation}), we easily obtain:
 $|\lambda ^{(1)}(\y(\varphi))-z|\leq k^{1-40\mu \delta }.$ From (\ref{2.8}) and
Hilbert identity it follows that
\begin{equation}\left\|\left(P_\m(H(\k^{(1)}(\varphi ))-z)P_\m\right)^{-1}\right\|\leq 8k^{-1+40\mu \delta }, \
\mbox{when }z\in \tilde C_1.\label{2.8"} \end{equation} Multiplying
the last two estimates, and using maximum principle we get
\begin{equation} \left \|(\lambda ^{(1)}(\y(\varphi
))-z)\left(P_\m(H(\k^{(1)}(\varphi
))-z)P_\m\right)^{-1}\right\|\leq8, \ \ \ \ z\in
\overline{\hbox{Int}\,\tilde C_1}. \label{July3a'}
\end{equation} Note that $z:=k^{2}\in\overline{\hbox{Int}\,\tilde
C_1}$. Indeed, by \eqref{metka1},
 $||\y(\varphi )|_{\R}^{2}-k^{2}|<\frac12 k^{\delta_*}<\frac14 k^{1-40\mu
\delta }$. Substituting $z=k^{2}$ in the last estimate, we get
\eqref{July3a}. \end{proof}

\subsection{Appendix 5}

First we rewrite $D_{i_0}(\tau _1)$ in the form: $D_{i_0}=\det \Big(E(\vec \varkappa )(H(\vec \varkappa )-k^2)E(\vec \varkappa ) \Big)$, where $\vec \varkappa =(\tau _1, \beta _1\tau _1+\beta _2)$,  $\tau _{1,0}-\sigma _{large}/100<\tau _1<\tau _{1,0}+\sigma _{large}/100$, $E(\vec \varkappa )$ is the spectral projection corresponding to $\lambda _{i_0}(\vec \varkappa )$ and $\lambda _{i_0+1}(\vec \varkappa )$. Thus, $EP=PE=E$, $P$ being the projector associated with $\MM_*$
under consideration. By $E_0$ we denote the corresponding spectral projection for $V=0$.

\begin{proposition}\label{app5}
The operators  $E_0(\vec \varkappa )$, and $H_0(\vec \varkappa )$ can be extended as holomorphic functions of $\tau _1$ to $T_0$. They have the following properties:
\begin{equation} E_0VE_0=0, \label{App5-8}
\end{equation}
\begin{equation} \|E_0(H_0-k^2)^{-1}\|=O(\sigma_{large}^{-2})\ \ \ \hbox{on}\ \partial T_0, \label{App5-10}
\end{equation}
\begin{equation} \|(P-E_0)(H_0-k^2)^{-1}\|=O(\Lambda ^{-1/2}p_\q^{-1}) \ \ \  \hbox{in}\ T_0, \label{App5-10'}
\end{equation}

The spectral projection $E(\vec \varkappa )$, $\vec \varkappa =(\tau _1, \beta _1\tau _1+\beta _2)$,  $\tau _{1,0}-\sigma _{large}/100<\tau _1<\tau _{1,0}+\sigma _{large}/100$, obeys the following asymptotics:
\begin{equation}
E=_{\Lambda \to \infty }E_0+O\left(\frac{\|V\|_\infty}{\Lambda ^{1/2}p_\q}\right), \label{App5-1}
\end{equation}
\begin{equation}
(E-E_0)(H-k^2)=_{\Lambda \to \infty }O\left(\|V\|_\infty \right),\label{App5-3}
\end{equation}
Here, the constants in $O(\cdot)$ do not depend on $V,\Lambda$ or $p_\q$.

The projection $E(\vec \varkappa )$ can be extended as a holomorphic function of $\tau _1$ to $T_0$, the estimates \eqref{App5-1}, \eqref{App5-3}  are preserved
 in $T_0$.
\end{proposition}
\begin{proof} Estimates \eqref{App5-10} and \eqref{App5-10'} easily follow from the definitions of $E_0$, $T_0$  and properties of $H_0(\vec \varkappa )$. In particular, \eqref{App5-10}
follows from the fact that $H_0(\vec \varkappa )$ is a quadratic polynomial with respect to $\tau _1$.
Next, for $V=0$ the eigenvalues $\tilde\lambda _{i_0},\tilde\lambda _{i_0+1}$ are equal to $(\tau _1+n_1p_\q)^2$ and $(\tau _1+ n_2p_\q)^2$, where $n_1,n_2$ have different signs and $|(n_1-n_2)p_\q|>\frac 14 \Lambda ^{1/2}$. It follows that $V_{n_1-n_2}=0$, since $V$ is a trigonometric polynomial (here we assume that $\Lambda^{1/2}>4Qp_\q$). Considering also that $V_{\bf 0}=0$, we obtain $E_0VE_0=0$.

Assume now that $\Lambda^{1/2}>4\|V\|_\infty p_\q^{-1}$. Formulas \eqref{App5-1}, \eqref{App5-3} follow from  a perturbation expansion for $E$:
\begin{equation}
E=E_0+\sum _{r=1}^{\infty} \hat G^{(r)},\ \  \hat G^{(r)}=\frac{(-1)^{r+1}}{2\pi i}\int _{\hat C}(H_0-z)^{-1}\left(V(H_0-z)^{-1}\right)^r dz, \label{App5-4'}
\end{equation}
where $\hat C$ is the contour around the pair of eigenvalues of the distance $R=\frac{1}{2}\Lambda ^{1/2}p_\q$ from them.
The r.h.s. of \eqref{App5-4'} can be easily extended as a holomorphic function of $\tau _1$ to $T_0$. Therefore, the estimates \eqref{App5-1}, \eqref{App5-3} are preserved in $T_0$.
\end{proof}

\begin{lemma} The determinant $D_{i_0}(\tau_1)$ has  the same number $J$ of zeros in $T_0$ as $D_{i_0,0}(\tau_1)$, $T_0$ being the $\sigma_{large}/100$-neighborhood of the zeros of $D_{i_0,0}$,  and
\begin{equation}
|D_{i_0}(\tau_1)|>(\sigma_{large}/100)^J/2\ \ \ \hbox{on}\ \partial T_0. \label{App5-7}
\end{equation}\end{lemma}
\begin{proof}
First, we prove that \begin{equation}\det \Big(P\big(H(\vec \varkappa )-k^2I\big)P\Big) \Big(P\big(H_0(\vec \varkappa )-k^2I\big)P\Big)^{-1}=1+O\left(\frac{Q\|V\|_{\infty }}{\sigma _{large}p_\q^{1/2}\Lambda ^{1/4}}\right), \  \mbox{when}\ \tau _1\in \partial T_0. \label{zaliv2} \end{equation}
Indeed, this determinant can be rewritten in the form $\det (I+S)$, where
$$S=\Big(P\big(H_0(\vec \varkappa )-k^2I\big)P\Big)^{-1/2}V\Big(P\big(H_0(\vec \varkappa )-k^2I\big)P\Big)^{-1/2}.$$
By \eqref{App5-8}, $E_0SE_0=0$. Using now \eqref{App5-10}, \eqref{App5-10'}, we obtain $\|S\|_1=O\left(\frac{Q\|V\|_{\infty }}{\sigma _{large}p_\q^{1/2}\Lambda ^{1/4}}\right)$. The last estimate yields \eqref{zaliv2}.
Now, by \eqref{determinants1} and Rouch\'{e}'s Theorem, $\det \Big(P\big(H(\vec \varkappa )-k^2I\big)P\Big) $ and $\det \Big(P\big(H_0(\vec \varkappa )-k^2I\big)P\Big)$ have the same number of zeros in $T_0$, when $\Lambda $ is sufficiently large: $\Lambda =\Lambda (V)$.

Next, note that the determinant of operator $(I-E_0)P\big(H_0(\vec \varkappa )-k^2I\big)P$ has no zeros inside $T_0$, see \eqref{App5-10'}.
Considering  \eqref{App5-10}--\eqref{App5-3},  it is not difficult to show that
\begin{equation}\frac{\det \Big((I-E)P\big(H(\vec \varkappa )-k^2I\big)P\Big)}{\det  \Big((I-E_0)P\big(H_0(\vec \varkappa )-k^2I\big)P\Big)}=1+O\left(\frac{\|V\|_{\infty }}{p_\q\Lambda ^{1/2}}\right), \  \mbox{when}\ \tau _1\in \partial T_0. \label{zaliv3} \end{equation}
It follows that $\det \Big((I-E)P\big(H(\vec \varkappa )-k^2I\big)P\Big) $ has no zeros inside $T_0$. Now, we obtain from \eqref{zaliv2} and \eqref{zaliv3} that
\begin{equation}
\frac{D_{i_0}}{D_{i_0,0}}=1+O\left(\frac{Q\|V\|_{\infty }}{\sigma _{large}p_\q^{1/2}\Lambda ^{1/4}}\right),\ \mbox{when}\ \ \tau _1\in \partial T_0, \end{equation}
and, hence (choosing $\Lambda=\Lambda(V)$ to be sufficiently large),
$D_{i_0,0}$ and $D_{i_0}$ have the same number of zeros inside $T_0$ and even twice smaller neighborhood. Since $D_{i_0,0}$ satisfies \eqref{Jan8-14}, $D_{i_0}$  satisfies the analogous estimate.
\end{proof}

\subsection{Appendix 6}

\begin{lemma}\label{L:r_0} Let $R$ be the smallest positive integer for which \eqref{r_0} holds.
We have $R>k^{(\frac\gamma2+2\delta_0)r_1-\delta_*-2\delta}$.
\end{lemma}

\begin{proof}
Notice that
$$
{\cal A}_{R}=\sum_{i_1,\dots,i_{R}=0}^{\hat{I}}{\cal
A}_{i_1,\dots,i_{R}},
$$
where
$$
{\cal
A}_{i_1,\dots,i_{R}}:=P^{\partial}(\hat{H}_0-k^{2})^{-1}W\left[\prod_{r=1}^{R}
\left(P_{i_r}(\hat{H}_0-k^{2})^{-1}W\right)\right]P^{(int)},
$$
$P_i$ being as in \eqref{opP'}, $i\geq 0$.
Here we used that $R$ is the smallest positive integer for which
${\cal A}_{R}\not=0$.

In principal, everything is defined by the case where all $i_q$ are
equal to zero. However,  to include impurities of non-resonant and white
clusters, we need an additional construction. Consider a particular
${\cal A}_{i_1,\dots,i_{R}}$. From the sequence $i_1,\dots,i_{R}$ we
take a subsequence of {\it all} non-zero indices
$i_{r_1},\dots,i_{r_s}$, $1\leq r_1<\dots<r_s\leq R$ (this sequence
can be empty).  In this subsequence we construct a subsubsequence ${j_1},\dots,{j_p}$
($p\leq s$) of non-repeating indices as follows. We choose
${j_1}=i_{r_1}$. If $i_{r_1}$ is not equal to any other $i_{r_t}$,
$t=2,\dots,s$, then ${j_2}=i_{r_2}$. If there is one or more
$i_{r_t}$ equal to $i_{r_1}$ then we denote the segment between the
first and the last $i_{r_1}$ as $I_1$, both ends being included. The next term after $I_1$ we
choose to be an ${j_2}$ etc. Thus, with a slight abuse of the
notation we have:
$$
i_{r_1},\dots,i_{r_s}=I_1, I_2,\dots,I_p,\ \ \ p\leq s,
$$
where $I_q$ may contain just one element.
Now, the initial sequence $i_1,\dots,i_{R}$ can be represented as
$I_1^0,\tilde{I}_1,I_2^0,\tilde{I}_2,\dots,\tilde{I}_p,I_{p+1}^0$,
where each $I_q^0$ is a sequence of only zeros (it can be empty) and
$\tilde{I}_q$ is $I_q$ with possibly some zeros inside. Put
$$
P_{j_q}BP_{j_q}:=P_{j_q}(\hat{H}_0-k^{2})^{-1}WP_{i'}(\hat{H}_0-k^{2})^{-1}W\dots(\hat{H}_0-k^{2})^{-1}WP_{j_q}.
$$
Here for all internal projectors $P_{i'}$ we have either $i'\in I_q$
or $i'=0$, all indices in  $\tilde I_q$ being included. We notice that $P_{j_q}BP_{j_q}$ has a block form.
Now we can represent ${\cal
A}_{i_1,\dots,i_{R}}$ as follows:
\begin{equation} \label{Apr12}
\begin{split}
&{\cal
A}_{i_1,\dots,i_{R}}=P^{\partial}(\hat{H}_0-k^{2})^{-1}W\left[\prod_{q=1}^{p}
\left(P_{0}(\hat{H}_0-k^{2})^{-1}W\right)^{s_q}P_{j_q}BP_{j_q}(\hat{H}_0-k^{2})^{-1}W\right]\cr
&
\times\left(P_{0}(\hat{H}_0-k^{2})^{-1}W\right)^{s_{p+1}}P^{(int)},
\end{split}
\end{equation}
where $s_q$ is the number of elements in $I_q^0$, $s_q\geq 0$; $j_q$
is a non-zero index corresponding to $I_{q}$. Obviously,
$\left((\hat{H}_0-k^{2})^{-1}W\right)_{\m\m'}=0$ if
$\||\cdot\||$-distance between the cluster containing $\p_\m$ and
the cluster containing $\p_{\m'}$ is greater than $k^\delta$ (here,
as usual, we consider the points in the range of $P_0$ as $1\times1$
clusters). Next, if $P_{j_q}$ is the projection on a non-resonant
cluster, then
$$
\left(P_{j_q}BP_{j_q}\right)_{\m\m'}=0,\ \ \ \hbox{for}\
\||\p_{\m-\m'}\||>k^{\delta_*+\delta},
$$
since a non-resonant cluster has the size not greater than $k^{\delta_*+\delta }$.
Let $p'$ be the number of non-resonant projections in the sequence
$\{P_{j_q}\}_{q=1}^p$. Hence, $p-p'$ is the number of white
clusters. The operator ${\cal A}_{i_1,\dots,i_{R}}$ can be non-zero
only if
\begin{equation}\label{odin*}
\frac D2 \leq k^\delta\sum_{q=1}^{p+1}(s_q+1)+k^{\delta_*+\delta }p'
+\sum_{m=1}^{p-p'} d_m,
\end{equation}
where $\frac D2$ is the $\||\cdot \||$ distance between the supports of $P^{\partial}$ and $P^{(int)}$ and $d_m$ is the size of a white cluster. Next, we
prove that
\begin{equation}\label{dva*}
\sum_{q=1}^{p+1}(s_q+1)+p' \geq
\frac{1}{4}k^{(\frac\gamma2+2\delta_0)r_1-\delta_*-\delta}.
\end{equation}
Assume that \eqref{dva*} does not hold. Then, by \eqref{odin*}
\begin{equation}\label{tri*}
\sum_{m=1}^{p-p'} d_m\geq \frac 14 k^{(\frac\gamma2+2\delta_0)r_1},
\end{equation}
since $D=k^{(\frac\gamma2+2\delta_0)r_1}$.
Obviously,
\begin{equation}\label{April29-14a} d_m\leq n_mk^{\frac{\gamma r_1}{6}},\end{equation}
where $n_m$ is the number of $\MM^{(2)}$ points in the white cluster number $m$,
$\ m=1,\dots,p-p'$. Let $\ell$ be the size of a
minimal box containing all these white clusters.  It is
easy to see  that
\begin{equation}\label{April29-14b}\ell\leq k^{\delta }\sum_{q=1}^{p+1}(s_q+1)+k^{\delta_* +\delta}p' +\sum_{m=1}^{p-p'} d_m\leq
\frac 14 k^{(\frac\gamma2+2\delta_0)r_1}+\sum_{m=1}^{p-p'}
d_m\leq2\sum_{m=1}^{p-p'} d_m.\end{equation} Here we also used \eqref{tri*} and
the inequality opposite to \eqref{dva*}. By Lemma \ref{L:2/3-1}
\begin{equation}\label{vot1}
\sum_{m=1}^{p-p'} n_m\leq C\ell^{2/3}k.
\end{equation}
Combining  inequalities \eqref{April29-14a}-\eqref{vot1} and solving for
$\sum_{m=1}^{p-p'} d_m $, we obtain: $\sum_{m=1}^{p-p'}
d_m<k^{\frac{\gamma r_1}{2}+3}$.
This contradicts to \eqref{tri*}. Thus, we proved \eqref{dva*}.
Using \eqref{dva*} and the obvious inequality $\sum _{q=1}^{p+1}(s_q+1)+p'\leq R+1$
proves the lemma.

\end{proof}
\subsection{Appendix 7}
\begin{lemma}\label{L:r_0-b} Let $R$ be the smallest positive integer for which \eqref{r_0} holds in the case of a black cluster.
We have $R>k^{\gamma r_1+\delta _0r_1-\delta_*-2\delta}$.
\end{lemma}
\begin{proof} We again use formula \eqref{Apr12}, where $P_{j_q}$ are projections on non-resonant, white and grey components in a component of a black region.
Assume first that all components $\Pi_{j_q}$ can be placed in one
$\||\cdot \||$ box of the size $4k^{\gamma r_1+\delta _0r_1}$.
Obviously,
\begin{equation} \label{Apr12-1}\frac D2\leq k^{\delta }\sum _{q=1}^{p+1}(s_q+1) +k^{\delta_*+\delta }p'+\sum  _m d_m^w +\sum _{\tilde m} d_{\tilde m}^g,\end{equation}
where $\frac D2$ is the $\||\cdot \||$ distance between the supports of $P^{\partial}$ and $P^{(int)}$, $p'$ is the number of non-resonant components in the black component,
$\sum  _m d_m^w$ and $\sum  _{\tilde m} d_{\tilde m}^g$ are the total lengths of white and grey components
 in the black component. Let us prove first that
 $k^{\delta }\sum (s_q+1) +k^{\delta_*+\delta }p'+\sum  _m d_m^w >\frac{1}{4}k^{\gamma r_1+\delta _0r_1}$. Suppose that it is not so.
 Then, by \eqref{Apr12-1},  $\sum _{\tilde m} d_{\tilde m}^g>\frac{1}{4}k^{\gamma r_1+\delta _0r_1}$, since $D=\frac{1}{2}k^{\gamma r_1+\delta _0r_1}.$
The $4k^{\gamma r_1+\delta _0r_1}$-box containing all components $\Pi_{j_q}$,
consists of no more than $4^4k^{4\delta_0r_1}$ boxes of the size
$k^{\gamma r_1}$. Since all $\Pi _{j_q}$ are in white $k^{\gamma
r_1}$ -boxes, the total number of points of $\MM^{(2)}$ in these
white boxes does not exceed $ck^{\frac{1}{2}\gamma r_1+\delta
_0r_1}\cdot k^{4\delta _0r_1}$. Since each grey box contains more
than $k^{\frac{1}{6}\gamma r_1-\delta _0r_1}$ points, the total
number of grey boxes is less than $ck^{\frac{1}{2}\gamma r_1+\delta
_0r_1}\cdot k^{4\delta _0r_1}\cdot k^{-\frac{1}{6}\gamma r_1+\delta
_0r_1}=ck^{\frac{1}{3}\gamma r_1+6\delta _0r_1}$. Therefore, the
total size of the grey region is less than $ck^{\frac{1}{3}\gamma
r_1+6\delta _0r_1}\cdot k^{\frac{1}{2}\gamma r_1+2\delta
_0r_1}=ck^{\frac{5}{6}\gamma r_1+8\delta _0r_1}$, see the definition of a grey region. Since $\delta
_0<\frac{1}{48}\gamma$, it is much less than $\frac{1}{4}k^{\gamma
r_1+\delta _0r_1}$. We have arrived to the contradiction with the
assumption $\sum _{\tilde m} d_{\tilde m}^g>\frac{1}{4}k^{\gamma
r_1+\delta _0r_1}$.  Therefore, $k^{\delta }\sum (s_q+1) +k^{\delta_*+\delta
}p'+\sum  _m d_m^w > \frac{1}{4}k^{\gamma r_1+\delta _0r_1}$.
Considering again that the total number of $\MM^{(2)}$ points in the
white boxes of the $4k^{\gamma r_1+\delta _0}$-box does not exceed
$ck^{\frac{\gamma r_1}{2}+5\delta _0r_1}$, we obtain $\sum _m
d_m^w<ck^{\frac{\gamma r_1}{2}+5\delta _0r_1}\cdot 4k^{\frac{\gamma
r_1}{6}}<\frac{1}{20}k^{\gamma r_1+\delta _0r_1}$, see the definition of a white region. It follows
$k^{\delta }\sum (s_q+1) +k^{\delta_*+\delta }p'
>\frac{1}{5}k^{\gamma r_1+\delta _0r_1}$.  By construction, $R+1\geq\sum
(s_q+1) +p' $. Therefore, $R>
>k^{\gamma r_1+\delta _0r_1-\delta_*-2\delta}$.

Assume that we cannot put all the components $\Pi_{j_q}$  in one
$\||\cdot \||$-box of the size $4k^{\gamma r_1+\delta _0r_1}$. Let
us consider the box of this size around   $\Pi_{j_1}$. Let $K$ be a
number such that all $\Pi_{j_q}$, $q=1,...K$ are completely in the box and
$\Pi_{j_{K+1}}$ is not. Then, instead of \eqref{Apr12} we consider
just its piece \begin{equation} \label{Apr12-2}
\begin{split}&
P_{j_1}W(\hat{H}_0-k^{2})^{-1}\left[\prod_{q=2}^{K}
\left(P_{0}W(\hat{H}_0-k^{2})^{-1}\right)^{s_q}P_{j_q}BP_{j_q}W(\hat{H}_0-k^{2})^{-1}\right]
 \cr &
 \times\left(P_{0}W(\hat{H}_0-k^{2})^{-1}\right)^{s_{K+1}}P_{j_{K+1}}.
\end{split}
\end{equation}
Further considerations are the same as in the previous case since by
construction the distance between $\Pi_{j_1}$ and $\Pi_{j_{K+1}}$ is
at least $\frac12 k^{\gamma r_1+\delta _0r_1}$. \end{proof}

\subsection{Appendix 8. On Application of Bezout Theorem}

Let $D(\k, \lambda )$ be the determinant of the truncated operator
$H(\k)-\lambda $ of the size $k^{r_*}$, $r_*\geq 1$. Obviously, $D$ is the
polynomial of the degree $k^{4r_*}$ with respect to
$\varkappa_1,\varkappa_2$ and {\it a line is not a solution of the
equation} $D(\k, \lambda )=0$. Let $\lambda $ be fixed, $\lambda
=k^{2}$.
\begin{definition}\label{elementary} We call a piece of $D(\k, \lambda )=0$ elementary, if

1) it can be parameterized by $\varkappa_1$:
$\varkappa_2=\varkappa_2(\varkappa_1)$ with
$|\varkappa_2'(\varkappa_1)|\leq1$ or by $\varkappa_2$:
$\varkappa_1=\varkappa_1(\varkappa_2)$ with
$|\varkappa_1'(\varkappa_2)|\leq1$;

2) function $\varkappa_1=\varkappa_1(\varkappa_2)$ (or
$\varkappa_2=\varkappa_2(\varkappa_1)$) is monotone and continuously
differentiable;

3) it does not have inflection points inside;

4) it has a length not greater than $1$.
\end{definition}
We will show that the curve $D(\k, \lambda )=0$ can be split into elementary pieces  and estimate the number of such pieces. In
the proof we will apply several times the following statement (which
is a simplified version of Bezout Theorem).
\begin{theorem}\label{bezout}
Let $P$ and $Q$ be two plane real-valued polynomials of degree $p$
and $q$ respectively. If $P$ and $Q$ do not contain common factors
then the total number of points satisfying $P(\k)=0=Q(\k)$ (i.e.
number of points of intersection) does not exceed $pq$.
\end{theorem}
We have
\begin{lemma} \label{elpieces} The set $D(\k,\lambda )=0$ can be split into $k^{17 r_*}$ or less elementary pieces.
\end{lemma}
\begin{proof}
First, $D(\k ,\lambda )$ can be represented as a product of simple (i.e.
irreducible) factors (counting multiplicity). The total number of
factors is less than $k^{4r_*}$ which is also the bound for their
total degree. We consider one of such simple factors $P$ and denote
by $p$ its degree (note that we do ignore the multiplicity of the
factor). Let us consider the points
\begin{equation}\label{odinindapp}
P(\k)=0,\ \ \frac{\partial}{\partial \varkappa_2}P(\k)=0.
\end{equation}
Since $P$ is irreducible and $\frac{\partial}{\partial
\varkappa_2}P(\k)$ has degree less than $p$ (we also notice that
$\frac{\partial}{\partial \varkappa_2}P(\k)$ is not identically zero
since $D(\k,\lambda)=0$ does not contain lines) they do not have
common factors. Thus, the number of such points $\k$ does not exceed
$p(p-1)$. Next, by the same reasons the number of points
\begin{equation}\label{odinindapp-1}
P(\k)=0,\ \ \frac{\partial}{\partial \varkappa_1}P(\k)=0
\end{equation}
 does not exceed $p(p-1)$ and the number of points
\begin{equation}\label{dvaindapp}
P(\k)=0,\ \ \frac{\partial P}{\partial
\varkappa_2}(\k)=\pm\frac{\partial P}{\partial \varkappa_1}(\k)
\end{equation}
does not exceed $2p(p-1)$. We split each previous piece by such
points. Thus, we have at most $4p(p-1)+1$ pieces, each end
satisfying \eqref{odinindapp} or \eqref{dvaindapp}. The sign of $
(\frac{\partial}{\partial \varkappa_2}P)^2-(\frac{\partial}{\partial
\varkappa_1}P)^2$ is constant on each piece, i.e. the piece admits
parametrization as in the property 2 of Definition~\ref{elementary}.
Making parametrization by $\varkappa_1$ or $\varkappa_2$, depending
on the sign, we obtain that the length of a piece does not exceed
$\sqrt{2}\cdot 4k^{r_*}$ (obviously, $|\varkappa_j|<2k^{r_*}$).
Therefore the total length of the curve $P=0$ does not exceed
$18p^2k^{r_*}$. Next, for each piece where $\frac{\partial}{\partial
\varkappa_2}P(\k)\not=0$ inflection points of $P(\k)=0$ are
described by the system
\begin{equation}\label{starindapp}
P=0,\ \ P_{\varkappa_2\varkappa_2}(P_{\varkappa_1})^2-
2P_{\varkappa_1\varkappa_2}P_{\varkappa_1}P_{\varkappa_2}+
P_{\varkappa_1\varkappa_1}(P_{\varkappa_2})^2=0.
\end{equation}
Again, since $P$ is irreducible and no line is a solution, we have no
common factors here and can apply Bezout Theorem. The number of
points satisfying \eqref{starindapp} does not exceed
$p(2(p-1)+(p-2))=3p^2-4p$. Therefore, we have at most $12p^4$ pieces
with the ends satisfying \eqref{odinindapp} or \eqref{dvaindapp} or
\eqref{starindapp}. At last, we split each of these concave pieces
into pieces with the length not greater than $1$. Considering that
the total length of $P(\k)=0$ is less that $18p^2k^{r_*}$, we obtain that
the total number of elementary pieces does not exceed $18p^2k^{r_*}+12p^4$.
Taking the sum over all simple factors of $D$ we prove the lemma.
\end{proof}


\subsection{Appendix 9. On the Proof of Geometric Lemmas Allowing to Deal with Clusters instead of Boxes \label{App9}}

In the proof of Lemma~\ref{norm2/3} it is important that we deal
with the same curve generated by the determinant and just change the
argument $\k$. At the same time, a priori we have the estimates for
the resolvent of the operator reduced onto a particular cluster. The
form of clusters can vary which formally changes the projector and
thus the determinant and the curve. Here we explain how to deal with
this situation. We will show that every cluster (white, grey or
black) can be embedded into a box of the fixed size (depending on
the color of the cluster) such that the estimate for the resolvent
on this box is  essentially the same as for the embedded cluster. We
also notice that the estimate for the number of points of
$\MM^{(2)}$ inside these boxes is the same as the worst possible
estimate for the corresponding cluster (see
Lemmas~\ref{L:black},~\ref{L:grey},~\ref{L:white}). This justifies
the application of Lemma~\ref{norm2/3} in the proof of
Lemma~\ref{L:2/3-1ind}.

By construction, white clusters are separated from each other by the
distance no less than $k^{\gamma r_1/6}$. Grey and black clusters
are separated by the distance at least $k^{\gamma
r_1/2+2\delta_0r_1}$ and $k^{\gamma r_1+\delta_0r_1}$, respectively.
Consider, first, a white cluster. Let $\Pi_w$ be a singular white
cluster, namely,
\begin{equation}\label{ap9-1}
\|(P_w(H-k^{2})P_w)^{-1}\|>k^\xi,\ \ \ \xi>k^{\gamma
r_1/6-2\delta_*},
\end{equation}
here and below $H=H(\k^{(2)}(\varphi_0))$, $P_w$ is the projector corresponding to $\Pi_w$. By construction,
$\Pi_w$ belongs to a small white box and its neighbors. Let us refer to
it as expanded small white box. Its size is $3k^{\gamma
r_1/2+2\delta_0r_1}$ and it contains less than $k^{\gamma
r_1/6-\delta_0r_1}$ elements of $\MM^{(2)}$.
\begin{lemma}\label{propw}
If \eqref{ap9-1} holds for a white cluster $\Pi_w$ then
\begin{equation}\label{ap9-2}
\|(P(H-k^{2})P)^{-1}\|>ck^{k^{\gamma r_1/6-2\delta_*}},
\end{equation}
$P$ being the projector corresponding to the expanded small white
box $\Pi $ containing $\Pi_w$. The box  $\Pi $ has the size $3k^{\gamma
r_1/2+2\delta_0r_1}$ and contains less than $k^{\gamma
r_1/6-\delta_0r_1}$ elements of $\MM^{(2)}$.
\end{lemma}
\begin{proof}
Assume \eqref{ap9-1} holds, but \eqref{ap9-2} does not. Let $f\in
P_w\ell^2$ be such that $\|f\|=1$, $P_w(H-k^{2})f=o(k^{-\xi})$, $\xi =k^{\gamma
r_1/6-2\delta_*}$.
Let us define
$$
g:=f-(P(H-k^{2})P)^{-1}(P-P_w)Vf.
$$
Now we have:
\begin{equation} \label{392}
\begin{split}&
P(H-k^{2})Pg=P(H-k^{2})f-(P-P_w)Vf=\cr &
 P_w(H-k^{2})f+(P-P_w)(H-k^{2})f -(P-P_w)Vf=\cr &
P_w(H-k^{2})P_w f+(P-P_w)(H_0-k^{2})P_wf=P_w(H-k^{2})P_w f=o(k^{-\xi}).
\end{split}
\end{equation}
If we show that
\begin{equation}\label{ap9-3}
\|P_w(P(H-k^{2})P)^{-1}(P-P_w)Vf\|=o(1),
\end{equation}
which means $\|g\|\geq1+o(1)$, then the lemma easily follows by the way
of contradiction. Thus, it remains to prove \eqref{ap9-3}. Denote
$\tilde f :=(P-P_w)Vf$. Let $\tilde{H}^{(2)}_w$ be the operator
consisting of $k^{\delta }$-clusters in $\Pi $. Namely,
$\tilde{H}^{(2)}_w=\sum _iP_iHP_i$,  $P_i$ being projectors onto
$k^{\delta }$-clusters.
 Formally,
\begin{equation}\label{393}
\begin{split}&
(P(H-k^{2})P)^{-1}\tilde f =
\sum\limits_{r=0}^{R_0}(\tilde{H}^{(2)}_w-k^{2})^{-1}\left(-(H-\tilde{H}^{(2)}_w)(\tilde{H}^{(2)}_w-k^{2})^{-1}\right)^r\tilde
f +\cr &
\left(P(H-k^{2})P\right)^{-1}\left(-(H-\tilde{H}^{(2)}_w)(\tilde{H}^{(2)}_w-k^{2})^{-1}\right)^{R_0+1}\tilde
f,\ \ \ R_0=[k^{\gamma r_1/6-\delta_*-\delta}]-1.
\end{split}
\end{equation}
Some of $k^{\delta }$-clusters $P_iHP_i$ are strongly resonant. However, their distance to the boundary of any white cluster is greater than $k^{\gamma r_1/6}$. Using this fact and considering
as in the proof of
\eqref{5-2}, we obtain
\begin{equation}\label{pochtivse}
\left\|\left(-(\tilde{H}^{(2)}_w-k^{2})^{-1}(H-\tilde{H}^{(2)}_w)(\tilde{H}^{(2)}_w-k^{2})^{-1}\right)^s\tilde
f \right\|<ck^{-s\delta_*/8}\ \ \  \mbox{when  }\ \ 2s\leq R_0+1, \end{equation}
since $(\tilde f)_{\m}\neq 0$ only near the boundary of a white cluster.
Hence, the right hand part of \eqref{393} is well
defined.
Now, substituting \eqref{pochtivse} into \eqref{393} and using the
estimate opposite to \eqref{ap9-2} we estimate the last term in \eqref{393}. Thus (cf. the proof of \eqref{5-}), we have
$$
P_w(P(H-k^{2})P)^{-1}\tilde f=P_w\left[(I+O(k^{-\delta_*/8}))(\tilde{H}^{(2)}_w-k^{2})^{-1}+O(k^{-\delta_*/8})\right](P-P_w)VP_w f.
$$
Using the identity $P_w(\tilde{H}^{(2)}_w-k^{2})^{-1}(P-P_w)=0$ we estimate
\begin{equation}\label{pochtivse1}
\|P_w(P(H-k^{2})P)^{-1}\tilde f\|\leq o(1)+O(k^{-\delta_*/8})\|(\tilde{H}^{(2)}_w-k^{2})^{-1}(P-P_w)VP_w\|.
\end{equation}
We notice that near the boundary of $\Pi_w$ all clusters satisfy the estimate \eqref{Pnr}. Thus we can apply the constructions from the proof of Theorem \ref{Thm2} (see in particular the proof of \eqref{||A||2} and especially the proof of \eqref{||A||2-4}, \eqref{||A||2-5} and \eqref{||A||2-7}). We get
$$
\|(\tilde{H}^{(2)}_w-k^{2})^{-1}(P-P_w)VP_w\|=O(k^{214\mu\delta}),
$$
which together with \eqref{pochtivse1} proves \eqref{ap9-3}.
\end{proof}

For singular grey and black clusters the proof is very similar. So,
we just introduce corresponding objects and formulate the results.

Let $\Pi_g$ be a singular grey cluster, i.e.
\begin{equation}\label{ap9-4}
\|(P_g(H-k^{2})P_g)^{-1}\|>k^\xi,\ \ \ \xi>k^{\gamma
r_1/2+2\delta_0r_1-2\delta_*},
\end{equation}
$P_g$ being the projector corresponding to $\Pi_g$. By construction,
$\Pi_g$ belongs to a big white box and its neighbors. We refer to it as
expanded big white box. Its size is $3k^{\gamma r_1}$ and it
contains less than $k^{\gamma r_1/2+\delta_0r_1}$ elements of
$\MM^{(2)}$.
\begin{lemma}\label{propg}
If \eqref{ap9-4} holds for a grey cluster $\Pi_g$, such that all the white clusters imbedded into it do not satisfy \eqref{ap9-2}, then
\begin{equation}\label{ap9-5}
\|(P(H-k^{2})P)^{-1}\|>ck^{k^{\gamma r_1/2+2\delta_0r_1-2\delta_*}},
\end{equation}
$P$ being the projector corresponding to the expanded big white box $\Pi $
containing $\Pi_g$. The box $\Pi $ has the size  $3k^{\gamma r_1}$ and it
contains less than $k^{\gamma r_1/2+\delta_0r_1}$ elements of
$\MM^{(2)}$.

\end{lemma}

The proof is analogous to that of Lemma \ref{propw} up to the
obvious changes: instead of $P_w$ we take $P_g$, $R_0=[k^{\frac 12
\gamma r_1+2\delta _0 r_1-\delta_*-\delta}]-1$ and $\tilde H^{(2)}_w$ is
replaced by $\tilde H^{(2)}_g$ which consists of $k^{\delta }$
non-resonance clusters and white clusters, which do not satisfy
\eqref{ap9-2}.

Let $\Pi_b$ be a singular black cluster, i.e.
\begin{equation}\label{ap9-6}
\|(P_b(H-k^{2})P_b)^{-1}\|>k^\xi,\ \ \ \xi>k^{\gamma
r_1+\delta_0r_1-2\delta_*},
\end{equation}
$P_b$ being the projector corresponding to $\Pi_b$. By
Lemma~\ref{L:black} any black cluster can be covered by a box of the
size $ck^{3\gamma r_1/2+3}$ containing less than $ck^{\gamma r_1
+3}$ elements of $\MM^{(2)}$. We refer to it as expanded black box.
\begin{lemma}\label{propb}
If \eqref{ap9-6} holds for a black cluster $\Pi_b$, such that all the white and grey clusters imbedded into it do not satisfy \eqref{ap9-2},  \eqref{ap9-5}, then
\begin{equation}\label{ap9-7}
\|(P(H-k^{2})P)^{-1}\|>ck^{k^{\gamma r_1 + \delta_0r_1-2\delta_*}},
\end{equation}
$P$ being the projector corresponding to the expanded black box
containing $\Pi_b$. The box  $\Pi $ has the size $ck^{3\gamma
r_1/2+3}$ and it contains less than $ck^{\gamma
r_1+3}$ elements of $\MM^{(2)}$.

\end{lemma}

\subsection{Appendix 10}

We consider $\k(\tau _1) =\b +\tau _1\a $. Let $b_\q=(\b+\p_\m, \vec \nu _\q)$, $a_\q=(\a, \vec \nu _\q)$, $t_\q=b_\q+a_\q\tau _1$ and 
$b_\q^{\bot }=(\b+\p_\m, \vec \nu _\q^{\bot })$, $a_\q^{\bot }=(\a, \vec \nu _\q^{\bot })$, $t_\q^{\bot }=b_\q^{\bot }+a_\q^{\bot }\tau _1$, $\m$ being the central point of $\MM_2^{j,s}$. Let us consider the corresponding periodic operator
$H_1^{per}(b_\q)$ associated with $\MM_2^{j,s}$. Let $\mu \leq 6C_*(V,\Lambda)$ (see \eqref{star111}). In our case $$\mu=k^2-(b_\q^{\bot })^2=k^2-(t_\q^{\bot })^2+(t_\q^{\bot })^2-(b_\q^{\bot })^2\leq\Lambda+5C_*\leq 6C_*$$ by assumption. Let $\lambda _{n_0}(b_\q)$, $\lambda _{n_0+1}(b_\q)$ be the eigenvalues of the operator being at the distance less than $\frac 18 d$ from $\mu $, where $d$
is the length of the shortest zone of the operator $H_1^{per}$. 
If we have just one such eigenvalue or none, the consideration is analogous, just simpler. It is easy to see that
$\lambda _{n_0}(t_\q)$, $\lambda _{n_0+1}(t_\q)$ are holomorphic functions of $\tau _1$ in a neighborhood of  zero, the size of the neighborhood depending only on $V,\Lambda $.
The projection $E(\tau _1)$, corresponding to the pair of eigenvalues,  can be extended as a  a holomorphic operator-function of $\tau _1$ in a similar neighborhood of zero and $\|E(\tau _1)\|<c(V, \Lambda)$ in this neighborhood. The resolvent $(H_1^{per}(t_\q)-\mu )^{-1}$ is a meromorphic function of $\tau _1$ for any fixed $\mu$ and $(H_1^{per}(t_\q)-\mu )^{-1}(I-E(\tau _1))$ is a holomorphic function. The size of a neighborhood of zero where this holds, depends only on $V, \Lambda $ when $\mu <6C_*$. Let 
$D_0(\tau _1)=(\lambda _{n_0}(t_\q)-\mu )(\lambda _{n_0+1}(t_\q)-\mu )$. Obviously, $(H_1^{per}(t_\q)-\mu )^{-1}E(\tau _1))D_0(\tau _1)$ is a holomorphic function bounded in norm by $C(V, \Lambda )$. Next, let $\mu _1(\tau _1)=(b_\q^{\bot })^2- (t_\q^{\bot })^2=-2a_\q^{\bot }b_\q^{\bot }\tau _1-(a_\q^{\bot })^2\tau _1^2.$ Obviously, $(H_1^{per}(t_\q)-\mu -\mu _1)^{-1}$ is a meromorphic function of $\tau _1$ . Perturbative arguments yield that $(H_1^{per}(t_\q)-\mu -\mu _1)^{-1}(I-E(\tau _1))$ is a holomorphic function of $\tau _1$ in a sufficiently small neighborhood of the origin, the size of the neighborhood depending on $V, \Lambda $ only. This part of the resolvent is bounded by $C(V, \Lambda )$.  Let
$D_1(\tau _1)=(\lambda _{n_0}(t_\q)-\mu -\mu _1)(\lambda _{n_0+1}(t_\q)-\mu -\mu _1)$. Noting that $(H_1^{per}(t_\q)-\mu -\mu _1)^{-1}E(\tau _1)D_1(\tau _1)= (H_1^{per}(t_\q)-\mu )^{-1}E(\tau _1)D_0(\tau _1)-\mu _1(\tau _1)E(\tau _1)$, we obtain that
$$\|(H_1^{per}(t_\q)-\mu -\mu _1)^{-1}E(\tau _1)D_1(\tau _1)\|<C(V,\Lambda )(1+|\mu _1|). $$
By assumption, $|\mu _1|<5C_*(V,\Lambda) $. Hence,
$$\|(H_1^{per}(t_\q)-\mu -\mu _1)^{-1}\|<\frac{C(V, \Lambda )}{D_1(\tau _1)}.$$
The product $D_1(\tau _1)$ consists of two multipliers. Let us consider one of them: $\lambda _{n_0}(t_\q)-\mu -\mu _1$. 
Using Taylor expansion of the third order for an eigenvalue, we obtain that  $\lambda _{n_0}(t_\q)-\mu -\mu _1=P_3(\tau _1)+O(\tau _1^4)$, where $P_3$ is a polynomial of order three and  $|O(\tau _1^4)|<C(V,\Lambda )\tau _1^4$. By \cite{Kor}, $|\lambda _{n_0}''(t_\q)|+|\lambda _{n_0}'''(t_\q)|>c(V, \Lambda )>0$.  Applying Rouch\'{e}'s Theorem, we obtain that
$\lambda _{n_0}(t_\q)-\mu -\mu _1$ has no more than three zeros in a $r(V,\Lambda )$- neighborhood of zero and $|\lambda _{n_0}(t_\q)-\mu -\mu _1|>c(V,\Lambda )\varepsilon^3$ at the distance $\varepsilon $ from the nearest zero, $\varepsilon <\varepsilon _0(V,\Lambda )$.  Hence $|D_1|>c(V,\Lambda )\varepsilon^6$ and
$$\|(H_1^{per}(t_\q)-\mu -\mu _1)^{-1}\|<C(V, \Lambda )\varepsilon^{-6}.$$

\section{List of the main notations}

Here, for the sake of convenience, we provide the list of main notations and definitions with the directions where they are introduced in the text.

\vskip 0.5cm

operator $H$ - formula \eqref{main0}

$C,c$ are constants depending only on $V$, $C_0,c_0$ are absolute constants.

$Q$ and irrational number $\alpha$ - formula \eqref{V}

measure of irrationality $\mu$ - formulae \eqref{geq}, \eqref{leq}

norm $\||\p_\s\||$ - just before Lemma~\ref{psnorms}

$\delta,\ \Omega(\delta),\ \tilde\Omega(\delta)$ - at the beginning of Section 3.1

$\OO_\m$ - formulae \eqref{resonance} and \eqref{resonance1}

$\MM,\,\tilde\MM,\,\MM_1,\,\tilde\MM_1,\,\MM_2,\,\tilde\MM_2,\,\MM_2^j,\,\tilde\MM_2^j,\,\MM'$ - formula \eqref{M} and text below

$\gamma$ - after the proof of Lemma~\ref{L:2/3-1}

simple, black, grey, white and non-resonant regions - Subsection~\ref{MOforStep3} (Step II) and Subsection~5.3.2 (Step III)

$\delta_0$ - see definition of the black regions in Subsection~\ref{MOforStep3}

$k_*$ - at the beginning of Section 5

$\beta$ and operator $\tilde H^{(2)}$ - at the beginning of Subsection 5.1

$r_n,\ r_n'$ - relations \eqref{indrn}, \eqref{r_2IV}, \eqref{r_2}, \eqref{r_2'}, Lemma~\ref{L:2/3-1ind},  relations \eqref{Aug13-1} and the beginning of Subsection~\ref{MOforStep3}

operator $\tilde H^{(n)}$ - after Lemma~\ref{May24-2}, at the beginning of Subsections~7.3, 6.1, 5.1 and formula \eqref{gulf1}

$\MM^{(n)}$ - formulae \eqref{M^n}, \eqref{M^3}, \eqref{M^2}

$\OO^{(n)}$ - formulae \eqref{Olast}, \eqref{O4}, \eqref{O3}, \eqref{O2}, \eqref{52a}

${\cal W}^{(n)}$ - formulae \eqref{Wlast}, \eqref{W4}, \eqref{W3}, \eqref{W2}, \eqref{W1}

$\omega^{(n)}$ - formulae \eqref{wlast}, \eqref{w4}, \eqref{w3}, \eqref{w2}, \eqref{omega}

$g^{(n)}_r$ - formulae \eqref{g3last}, \eqref{g3IV}, \eqref{g3}, \eqref{g2}, \eqref{g}

$G^{(n)}_r$ - formulae \eqref{G3last}, \eqref{G3IV}, \eqref{G3}, \eqref{G2}, \eqref{G}

contour $C_n$ -  formulae \eqref{G3last}, \eqref{G3IV}, Lemmas~\ref{estnonres0-1}, \ref{estnonres0} and formula \eqref{G-4}

$\lambda^{(n)}$ - formulae \eqref{eigenvalue-3last}, \eqref{eigenvalue-3IV}, \eqref{eigenvalue-3}, \eqref{eigenvalue-2}, \eqref{eigenvalue}

spectral projector $\E^{(n)}$ - formulae \eqref{sprojector-3last}, \eqref{sprojector-3IV}, \eqref{sprojector-3}, \eqref{sprojector-2}, \eqref{sprojector}

$d^{(n)}(\s,\s')$ - formulae \eqref{Feb6b-3last}, \eqref{Feb6b-3IV}, \eqref{Feb6b-3}, \eqref{Feb6b}, \eqref{matrix elements}

$\varkappa^{(n)}$ and $h^{(n)}$ - Lemmas~\ref{ldk-3IVlast}, \ref{ldk-3IV}, \ref{ldk-3}, \ref{ldk-2}, \ref{ldk}

$\SS^{(n)}$ - formulae \eqref{triind-last}, \eqref{triind}, \eqref{Aug25-1}

$\Omega_s^{(j)}(r_n)$ - formulae \eqref{Omega-j}, \eqref{Omega-s}, \eqref{se} and description of the simple region in Subsection~\ref{MOforStep3}


\begin{thebibliography}{100}

\bibitem[1]{KaSh}  Yu. Karpeshina, R. Shterenberg, {\em Multiscale analysis in momentum space for quasi-periodic potential in dimension two}, J. Math. Physics {\bf 54} (2013), 7. 
\bibitem[2]{DiSi} E. I. Dinaburg, Ya. Sinai,
{\em The One-dimensional Schr\"{o}dinger Equation with a
Quasiperiodic Potential}, Funct. Anal. Appl. {\bf 9} (1975),
279--289.
\bibitem[3]{R}  H. R\^{u}ssmann,
{\em On the one dimensional Schr\"{o}dinger equation with a
quasi-periodic potential}, Ann. N. Y. Acad. Sci. {\bf 357} (1980),
90--107.
\bibitem[4]{3M} R. Johnson, J. Moser, {\em The rotation number for almost periodic
potentials}, Commun. Math. Phys. {\bf 84} (1982), 403--438.
\bibitem[5]{MP}
J. Moser, J. P\"{u}schel, {\em An extension of a result by Dinaburg
and Sinai on quasi-periodic potentials}, Comment. Math. Helvetic
{\bf 59} (1984), 39--85.
\bibitem[6]{CFKS} H. L. Cycon,  R. G. Froese, W. Kirsch,  B. Simon, {\em   Schr\"{o}dinger Operators}, Berlin: Springer Verlag, 1987, corrected and extended 2nd printing, Springer Verlag, 2008.
\bibitem[7]{PF} L. Pastur, A. Figotin, {\em Spectra of Random and Almost-Periodic
Operators}, Springer-Verlag, 1992, 583 pp.
\bibitem[8]{E} L. H. Eliasson, {\em Floquet Solutions for the One-dimensional Quasi-periodic Schr\"{o}dinger Equation}, Comm. Math. Phys. {\bf 146 } (1992), no. 3, 447--482.
\bibitem[9]{J1} S. Jitomirskaya, {\em Metal-Insulator Transition for the Almost Mathieu
Operator}, Ann. of Math. {\bf 150} (1999), 1159--1175.
\bibitem[10]{S2} B. Simon, {\em Schr\"{o}dinger Operators in the Twentieth
Century}, J. Math. Phys. {\bf 41} (2000), no. 6, 3523--3355.
\bibitem[11]{FK-1}  A. Fedotov, F. Klopp,
{\em On the Singular Spectrum for Adiabatic Quasi-periodic
Schr\"odinger Operators on the real line}, Ann. Henri Poincar\'e
{\bf 5} (2004), no. 5, 929--978.
\bibitem[12]{FK-2} A. Fedotov, F. Klopp,
{\em On the Absolutely Continuous Spectrum of One-dimensional
Quasi-periodic Schr\"odinger operators in the Adiabatic Limit},
Trans. Amer. Math. Soc. {\bf 357} (2005), no. 11, 4481--4516.
\bibitem[13]{FK}  A. Fedotov, F. Klopp, {\em Strong Resonant Tunneling, Level Repulsion and Spectral Type for
One-dimensional Adiabatic Quasi-periodic Schr\"odinger Operators},
Ann. Sci. \'Ecole Norm. Sup. {\bf (4) 38} (2005), no. 6, 889--950.
\bibitem[14]{Avila1} A. Avila, {\em Global theory of one-frequency Schrodinger operators I: stratified analyticity of the Lyapunov exponent and the boundary of nonuniform hyperbolicity}, arXiv:0905.3902.
\bibitem[15]{Avila2} A. Avila, {\em Global theory of one-frequency Schrodinger operators II: acriticality and finiteness of phase transitions for typical potentials}, http://w3.impa.br/~avila/global2.pdf
\bibitem[16]{Sh1} M. A. Shubin, {\em Density of States for Selfadjoint Elliptic Operators
with Almost Periodic Coefficients}, Trudy sem. Petrovskii (Moscow
University) {\bf 3} (1978), 243--281.
\bibitem[17]{BLS} J. Bellissard, R. Lima, and E. Scoppola,  {\em Localization in $n$-dimensional
incommensurable structures}, Commun. Math. Phys. {\bf 88}, (1983)
465--477.
\bibitem[18]{FP} A. L. Figotin, L. A. Pastur, {\em An Exactly Solvable Model of a Multidimensional
Incommensurate Structure}, Commun. Math. Physics {\bf 95} (1984),
401--425.
\bibitem[19]{12} J. Bellissard, {\em Almost periodicity in solid state Physics and
C*-algebras}, Mat.-Fys. Medd danske Vid. Selsk. {\bf 42} (1989), no. 3,
35--75.
\bibitem[20]{ChD} V. Chulaevsky, E. I. Dinaburg, {\em Methods of KAM theory for Long-Range
Quasiperiodic Potentials on $\Z^{\vec \nu }$. Pure Point Spectrum},
Commun. Math. Physics {\bf 153}
 (1993), no. 3, 559--577.
\bibitem[21]{74}  J. Bourgain, M. Goldstein, W. Schlag, {\em  Anderson
 Localization on $\Z^2$
with Quasi-Periodic Potential}, Acta Math. {\bf 188} (2002), 41--87.
\bibitem[22]{73} J. Bourgain, {\em On Quasi-Periodic Lattice
 Schr\"{o}dinger Operators}, Discrete and Continuous Dynamical
 Systems {\bf 10} (2004), no. 1\&2, 75--88.
\bibitem[23]{PS}L.~Parnovski, R.~Shterenberg, {\em Complete Asymptotic Expansion of
the Integrated Density of States of Multidimensional Almost-periodic
Schr\"odinger Operators}, Ann. of Math. {\bf 176} (2012), no. 2, 1039--1096.
\bibitem[24]{68a}
S.~Morozov, L.~Parnovski, R.~Shterenberg, {\em Complete Asymptotic Expansion of
the Integrated Density of States of Multidimensional Almost-periodic
Pseudo-differential Operators}, Ann. H. Poincar\'e {\bf 15} (2014), no. 2, 263--312.
\bibitem[25]{Bou1} J. Bourgain, {\em Anderson localization for quasi-periodic lattice Schr\"odinger operators on $\Z^d$, $d$ arbitrary},
Geom. Funct. Anal. {\bf 17} (2007), no. 3, 682--706.
\bibitem[26]{Bou2} J. Bourgain, {\em Green's function estimates for lattice Schr\"odinger operators and applications}, Ann. of
Math. Studies, {\bf 158}, Princeton University Press, 2005.
\bibitem[27]{14r}  V. N. Popov, M. M. Skriganov,  {\em Remark on the
Structure of the Spectrum of a Two-Dimensional Schr\"{o}dinger
Operator with Periodic Potential}, Zap. Nauchn. Sem. Leningrad.
Otdel. Mat. Inst. Steklov. (LOMI) {\bf 109} (1981), 131--133;
English transl.: J. Soviet Math. {\bf 24} (1984), no. 2, 239--240.
\bibitem[28]{PS} L.Parnovski, A.Sobolev, {\em Bethe-Sommerfeld conjecture for periodic operators with strong perturbations}, Invent. Math., 181(3) (2010), 467--540.
\bibitem[29]{SS}  M. M. Skriganov, A. V. Sobolev,  {\em On the Spectrum of a
Limit-Periodic Schr\"{o}dinger Operator}, Algebra i Analiz, {\bf 17}
(2005), no. 5; Engl. Transl.: St. Petersburg Math. J. {\bf 17}
(2006), 815--833.
\bibitem[30]{KL1} Yu. Karpeshina, Y.-R. Lee,
{\em Spectral properties of polyharmonic operators with
limit-periodic potential in dimension two}, D'Analyse
Math{\'e}matique {\bf 102} (2007), 225--310.
\bibitem[31]{KL2} Yu. Karpeshina, Y.-R. Lee,
{\em Absolutely Continuous Spectrum of a Polyharmonic Operator with
a Limit Periodic Potential in Dimension Two}, Communications in
Partial Differential Equations {\bf 33} (2008), no. 9, 1711--1728.
 \bibitem[32]{KL3}
Yu. Karpeshina, Y.-R. Lee, {\em Spectral properties of a limit-periodic Schr\"odinger operator in dimension two}, J. Anal. Math.  {\bf 120}  (2013), 1--84.
 \bibitem[33]{2} G. Gallavotti,
{\em Perturbation Theory for Classical Hamiltonian Systems,} in {\it
Scaling and Self-Similarity in Physics} edited by J. Froehlich,
Birkh\"auser, Basel, Switzerland, 1983, 359--424.
\bibitem[34]{22} L.E. Thomas, S.R. Wassel,
{\em Stability of Hamiltonian systems at high evergy}, J. Math.
Phys. {\bf 33(10)}, (1992), 3367--3373.
 \bibitem[35]{3} L. E. Thomas
and S. R. Wassel, {\em Semiclassical Operators at High Energy},
Lecture Notes in Physics {\bf 403}, edited by E. Balslev,
Springer-Verlag, 1992,  194--223.
\bibitem[36]{FrSp} J. Fr\"{o}lich, T. Spencer, {\em Absence of Diffusion in the Anderson Tight
Binding Model for Large Disorder and Low Energy}, Commun. Math.
Physics {\bf 88} (1983), no. 2, 151--184.
\bibitem[37]{BG} J. Bourgain, M. Goldstein, {\em  On Nonperturbative Localization
with Quasi-Periodic Potential}, Ann. of Math. {\bf 152} (2000),
no. 3, 835--879.
\bibitem[38]{B3} J. Bourgain, {\em Quasiperiodic Solutions of Hamilton Perturbations of $2D$
Linear Schr\"{o}dinger Equation}, Ann. of Math. {\bf 148}
(1998), no. 2, 363--439.
\bibitem[39]{K} T.~Kato, {\em Perturbation theory for linear operators}, Springer-Verlag, Berlin, 1995.
\bibitem[40]{RS} M. Reed, B. Simon, {\em Methods of Modern Mathematical Physics}, Vol IV, Academic Press, 3rd ed., New York-San Francisco-London (1987), 396 pp.
\bibitem[41]{AvJit} A. Avila, S. Jitomirskaya, {\em Almost Reducibility and Almost Localization}, JEMS {\bf 12} (2010), no. 1, 93--131.
\bibitem[42]{BouJit} J. Bourgain, S. Jitomirskaya, {\em Absolutely Continuous Spectrum for 1D Quasiperiodic Operators}, Invent. Math. {\bf 148} (2002), no. 3, 453--463.
\bibitem[43]{Kor} E. Korotyaev, {\em Some properties of the quasimomentum of the one-dimensional Hill operator}, (Russian)  Zap. Nauchn. Sem. Leningrad. Otdel. Mat. Inst. Steklov. (LOMI) {\bf 195}  (1991),  Mat. Vopr. Teor. Rasprostr. Voln. 21, 48--57, 180;  translation in  J. Soviet Math.  {\bf 62}  (1992),  no. 6, 3081--3087.
\bibitem[44]{Pos} J. P\"oschel, {\em Examples of discrete Schr\"odinger operators with pure point spectrum}, Comm. Math. Phys. 88 (1983), 447--463.
\bibitem[45]{DamG} D. Damanik, Z. Gan, {\em Limit-periodic Schr\"odinger operators on $\Z^d$:
Uniform localization}, J. Func. Anal. {\bf 265} (2013), 435--448.


\end{thebibliography}
\end{document}